\documentclass{mscs}

\title[On Recursive Operations Over LLTS]
  {On Recursive Operations Over Logic LTS \thanks{
This work received financial support of the National Natural Science of
China (No. 60973045), Fok Ying-Tung Education Foundation and NSF of the Jiangsu Higher Education Institutions (No. 13KJB520012).}}
\author[Yan Zhang, Zhaohui Zhu and Jinjin Zhang]
  {Y\ls A\ls N\ns Z\ls H\ls A\ls N\ls G$^1$\ns
   Z\ls H\ls A\ls O\ls H\ls U\ls I\ns Z\ls H\ls U\ls$^1$ \thanks{
Corresponding author. E-mail address: zhaohui@nuaa.edu.cn,
zhaohui.nuaa@gmail.com \quad (Zhaohui Zhu)} \ns J\ls I\ls N\ls J\ls I\ls N\ns Z\ls H\ls A\ls N\ls G$^2$ \\
$^1$ College of Computer Science, Nanjing University of Aeronautics and Astronautics,\\
Nanjing P. R. China, 210016\\
$^2$ School of Information Science, Nanjing Audit University, Nanjing P. R. China, 211815}

\pubyear{2001}
\pagerange{\pageref{firstpage}--\pageref{lastpage}}
\volume{00}
\doi{S0960129501003425}

\newtheorem{theorem}{Theorem}[section]
\newtheorem{lemma}{Lemma}[section]
\newtheorem{proposition}{Proposition}[section]
\newtheorem{corollary}{Corollary}[section]

\newtheorem{mydefn}{Definition}[section]

\newtheorem{convention}{Convention}[section]
\newtheorem{rmk}{Remark}[section]

\usepackage{amsmath}

\usepackage{amssymb}

\begin{document}

\label{firstpage}
\maketitle

\begin{abstract}
Recently, in order to mix algebraic and logic styles of specification in a uniform framework, the notion of a logic labelled transition system (Logic LTS or LLTS for short) has been introduced and explored. A variety of constructors over LLTS, including usual process-algebraic operators, logic connectives ($conjunction$ and $disjunction$) and standard temporal operators ($always$ and $unless$), have  been given. However, no attempt has made so far to develop general theory concerning (nested) recursive operations over LLTS and a few fundamental problems are still open. This paper intends to study this issue in pure process-algebraic style.
A few fundamental properties, including precongruence and the uniqueness of consistent solutions of equations, will be established.
\end{abstract}


\ifprodtf \newpage \else \vspace*{-1\baselineskip}\fi

\section{Introduction}

Algebra and logic are two dominant approaches for the specification,
verification and systematic development of reactive and concurrent systems.
They take different standpoints for looking at specifications and
verifications, and offer complementary advantages.

Logical approaches devote themselves to specifying and verifying abstract
properties of systems. In such frameworks, the most common reasonable
property of concurrent systems, such as safety, liveness, etc., can be
formulated in terms of logic formulas without resorting to operational
details and verification is a deductive or model-checking activity. However,
due to their global perspective and abstract nature, logical approaches often
give little support for modular designing and compositional reasoning.

Algebraic approaches put attention to behavioral aspects of systems, which have tended
to use formalisms in algebraic style. These formalisms are referred to as
process algebra or process calculus. In such a paradigm, a specification and
its implementation usually are formulated by terms (expressions) of a formal
language built from a number of operators, and the underlying semantics is
often assigned operationally.
The verification amounts to comparing terms, which is often referred to as \emph{implementation verification} or \emph{equivalence checking}.
Algebraic approaches often support compositional construction
and reasoning, which bring us advantages in developing systems, such as,
supporting modular design and verification, avoiding verifying the whole
system from scratch when its parts are modified, allowing reusability of
proofs and so on \cite{Andersen94}.
Thus such approaches offer significant
support for rigorous systematic development of reactive and concurrent
systems. However, since algebraic approaches specify a system by means of
prescribing in detail how the system should behave, it is often difficult
for them to describe abstract properties of systems, which is a major
disadvantage of such approaches.

In order to take advantage of these two approaches when designing systems,
so-called heterogeneous specifications have been proposed, which uniformly
integrate these two specification styles. Amongst, based on B\"uchi automata
and labelled transition system (LTS) augmented with a predicate, Cleaveland
and L\"uttgen provide a semantic framework for heterogenous system design (Cleaveland
and L\"uttgen 2000, 2002).
In this framework, not only usual operational operators but also logic
connectives are considered, and must-testing preorder presented in \cite{Nicola83} is
adopted to capture refinement relation.
Unfortunately, this setting does not
support compositional reasoning since must-testing preorder is not a
precongruence in this situation.
Moreover, the logic connective conjunction
in this framework lacks the desired property that $r$ is an implementation
of a given specification $p\wedge q$ if and only if $r$ implements both $p$
and $q$.

Recently, L\"uttgen and Vogler have introduced the notion of a Logic LTS (LLTS), which combines operational and logic styles of specification in
one unified framework (L\"uttgen and Vogler 2007, 2010, 2011). 
In addition to usual operational constructors,
e.g., CSP-style parallel composition, hiding and so on, logic connectives
(conjunction and disjunction) and standard modal operators ($always$ and $%
unless$) are also integrated into this framework.
Moreover, the drawbacks in (Cleaveland
and L\"uttgen 2000, 2002)
mentioned above have been remedied by adopting
ready-tree semantics \cite{Luttgen07}.
In order to support compositional reasoning in the
presence of the parallel constructor, a variant of the usual notion of ready
simulation is employed to capture the refinement relation, which has been
shown to be the largest precongruence satisfying some desired properties \cite{Luttgen10}.

Along the direction suggested by L\"uttgen and Vogler in \cite{Luttgen10}, a process
calculus called CLL is presented in \cite{Zhang11}, which reconstructs their setting in
pure process-algebraic style. Moreover, a sound and ground-complete
proof system for CLL is provided. In effect, it gives an axiomatization
of ready simulation in the presence of logic operators.
However, CLL is lack of capability of describing infinite behaviour, that is important for representing reactive systems.


It is well known that recursive operations are fundamental mechanisms for
representing objects with infinite behavior in terms of finite expressions
(see, for instance \cite{Bergstra01}).
We extend CLL with recursive operations and propose a process calculus named CLL$_R$.
Since LLTS involves consideration of inconsistencies, it is far from straightforward to re-establish existent results concerning recursive operations in this framework.
A solid effort is required, especially for handling inconsistencies.
This paper intends to explore recursive operations over
LLTS in pure process-algebraic style.
A behavioral theory of CLL$_R$ will be established, especially, we prove that the behavioral relation (i.e., ready simulation mentioned above) is precongruent w.r.t all operators in CLL$_R$.
Moreover, the uniqueness of solution of equations  will be obtained, provided conditions that, given an equation $X=t_X$, $X$ is strongly guarded and does not occur in the scope of any conjunction in $t_X$.


The remainder of this paper is organized as follows. The next section
recalls some related notions. Section~3 introduces SOS rules of $\text{CLL}_R
$. In section~4, the existence and uniqueness of stable transition model for
$\text{CLL}_R$ is demonstrated, and a few of basic properties of the LTS
associated with $\text{CLL}_R$ are given. More further properties are
considered in Section~5. In section~6, we shall show that the variant of
ready simulation presented by L\"uttgen and Vogler is precongruent in the
presence of (nested) recursive operations.
In section~7, a theorem on the uniqueness of solution of equations is obtained. Finally, a brief
conclusion and discussion are given in Section~8.

\section{Preliminaries}

\subsection{Logic LTS and ready simulation}

This subsection will set up notations and briefly recall the notions of Logic LTS and ready simulation presented by L{\" u}ttgen and Vogler. For motivation behind these notions we refer the reader to (L\"uttgen and Vogler 2007, 2010, 2011).

Let $Act$ be the set of visible actions ranged over by letters $a$, $b$, etc., and let $Act_{\tau}$ denote $Act \cup \{\tau\}$ ranged over by $\alpha$ and $\beta$, where $\tau$ represents invisible actions.
A labelled transition system (LTS) with a predicate is a quadruple $(P,Act_{\tau},\longrightarrow,F)$, where $P$ is a set of processes (states), $\longrightarrow \subseteq P\times Act_{\tau}\times P$ is the transition relation and $F\subseteq P$.

As usual, we write $p \stackrel{\alpha}{\longrightarrow} q$ if $(p,\alpha,q)\in \longrightarrow$.
$q$ is an $\alpha$-derivative of $p$ if $p \stackrel{\alpha}{\longrightarrow}q$.
We write $p \stackrel{\alpha}{\longrightarrow}$ (or, $p \not \stackrel{\alpha}{\longrightarrow}$) if $\exists q\in P.p\stackrel{\alpha}{\longrightarrow}q$ ($\nexists q\in P.p  \stackrel{\alpha}{\longrightarrow}q$ respectively).
Given a process $p$, the ready set $\{\alpha \in Act_{\tau}|p \stackrel{\alpha}{\longrightarrow}\}$ of $p$ is denoted by $\mathcal{I}(p)$.
A state $p$ is stable if it cannot engage in any $\tau$-transition, i.e., $p \not\stackrel{\tau}{\longrightarrow}$.
The list below contains some useful decorated transition relations:

$p \stackrel{\alpha}{\longrightarrow}_F q$ iff $p \stackrel{\alpha}{\longrightarrow} q$ and $p,q\notin F$.

$p \stackrel{\epsilon}{\Longrightarrow}q$ iff $p (\stackrel{\tau}{\longrightarrow})^* q$, where $(\stackrel{\tau}{\longrightarrow})^* $ is the transitive reflexive closure of $\stackrel{\tau}{\longrightarrow}$.

$p \stackrel{\alpha}{\Longrightarrow}q$ iff $\exists r,s\in P.p \stackrel{\epsilon}{\Longrightarrow} r \stackrel{\alpha}{\longrightarrow}s \stackrel{\epsilon}{\Longrightarrow} q$.

$p \stackrel{\gamma}{\Longrightarrow}|q$ iff $p \stackrel{\gamma}{\Longrightarrow}q \not\stackrel{\tau}{\longrightarrow}$ with $\gamma \in Act_{\tau}\cup \{\epsilon\}$.

$p\stackrel{\epsilon }{\Longrightarrow }_Fq$ iff there exists a sequence of $\tau -$labelled transitions from $p$ to $q$ such that all states along this sequence, including $p$ and $q$, are not in $F$. The decorated transition $p \stackrel{\alpha }{\Longrightarrow }_Fq$ may be defined similarly.

$p \stackrel{\epsilon}{\Longrightarrow}_F|q$ (or, $p \stackrel{\alpha}{\Longrightarrow}_F|q$) iff $p \stackrel{\epsilon}{\Longrightarrow}_F q$ ($p \stackrel{\alpha}{\Longrightarrow}_F q$ respectively) and $q$ is stable.

\begin{rmk}\label{R:TRANSITION_CONSISTENCY}
  Notice that some notations above are slightly different
from ones adopted by L\"uttgen and Vogler.
In
(L\"uttgen and Vogler 2010, 2011)
the notation $p%
\stackrel{\epsilon }{\Longrightarrow }\mspace{-8mu}|q$ (or, $p\stackrel{\alpha }{%
\Longrightarrow }\mspace{-8mu}|q$) has the same meaning as $p\stackrel{\epsilon }{%
\Longrightarrow }_F|q$ ($p\stackrel{\alpha }{\Longrightarrow }_F|q$ respectively)
in this paper, while $p\stackrel{\epsilon }{\Longrightarrow }|q $ in this paper does not involve any requirement on consistency.
\end{rmk}

\removebrackets
\begin{mydefn}[\cite{Luttgen10}]\label{D:LLTS}
 An LTS $(P,Act_{\tau},\longrightarrow,F)$ is an LLTS if, for each $p \in P$,

\noindent\textbf{(LTS1) }$p \in F$ if $\exists\alpha\in \mathcal{I}(p)\forall q\in P(p \stackrel{\alpha}{\longrightarrow}q \;\text{implies}\; q\in F)$;

\noindent\textbf{(LTS2)} $p \in F$ if $\nexists q\in P.p \stackrel{\epsilon}{\Longrightarrow}_F|q$.
\end{mydefn}

Here the predicate $F$ is used to denote the set of all inconsistent states.
In the sequel, we shall use the phrase ``{\it %
inconsistency} {\it predicate}'' to refer to $F$.
The condition (LTS1)
formalizes the backward propagation of inconsistencies, and (LTS2) captures
the intuition that divergence (i.e., infinite sequences of $\tau $%
-transitions) should be viewed as catastrophic.

Compared with usual LTSs, it is one distinguishing feature of LLTS that it
involves consideration of inconsistencies.
Roughly speaking, the motivation
behind such consideration lies in dealing with inconsistencies caused by
conjunctive composition.
For example, consider a simple composition $a.0 \wedge b.0$, it cannot be interpreted as deadlock 0 because a run of a process cannot begin with both actions $a$ and $b$, that is $a.0$ and $b.0$ specify processes with different ready sets.
It is proper to tag it as an inconsistent specification.
Moreover, inconsistencies could propagate backward.
$a.b.0 \wedge a.c.0$ specifies the absence of any alternative $a$-transition leading to a consistent state.
It should  also be tagged as an inconsistent one.
For more intuitive ideas and motivation about inconsistency, the reader may
refer to (L\"uttgen and Vogler 2007, 2010).

\removebrackets
\begin{mydefn}[\cite{Luttgen10}]
An LTS $(P,Act_{\tau},\longrightarrow,F)$ is {$\tau$}-pure if, for each $p \in P$, $p\stackrel{\tau}{\longrightarrow}$ implies $\nexists a\in Act.\;p\stackrel{a}{\longrightarrow}$.
\end{mydefn}

Hence, for any state $p$ in a $\tau$-pure LTS, either ${\mathcal I}(p)=\{\tau\}$ or ${\mathcal I}(p)\subseteq Act$ and intuitively, it represents either an external or internal (disjunctive) choice between its outgoing transitions.
Following \cite{Luttgen10}, this paper will focus on $\tau$-pure LLTSs.

In (L\"uttgen and Vogler 2010, 2011), the notion of ready simulation below is adopted to capture the refinement relation, which is a variant of the usual notion of weak ready simulation.
Such kind of ready simulation cares only stable consistent states.

\begin{mydefn}[Ready simulation on LLTS]\label{D:RS}
Let $(P, Act_{\tau}, \longrightarrow , F)$ be a  LLTS.
A relation ${\mathcal R} \subseteq P \times P$ is a stable ready simulation relation, if for any $(p,q) \in {\mathcal R}$ and $a \in Act $\\
\textbf{(RS1)} both $p$ and $q$ are stable;\\
\textbf{(RS2)} $p \notin F$ implies $q \notin F$;\\
\textbf{(RS3)} $p \stackrel{a}{\Longrightarrow}_F|p'$ implies $\exists q'.q \stackrel{a}{\Longrightarrow}_F|q'\; \textrm{and}\;(p',q') \in {\mathcal R}$;\\
\textbf{(RS4)} $p\notin F$ implies ${\mathcal I}(p)={\mathcal I}(q)$.

\noindent We say that $p$ is stable ready simulated by $q$, in symbols $p \underset{\thicksim}{\sqsubset}_{RS} q$, if there exists a stable ready simulation relation $\mathcal R$ with $(p,q) \in {\mathcal R}$.
 Further, $p$ is ready simulated by $q$, written $p\sqsubseteq_{RS}q$, if
 $\forall p'(p\stackrel{\epsilon}{\Longrightarrow}_F| p' \;\text{implies}\; \exists q'(q \stackrel{\epsilon}{\Longrightarrow}_F| q'\; \text{and}\;p' \underset{\thicksim}{\sqsubset}_{RS} q'))$.
 The kernels of $\underset{\thicksim}{\sqsubset}_{RS}$ and $\sqsubseteq_{RS}$ are denoted by $\approx_{RS}$ and $=_{RS}$ respectively.
 It is easy to see that $\underset{\thicksim}{\sqsubset}_{RS}$ is a stable ready simulation relation and both $\underset{\thicksim}{\sqsubset}_{RS}$ and $\sqsubseteq_{RS}$ are pre-order (i.e., reflexive and transitive).
\end{mydefn}


\subsection{Transition system specification}
Structural Operational Semantics (SOS) is proposed by G.~Plotkin in \cite{Plotkin81}, which adopts a syntax oriented view on operational semantics, and gives operational semantics in logical style.
Transition System Specifications (TSSs), as presented by Groote and Vaandrager in \cite{Groote92}, are formalizations of SOS.
This subsection recalls basic concepts related to TSS.
For further information on this issue we refer the reader to \cite{Aceto01,Bol96,Groote92}.

 Let $V_{AR}$ be an infinite set of variables and $\Sigma$  a signature. The set of $\Sigma $-terms over $V_{AR}$, denoted by $T(\Sigma ,V_{AR})$, is the least set such that (I) $V_{AR} \subseteq T(\Sigma ,V_{AR})$ and (II) if $f\in \Sigma$ and $t_1,\dots,t_n \in T(\Sigma ,V_{AR})$, then $f(t_1,\dots,t_n) \in T(\Sigma,V_{AR})$, where $n$ is the arity of $f$.
 $T(\Sigma ,\emptyset )$ is abbreviated by $T(\Sigma )$, elements in $T(\Sigma )$ are called closed or ground terms.

    A substitution $\sigma $ is a mapping from $V_{AR}$ to $T(\Sigma ,V_{AR})$.
    As usual, a substitution $\sigma $ may be lifted to a mapping $T(\Sigma ,V_{AR}) \rightarrow T(\Sigma ,V_{AR})$ by $\sigma (f(t_1,...,t_n))\triangleq f(\sigma(t_1),\dots ,\sigma (t_n))$ for any n-arity $f\in \Sigma$ and $t_1,\dots ,t_n \in T(\Sigma ,V_{AR})$. A substitution is closed if it maps all variables to ground terms.

    A TSS is a quadruple $\mathcal{P}=(\Sigma,\mathbb{A},\mathbb{P},\mathbb{R})$, where $\Sigma$ is a signature, $\mathbb A$ is a set of labels, $\mathbb P$ is a set of predicate symbols and $\mathbb R$ is a set of rules.
    Positive literals are all expressions of the form $t \stackrel{a}{\longrightarrow} t'$ or $tP$, while negative literals are all expressions of the form $t \not\stackrel{a}{\longrightarrow}$ or $t\neg P$, where $t,t'\in  T(\Sigma ,V_{AR})$, $a\in \mathbb{A}$ and $P\in {\mathbb P}$.
    A literal is a positive or negative literal, and $\varphi$, $\psi$, $\chi$ are used to range over literals.
    A literal is ground or closed if all terms occurring in it are ground.
    A rule $r\in \mathbb R$ has the form like $\frac{H}{C}$, where $H$, the premises of the rule $r$, denoted  $prem(r)$, is a set of literals, and $C$, the conclusion of the rule $r$, denoted $conc(r)$, is a positive literal.
    Furthermore, we write $pprem(r)$ for the set of positive premises of $r$ and $nprem(r)$ for the set of negative premises of $r$.
    A rule $r$ is positive if $nprem(r)= \emptyset$.
    A TSS is  positive if it has only positive rules.
    Given a substitution $\sigma$ and a rule $r \in \mathbb{R}$, $\sigma(r)$ is the rule obtained from $r$ by replacing each variable in $r$ by its $\sigma$-image, that is, $\sigma(r) \triangleq \frac{\{\sigma(\varphi)|\varphi \in prem(r)\}}{\sigma(conc(r))}$.
    Moreover, if $\sigma$ is closed then $\sigma(r)$ is a ground instance of $r$.

\begin{mydefn}[Proof in positive TSS]\label{D:PROOF}
    Let $\mathcal{P}=(\Sigma,\mathbb{A},\mathbb{P},\mathbb{R})$ be a positive TSS. A proof of a closed positive literal $\psi$ from $\mathcal{P}$ is a well-founded, upwardly branching tree, whose nodes are labelled with closed positive literals, such that\\
    --- the root is labelled with $\psi$,\\
    --- if $\chi$ is the label of a node $q$ and $\{\chi_i:i\in I\}$ is the set of labels of the nodes directly above $q$, then there is a rule $\{\varphi_i:i \in I\} \slash \varphi$ in $\mathbb{R}$ and a closed substitution $\sigma$ such that $\chi=\sigma(\varphi)$ and $\chi_i=\sigma(\varphi_i)$ for each $i \in I$.

    If a proof of $\psi$ from $\mathcal{P}$ exists, then $\psi$ is provable from $\mathcal{P}$, in symbols $\mathcal{P}\vdash \psi$.
\end{mydefn}

A natural and simple method of describing the operational nature of closed terms is in terms of LTSs.
Given a TSS, an important problem is how to associate LTS with any given closed terms.
For positive TSS, the answer is straightforward.
However, this problem is far from trivial for TSS containing negative premises.
The notions of stable model and stratification of TSS play an important role in dealing with this issue.
The remainder of this subsection intends to recall these notions briefly.


Given a TSS $\mathcal{P}=(\Sigma,\mathbb{A},\mathbb{P},\mathbb{R})$, a transition model $M$ is a subset of $Tr(\Sigma ,\mathbb{A})\cup Pred(\Sigma ,\mathbb{P} )$, where $Tr(\Sigma ,\mathbb{A})=T(\Sigma)\times \mathbb{A}\times T(\Sigma )$ and $Pred(\Sigma ,\mathbb{P} )=T(\Sigma )\times \mathbb{P}$, elements $(t,a,t^{\prime })$ and $(t,P)$ in $M$ are written as $t \stackrel{a}{\longrightarrow} t'$ and $tP$ respectively.
     A positive closed literal $\psi$ holds in $M$ or $\psi$ is valid in $M$, in symbols $M\models \psi$, if $\psi \in M$. A negative closed literal  $t \not \stackrel{a}{\longrightarrow}$ (or, $t \neg P$) holds in $M$, in symbols $M \models t \not \stackrel{a}{\longrightarrow}$ ($M \models t \neg P$ respectively), if there is no $t'$ such that $t \stackrel{a}{\longrightarrow} t' \in M$($tP \notin M$ respectively). For a set of closed literals $\Psi$, we write $M \models \Psi$ iff $M \models \psi$ for each $\psi \in \Psi$.
     $M$ is a model of $\mathcal{P}$ if, for each $r\in \mathbb{R}$ and $\sigma:V_{AR}\longrightarrow T(\Sigma)$, we have $M \models conc(\sigma(r))$ whenever $M \models prem(\sigma(r))$.
    $M$ is supported by $\mathcal{P}$ if, for each $\psi \in M$, there exists $r \in \mathbb{R}$ and $\sigma:V_{AR} \longrightarrow T(\Sigma)$ such that $M \models prem(\sigma(r))$ and $conc(\sigma(r))=\psi$.
    $M$ is a supported model of $\mathcal{P}$ if $M$ is supported by $\mathcal{P}$ and $M$ is a model of $\mathcal{P}$.

\removebrackets
\begin{mydefn}[\cite{Aceto01,Bol96}]\label{D:STRATIFICATION}
    Let $\mathcal{P}=(\Sigma,\mathbb{A},\mathbb{P},\mathbb{R})$ be a TSS and $\alpha$ an ordinal number. A function $S:Tr(\Sigma,\mathbb{A})\cup Pred(\Sigma, \mathbb{P}) \longrightarrow \alpha$ is a stratification of $\mathcal{P}$ if, for every rule $r \in \mathbb{R}$ and every substitution $\sigma :V_{AR} \longrightarrow T(\Sigma)$, the following conditions hold.\\
    (1) $S(\psi)\leq S(conc( \sigma(r)))$ for each $\psi \in pprem(\sigma (r))$,\\
    (2) $S(tP)<S(conc(\sigma(r)))$ for each $t \neg P \in nprem(\sigma(r))$, and\\
    (3) $S(t \stackrel{a}{\longrightarrow} t') < S(conc(\sigma(r)))$ for each $t' \in T(\Sigma)$ and $t \not\stackrel{a}{\longrightarrow} \in nprem(\sigma(r))$.

    A TSS is stratified iff there exists a stratification function for it.
\end{mydefn}

\removebrackets
\begin{mydefn}[\cite{Bol96,Gelfond88}] \label{D:STABLE}
   Let $\mathcal{P}=(\Sigma,\mathbb{A},\mathbb{P},\mathbb{R})$ be a TSS and $M$ a transition model. $M$ is a stable transition model for $\mathcal P$ if
    \[M=M_{Strip({\mathcal P},M)},\]
    where
    $Strip({\mathcal P},M)$ is the TSS $(\Sigma,\mathbb{A},\mathbb{P},Strip(\mathbb{R},M))$ with
    \[Strip(\mathbb{R},M)\triangleq\left\{\frac{pprem(r)}{conc(r)}|\quad r\in \mathbb{R}_{ground}\;\text{and}\;M\models nprem(r)\right\},
    \]
    where $\mathbb{R}_{ground}$ denotes the set of all ground instances of rules in $\mathbb{R}$,
    and $M_{Strip({\mathcal P},M)}$ is the least transition model of the positive TSS $Strip({\mathcal P},M)$.
\end{mydefn}

As is well known, stable models are supported models and each stratified TSS $\mathcal P$ has a unique stable model \cite{Bol96}; moreover, such stable model does not depend on particular stratification function \cite{Groote93}.

\section{Syntax and SOS rules of $\text{CLL}_R$}
The calculus $\text{CLL}_R$ is obtained from CLL by enriching it with recursive operations.
Following \cite{Baeten08}, this paper adopts the notation $\langle X | E \rangle$ to denote recursive operations, which encompasses both the CCS operator $recX.t$ and standard way of expressing recursion in ACP.
Formally, the terms in $\text{CLL}_{R}$ are defined by BNF:
\[ t::= 0\;|\perp\;|\;(\alpha.t) \;|\; (t\Box t)\;|\;(t\wedge t)\;|\;(t\vee t)\;|\;(t\parallel_A t)\;|\;X\; | \;\langle X|E \rangle \]
where $X \in V_{AR}$, $\alpha\in Act_\tau$, $A\subseteq Act$ and recursive specification
$E = E(V)$ is a nonempty finite set of equations $E = \{X = t| X \in V\}$.
As usual, 0 encodes deadlock.
The prefix $\alpha.t$ has a single capability, expressed by $\alpha$; the process $t$ cannot proceed until $\alpha$ has been exercised.
$\Box$ is an external choice operator.
$\parallel_A $ is a CSP-style parallel operator, $t_1\parallel_A t_2$ represents a process that behaves as $t_1$ in parallel with $t_2$ under the synchronization set $A$.
$\bot$ represents an inconsistent process with empty behavior.
$\vee$ and $\wedge$ are logical operators, which are intended for describing logical combinations of processes.

In the sequel, we often denote $\langle X|\{X=t_X\}\rangle$ briefly by $\langle X|X=t_X\rangle$.
Given a term $\langle X|E \rangle$ and variable $Y$, the phrase
\textquotedblleft $Y$ occurs in $\langle X|E \rangle$" means that $Y$ occurs in $t_Z$ for some $Z= t_Z \in E$.
Moreover, the scope of a recursive operation $\langle X|E \rangle$ exactly consists of all $t_Z$ with $Z = t_Z \in E$.
An occurrence of a variable $X$ in a given $t$ is  free if it does not occur in the scope of any recursive operation $\langle Y|E \rangle$ with $E = E(V)$ and $X \in V$.
A variable $X$ in term $t$ is a free variable if all occurrences of $X$ in $t$ are free, otherwise $X$ is a recursive variable in $t$.

\begin{convention}\label{C:REC_VAR}
  Throughout this paper, as usual, we make the assumption that recursive variables are distinct from each other. That is, for any two recursive specifications $E(V_1)$ and $E'(V_2)$ we have $V_1 \cap V_2 = \emptyset$. Moreover, we will tacitly restrict our attention to terms where no recursive variable  has free occurrences. For example we will not consider terms such as $X \Box \langle X| X= a.X \rangle$ because this term could be replaced by the clear term $X \Box \langle Y | Y = a.Y \rangle$ with the same meaning.
\end{convention}

On account of the above convention, given a term $t$, the set $FV(t)$ of all free variables of $t$ may be defined recursively as:
\begin{itemize}
  \item $FV(X) = \{X\}$; $FV(0)=FV(\bot)= \emptyset$; $FV(\alpha.t) = FV(t)$;
  \item $FV(t_1 \odot t_2) = FV(t_1) \cup FV(t_2)$ with $\odot \in \{\vee,\wedge,\Box,\parallel_A\}$;
  \item $FV(\langle Y|E \rangle) = \bigcup_{Z = t_Z \in E}FV(t_Z) - V$ where $E = E(V)$.
\end{itemize}

As usual, a term $t$ is closed if $FV(t) = \emptyset$.
The set of all closed terms of $\text{CLL}_R$ is denoted $T(\Sigma_{\text{CLL}_R} )$.
In the following, a term is a \emph{process} iff it is closed.
Unless noted otherwise we use $p,q,r$ to represent processes.
We shall always use $t_1 \equiv t_2$ to mean that expressions $t_1$ and $t_2$ are syntactically identical.
In particular, $\langle Y | E \rangle \equiv \langle Y'| E' \rangle$ means that $Y \equiv Y'$ and for any $Z$ and $t_Z$, $Z = t_Z \in E$ iff  $Z = t_Z \in E'$.

\begin{mydefn}
  For any recursive specification $E(V)$ and term $t$, we define $\langle t|E \rangle$ to be $t\{\langle X|E \rangle/X: X \in V\}$, that is, $\langle t|E \rangle$ is obtained from $t$ by simultaneously replacing all free occurrences of each $X(\in V)$ by $\langle X|E \rangle$.
\end{mydefn}

For example, consider $t \equiv X \Box a.\langle Y | Y = X \ \Box Y \rangle$ and $E(\{X\})=\{X=t_X\}$ then $\langle t| E\rangle \equiv \langle X|X =t_X\rangle \Box a.\langle Y | Y = \langle X|X=t_X\rangle \Box Y \rangle$.
In particular, for any recursive specification $E(V)$ and $t \equiv X$, $\langle t|E \rangle \equiv \langle X|E\rangle$ whenever $X \in V$ and $\langle t|E \rangle \equiv X$ if $X \notin V$.


As usual, an occurrence of $X$ in $t$ is strongly (or, weakly) guarded if such occurrence is within some subexpression $a.t_1$ with $a \in Act$ ($\tau.t_1$ or $t_1 \vee t_2$ respectively).
A variable $X$ is strongly (or, weakly) guarded in $t$ if each occurrence of $X$ is strongly (weakly respectively) guarded.
Notice that, since the first move of $r \vee s$ is a $\tau$-labelled transition (see Table~\ref{Ta:OPERATIONAL_RULES}), which is independent of $r$ and $s$, any occurrence of $X$ in $r \vee s$ is treated as being weakly guarded.
A recursive specification $E(V)$ is guarded if for each $X \in V$ and $Z = t_Z \in E$, each occurrence of $X$ in $t_Z$ is (weakly or strongly) guarded.

\begin{convention}\label{C:REC_SPEC}
  It is well known that unguarded processes cause many problems in many aspects of the theory \cite{Milner83} and  unguarded recursion is incompatible with negative rules \cite{Bloom94}.
  As usual, this paper will focus on guarded recursive specifications.
  That is, we assume that all recursive specifications considered in the remainder of this paper are guarded.
\end{convention}


We now provide SOS rules to specify the behavior of processes (i.e., closed terms) formally.
All SOS rules are divided into two parts: operational and predicate rules.

\begin{table}
\begin{center}
    $\begin{array}{ll}
    \displaystyle \qquad (Ra_1)\frac{-}{\alpha.x_1 \stackrel{\alpha}{\longrightarrow} x_1}  &
    \displaystyle \qquad (Ra_2)\frac{x_1 \stackrel{a}{\longrightarrow} y_1, x_2 \not \stackrel{\tau}{\longrightarrow}}{x_1 \Box x_2 \stackrel{a}{\longrightarrow} y_1}\\
      & \\
    \displaystyle \qquad (Ra_3)\frac{x_1 \not\stackrel{\tau}{\longrightarrow} , x_2 \stackrel{a}{\longrightarrow} y_2 }{x_1 \Box x_2 \stackrel{a}{\longrightarrow} y_2}&
    \displaystyle \qquad(Ra_4)\frac{x_1 \stackrel{\tau}{\longrightarrow} y_1}{x_1 \Box x_2 \stackrel{\tau}{\longrightarrow} y_1 \Box x_2}\\
     & \\
    \displaystyle \qquad (Ra_5)\frac{x_2 \stackrel{\tau}{\longrightarrow} y_2}{x_1 \Box x_2 \stackrel{\tau}{\longrightarrow} x_1 \Box y_2}&
    \displaystyle \qquad (Ra_6)\frac{x_1 \stackrel{a}{\longrightarrow} y_1, x_2 \stackrel{a}{\longrightarrow}y_2}{x_1 \wedge x_2 \stackrel{a}{\longrightarrow} y_1 \wedge y_2}\\
    & \\
    \displaystyle \qquad (Ra_7)\frac{x_1 \stackrel{\tau}{\longrightarrow} y_1}{x_1 \wedge x_2 \stackrel{\tau}{\longrightarrow} y_1 \wedge x_2} &
    \displaystyle \qquad (Ra_8)\frac{x_2 \stackrel{\tau}{\longrightarrow} y_2}{x_1 \wedge x_2 \stackrel{\tau}{\longrightarrow} x_1 \wedge y_2} \\
    & \\
    \displaystyle \qquad(Ra_9)\frac{-}{x_1 \vee x_2 \stackrel{\tau}{\longrightarrow} x_1} &
    \displaystyle \qquad(Ra_{10})\frac{-}{x_1 \vee x_2 \stackrel{\tau}{\longrightarrow} x_2} \\
    & \\
    \displaystyle \qquad(Ra_{11})\frac{x_1 \stackrel{\tau}{\longrightarrow} y_1}{x_1 \parallel_A x_2 \stackrel{\tau}{\longrightarrow} y_1\parallel_A x_2} &
    \displaystyle \qquad(Ra_{12})\frac{x_2 \stackrel{\tau}{\longrightarrow} y_2}{x_1 \parallel_A x_2 \stackrel{\tau}{\longrightarrow} x_1 \parallel_A y_2} \\
    & \\
    \displaystyle \qquad(Ra_{13})\frac{x_1 \stackrel{a}{\longrightarrow} y_1 , x_2 \not \stackrel{\tau}{\longrightarrow} }{x_1 \parallel_A x_2 \stackrel{a}{\longrightarrow} y_1 \parallel_A x_2}(a\notin A)&
    \displaystyle \qquad(Ra_{14})\frac{x_1 \not\stackrel{\tau}{\longrightarrow} , x_2 \stackrel{a}{\longrightarrow} y_2 }{x_1 \parallel_A x_2 \stackrel{a}{\longrightarrow} x_1 \parallel_A y_2}(a\notin A)\\
    & \\
    \displaystyle \qquad(Ra_{15})\frac{x_1 \stackrel{a}{\longrightarrow} y_1, x_2 \stackrel{a}{\longrightarrow}y_2}{x_1\parallel_A x_2 \stackrel{a}{\longrightarrow} y_1 \parallel_A y_2} (a\in A)&
    \displaystyle \qquad(Ra_{16})\frac{\langle t_X| E \rangle  \stackrel{\alpha}{\longrightarrow} y}{\langle X|E \rangle \stackrel{\alpha}{\longrightarrow} y}(X=t_X \in E)\\
    &
    \end{array}
    $

\caption{Operational rules\label{Ta:OPERATIONAL_RULES}}
\end{center}
\end{table}

Operational rules $Ra_i(1 \leq i \leq 16)$ are listed in Table~\ref{Ta:OPERATIONAL_RULES}, where $a \in Act$, $\alpha \in Act_{\tau}$ and $A \subseteq Act$.
Negative premises in Rules $Ra_2$, $Ra_3$, $Ra_{13}$ and $Ra_{14}$ give $\tau$-transition precedence over transitions labelled with visible actions, which guarantees that the transition model of $\text{CLL}_{R}$ is $\tau$-pure.
Rules $Ra_9$ and $Ra_{10}$ illustrate that the operational aspect of $t_1\vee t_2$ is same as internal choice in usual process calculus.
Rule $Ra_6$ reflects that conjunction operator is a synchronous product for visible transitions.
The operational rules of the other operators are as usual.

\begin{table}[ht]
\begin{center}
    $\begin{array}{ll}
    \displaystyle \qquad(Rp_1)\frac{-}{\bot F}&
    \displaystyle \qquad(Rp_2)\frac{x_1 F}{\alpha .x_1 F}\\
    &\\
    \displaystyle \qquad(Rp_3)\frac{x_1 F, x_2 F}{x_1\vee x_2 F}&
    \displaystyle \qquad(Rp_4)\frac{x_1 F}{x_1\Box x_2 F}\\
    &\\
    \displaystyle \qquad(Rp_5)\frac{x_2 F}{x_1\Box x_2 F}&
    \displaystyle \qquad(Rp_6)\frac{x_1 F}{x_1\parallel_A x_2 F}\\
    &\\
    \displaystyle \qquad(Rp_7)\frac{x_2 F}{x_1\parallel_A x_2 F}&
    \displaystyle \qquad(Rp_8)\frac{x_1 F}{x_1\wedge x_2 F}\\
    &\\
    \displaystyle \qquad(Rp_9)\frac{x_2 F}{x_1\wedge x_2 F}&
    \displaystyle \qquad(Rp_{10})\frac{x_1 \stackrel{a}{\longrightarrow} y_1, x_2 \not\stackrel{a}{\longrightarrow}, x_1 \wedge x_2 \not\stackrel{\tau}{\longrightarrow}}{x_1 \wedge x_2 F}\\
    &\\
    \displaystyle \qquad(Rp_{11})\frac{x_1 \not\stackrel{a}{\longrightarrow} , x_2 \stackrel{a}{\longrightarrow} y_2, x_1 \wedge x_2 \not\stackrel{\tau}{\longrightarrow}}{x_1 \wedge x_2 F}&
    \displaystyle \qquad(Rp_{12})\frac{x_1 \wedge x_2 \stackrel{\alpha}{\longrightarrow} z, \{yF:x_1 \wedge x_2 \stackrel{\alpha}{\longrightarrow}y\}}{x_1 \wedge x_2 F} \\
    &\\
    \displaystyle \qquad(Rp_{13})\frac{\{yF:x_1 \wedge x_2 \stackrel{\epsilon}{\Longrightarrow}|y\}}{x_1 \wedge x_2 F} &
    \displaystyle \qquad(Rp_{14})\frac{\langle t_X|E \rangle F}{\langle X|E \rangle F}(X = t_X \in E) \\
    &\\
    \displaystyle \qquad(Rp_{15})\frac{\{yF:\langle X|E \rangle \stackrel{\epsilon}{\Longrightarrow}|y\}}{\langle X|E \rangle F}&\\
    &
    \end{array}
    $
\caption{Predicate rules}\label{Ta:PREDICATIVE_RULES}
\end{center}
\end{table}
Predicate rules in Table~\ref{Ta:PREDICATIVE_RULES} specify  the inconsistency predicate $F$.
$0$ and $\bot$ represent different processes.
Rule $Rp_1$ says that $\bot$ is inconsistent.
Thus $\bot$ cannot be implemented.
While $0$ is consistent, which is an implementable process.
Rule $Rp_3$ reflects that if both two disjunctive parts are inconsistent then so is the disjunction.
Rules $Rp_4-Rp_9$ describe the system design strategy that if one part is inconsistent, then so is the whole composition.
Rules $Rp_{10}$ and $Rp_{11}$ reveal that a stable conjunction is inconsistent if its conjuncts have distinct ready sets.


Rules $Rp_{13}$ and $Rp_{15}$ are used to capture (LTS2) in Def.~\ref{D:LLTS}, which are the abbreviation of the rules with the format
\[\frac{\{yF: \exists y_0,y_1,\dots,y_n(z \equiv y_0 \stackrel{\tau}{\longrightarrow}y_1\stackrel{\tau}{\longrightarrow} \dots \stackrel{\tau}{\longrightarrow}y_n \equiv y\;\text{and}\; y \not\stackrel{\tau}{\longrightarrow})\}}{zF}\]
with $z \equiv x_1 \wedge x_2$ or $\langle X|E \rangle$.
Intuitively, these two rules say that if all stable $\tau$-descendants of $z$  are inconsistent, then $z$ itself is inconsistent.
Notice that, especially for readers who are familiar with notations used in \cite{Luttgen10}, the transition relation $\stackrel{\epsilon}{\Longrightarrow}|$ occurring in these two rules does not involve any requirement on consistency (see Remark~\ref{R:TRANSITION_CONSISTENCY} and notations above it).

Since the behavior of any process in CLL is finite, each process can reach a stable state, and Rules $Rp_{1}-Rp_{12}$  suffice to capture the inconsistency predicate $F$.
In particular, these rules guarantee that the LTS associated with CLL satisfies (LTS1) and (LTS2) in Def.~\ref{D:LLTS} \cite{Zhang11}.
However, for $\text{CLL}_R$, Rules $Rp_{1}-Rp_{12}$ are insufficient even if the usual rule for recursive operations (i.e. $Rp_{14}$) is added.
For instance, consider processes $q \equiv \langle X|X= \tau.X\rangle$ and $p \equiv \langle X| X = X \vee 0 \rangle \wedge a.0$, it is not difficult to see that neither $qF$ nor $pF$ can be inferred by using only Rules $Rp_{1} - Rp_{12}$ and $Rp_{14}$, however, both $p$ and $q$ should be inconsistent due to (LTS2).
Fortunately, an inference of $pF$ (or, $qF$) is at hand by admitting Rule $Rp_{13}$ ($Rp_{15}$ respectively).


Summarizing, the TSS for $\text{CLL}_{R}$ is $\mathcal{P}_{\text{CLL}_{R}} =(\Sigma_{\text{CLL}_{R}},Act_{\tau},\mathbb{P}_{\text{CLL}_{R}},\mathbb{R}_{\text{CLL}_{R}})$, where
\begin{itemize}
  \item $\Sigma_{\text{CLL}_{R}} = \{ \Box,\wedge,\vee,0,\bot\}\cup\{\alpha.(): \alpha \in Act_\tau\}\cup\{\parallel_A: A\subseteq Act\} \cup \{\langle X|E \rangle : E= E(V)\;\text{is a guarded recursive specification with } X \in V\}$,
  \item $\mathbb{P}_{\text{CLL}_{R}}=\{F\}$, and
  \item $\mathbb{R}_{\text{CLL}_{R}}=\{Ra_1, \dots, Ra_{16}\}\cup\{Rp_1, \dots ,Rp_{15}\}$.
\end{itemize}

\section{Stable transition model of $\mathcal{P}_{\text{CLL}_{R}}$}


This section will consider the well-definedness of the TSS ${\mathcal P}_{\text{CLL}_{R}}$ (i.e., the existence and uniqueness of the stable model of ${\mathcal P}_{\text{CLL}_{R}}$) and provide a few basic properties of the LTS associated with ${\mathcal P}_{\text{CLL}_R}$.




    As we know, it is not trivial that a TSS with rules with negative premises and recursion has a unique stable model.
    In order to demonstrate that ${\mathcal P}_{\text{CLL}_{R}}$ has one, it is sufficient to give a stratification function of ${\mathcal P}_{\text{CLL}_R}$.
    To this end, a few preliminary notations are introduced. Given a term $t$, the degree of $t$, denoted by $|t|$, is inductively defined as:
    \begin{itemize}
      \item $|0|=|\bot|=|\langle X|E \rangle |\triangleq1$;
      \item $|t_1 \odot t_2|\triangleq|t_1|+|t_2|+1$ for each $\odot \in \{\wedge,\Box,\vee,\parallel_A\}$;
      \item $|\alpha .t|\triangleq|t|+1$ with $\alpha \in Act_{\tau}$.
    \end{itemize}

The function $G: T(\Sigma_{\text{CLL}_{R}}) \longrightarrow \mathbb{N}$ is defined by:
\begin{itemize}
  \item $G(\langle X|E\rangle)\triangleq 1$;
  \item $G(0)=G(\bot)=G(\alpha.t)=G(t_1 \vee t_2)  \triangleq 0$ with $\alpha \in Act_{\tau}$;
  \item $G(t_1 \odot t_2) \triangleq G(t_1)+G(t_2)$ for each $\odot \in \{\wedge,\Box, \parallel_A\}$.
\end{itemize}

Clearly, given a term $t$, $G(t)$ is the number of unguarded recursive operations occurring in $t$.
Further, the function $S_{{\mathcal P}_{\text{CLL}_{R}}}$ from  $Tr(\Sigma_{\text{CLL}_{R}},Act_{\tau})\;\cup\; Pred(\Sigma_{\text{CLL}_{R}},{\mathbb P}_{\text{CLL}_{R}})$ to $\omega \times 2+1$ is given below, where $\omega$ is the initial limit ordinal,
\begin{itemize}
  \item $S_{{\mathcal P}_{\text{CLL}_{R}}}(t\stackrel{\alpha}{\longrightarrow} t')\triangleq G(t)\times \omega + |t|$;
  \item $S_{{\mathcal P}_{\text{CLL}_{R}}}(tF)\triangleq \omega \times 2$.
\end{itemize}

Since each recursive specification is assumed to be guarded (see, Convention~\ref{C:REC_SPEC}), it is not difficult to check that this function $S_{{\mathcal P}_{\text{CLL}_{R}}}$ is a stratification of ${\mathcal P}_{\text{CLL}_{R}}$. Moreover, since each stratified TSS has a unique stable model \cite{Bol96}, ${\mathcal P}_{\text{CLL}_{R}}$ has a unique stable transition model.
From now on, we use $M_{\text{CLL}_{R}}$ to denote such stable model.

\begin{mydefn}
    The LTS associated with $\text{CLL}_{R}$, in symbols $LTS(\text{CLL}_{R})$, is the quadruple
    $(T(\Sigma_{\text{CLL}_{R}}),Act_{\tau},\longrightarrow_{\text{CLL}_{R}},F_{\text{CLL}_{R}})$, where \\
    --- $p \stackrel{\alpha}{\longrightarrow}_{\text{CLL}_{R}} p'$ iff $p\stackrel{\alpha}{\longrightarrow} p' \in M_{\text{CLL}_{R}}$;\\
    --- $p\in F_{\text{CLL}_{R}}$ iff $pF \in M_{\text{CLL}_{R}}$.
\end{mydefn}

Therefore, $p \stackrel{\alpha}{\longrightarrow}_{\text{CLL}_{R}} p'$ (or, $p \in F_{\text{CLL}_R}$) if and only if  $Strip({\mathcal P}_{\text{CLL}_{R}}, M_{\text{CLL}_{R}}) \vdash p\stackrel{\alpha}{\longrightarrow}p'$ ($pF$ respectively) for any processes $p$, $p'$ and $\alpha \in Act_{\tau}$.
This allows us to proceed by induction on the depth of inferences when demonstrating propositions concerning $\longrightarrow_{\text{CLL}_R}$ and $F_{\text{CLL}_R}$.

\begin{convention}
  For the sake of convenience, in the remainder of this paper, we shall omit the subscript in labelled transition relations $\stackrel{\alpha}{\longrightarrow}_{\text{CLL}_R}$, that is, we shall use $\stackrel{\alpha}{\longrightarrow}$ to denote transition relation in $LTS(\text{CLL}_R)$.
  Thus, the notation $\stackrel{\alpha}{\longrightarrow}$ has double utility:
  predicate symbols in the TSS ${\mathcal P}_{\text{CLL}_R}$ and labelled transition relations on processes in $LTS(\text{CLL}_R)$.
  However, it usually does not lead to confusion in a given context.
  Similarly, the notation $F_{\text{CLL}_R}$ will be abbreviated by $F$.
  Hence the symbol $F$ is overloaded, predicate symbol in the TSS ${\mathcal P}_{\text{CLL}_R}$ and the set of all inconsistent processes within $LTS(\text{CLL}_R)$, in each case the context of use will allow us to make the distinction.
\end{convention}

In the following, we intend to provide a number of simple properties of $LTS({\text{CLL}_{R}})$. In particular, we will show that $LTS({\text{CLL}_{R}})$ is a $\tau$-pure LLTS.

\begin{lemma}\label{L:F_NORMAL}
Let $p$ and $q$ be any two processes.
  \begin{enumerate}[(1)]
\renewcommand{\theenumi}{(\arabic{enumi})}
    \item  $p \vee q \in F $  iff $p,q \in F $.
    \item  $\alpha.p \in F $ iff $p \in F $ for each $\alpha \in Act_{\tau}$.
    \item  $p \odot q \in F $  iff either $p \in F $ or $q \in F $  with $\odot \in \{\Box, \parallel_A\}$.
    \item Either $p \in F $ or $q \in F $ implies $p \wedge q \in F $.
    \item $0 \notin F $ and $\bot \in F $.
    \item $\langle X|X = \tau.X \rangle \in F$.
    \item If $\forall q(p \stackrel{\epsilon}{\Longrightarrow}|q\;\text{implies}\;q \in F)$ then $p \in F$.
    \item $\langle X|E \rangle  \in F $ iff $\langle t_X|E \rangle \in F $ for each $X$ with $X = t_X \in E$.
  \end{enumerate}
\end{lemma}
\begin{proof}
  Items (1) - (6) are straightforward.
  For item (7), it proceeds by induction on $p$, in particular, for the case where $p$ is of the format $p_1 \wedge p_2$ (or  $\langle X |E \rangle$), the conclusion immediately follows due to Rule $Rp_{13}$ ($Rp_{15}$ respectively).

For item (8), the implication from right to left is straightforward.
The argument of the converse implication splits into two cases based on the last rule applied in the proof tree of  $Strip({\mathcal P}_{\text{CLL}_{R}}, M_{\text{CLL}_{R}}) \vdash \langle X |E \rangle F$.
If Rule $Rp_{14}$ is the last rule then the proof is trivial.
For the other case where Rule $Rp_{15}$ is used, it is also straightforward by applying item (7) in this lemma and the fact that $\langle X|E \rangle \stackrel{\tau}{\longrightarrow}r$ iff $\langle t_X |E \rangle \stackrel{\tau}{\longrightarrow} r$ for any $r$.
\end{proof}

The notion of $\tau$-purity is a technical constraint for LLTSs (L\"uttgen and Vogler 2007, 2010).
The result below shows that $LTS({\text{CLL}_{R}})$ is indeed $\tau$-pure.

\begin{theorem}\label{L:TAU_PURE}
    $LTS({\text{CLL}_{R}})$ is $\tau$-pure.
\end{theorem}
\begin{proof}
Suppose $p \stackrel{\tau}{\longrightarrow}$.
Hence $p \stackrel{\tau}{\longrightarrow} q$ for some $q$.
Then the lemma would be established by proving that $p \not\stackrel{a}{\longrightarrow}$ for any $a \in Act$.
It is straightforward by induction on  the depth of the proof tree of $Strip({\mathcal P}_{\text{CLL}_{R}}, M_{\text{CLL}_{R}}) \vdash p \stackrel{\tau}{\longrightarrow}q$.
\end{proof}

%
In order to prove that $LTS(\text{CLL}_R)$ is a LLTS, the result below is needed.
Its converse is an instance of (LTS1) with $\alpha= \tau$, and hence also holds by Theorem~\ref{L:LLTS}.

\begin{lemma}\label{L:FAILURE_TAU_I}
For any process $p$ with $\tau \in {\mathcal I}(p)$, if $p \in F $ then $\forall q(p\stackrel{\tau}{\longrightarrow} q\;\text{implies}\; q\in F)$.
\end{lemma}
\begin{proof}
Suppose $p\stackrel{\tau}{\longrightarrow} q$. We may prove $q \in F$ by induction on the depth of the proof tree $\mathcal T$ of $Strip({\mathcal P}_{\text{CLL}_{R}}, M_{\text{CLL}_{R}}) \vdash p\stackrel{\tau}{\longrightarrow}q$.
It proceeds by distinguishing different cases based on the form of $p$.
Here we  handle only three cases as examples.\\

\noindent Case 1 $p \equiv p_1 \Box p_2$.

W.l.o.g, assume the last rule applied in $\mathcal T$  is $\frac{p_1 \stackrel{\tau}{\longrightarrow}p_1'}{p_1 \Box p_2 \stackrel{\tau}{\longrightarrow} p_1' \Box p_2}$.
Hence $q \equiv p_1' \Box p_2$.
Since $p \in F$, by Lemma~\ref{L:F_NORMAL}(3), $p_1 \in F$ or $p_2 \in F$.
If $p_2 \in F$ then it immediately follows from Lemma~\ref{L:F_NORMAL}(3) that $q \equiv p_1' \Box p_2 \in F$.
If $p_1 \in F$ then $p_1' \in F$ by induction hypothesis (IH, for short). Hence $p_1' \Box p_2 \in F$, as desired.\\

\noindent Case 2 $p \equiv \langle X|E \rangle$.

The last rule applied in $\mathcal T$  is $\frac{\langle t_X |E \rangle \stackrel{\tau}{\longrightarrow}q}{\langle X| E \rangle \stackrel{\tau}{\longrightarrow} q}$ with $X = t_X \in E$.
Since $p \in F$, by Lemma~\ref{L:F_NORMAL}(8), we have $\langle t_X |E \rangle \in F$.
Then $q \in F$ by applying IH.\\

\noindent Case 3 $p \equiv p_1 \wedge p_2$.

W.l.o.g, assume the last rule applied in $\mathcal T$ is $\frac{p_1 \stackrel{\tau}{\longrightarrow}p_1'}{p_1 \wedge p_2 \stackrel{\tau}{\longrightarrow} p_1' \wedge p_2}$.
Hence $q \equiv p_1' \wedge p_2$.
In the following, we intend to show $q \in F$ by distinguishing four cases based on the last rule applied in the inference of $Strip({\mathcal P}_{\text{CLL}_{R}}, M_{\text{CLL}_{R}}) \vdash p_1 \wedge p_2F$.\\

\noindent Case 3.1 $\frac{p_1 F}{p_1\wedge p_2 F}$ or $\frac{p_2 F}{p_1\wedge _2 F}$.

Similar to Case 1, omitted.\\

\noindent Case 3.2  $\frac{p_1 \stackrel{a}{\longrightarrow} r, p_2 \not\stackrel{a}{\longrightarrow}, p_1 \wedge p_2 \not\stackrel{\tau}{\longrightarrow}}{p_1 \wedge p_2 F}$ or $\frac{p_1 \not\stackrel{a}{\longrightarrow} , p_2 \stackrel{a}{\longrightarrow} r, p_1 \wedge p_2 \not\stackrel{\tau}{\longrightarrow}}{p_1 \wedge p_2 F}$.

This case is impossible because of $\tau \in {\mathcal I}(p_1 \wedge p_2)$.\\

\noindent Case 3.3 $\frac{p_1 \wedge p_2 \stackrel{\alpha}{\longrightarrow} r, \{r'F:p_1 \wedge p_2 \stackrel{\alpha}{\longrightarrow}r'\}}{p_1 \wedge p_2 F}$.

Since $LTS(\text{CLL}_R)$ is $\tau$-pure and $p_1 \wedge p_2 \stackrel{\tau}{\longrightarrow}$, we have $\alpha = \tau$.
Hence $q \in F$ immediately.\\

\noindent Case 3.4 $\frac{\{rF:p_1 \wedge p_2 \stackrel{\epsilon}{\Longrightarrow}|r\}}{p_1 \wedge p_2 F}$.

Assume $q \equiv p_1'\wedge p_2 \stackrel{\epsilon}{\Longrightarrow}| r'$.
Thus $r' \in F$ due to $p \stackrel{\tau}{\longrightarrow}p_1'\wedge p_2\stackrel{\epsilon}{\Longrightarrow}| r'$.
Hence $p_1' \wedge p_2 \in F$ by applying Rule $Rp_{13}$.
\end{proof}

Now we are ready to show that  $LTS({\text{CLL}_{R}})$ is a LLTS.
\begin{theorem}\label{L:LLTS}
    $LTS({\text{CLL}_{R}})$ is a LLTS.
\end{theorem}
\begin{proof}
\noindent \textbf{(LTS1)} Suppose $\alpha \in {\mathcal I}(p)$ and $\forall r(p \stackrel{\alpha}{\longrightarrow}r\;\text{implies}\;r \in F)$.
Then $p \stackrel{\alpha}{\longrightarrow} q$ for some $q$.
To complete the proof, we intend to show $ p \in F$.
 It proceeds by induction on the depth of the proof tree $\mathcal T$ of $Strip({\mathcal P}_{\text{CLL}_{R}}, M_{\text{CLL}_{R}}) \vdash p\stackrel{\alpha}{\longrightarrow}q$.
We distinguish different cases based on the form of $p$.
In particular, the proof for the case $p \equiv p_1 \wedge p_2$ is immediate by Rule $Rp_{12}$.
In the following, we give the proof for the case $p \equiv p_1 \parallel_A p_2$, the other cases are left to the reader.
The argument splits into two cases depending on $\alpha$.\\

\noindent Case 1 $\alpha = \tau$.

W.l.o.g, assume the last rule applied in $\mathcal T$ is $\frac{p_1 \stackrel{\tau}{\longrightarrow}p_1'}{p_1 \parallel_A p_2 \stackrel{\tau}{\longrightarrow} p_1' \parallel_A p_2}$.
Thus $q \equiv p_1' \parallel_A p_2$.
        If $p_2 \in F  $ then $p_1 \parallel_A p_2 \in F $ follows from Lemma~\ref{L:F_NORMAL}(3) at once.
         For the other case $p_2 \notin F$, it is not difficult to see that each $\tau$-derivative of $p_1$ is inconsistent, that is $\forall p_1'' (p_1  \stackrel{\tau}{\longrightarrow}  p_1'' \;\text{implies}\; p_1'' \in F)$. Hence $p_1 \in F$ by IH.
        Therefore it follows from Lemma~\ref{L:F_NORMAL}(3) that $p_1\parallel_A p_2 \in F$, as desired.\\

  \noindent Case 2 $\alpha \in Act$.

         In this situation, the last rule applied in $\mathcal T$ has one of the following three formats:

\noindent  (1) $\frac{p_1 \stackrel{\alpha}{\longrightarrow} p_1',p_2 \not \stackrel{\tau}{\longrightarrow}}{p_1 \parallel_A p_2 \stackrel{\alpha}{\longrightarrow} p_1' \parallel_A p_2} (\alpha\notin A)$;
  (2) $\frac{p_2 \stackrel{\alpha}{\longrightarrow} p_2',p_1 \not \stackrel{\tau}{\longrightarrow}}{p_1 \parallel_A p_2 \stackrel{\alpha}{\longrightarrow} p_1\parallel_A p_2'}(\alpha\notin A)$;
  (3)  $\frac{p_1 \stackrel{\alpha}{\longrightarrow} p_1',p_2 \stackrel{\alpha}{\longrightarrow} p_2'}{p_1 \parallel_A p_2 \stackrel{\alpha}{\longrightarrow} p_1' \parallel_A p_2'}( \alpha\in A) $.

        We consider only (3), the other two may be handled in a similar manner as the case $\alpha = \tau$.
Since $\forall r(p_1 \parallel_A p_2 \stackrel{\alpha}{\longrightarrow}r\;\text{implies}\;r \in F)$, by Lemma~\ref{L:F_NORMAL}(3), it is easy to see that either $\forall r(p_1 \stackrel{\alpha}{\longrightarrow}r\;\text{implies}\;r \in F)$ or $\forall r(p_2 \stackrel{\alpha}{\longrightarrow}r\;\text{implies}\;r \in F)$.
Furthermore, due to $\alpha \in {\mathcal I}(p_1)$ and $\alpha \in {\mathcal I}(p_2)$, by IH, we have $p_1 \in F$ or $p_2 \in F$, which implies $p_1 \parallel_A p_2 \in F$.\\

\noindent \textbf{(LTS2)} It suffices to show that, for each $p$, if $p\notin F $ then $p \stackrel{\epsilon}{\Longrightarrow}_F|q$ for some $q$.
    Suppose $p\notin F $.
    By Lemma~\ref{L:F_NORMAL}(7),  there exists $q$ such that $p\stackrel{\epsilon}{\Longrightarrow} |q$ and $q \notin F $.
    Then it immediately follows from Lemma~\ref{L:FAILURE_TAU_I} that $p\stackrel{\epsilon}{\Longrightarrow}_F |q$, as desired.
\end{proof}

\begin{rmk}
It is worth pointing out that Lemma~\ref{L:FAILURE_TAU_I} does not always hold for LLTS.
In fact, the property ``\emph{$p \in F$ implies $q \in F$ for each $\tau$-derivative $q$ of $p$}" is logically independent of Def.~\ref{D:LLTS}.
It is SOS rules adopted in this paper that bring such additional property.
Hence this paper restricts itself to specific LLTSs, which makes reasoning about inconsistency a bit easier than in the general LLTS setting.
%
%
\end{rmk}

A simple observation on proof trees for $Strip({\mathcal P}_{\text{CLL}_{R}}, M_{\text{CLL}_{R}}) \vdash p \wedge q F$ is given below, which will be used in establishing a fundamental property of conjunctive compositions.

\begin{lemma} \label{L:RS_CON_Pre}
For any finite sequence $p_0 \wedge q_0 \stackrel{\tau}{\longrightarrow} , ..,\stackrel{\tau}{\longrightarrow}  p_i \wedge q_i \stackrel{\tau}{\longrightarrow} ,..,
            \stackrel{\tau}{\longrightarrow}  |p_n \wedge q_n(n \geq 0)$, if $p_i \wedge q_i \in F  $ and $p_i,q_i\notin F  $ for each $i \leq n$, then the inference of $p_0 \wedge q_0F$ essentially depends on $p_n \wedge q_nF$, that is,
            each proof tree for $Strip({\mathcal P}_{\text{CLL}_{R}}, M_{\text{CLL}_{R}}) \vdash p_0 \wedge q_0F$ has a subtree with root $p_n \wedge q_nF$, in particular, such subtree is proper if $n \geq 1$.
\end{lemma}
\begin{proof}
            We prove the statement by induction on $n$.
            For the inductive basis $n=0$, it holds trivially due to $p_0 \wedge q_0 \equiv p_n \wedge q_n$.
            For the inductive step, assume that $p_0 \wedge q_0 \stackrel{\tau}{\longrightarrow} p_1 \wedge q_1 (\stackrel{\tau}{\longrightarrow} )^k |p_{k+1} \wedge q_{k+1}$.
            Let $\mathcal T$ be any proof tree for $Strip({\mathcal P}_{\text{CLL}_{R}}, M_{\text{CLL}_{R}}) \vdash p_0 \wedge q_0 F$.
            Since $p_0,q_0 \notin F $ and $p_0 \wedge q_0  \stackrel{\tau}{\longrightarrow} $,
            the last rule applied in $\mathcal T$ is
            \[
            \text{either}\;\frac{p_0 \wedge q_0 \stackrel{\alpha}{\longrightarrow} r',\{rF:p_0 \wedge q_0 \stackrel{\alpha}{\longrightarrow} r\}}{p_0 \wedge q_0  F}\;\text{or}\;\frac{\{rF:p_0 \wedge q_0 \stackrel{\epsilon}{\Longrightarrow}|r\}}{p_0 \wedge q_0 F}.
            \]

                For the first alternative, since $LTS(\text{CLL}_{R})$ is $\tau$-pure, we have $\alpha = \tau$.
                Then it follows from $p_0 \wedge q_0 \stackrel{\tau}{\longrightarrow} p_1 \wedge q_1$ that, in the proof tree $\mathcal T$, one of nodes directly above the root is labelled with  $p_1 \wedge q_1 F$.
                Thus, by IH, $\mathcal T$  has a proper subtree with root $p_{k+1} \wedge q_{k+1}F$.

            For the second alternative, since $p_0 \wedge q_0 \stackrel{\epsilon}{\Longrightarrow} |p_{k+1} \wedge q_{k+1}$,  one of nodes directly above the root of $\mathcal T$ is labelled with $ p_{k+1} \wedge q_{k+1} F $, as desired.
\end{proof}

The next three results has been obtained for CLL in pure process-algebraic style in \cite{Zhang11}, where the proof essentially depends on the fact that, for any process $p$ within CLL and $\alpha \in
Act_\tau $, $p$ is of more complex structure than its $\alpha $-derivatives.
Unfortunately, such property does not always hold for $\text{CLL}_R$.
For instance, consider the process $\langle X| X= a.X \parallel_{\emptyset} a.b.X \rangle$.
Here we give another proof along lines presented in \cite{Zhu13}.

\begin{lemma}\label{L:RS_CON}
If $p_1 \underset{\thicksim}{\sqsubset}_{RS} p_2$, $p_1 \underset{\thicksim}{\sqsubset}_{RS} p_3$  and $p_1 \notin F $ then $p_2 \wedge p_3 \notin F $.
\end{lemma}
\begin{proof}
  Let $\Omega = \{q \wedge r :   p \underset{\thicksim}{\sqsubset}_{RS} q, p \underset{\thicksim}{\sqsubset}_{RS} r\;\text{and}\; p  \notin F\}$.
Clearly, it suffices to prove that $F \cap \Omega = \emptyset$.
  Conversely, suppose that $F \cap \Omega \neq \emptyset$.
  In the following, we intend to prove that, for each $t \in \Omega$, any proof tree of $tF$  is not well-founded.
  Then a contradiction arises at this point due to Def.~\ref{D:PROOF}.
  Thus, to complete the proof, it suffices to show the claim below. \\

\noindent \textbf{Claim} For any $s \in \Omega$,  each proof tree of $sF$ has a proper subtree with root $s'F$ for some $s' \in \Omega$.\\

Suppose $q \wedge r \in \Omega$.
    Then  $p \underset{\thicksim}{\sqsubset}_{RS} q$, $p \underset{\thicksim}{\sqsubset}_{RS} r$ and $p \notin F$ for some $p$.
    Thus it follows  that
    \[q \notin F, r \notin F \text{ and }{\mathcal I}(p)={\mathcal I}(q)={\mathcal I}(r).\tag{\ref{L:RS_CON}.1}\]
    Let $\mathcal T$ be any proof tree of $Strip({\mathcal P}_{\text{CLL}_{R}}, M_{\text{CLL}_{R}}) \vdash q \wedge rF$.
    By (\ref{L:RS_CON}.1), the last rule applied in $\mathcal T$ is of the form
\[\text{either}\;\frac{\{sF:q \wedge r \stackrel{\epsilon}{\Longrightarrow}|s\}}{q \wedge r  F}\;\text{or}\; \frac{q \wedge r \stackrel{\alpha}{\longrightarrow} s',\{sF:q \wedge r \stackrel{\alpha}{\longrightarrow} s\}}{q \wedge r  F}.\]
    Since both $q$ and $r$ are stable,  so is $q \wedge r$.
    Then, for the first alternative, the label of the node directly above the root of $\mathcal T$ is $q \wedge rF$ itself, as desired.

    Next we consider the second alternative.
        In this case, $\tau \neq \alpha \in {\mathcal I}(q \wedge r) $ and
        \[\forall s(q \wedge r \stackrel{\alpha}{\longrightarrow}  s \;\text{implies}\;s \in F).\tag{\ref{L:RS_CON}.2}\]
        Hence $\alpha \in {\mathcal I}(q) \cap {\mathcal I}(r)$.
        Then $\alpha \in {\mathcal I}(p)$ due to (\ref{L:RS_CON}.1).
        Further, since $p\notin F $, by Theorem~\ref{L:LLTS}, we get
         \[p \stackrel{\alpha}{\longrightarrow}_F p' \stackrel{\epsilon}{\Longrightarrow}_F| p''\text{ for some } p'\;\text{and}\;p''.\tag{\ref{L:RS_CON}.3}\]
        Then it immediately follows from   $p \underset{\thicksim}{\sqsubset}_{RS} q$ and $p \underset{\thicksim}{\sqsubset}_{RS} r$ that
        \[q \stackrel{\alpha}{\longrightarrow}_F q'\stackrel{\epsilon}{\Longrightarrow}_F| q''\;\text{and}\; p'' \underset{\thicksim}{\sqsubset}_{RS} q''\;\text{for some}\; q',q'',\text{and} \tag{\ref{L:RS_CON}.4}\]
        \[r \stackrel{\alpha}{\longrightarrow}_F r'\stackrel{\epsilon}{\Longrightarrow}_F| r''\;\text{and}\; p'' \underset{\thicksim}{\sqsubset}_{RS} r''\;\text{for some}\; r',r''. \tag{\ref{L:RS_CON}.5}\]
        So,  $q \wedge r  \stackrel{\alpha}{\longrightarrow} q'\wedge r' $.
        Then $q' \wedge r' \in F$ by (\ref{L:RS_CON}.2).
        Moreover, we obtain
        $q' \equiv q_0 \stackrel{\tau}{\longrightarrow}_{F},\dots,\stackrel{\tau}{\longrightarrow}_{F}| q_n\equiv q''$ for some $q_i(0 \leq i \leq n)$,
        and $r'\equiv r_0 \stackrel{\tau}{\longrightarrow}_{F},\dots,\stackrel{\tau}{\longrightarrow}_{F}| r_m\equiv r''$ for some $r_j(0 \leq j \leq m)$.
        Then
\[
            q' \wedge r' \equiv q_0 \wedge r_0  \stackrel{\tau}{\longrightarrow} ,..,\stackrel{\tau}{\longrightarrow}  q_n \wedge r_0 \stackrel{\tau}{\longrightarrow}q_n \wedge r_{1} ,..,
            \stackrel{\tau}{\longrightarrow}  |q_n \wedge r_m \equiv q''\wedge r''. \tag{\ref{L:RS_CON}.6}
\]
        By Lemma~\ref{L:FAILURE_TAU_I}, it follows from $q' \wedge r' \in F$ that
        \[q_i \wedge r_j \in F \;\text{for each}\; q_i \wedge r_j\text{ occurring in }
        (\ref{L:RS_CON}.6). \tag{\ref{L:RS_CON}.7}\]
        It follows from (\ref{L:RS_CON}.3), (\ref{L:RS_CON}.4) and (\ref{L:RS_CON}.5) that $q_n \wedge r_m \equiv q''\wedge r'' \in \Omega$.
        Moreover, since one of nodes directly above the root of $\mathcal T$ is labelled with $q' \wedge r'F$, by (\ref{L:RS_CON}.6), (\ref{L:RS_CON}.7) and Lemma~\ref{L:RS_CON_Pre}, it follows from $q_i \notin F(0 \leq i \leq n)$ and $r_j \notin F(0 \leq j \leq m)$ that $\mathcal T$ has a proper subtree with root $q_n \wedge r_mF$.
\end{proof}

\begin{lemma}\label{L:WEDGE_PRECONGRUENCE}
  If $p  \underset{\thicksim}{\sqsubset}_{RS} q$ and $p \underset{\thicksim}{\sqsubset}_{RS}r$ then $p  \underset{\thicksim}{\sqsubset}_{RS} q \wedge r$.
\end{lemma}
\begin{proof}
  Set
\[{\mathcal R} = \{(p_1,p_2 \wedge p_3):p_1 \underset{\thicksim}{\sqsubset}_{RS} p_2\;\text{and}\; p_1\underset{\thicksim}{\sqsubset}_{RS}p_3\}.\]
It suffices to show that $\mathcal R$ is a stable ready simulation relation, which is almost immediate by using Lemma~\ref{L:RS_CON} to handle (RS2) and (RS3).
\end{proof}

We conclude this section with recalling a result obtained in \cite{Luttgen10} and \cite{Zhang11}  in different style, which reveals that $\sqsubseteq_{RS}$ is precongruent w.r.t the operators $\Box$, $\parallel_A$, $\vee$ and $\wedge$.
Formally,

\begin{theorem}\label{L:pre_precongruence} \hfill
  \begin{enumerate}[(1)]
\renewcommand{\theenumi}{(\arabic{enumi})}
    \item For each  $\odot \in \{\Box, \parallel_A, \wedge\}$, if $p \underset{\thicksim}{\sqsubset}_{RS} q$ and  $s \underset{\thicksim}{\sqsubset}_{RS} r$ then $p \odot s \underset{\thicksim}{\sqsubset}_{RS} q \odot r$.
\item For each  $\odot \in \{\Box, \parallel_A, \vee, \wedge\}$, if $p \sqsubseteq_{RS} q$ and  $s \sqsubseteq_{RS} r$ then $p \odot s \sqsubseteq_{RS} q \odot r$.
  \end{enumerate}
\end{theorem}
\begin{proof}
The item (2) follows from item (1).
For item (1), the proofs are not much different from ones given in \cite{Zhang11}. In particular, Lemma~\ref{L:WEDGE_PRECONGRUENCE} is applied in the proof for the case $\odot = \wedge$.
\end{proof}


\section{Basic properties of unfolding, context and transitions}

This section will provide a number of useful results that will be used in the following sections.
Subsection~5.1 will recall the notion of unfolding and give some elementary properties of it.
In subsection~5.2, we will be concerned with capturing one-step transitions in terms of contexts and substitutions.
A treatment of a more general case involving multi $\tau$-transitions will be considered in subsection 5.3.

\subsection{Unfolding}

The notion of unfolding plays an important role when dealing with recursive operators.
This subsection will give a few results concerning it.
We begin with recalling the notion of unfolding.

\begin{mydefn}
  Let $X$ be a free variable in a given term $t$. An occurrence of $X$  in $t$ is unfolded, if this occurrence does not occur  in the scope of any recursive operation $\langle Y|E \rangle$.
Moreover, $X$ is unfolded if all occurrences of $X$ in $t$ are unfolded.
\end{mydefn}

\removebrackets
\begin{mydefn}[\cite{Baeten08}]\label{D:UNFOLDING}
  A series of binary relations $\Rrightarrow_k$ over terms with $k < \omega$ is defined inductively as:
\begin{itemize}
  \item $t \Rrightarrow_0 s$ if $t \equiv s$;
  \item $t \Rrightarrow_1 s$ if $t$ has a subterm $\langle Y|E \rangle$ with $Y =t_Y \in E$ which is not in the scope of any recursive operation, and $s$ is obtained from $t$ by replacing this subterm by $\langle t_Y|E \rangle$;
  \item  $t \Rrightarrow_{k+1} s$ if  $t \Rrightarrow_k t'$ and  $t' \Rrightarrow_1 s$ for some term $t'$.
\end{itemize}
Moreover,  $ \Rrightarrow \triangleq \underset{0 \leq k <\omega}{\bigcup} \Rrightarrow_k$. For any $t$ and $s$, $s$ is a multi-step unfolding of $t$ if $t \Rrightarrow s$.
\end{mydefn}

For instance, consider $t \equiv (\langle X | X = a.X \Box b.\langle Y|Y = c.Y\rangle \rangle \Box d.0)\Box Z$, we have
\[t \Rrightarrow_1 ((a.\langle X | X = a.X \Box b.\langle Y|Y = c.Y\rangle \rangle  \Box b.\langle Y|Y=c.Y\rangle ) \Box d.0) \Box Z,\]
but it does not hold that $t \Rrightarrow_1 (\langle X | X = a.X \Box b.c.\langle Y|Y = c.Y\rangle \rangle \Box d.0) \Box Z$ because the subterm $\langle Y|Y = c.Y\rangle$ is in the scope of the recursive operation $\langle X | X = a.X \Box b.\langle Y|Y = c.Y\rangle \rangle$.
The simple result below provides an equivalent formulation of the binary relation $\Rrightarrow_1$.

\begin{lemma}\label{L:ONE_STEP_UNFOLDING_SYNTAX}
  For any term $t_1$ and $t_2$, $t_1  \Rrightarrow_1 t_2$ iff there exists a  term $s$ and variable $X$ such that
  \begin{enumerate}[(1$_\Rrightarrow$)]
 \renewcommand{\theenumi}{(\arabic{enumi}$_\Rrightarrow$)}
    \item $X$ is a unfolded variable in $s$,
    \item $X$ occurs in $s$ exactly once, and
    \item $t_1 \equiv s \{\langle Y|E \rangle/ X\}$ and $t_2 \equiv s\{\langle t_Y|E \rangle/ X\}$ for some $Y,E$ with $Y = t_Y \in E$.
  \end{enumerate}
\end{lemma}
\begin{proof}
Immediately follows from Def.~\ref{D:UNFOLDING}.
\end{proof}

A few trivial but useful properties concerning $\Rrightarrow_n$ are listed in the next lemma.

\begin{lemma}\label{L:ONE_STEP_UNFOLDING_VARIABLE}
  For any term $t,s$ and $X \in FV(t)$, if $t \Rrightarrow_n s$ then
 \begin{enumerate}[(1)]
 \renewcommand{\theenumi}{(\arabic{enumi})}
   \item  
         if $X$ is unfolded in $t$ then so it is in $s$ and the number of occurrences of $X$ in $s$ is equal to that in $t$;
   \item  the number of unguarded occurrences of $X$ in $s$ is not more than that in $t$;
   \item  if $X$ is (strongly) guarded in $t$ then so it is in $s$;
   \item $FV(s)\subseteq FV(t)$;
   \item  if $X$  occurs in the scope of conjunction in $s$ (that is, there exists a subterm $t_1 \wedge t_2$ of $s$ such that $X$ occurs in either $t_1$ or $t_2$) then so does it in $t$.
 \end{enumerate}
\end{lemma}
\begin{proof}
By Lemma~\ref{L:ONE_STEP_UNFOLDING_SYNTAX} and Convention~\ref{C:REC_SPEC}, it is straightforward by induction on $n$.
\end{proof}


Notice that the clause (2) in the above lemma does not always hold for guarded occurrences.
For example, consider $t\equiv \langle X|X=a.X \wedge b.Y \rangle$, we have $t \Rrightarrow_1 a.\langle X|X=a.X \wedge b.Y \rangle \wedge b.Y$, and $Y$ guardedly occurs in the latter twice but occurs in $t$ only once.
Clearly, the clause (2) strongly depends on Convention~\ref{C:REC_SPEC}.
Moreover, the clause (4) cannot be strengthened to ``$FV(s)= FV(t)$''.
Consider $t \equiv \langle X_1|\{X_1=a.0,X_2=b.X_1\Box Y\}\rangle$ and $t\Rrightarrow_1 a.0$, then we have $FV(t)=\{Y\}$ and $FV(a.0)=\emptyset$.

Given a variable $X$ and term $t$, the folding number of $X$ in $t$, in symbols $FN(t,X)$, is defined as  the sum of depths of nested recursive operations surrounding all \emph{unguarded} occurrences of $X$ in $t$.
Formally:

\begin{mydefn}[Folding number]
 Given a term $t$ and $X \in FV(t)$, the folding number of $X$ in $t$, denoted by $FN(t,X)$, is defined recursively below, where $UFV(t)$ is the set of all free variables which have unguarded occurrence in $t$.
 \begin{equation*}
    \begin{aligned}
      -\mspace{-3mu}-\;  & FN(0,X)  =  FN(\bot,X)  = FN(Y,X)=FN(t_1 \vee t_2,X)=  FN(\alpha.t,X) \triangleq 0;\\
      -\mspace{-3mu}-\; & FN(t_1 \odot t_2,X) \triangleq FN(t_1,X) + FN(t_2,X) \;\text{with}\;  \odot \in \{\Box, \parallel_A, \wedge\};\\
      -\mspace{-3mu}-\;  & FN(\langle Y|E \rangle,X)  \triangleq
     \left\{
          \begin{array}{ll}
           1 + \underset{Z =t_Z \in E}{\sum} FN(t_Z,X), &\text{if}\; X \in UFV(\langle Y|E \rangle); \\
             & \\
           0,  & \text{otherwise.}
          \end{array}
        \right.
    \end{aligned}
 \end{equation*}
\end{mydefn}
For instance, consider $t \equiv \langle X| X= a.X \vee Y_1 \rangle \Box \langle Z | Z = c.Z \Box Y_2 \rangle$, then $FN(t,Y_1)=0$ and $FN(t,Y_2)= 1$.

\begin{lemma}\label{L:STABILIZATION_PRE}
   For any  term $t$, there exits a term $s$ such that $t  \Rrightarrow s$ and each unguarded occurrence of any free variable in $s$ is unfolded.
\end{lemma}
\begin{proof}
It proceeds by induction on $n = \underset{X\in UFV(t)}{\sum}FN(t,X)$.
For the induction base $n =0$, it is easy to see that for each $X \in FV(t)$, any unguarded occurrence of $X$ in $t$ must be unfolded.
Thus $t$ itself meets our requirement because of $t \Rrightarrow t$.
For the inductive step $n=k+1$,
 due to $n = k +1 > 0$, $t$ is of the format either $t_1 \odot t_2$ with $\odot \in \{\wedge, \parallel_A \Box\}$ or $\langle Y |E \rangle$ .
 In the following, we shall proceed by induction on the structure of $t$.
 In case $t \equiv t_1 \odot t_2$ with $\odot \in \{\wedge, \parallel_A, \Box\}$, it is straightforward by applying IH on $t_1$ and $t_2$.
Next we consider the case $t \equiv \langle Y|E \rangle$ with $Y = t_Y \in E$.

 Clearly, $UFV(\langle Y|E \rangle) \neq \emptyset$ because of $n >0$.
  Since $\langle Y|E \rangle \Rrightarrow_1 \langle t_Y|E \rangle$, by Lemma~\ref{L:ONE_STEP_UNFOLDING_VARIABLE}(2)(4), we have
  \[UFV(\langle t_Y|E \rangle ) \subseteq UFV(\langle Y|E \rangle).\]
  Moreover, by Convention~\ref{C:REC_SPEC} and the definition of $\langle t_Y|E \rangle$,  it is not difficult to get
  $FN(\langle Y|E\rangle,X) > FN(\langle t_Y |E \rangle,X)$ for each $ X \in UFV(\langle t_Y|E \rangle )$.
  Hence
  \[\underset{X\in UFV(\langle t_Y|E \rangle)}{\sum}FN(\langle t_Y|E \rangle ,X) < \underset{X\in UFV(\langle Y|E \rangle)}{\sum}FN(\langle Y|E \rangle ,X).\]
  Then, by IH on $n$, there exists $s$ such that $\langle t_Y |E \rangle \Rrightarrow s$ and each unguarded occurrence of any free variable is unfolded in $s$.
  Moreover, $\langle Y|E \rangle \Rrightarrow s$ due to $\langle Y|E \rangle \Rrightarrow_1 \langle t_Y|E \rangle$.
\end{proof}


\subsection{Contexts and transitions}
Due to Rules $Rp_{12}$, $Rp_{13}$ and $Rp_{15}$, in order to obtain further properties of the inconsistency predicate $F$, we often need to capture the connection between the formats of $p$ and $q$ for a given transition $p \stackrel{\alpha}{\longrightarrow}q$.
Clearly, if $p$ involves recursive operations, $q$ is not always a subterm of $p$ and its format often depends on some unfolding of $p$.
This subsection intends to explore this issue.

%
%
%

\begin{mydefn}[Context]
 A context $C_{\widetilde{X}}$ is a term whose free variables are among a $n$-tuple distinct variables $\widetilde{X}=(X_1,...,X_n)$ ($n \geq 0$).
 Given a $n$-tuple processes $\widetilde{p}=(p_1,\dots,p_n)$, the term $C_{\widetilde{X}}\{p_1/X_1,...,p_n/X_n\}$ ($C_{\widetilde{X}}\{\widetilde{p}/\widetilde{X}\}$, for short) is obtained from $C_{\widetilde{X}}$ by replacing $X_i$ by $p_i$ for each $i < n$ simultaneously.
 In particular, we use $C_{\widetilde{X}}\{p/\widetilde{X}\}$ to denote the result of replacing all variables in $\widetilde{X}$ by $p$.
  A context $C_{\widetilde{X}}$ is said to be stable if $C_{\widetilde{X}}\{0/\widetilde{X}\} \not\stackrel{\tau}{\longrightarrow} $.
\end{mydefn}

In the remainder of this paper, whenever the expression $C_{\widetilde{X}}\{\widetilde{p}/\widetilde{X}\}$ occurs, we always assume that $|\widetilde{p}|= |\widetilde{X}|$ and $C_{\widetilde{X}}\{\widetilde{p}/\widetilde{X}\}$ is subject to Convention~\ref{C:REC_VAR} (recursive variables occurring in $\widetilde{p}$ may be renamed if it is necessary), where $|\widetilde{X}|$ is the length of the tuple $\widetilde{X}$.


\begin{mydefn}[Active]
An occurrence of a free variable $X$ in term $t$ is \emph{active} if such occurrence is unguarded and unfolded.
A free variable $X$ in term $t$ is active if all its occurrences are active.
A free variable $X$ in term $t$ is \emph{1-active} if $X$ occurs in $t$ exactly once and such occurrence is active.
\end{mydefn}

For example, $X$ is 1-active in $\langle Y|Y=a.Y \rangle \Box X$.
Moreover, it is evident that, for any context $C_{\widetilde{X}}$, if there exists an active occurrence of some variable within $C_{\widetilde{X}}$, then $C_{\widetilde{X}}$ is not of the form   $\alpha.B_{\widetilde{X}}$, $B_{\widetilde{X}} \vee D_{\widetilde{X}}$ and $\langle Y|E \rangle$.
This fact is used in demonstrating the next two lemmas, which  give some properties of 1-active place-holder.
Before presenting them, for simplicity of notation, we introduce the notation below.\\

\noindent \textbf{Notation} Given $n$-tuple processes $\widetilde{p}=(p_1, \dots, p_n)$ and $p'$, we use  $\widetilde{p}\,[p'/p_i]$  to denote $(p_1, \dots, p_{i-1},p',p_{i+1}, \dots, p_n)$.


\begin{lemma}\label{L:ONE_ACTION_TAU_GF}
  For any $C_{\widetilde{X}}$ with 1-active variable $X_{i_0}$ and $\widetilde{p}$ with $p_{i_0} \stackrel{\tau}{\longrightarrow} p'$,
  $C_{\widetilde{X}}\{\widetilde{p}/\widetilde{X}\}\stackrel{\tau}{\longrightarrow}  C_{\widetilde{X}}\{\widetilde{p}\,[p'/p_{i_0}]/\widetilde{X}\}$.
\end{lemma}
\begin{proof}
  Proceed by induction on the structure of $C_{\widetilde{X}}$.
\end{proof}

This result does not always hold for visible transitions.
For instance, consider $C_X \equiv X \Box \tau.r$ and $p \equiv a.q$, although $p \stackrel{a}{\longrightarrow}q$ and $X$ is 1-active in $C_X$, it is false that $C_X\{p/X\} \stackrel{a}{\longrightarrow}$.

\begin{lemma}\label{L:FAILURE_GF}
  For any $p$ and  $C_X$ with 1-active variable $X$, if $p \in F $ then $C_X\{p/X\} \in F $.
\end{lemma}
\begin{proof}
  By a straightforward induction on $C_X$.
\end{proof}

In order to prove that $\sqsubseteq_{RS}$ still is precongruent in the presence of recursive operations, it is necessary to formally describe the contribution of $C_{\widetilde{X}}$ and $\widetilde{p}$ for a given transition $C_{\widetilde{X}}\{\widetilde{p}/\widetilde{X}\} \stackrel{\alpha}{\longrightarrow} r$.
In the following, we shall provide a few of results concerning this.
We begin with considering $\tau$-labelled transitions.
Before giving the next lemma formally, we illustrate the intuition behind it by means of an example.
Consider $C_X \equiv (a.0 \vee X) \Box X$, $B_X \equiv \langle Y|Y= X \Box b.Y \rangle \Box X$, $p \equiv c.0\vee e.0$ and $q \equiv d.0$, then we have two $\tau$-labelled transitions
\[C_{X}\{q/X\}   \stackrel{\tau}{\longrightarrow} a.0 \Box d.0 \]
and
\[B_{X}\{p/X\}  \stackrel{\tau}{\longrightarrow} (e.0 \Box b.\langle Y|Y = (c.0 \vee e.0) \Box b.Y \rangle) \Box (c.0 \vee e.0).\]
It is not difficult to see that these two $\tau$-labelled transitions depend on the capability of context $C_X$ and substitution $p$ respectively.
For the former, no matter what $q$ is, the corresponding $\tau$-transition still exists for $C_{X}\{q/X\}$.
Moreover, the target has the same pattern.
Set $C_X'\equiv a.0 \Box X$.
Clearly, $C_{X}\{q/X\}   \stackrel{\tau}{\longrightarrow} C_X'\{q/X\}$ for any $q$.
The latter is much more trick.
Intuitively, one instance of $p$ first exposes itself and then performs a $\tau$-transition.
Since there are multi instances of $p$ and some of them are nested by recursive operations, we should identify the real performer of the $\tau$-transition and this identification is very helpful when we deal with multi $\tau$-transitions.
As long as $p$ can perform $\tau$-transition, so can $B_{X}\{p/X\}$.
Similarly, the target also has a pattern.
Set $B_{X,Z}' \equiv (Z \Box b.\langle Y|Y = X \Box b.Y \rangle) \Box X$.
It is easy to see that $B_{X}\{p/X\}  \stackrel{\tau}{\longrightarrow} B_{X,Z}'\{p/X,p'/Z\}$ for any $p \stackrel{\tau}{\longrightarrow}p'$.
We summarize this observation formally as follows, where two clauses capture $\tau$-transitions exited by  contexts and substitutions respectively; moreover, some simple properties on contexts are also listed in (C-$\tau$-3) which will be used in the sequel.




\begin{lemma}\label{L:ONE_ACTION_TAU}
  For any $C_{\widetilde{X}}$ and $\widetilde{p}$, if $C_{\widetilde{X}}\{\widetilde{p}/\widetilde{X}\} \stackrel{\tau}{\longrightarrow} r$ then one of conclusions below holds.
  \begin{enumerate}[(1)]
\renewcommand{\theenumi}{(\arabic{enumi})}
    \item  There exists $C_{\widetilde{X}}'$ such that

        \textbf{(C-$\tau$-1)} $r \equiv C_{\widetilde{X}}'\{\widetilde{p}/\widetilde{X}\}$;

        \textbf{(C-$\tau$-2)} for any processes $\widetilde{q}$, $C_{\widetilde{X}}\{\widetilde{q}/\widetilde{X}\} \stackrel{\tau}{\longrightarrow} C_{\widetilde{X}}'\{\widetilde{q}/\widetilde{X}\}$;

        \textbf{(C-$\tau$-3)} for each $X \in \widetilde{X}$,

          \quad \textbf{(C-$\tau$-3-i)} if $X$ is active in $C_{\widetilde{X}}$  then so it is in $C_{\widetilde{X}}'$ and the number of occurrences of $X$ in $C_{\widetilde{X}}'$ is equal to that in $C_{\widetilde{X}}$;

          \quad \textbf{(C-$\tau$-3-ii)}   if $X$ is unfolded in $C_{\widetilde{X}}$  then so it is in $C_{\widetilde{X}}'$ and the number of occurrences of $X$ in $C_{\widetilde{X}}'$ is not more than that in $C_{\widetilde{X}}$;

          \quad \textbf{(C-$\tau$-3-iii)}    if $X$ is strongly guarded in $C_{\widetilde{X}}$ then so it is in $C_{\widetilde{X}}'$;

          \quad \textbf{(C-$\tau$-3-iv)}   if $X$ does not occur in the scope of any conjunction in $C_{\widetilde{X}}$ then neither does it in $C_{\widetilde{X}}'$.

    \item  There exist $C_{\widetilde{X}}'$, $C_{\widetilde{X},Z}''$ with $Z \notin \widetilde{X}$ and $i\leq |\widetilde{X}|$ such that

        \textbf{(P-$\tau$-1)} $C_{\widetilde{X}}  \Rrightarrow  C_{\widetilde{X}}'$, in particular, if $X_i$ is active in $C_{\widetilde{X}}$ then $C_{\widetilde{X}}' \equiv C_{\widetilde{X}}$;

        \textbf{(P-$\tau$-2)} $p_i \stackrel{\tau}{\longrightarrow} p'$ and $r \equiv C_{\widetilde{X},Z}''\{\widetilde{p}/\widetilde{X},p'/Z\}$ for some $p'$;

        \textbf{(P-$\tau$-3)} $C_{\widetilde{X},Z}''\{X_i/Z\} \equiv C_{\widetilde{X}}'$ and $Z$ is 1-active in  $C_{\widetilde{X},Z}''$;

        \textbf{(P-$\tau$-4)} for any processes $\widetilde{q}$ with $q_i \stackrel{\tau}{\longrightarrow} q'$, $C_{\widetilde{X}}\{\widetilde{q}/\widetilde{X}\} \stackrel{\tau}{\longrightarrow} C_{\widetilde{X},Z}''\{\widetilde{q}/\widetilde{X},q'/Z\}$.
  \end{enumerate}
\end{lemma}
\begin{proof}
  It proceeds by induction on the depth of the inference of $Strip(\mathcal{P}_{\text{CLL}_R} ,M_{\text{CLL}_R} ) \vdash C_{\widetilde{X}}\{\widetilde{p}/\widetilde{X}\} \stackrel{\tau}{\longrightarrow}r$.
  We distinguish six cases based on the form of $C_{\widetilde{X}}$ as follows.\\

  \noindent Case 1 $C_{\widetilde{X}}$ is  closed.

  Set $C_{\widetilde{X}}' \triangleq r$.
  Then (C-$\tau$-1,2,3) hold trivially. \\

   \noindent Case 2 $C_{\widetilde{X}} \equiv X$ with $X \in \widetilde{X}$.

  Put $C_{\widetilde{X}}' \triangleq X$ and $C_{\widetilde{X},Z}'' \triangleq Z$ with $Z \notin \widetilde{X}$.
   Then it is easy to check that (P-$\tau$-1) -- (P-$\tau$-4) hold.\\

\noindent Case 3 $C_{\widetilde{X}} \equiv \alpha.B_{\widetilde{X}}$.

   Thus $\alpha = \tau$ and $r \equiv B_{\widetilde{X}}\{\widetilde{p}/\widetilde{X}\}$.
   Then it is not difficult to see that (C-$\tau$-1,2,3) hold by taking $C_{\widetilde{X}}' \triangleq B_{\widetilde{X}}$.\\

\noindent Case 4 $C_{\widetilde{X}} \equiv B_{\widetilde{X}}\vee D_{\widetilde{X}}$.

   Obviously, $r\equiv B_{\widetilde{X}}\{\widetilde{p}/\widetilde{X}\} $ or $r\equiv D_{\widetilde{X}}\{\widetilde{p}/\widetilde{X}\} $.
   W.l.o.g, assume that $r\equiv B_{\widetilde{X}}\{\widetilde{p}/\widetilde{X}\} $.
   We set $C_{\widetilde{X}}' \triangleq B_{\widetilde{X}}$.
   Then it is straightforward that 
   (C-$\tau$-1,2) and (C-$\tau$-3-ii,iii,iv) hold.
   Moreover, since $C_{\widetilde{X}} \equiv B_{\widetilde{X}}\vee D_{\widetilde{X}}$, for each $X \in \widetilde{X}$, each occurrence of $X$
   is weakly guarded.
   Hence (C-$\tau$-3-i)  holds trivially.\\

\noindent Case 5 $C_{\widetilde{X}} \equiv B_{\widetilde{X}} \odot D_{\widetilde{X}}$ with $\odot \in \{\Box,\wedge,\parallel_A\}$.

   We consider the case $\odot = \Box$, others may be handled similarly and omitted.
   W.l.o.g, assume the last rule applied in the inference is
   \[\frac{B_{\widetilde{X}}\{\widetilde{p}/\widetilde{X}\} \stackrel{\tau}{\longrightarrow} r'}{B_{\widetilde{X}}\{\widetilde{p}/\widetilde{X}\} \Box D_{\widetilde{X}}\{\widetilde{p}/\widetilde{X}\} \stackrel{\tau}{\longrightarrow} r' \Box D_{\widetilde{X}}\{\widetilde{p}/\widetilde{X}\}}.\]
   Then $r \equiv r' \Box D_{\widetilde{X}}\{\widetilde{p}/\widetilde{X}\}$.
   For the $\tau$-labelled transition $B_{\widetilde{X}}\{\widetilde{p}/\widetilde{X}\} \stackrel{\tau}{\longrightarrow} r'$, by IH, either the clause (1) or (2) holds.

   For the former case, there exists $B_{\widetilde{X}}'$ that satisfies (C-$\tau$-1,2,3).
   Put $C_{\widetilde{X}}'\triangleq B_{\widetilde{X}}'\Box D_{\widetilde{X}}$.
   It immediately follows that $C_{\widetilde{X}}'$  satisfies (C-$\tau$-1,2,3).

    Next we consider the latter case. In this situation, there exist $B_{\widetilde{X}}'$, $B_{\widetilde{X},Z}''$ with $Z \notin \widetilde{X}$ and $i_0 \leq |\widetilde{X}|$ that satisfy (P-$\tau$-1) -- (P-$\tau$-4).
    Set
    \[C_{\widetilde{X}}' \triangleq B_{\widetilde{X}}' \Box D_{\widetilde{X}}\;\text{and}\;C_{\widetilde{X},Z}'' \triangleq B_{\widetilde{X},Z}'' \Box D_{\widetilde{X}}.\]
    We shall show that, for the $\tau$-labelled transition $C_{\widetilde{X}}\{\widetilde{p}/\widetilde{X}\} \stackrel{\tau}{\longrightarrow}r$, $C_{\widetilde{X}}'$, $C_{\widetilde{X},Z}''$ and $i_0$ realize (P-$\tau$-1) -- (P-$\tau$-4).

    \textbf{(P-$\tau$-1)} It follows from $B_{\widetilde{X}} \Rrightarrow  B_{\widetilde{X}}'$ that  $C_{\widetilde{X}} \equiv B_{\widetilde{X}} \Box D_{\widetilde{X}}  \Rrightarrow B_{\widetilde{X}}' \Box D_{\widetilde{X}} \equiv C_{\widetilde{X}}'$.
    If $X_{i_0}$ is active in $C_{\widetilde{X}}$ then so it is in $B_{\widetilde{X}}$, and hence $C_{\widetilde{X}}' \equiv C_{\widetilde{X}}$ due to $B_{\widetilde{X}}' \equiv B_{\widetilde{X}}$.

    \textbf{(P-$\tau$-2)} Since $B_{\widetilde{X}}'$ satisfies (P-$\tau$-2),  $p_{i_0} \stackrel{\tau}{\longrightarrow} p'$ and $r' \equiv B_{\widetilde{X},Z}''\{\widetilde{p}/\widetilde{X},p'/Z\}$ for some $p'$.
    Due to $Z \notin \widetilde{X}$, we have $r \equiv B_{\widetilde{X},Z}''\{\widetilde{p}/\widetilde{X},p'/Z\} \Box D_{\widetilde{X}}\{\widetilde{p}/\widetilde{X}\} \equiv C_{\widetilde{X},Z}''\{\widetilde{p}/\widetilde{X},p'/Z\}$.

    \textbf{(P-$\tau$-3)} It follows from $B_{\widetilde{X},Z}''\{X_{i_0}/Z\} \equiv B_{\widetilde{X}}'$ and $Z \notin \widetilde{X}$ that $C_{\widetilde{X},Z}''\{X_{i_0}/Z\} \equiv B_{\widetilde{X},Z}''\{X_{i_0}/Z\}\Box D_{\widetilde{X}} \equiv C_{\widetilde{X}}'$.
    Moreover, since $Z$ is 1-active in $B_{\widetilde{X},Z}''$, so it is in $C_{\widetilde{X},Z}''$.

    \textbf{(P-$\tau$-4)} Let $\widetilde{q}$ be any tuple with $|\widetilde{q}|=|\widetilde{p}|$ and $q_{i_0} \stackrel{\tau}{\longrightarrow} q'$.
        It follows from $B_{\widetilde{X}}\{\widetilde{q}/\widetilde{X}\} \stackrel{\tau}{\longrightarrow} B_{\widetilde{X},Z}''\{\widetilde{q}/\widetilde{X},q'/Z\}$ and $Z \notin \widetilde{X}$ that $C_{\widetilde{X}}\{\widetilde{q}/\widetilde{X}\} \stackrel{\tau}{\longrightarrow} C_{\widetilde{X},Z}''\{\widetilde{q}/\widetilde{X},q'/Z\}$.\\

\noindent Case 6 $C_{\widetilde{X}} \equiv \langle Y|E \rangle$.

Clearly, the last rule applied in the inference is
 \[\frac{\langle t_Y|E \rangle \{\widetilde{p}/\widetilde{X}\}\stackrel{\tau}{\longrightarrow}r}{\langle Y|E \rangle \{\widetilde{p}/\widetilde{X}\} \stackrel{\tau}{\longrightarrow}r}\;\text{with}\; Y=t_Y \in E.\]

For the $\tau$-labelled transition $\langle t_Y|E \rangle \{\widetilde{p}/\widetilde{X}\}\stackrel{\tau}{\longrightarrow}r$, by IH, either the clause (1) or (2) holds.

For the first alternative, there exists $C_{\widetilde{X}}'$ satisfying (C-$\tau$-1,2,3).
Then it is not difficult to check that, for the transition $\langle Y|E \rangle \{\widetilde{p}/\widetilde{X}\} \stackrel{\tau}{\longrightarrow}r$, $C_{\widetilde{X}}'$ also realizes the conditions (C-$\tau$-1,2,3).
Here $\langle Y|E \rangle \{\widetilde{p}/\widetilde{X}\} \Rrightarrow_1 \langle t_Y|E \rangle \{\widetilde{p}/\widetilde{X}\} $ and Lemma~\ref{L:ONE_STEP_UNFOLDING_VARIABLE}(3)(5) are used to assert (C-$\tau$-3-iii,iv) to be true.

For the second alternative, there exist $C_{\widetilde{X}}'$, $C_{\widetilde{X},Z}''$ with  $Z \notin \widetilde{X}$ and $i_0\leq |\widetilde{X}|$ that satisfy (P-$\tau$-1,2,3,4).
Clearly, $C_{\widetilde{X}}'$, $C_{\widetilde{X},Z}''$ and $i_0$ also realize (P-$\tau$-1,2,3,4) for the transition $\langle Y|E \rangle \{\widetilde{p}/\widetilde{X}\} \stackrel{\tau}{\longrightarrow}r$.
In particular, $\langle Y|E \rangle  \Rrightarrow C_{\widetilde{X}}' $ follows from  $\langle Y|E \rangle \Rrightarrow_1 \langle t_Y|E \rangle  \Rrightarrow C_{\widetilde{X}}' $.
\end{proof}

As an immediate consequence of Lemma~\ref{L:ONE_ACTION_TAU}, we have

\begin{lemma}\label{L:STABLE_CONTEXT_I}
For any context $C_{\widetilde{X}}$, $C_{\widetilde{X}}$ is stable iff $C_{\widetilde{X}}\{\widetilde{p}/\widetilde{X}\} \not\stackrel{\tau}{\longrightarrow} $ for some $\widetilde{p}$.
\end{lemma}
\begin{proof}
  Straightforward by Lemma~\ref{L:ONE_ACTION_TAU}.
\end{proof}

In the following, we intend to provide an analogue of Lemma~\ref{L:ONE_ACTION_TAU} for transitions labelled with visible actions.
To explain intuition behind the next result clearly, it is best to work with an example.
Consider $C_{X_1,X_2} \equiv ((X_1 \wedge \langle Y|Y = a.Y \rangle ) \Box a.b.0)\parallel_{\{b\}} (X_1 \wedge X_2) $, $p_1 \equiv a.0$ and $p_2 \equiv a.c.0$, we have three $a$-labelled transitions
\[C_{X_1,X_2}\{p_1/X_1,p_2/X_2\}  \stackrel{a}{\longrightarrow} (0 \wedge \langle  Y|Y=a.Y \rangle) \parallel_{\{b\}}(a.0\wedge a.c.0),\]
\[C_{X_1,X_2}\{p_1/X_1,p_2/X_2\}  \stackrel{a}{\longrightarrow} b.0 \parallel_{\{b\}}(a.0\wedge a.c.0),\]
and
\[C_{X_1,X_2}\{p_1/X_1,p_2/X_2\} \stackrel{a}{\longrightarrow} ((a.0 \wedge \langle  Y|Y=a.Y \rangle) \Box a.b.0) \parallel_{\{b\}}(0\wedge c.0).\]
These visible transitions starting from $C_{X_1,X_2}\{p_1/X_1,p_2/X_2\}$ are activated by three distinct events.
Clearly, both the context $C_{X_1,X_2}$ and the substitution $p_1$ contribute to the first transition, while two latter transitions depend merely on the capability of $C_{X_1,X_2}$ and $\widetilde{p_{1,2}}$ respectively.
These three situations may be described uniformly in the lemma below.
Here some additional properties on contexts are also listed in (CP-$a$-4), which will be useful in the sequel.

\begin{lemma}\label{L:ONE_ACTION_VISIBLE}
  For any $a \in Act$, $C_{\widetilde{X}}$ and $\widetilde{p}$, if $C_{\widetilde{X}}\{\widetilde{p}/\widetilde{X}\} \stackrel{a}{\longrightarrow} r$ then there exist $C_{\widetilde{X}}'$, $C_{\widetilde{X},\widetilde{Y}}'$ and $C_{\widetilde{X},\widetilde{Y}}''$ with $\widetilde{X} \cap \widetilde{Y} = \emptyset$  satisfying the conditions:

  \noindent \textbf{(CP-$a$-1)}  $C_{\widetilde{X}}  \Rrightarrow  C_{\widetilde{X}}' $;

  \noindent \textbf{(CP-$a$-2)}  for each $Y \in \widetilde{Y}$, $Y$ is 1-active in $C_{\widetilde{X},\widetilde{Y}}'$ and $C_{\widetilde{X},\widetilde{Y}}''$;

    \noindent \textbf{(CP-$a$-3)} there exist $i_Y \leq |\widetilde{X}|$ for each $Y \in \widetilde{Y}$ such that

        \textbf{(CP-$a$-3-i)} $C_{\widetilde{X},\widetilde{Y}}'\{\widetilde{X_{i_Y}}/\widetilde{Y}\} \equiv C_{\widetilde{X}}'$;

        \textbf{(CP-$a$-3-ii)} there exist $p_{Y}'$ such that $p_{i_Y} \stackrel{a}{\longrightarrow}  p_{Y}'$ for each $Y \in \widetilde{Y}$ and $r \equiv  C_{\widetilde{X},\widetilde{Y}}''\{\widetilde{p}/\widetilde{X},\widetilde{p_{Y}'}/\widetilde{Y}\}$;

        \textbf{(CP-$a$-3-iii)}  for any $\widetilde{q}$ with $|\widetilde{q}|=|\widetilde{X}|$ and $\widetilde{q'}$ such that $|\widetilde{q'}| = |\widetilde{Y}|$ and $q_{i_Y} \stackrel{a}{\longrightarrow} q_{Y}'$ for each $Y \in \widetilde{Y}$,
        if  $C_{\widetilde{X}}\{\widetilde{q}/\widetilde{X}\}$ is stable then $C_{\widetilde{X}}\{\widetilde{q}/\widetilde{X}\} \stackrel{a}{\longrightarrow} C_{\widetilde{X},\widetilde{Y}}''\{\widetilde{q}/\widetilde{X},\widetilde{q_{Y}'}/\widetilde{Y}\}$;

  \noindent \textbf{(CP-$a$-4)}  for each $X \in \widetilde{X}$,

        \textbf{(CP-$a$-4-i)} the number of occurrences of $X$ in $C_{\widetilde{X},\widetilde{Y}}''$ is not more than that in $C_{\widetilde{X},\widetilde{Y}}'$;

        \textbf{(CP-$a$-4-ii)} if $X$ is active in $C_{\widetilde{X},\widetilde{Y}}'$ then so it is in $C_{\widetilde{X},\widetilde{Y}}''$;

        \textbf{(CP-$a$-4-iii)} if $X$ does not occur in the scope of any conjunction in $C_{\widetilde{X}}$ then neither does it in $C_{\widetilde{X},\widetilde{Y}}''$.
\end{lemma}
\begin{proof}
  It proceeds by induction on the depth of the inference of $Strip(\mathcal{P}_{\text{CLL}_R} ,M_{\text{CLL}_R} ) \vdash C_{\widetilde{X}}\{\widetilde{p}/\widetilde{X}\} \stackrel{a}{\longrightarrow} r$.
  Due to $C_{\widetilde{X}}\{\widetilde{p}/\widetilde{X}\} \not\stackrel{\tau}{\longrightarrow} $,  it is impossible that $C_{\widetilde{X}} \equiv B_{\widetilde{X}} \vee D_{\widetilde{X}}$.
  Thus we can distinguish seven cases depending on the form of $C_{\widetilde{X}}$.\\

  \noindent Case 1 $C_{\widetilde{X}}$  is  closed.

    Set $C_{\widetilde{X}}' \equiv C_{\widetilde{X},\widetilde{Y}}' \triangleq C_{\widetilde{X}}$ and $C_{\widetilde{X},\widetilde{Y}}''\triangleq r$ with $\widetilde{Y} = \emptyset$.
    Clearly, these contexts realize conditions (CP-$a$-$i$) ($1 \leq i \leq 4$) trivially.\\

  \noindent Case 2 $C_{\widetilde{X}} \equiv X_{i_0}$ with $i_0 \leq |\widetilde{X}|$.

     Put $C_{\widetilde{X}}' \triangleq X_{i_0}$ and $C_{\widetilde{X},\widetilde{Y}}' \equiv C_{\widetilde{X},\widetilde{Y}}'' \triangleq Y $ with  $Y  \notin \widetilde{X}$.
     Then (CP-$a$-$i$) ($1 \leq i \leq 4$) follow immediately, in particular, for (CP-$a$-3), we take $i_Y \triangleq i_0$.\\

  \noindent Case 3 $C_{\widetilde{X}} \equiv \alpha.B_{\widetilde{X}}$.

     Then  $\alpha = a$ and $r \equiv B_{\widetilde{X}}\{\widetilde{p}/\widetilde{X}\}$.
     Put $C_{\widetilde{X}}' \equiv C_{\widetilde{X},\widetilde{Y}}'  \triangleq \alpha.B_{\widetilde{X}}$ and $ C_{\widetilde{X},\widetilde{Y}}'' \triangleq B_{\widetilde{X}}$ with $\widetilde{Y} = \emptyset$.
     Obviously, these contexts are what we seek.\\

  \noindent Case  4 $C_{\widetilde{X}} \equiv B_{\widetilde{X}} \Box D_{\widetilde{X}}$.

     W.l.o.g, suppose that the last rule applied in the inference is
     \[\frac{B_{\widetilde{X}}\{\widetilde{p}/\widetilde{X}\} \stackrel{a}{\longrightarrow} r, \; D_{\widetilde{X}}\{\widetilde{p}/\widetilde{X}\} \not\stackrel{\tau}{\longrightarrow}}{B_{\widetilde{X}}\{\widetilde{p}/\widetilde{X}\} \Box  D_{\widetilde{X}}\{\widetilde{p}/\widetilde{X}\} \stackrel{a}{\longrightarrow} r}.\]
     By IH, for the $a$-labelled transition $B_{\widetilde{X}}\{\widetilde{p}/\widetilde{X}\} \stackrel{a}{\longrightarrow} r$, there exist $B_{\widetilde{X}}'$, $B_{\widetilde{X},\widetilde{Y}}'$ and $B_{\widetilde{X},\widetilde{Y}}''$ with $\widetilde{X} \cap \widetilde{Y} = \emptyset$ that satisfy (CP-$a$-1) -- (CP-$a$-4).
     Set
     \[C_{\widetilde{X}}' \triangleq B_{\widetilde{X}}' \Box D_{\widetilde{X}}, C_{\widetilde{X},\widetilde{Y}}' \triangleq B_{\widetilde{X},\widetilde{Y}}' \Box D_{\widetilde{X}}\;\text{and}\; C_{\widetilde{X},\widetilde{Y}}'' \triangleq B_{\widetilde{X},\widetilde{Y}}''.\]
     Then it is not difficult to check that, for the $a$-labelled transition $C_{\widetilde{X}}\{\widetilde{p}/\widetilde{X}\}   \stackrel{a}{\longrightarrow} r$, these contexts above realizes (CP-$a$-1) -- (CP-$a$-4), as desired.\\

  \noindent Case 5 $C_{\widetilde{X}} \equiv B_{\widetilde{X}} \wedge D_{\widetilde{X}}$.

      In this situation, the last rule applied in the inference is
      \[\frac{B_{\widetilde{X}}\{\widetilde{p}/\widetilde{X}\} \stackrel{a}{\longrightarrow} r_1, D_{\widetilde{X}}\{\widetilde{p}/\widetilde{X}\} \stackrel{a}{\longrightarrow} r_2}{B_{\widetilde{X}}\{\widetilde{p}/\widetilde{X}\} \wedge  D_{\widetilde{X}}\{\widetilde{p}/\widetilde{X}\} \stackrel{a}{\longrightarrow} r_1 \wedge r_2}\]
      and $r \equiv r_1 \wedge r_2$.
       Then by IH,  there exist $B_{\widetilde{X}}'$, $B_{\widetilde{X},\widetilde{Y}}'$ and $B_{\widetilde{X},\widetilde{Y}}''$ with $\widetilde{X} \cap \widetilde{Y} = \emptyset$ and, $D_{\widetilde{X}}'$, $D_{\widetilde{X},\widetilde{Z}}'$ and $D_{\widetilde{X},\widetilde{Z}}''$ with $\widetilde{X} \cap \widetilde{Z} = \emptyset$  that realize (CP-$a$-1,2,3,4) for two $a$-labelled transitions involving in premises respectively.
    W.l.o.g, we may assume $\widetilde{Y} \cap \widetilde{Z} = \emptyset$.
     Then it is straightforward to verify that, for the $a$-labelled transition $C_{\widetilde{X}}\{\widetilde{p}/\widetilde{X}\} \stackrel{a}{\longrightarrow} r$, the contexts
      $C_{\widetilde{X}}' \triangleq B_{\widetilde{X}}' \wedge D_{\widetilde{X}}'$, $C_{\widetilde{X},\widetilde{V}}' \triangleq B_{\widetilde{X},\widetilde{Y}}' \wedge D_{\widetilde{X},\widetilde{Z}}'$ and $C_{\widetilde{X},\widetilde{V}}'' \triangleq B_{\widetilde{X},\widetilde{Y}}'' \wedge D_{\widetilde{X},\widetilde{Z}}''$ with $\widetilde{V} = \widetilde{Y} \cup \widetilde{Z}$ realize (CP-$a$-1) -- (CP-$a$-4), as desired.\\

  \noindent Case 6 $C_{\widetilde{X}} \equiv B_{\widetilde{X}} \parallel_A D_{\widetilde{X}}$.

   Then the last rule applied in the proof tree is one of the following:
   \begin{enumerate}
     \item [(6.1)]\  $\;\;\frac{B_{\widetilde{X}}\{\widetilde{p}/\widetilde{X}\} \stackrel{a}{\longrightarrow} r_1, D_{\widetilde{X}}\{\widetilde{p}/\widetilde{X}\} \stackrel{a}{\longrightarrow} r_2}{B_{\widetilde{X}}\{\widetilde{p}/\widetilde{X}\} \parallel_A  D_{\widetilde{X}}\{\widetilde{p}/\widetilde{X}\} \stackrel{a}{\longrightarrow} r_1 \parallel_A r_2}\;\text{with}\;a \in A$;
     \item [(6.2)] \  $\;\;\frac{B_{\widetilde{X}}\{\widetilde{p}/\widetilde{X}\} \stackrel{a}{\longrightarrow} r',\;D_{\widetilde{X}}\{\widetilde{p}/\widetilde{X}\} \not\stackrel{\tau}{\longrightarrow} }{B_{\widetilde{X}}\{\widetilde{p}/\widetilde{X}\} \parallel_A  D_{\widetilde{X}}\{\widetilde{p}/\widetilde{X}\} \stackrel{a}{\longrightarrow} r' \parallel_A D_{\widetilde{X}}\{\widetilde{p}/\widetilde{X}\}}\;\text{with}\;a \notin A$;
     \item [(6.3)] \  $\;\;\frac{D_{\widetilde{X}}\{\widetilde{p}/\widetilde{X}\} \stackrel{a}{\longrightarrow} r',\; B_{\widetilde{X}}\{\widetilde{p}/\widetilde{X}\} \not\stackrel{\tau}{\longrightarrow} }{B_{\widetilde{X}}\{\widetilde{p}/\widetilde{X}\} \parallel_A  D_{\widetilde{X}}\{\widetilde{p}/\widetilde{X}\} \stackrel{a}{\longrightarrow} B_{\widetilde{X}}\{\widetilde{p}/\widetilde{X}\} \parallel_A r'}\;\text{with}\;a \notin A$.
   \end{enumerate}
   Among them, the argument for (6.1) is similar to one for Case~5.
   We shall consider the case (6.2), and (6.3) may be handled similarly.
   In this situation, $r \equiv r' \parallel_A D_{\widetilde{X}}\{\widetilde{p}/\widetilde{X}\}$.
   Moreover, for the $a$-labelled transition $B_{\widetilde{X}}\{\widetilde{p}/\widetilde{X}\} \stackrel{a}{\longrightarrow} r'$, by IH, there exist $B_{\widetilde{X}}'$, $B_{\widetilde{X},\widetilde{Y}}'$ and $B_{\widetilde{X},\widetilde{Y}}''$ with $\widetilde{X} \cap \widetilde{Y} = \emptyset$ that satisfy (CP-$a$-1) -- (CP-$a$-4).
   Put
   \[C_{\widetilde{X}}' \triangleq B_{\widetilde{X}}' \parallel_A D_{\widetilde{X}}, C_{\widetilde{X},\widetilde{Y}}' \triangleq B_{\widetilde{X},\widetilde{Y}}' \parallel_A D_{\widetilde{X}}\;\text{and}\; C_{\widetilde{X},\widetilde{Y}}'' \triangleq B_{\widetilde{X},\widetilde{Y}}'' \parallel_A D_{\widetilde{X}}.\]
   Next we want to show that these contexts realize (CP-$a$-1) -- (CP-$a$-4).

    \textbf{(CP-$a$-1)} It is obvious because of $B_{\widetilde{X}}  \Rrightarrow B_{\widetilde{X}}' $.

    \textbf{(CP-$a$-2)} For each $Y \in \widetilde{Y}$, since $Y$ is 1-active in  $B_{\widetilde{X},\widetilde{Y}}''$ and $B_{\widetilde{X},\widetilde{Y}}'$, so it is in  $C_{\widetilde{X},\widetilde{Y}}''$ and $C_{\widetilde{X},\widetilde{Y}}'$ because of $\widetilde{X} \cap \widetilde{Y} = \emptyset$.

    \textbf{(CP-$a$-3)} By IH, there exist $i_Y \leq |\widetilde{X}|(Y\in \widetilde{Y})$ which realize subclauses (i)(ii)(iii) in (CP-$a$-3).
    In the following, we will verify that these $i_Y$ also work well for the induction step.
    Clearly, it follows from $B_{\widetilde{X},\widetilde{Y}}'\{\widetilde{X_{i_Y}}/\widetilde{Y}\}\equiv B_{\widetilde{X}}'$ and $\widetilde{X} \cap \widetilde{Y} = \emptyset$ that $C_{\widetilde{X},\widetilde{Y}}'\{\widetilde{X_{i_Y}}/\widetilde{Y}\}\equiv C_{\widetilde{X}}'$.
    Hence these $i_Y$ satisfy the subclause (CP-$a$-3-i) for the induction step.
    Moreover, due to $r' \equiv  B_{\widetilde{X},\widetilde{Y}}''\{\widetilde{p}/\widetilde{X},\widetilde{p_{Y}'}/\widetilde{Y}\}$ for some $p_Y'(Y\in \widetilde{Y})$ with $p_{i_Y} \stackrel{a}{\longrightarrow}p_Y'$ for each $Y \in \widetilde{Y}$ and $\widetilde{X} \cap \widetilde{Y} = \emptyset$, we have $r \equiv  r' \parallel_A D_{\widetilde{X}}\{\widetilde{p}/\widetilde{X}\} \equiv C_{\widetilde{X},\widetilde{Y}}''\{\widetilde{p}/\widetilde{X},\widetilde{p_{Y}'}/\widetilde{Y}\}$, that is, they realize (CP-$a$-3-ii) for the induction step.
    Finally, to verify that
     these $i_Y$ also meet the challenge of (CP-$a$-3-iii), we assume that $\widetilde{q}$ and $\widetilde{q'}$ be any tuple such that $|\widetilde{q}|=|\widetilde{X}|$, $q_{i_Y} \stackrel{a}{\longrightarrow} q_{Y}'$ for each $Y \in  \widetilde{Y}$ and $C_{\widetilde{X}}\{\widetilde{q}/\widetilde{X}\}$ is stable.
    So, $B_{\widetilde{X}}\{\widetilde{q}/\widetilde{X}\}$ and $D_{\widetilde{X}}\{\widetilde{q}/\widetilde{X}\}$ are stable.
    Further, since $B_{\widetilde{X},\widetilde{Y}}''$ satisfies (CP-$a$-3-iii) and $a \notin A$, it is easy to obtain that $C_{\widetilde{X}}\{\widetilde{q}/\widetilde{X}\} \stackrel{a}{\longrightarrow} C_{\widetilde{X},\widetilde{Y}}''\{\widetilde{q}/\widetilde{X},\widetilde{q_{Y}'}/\widetilde{Y}\}$.

    \textbf{(CP-$a$-4)} All subclauses immediately follow from IH and constructions of $C_{\widetilde{X}}'$, $C_{\widetilde{X},\widetilde{Y}}'$ and $C_{\widetilde{X},\widetilde{Y}}''$.\\

  \noindent Case 7 $C_{\widetilde{X}} \equiv \langle Y|E \rangle$.

    Clearly,  the last rule applied in the inference is
   \[\frac{\langle t_Y|E \rangle \{\widetilde{p}/\widetilde{X}\} \stackrel{a}{\longrightarrow} r}{\langle Y|E \rangle\{\widetilde{p}/\widetilde{X}\} \stackrel{a}{\longrightarrow} r} \;\text{with}\; Y =t_Y \in E.\]
   For the transition $\langle t_Y|E \rangle \{\widetilde{p}/\widetilde{X}\} \stackrel{a}{\longrightarrow} r$, by IH, there exist $C_{\widetilde{X}}'$, $C_{\widetilde{X},\widetilde{Y}}'$ and $C_{\widetilde{X},\widetilde{Y}}''$ with $\widetilde{X} \cap \widetilde{Y} = \emptyset$ that satisfy (CP-$a$-1) -- (CP-$a$-4).
    It is trivial to check that these contexts are what we need.
\end{proof}

Clearly, whenever all free variables occurring in $C_{\widetilde{X}}$ are guarded, any action labelled transition starting from $C_{\widetilde{X}}\{\widetilde{p}/\widetilde{X}\}$  must be performed by $C_{\widetilde{X}}$ itself.

\begin{lemma}\label{L:ONE_ACTION_VISIBLE_GUARDED}
Let $C_{\widetilde{X}}$ be a context such that $X$ is guarded for each $X \in \widetilde{X}$.
 If $C_{\widetilde{X}}\{\widetilde{p}/\widetilde{X}\} \stackrel{\alpha}{\longrightarrow} r$ then there exists $B_{\widetilde{X}}$ such that $r \equiv B_{\widetilde{X}}\{\widetilde{p}/\widetilde{X}\}$ and $C_{\widetilde{X}}\{\widetilde{q}/\widetilde{X}\} \stackrel{\alpha}{\longrightarrow}  B_{\widetilde{X}}\{\widetilde{q}/\widetilde{X}\}$ for any $\widetilde{q}$.
\end{lemma}
\begin{proof}
Firstly, we handle the case $\alpha =\tau$.
For the transition $C_{\widetilde{X}}\{\widetilde{p}/\widetilde{X}\} \stackrel{\tau}{\longrightarrow} r$, either the clause (1) or (2)  in Lemma~\ref{L:ONE_ACTION_TAU} holds.
It is a simple matter to see that the clause (1) implies what we desire.
The task is now to show that the clause (2) does not hold for such transition.
On the contrary, assume that the clause (2) holds.
Then there exist $C_{\widetilde{X}}'$, $C_{\widetilde{X},Z}''$ and $i_0 \leq |\widetilde{X}|$ satisfying (P-$\tau$-1,2,3,4).
For each $X \in \widetilde{X}$, since it is guarded in $C_{\widetilde{X}}$, by Lemma~\ref{L:ONE_STEP_UNFOLDING_VARIABLE}(3) and (P-$\tau$-1), so it is in $C_{\widetilde{X}}'$.
Hence a contradiction arises due to (P-$\tau$-3), as desired.

Next we treat the other case $\alpha \in Act$.
By Lemma~\ref{L:ONE_ACTION_VISIBLE}, for the transition $C_{\widetilde{X}}\{\widetilde{p}/\widetilde{X}\} \stackrel{\alpha}{\longrightarrow} r$, there exist $C_{\widetilde{X}}'$, $C_{\widetilde{X},\widetilde{Y}}'$ and $C_{\widetilde{X},\widetilde{Y}}''$  realizing (CP-$a$-1) -- (CP-$a$-4).
Clearly, if $\widetilde{Y} = \emptyset$ then $C_{\widetilde{X},\widetilde{Y}}''$ is exactly one that we need.
Thus, to complete the proof, it suffices to show that $\widetilde{Y}$ is indeed empty.
Since $X$ is guarded in $C_{\widetilde{X}}$ for each $X \in \widetilde{X}$ and $C_{\widetilde{X}} \Rrightarrow C_{\widetilde{X}}'$ (i.e., (CP-$a$-1)), by Lemma~\ref{L:ONE_STEP_UNFOLDING_VARIABLE}(3), all occurrences of free variables in $C_{\widetilde{X}}'$ are guarded.
Moreover, since  $C_{\widetilde{X},\widetilde{Y}}'$  satisfies (CP-$a$-2) and (CP-$a$-3-i), we get $\widetilde{Y} =\emptyset$, as desired.
\end{proof}

%
%

\begin{lemma}\label{L:ONE_STEP_UNFOLDING_TAU_ACTION}
 For any $Y,E$  with $Y = t_Y \in E$ and context $C_{X}$ with at most one occurrence of the unfolded variable $X$, we have

\noindent   (1) if $C_{X}\{\langle Y|E \rangle/ X\} \stackrel{\alpha}{\longrightarrow}  q$ then there exists $B_{X}$ such that
      \begin{enumerate}[({1.}1)]
      \renewcommand{\theenumi}{(1.\arabic{enumi})}
      \item \  $q \equiv B_{X}\{\langle Y|E \rangle/ X\}$,
       \item \  $C_{X}\{\langle t_Y|E \rangle/ X\} \stackrel{\alpha}{\longrightarrow} B_{X}\{\langle t_Y|E \rangle/ X\}$,  and
       \item \   $X$ occurs in $B_{X}$ at most once; moreover, such occurrence is unfolded;
       \end{enumerate}

\noindent   (2) if $C_{X}\{\langle t_Y|E \rangle/ X\} \stackrel{\alpha}{\longrightarrow}  q$ then there exists $B_{X}$ such that
        \begin{enumerate}[({2.}1)]
        \renewcommand{\theenumi}{(2.\arabic{enumi})}
      \item \  $q \equiv B_{X}\{\langle t_Y|E \rangle/ X\}$,
       \item \  $C_{X}\{\langle Y|E \rangle/ X\} \stackrel{\alpha}{\longrightarrow} B_{X}\{\langle Y|E \rangle/ X\}$, and
       \item \  $X$ occurs in $B_{X}$ at most once; moreover, such occurrence is unfolded.
       \end{enumerate}
\end{lemma}
\begin{proof}
    We prove only item (1), and the same reasoning applies to item (2).
    For (1), the argument is splitted into two parts based on $\alpha$.\\

\noindent Case 1 $\alpha = \tau$.

    Assume $C_{X}\{\langle Y|E \rangle/ X\} \stackrel{\tau}{\longrightarrow}  q$.
    Then, for such transition, by Lemma~\ref{L:ONE_ACTION_TAU}, either there exists $C_X'$ that satisfies (C-$\tau$-1,2,3) or there exist $C_X'$, $C_{X,Z}''$ with $Z \neq X$ that satisfy (P-$\tau$-1,2,3,4).

        For the first alternative, it follows from $C_X'$ satisfies (C-$\tau$-1,2) that $q \equiv C_{X}'\{\langle Y|E \rangle/ X\}$ and $C_X\{\langle t_Y|E \rangle/ X\} \stackrel{\tau}{\longrightarrow} C_X'\{\langle t_Y|E \rangle/ X\}$.
        Moreover, due to (C-$\tau$-3-ii), there is at most one occurrence of the unfolded variable $X$ in $C_X'$.
        Consequently, the context $C_X'$ is exactly one that we seek.

        For the second alternative, by (P-$\tau$-2), there exists $q'$ such that
        \[\langle Y|E \rangle \stackrel{\tau}{\longrightarrow}  q'\;\text{and}\;q \equiv C_{X,Z}''\{\langle Y|E \rangle/ X,q'/Z\}.\]
        Hence $\langle t_Y|E \rangle \stackrel{\tau}{\longrightarrow}  q'$.
        Then it follows from (P-$\tau$-4) that
        \[C_X\{\langle t_Y|E \rangle/ X\} \stackrel{\tau}{\longrightarrow} C_{X,Z}''\{\langle t_Y|E \rangle/ X,q'/Z\}.\]
        In addition, due to (P-$\tau$-1), by Lemma~\ref{L:ONE_STEP_UNFOLDING_VARIABLE}(1), there is at most one occurrence of the unfolded variable $X$ in $C_X'$.
        Moreover, since $C_X'$ and $C_{X,Z}''$ satisfy (P-$\tau$-3), we obtain $X \notin FV(C_{X,Z}'')$.
        Hence $q \equiv C_{X,Z}''\{\langle Y|E \rangle/ X,q'/Z\} \equiv C_{X,Z}''\{\langle t_Y|E \rangle/ X,q'/Z\}$.
        Then it is easy to see that $B_X \triangleq q$ is  what we need.\\
        
\noindent Case 2 $\alpha \in Act$.
        
  Let $C_{X}\{\langle Y|E \rangle/ X\} \stackrel{\alpha}{\longrightarrow}  q$.
Then, by Lemma~\ref{L:ONE_ACTION_VISIBLE}, there exist $C_X'$, $C_{X,\widetilde{Z}}'$ and $C_{X,\widetilde{Z}}''$ with $X \notin \widetilde{Z}$ that satisfy (CP-$a$-1) -- (CP-$a$-4).
Since $C_X\{\langle Y|E \rangle/X\}$ is stable, by item (2),
so is $C_X\{\langle t_Y|E \rangle/X\}$.

 If $\widetilde{Z} = \emptyset$, it follows trivially by (CP-$a$-3-iii) that $C_X\{\langle t_Y|E \rangle/ X\} \stackrel{\alpha}{\longrightarrow} C_{X,\widetilde{Z}}''\{\langle t_Y|E \rangle/ X \}$;
 moreover, by (CP-$a$-1), (CP-$a$-3-i), (CP-$a$-4-i) and Lemma~\ref{L:ONE_STEP_UNFOLDING_VARIABLE}(1), there is at most one occurrence of the unfolded variable $X$ in $C_{X,\widetilde{Z}}''$.
 Therefore, $C_{X,\widetilde{Z}}''$ is exactly the context that we need.

        We next deal with the other case $\widetilde{Z} \neq \emptyset$.
        Since $C_X'$ satisfies (CP-$a$-1), by Lemma~\ref{L:ONE_STEP_UNFOLDING_VARIABLE}(1), there is at most one occurrence of the unfolded variable $X$ in $C_{X}'$.
        Then it follows from  (CP-$a$-2), (CP-$a$-3-i) and (CP-$a$-4-i) that $|\widetilde{Z}|=1$ and $X \notin FV(C_{X,\widetilde{Z}}'')$.
        So, due to (CP-$a$-3-ii), there exists $q'$ such that
        \[\langle Y|E \rangle \stackrel{\alpha}{\longrightarrow}q'\;\text{and}\;q \equiv C_{X,\widetilde{Z}}''\{\langle Y|E \rangle /X,q'/\widetilde{Z}\}.\]
        Hence $\langle t_Y|E \rangle \stackrel{\alpha}{\longrightarrow}q'$.
        Then  $C_X\{\langle t_Y|E\rangle /X\} \stackrel{\alpha}{\longrightarrow} C_{X,\widetilde{Z}}''\{\langle t_Y|E \rangle /X,q'/\widetilde{Z}\}$  by (CP-$a$-3-iii) and $C_X\{\langle t_Y|E\rangle /X\} \not\stackrel{\tau}{\longrightarrow}$.
        Thus $q \equiv C_{X,\widetilde{Z}}''\{\langle t_Y|E \rangle /X,q'/\widetilde{Z}\}$ because of $X \notin FV(C_{X,\widetilde{Z}}'')$.
        Then it is easy to check that  $B_X \triangleq q$ is exactly what we seek.
\end{proof}

\subsection{Multi-$\tau$ transitions and more on unfolding}

Based on the result obtained in the preceding subsections, we shall give a few further properties of unfolding.
We first want to indicate some simple properties.

\begin{lemma}\label{L:MULTI_STEP_UNFOLDING_ACTION}
  The relation $ \Rrightarrow$ satisfies the forward and backward conditions, that is, for any $\alpha \in Act_{\tau}$ and $p,q$ such that $p \Rrightarrow  q$, we have
\begin{enumerate}[(1)]
\renewcommand{\theenumi}{(\arabic{enumi})}
  \item if $p \stackrel{\alpha}{\longrightarrow}  p'$ then $q \stackrel{\alpha}{\longrightarrow}  q'$ and $p' \Rrightarrow  q'$ for some $q' $;
  \item if $q \stackrel{\alpha}{\longrightarrow}  q'$ then $p \stackrel{\alpha}{\longrightarrow}  p'$ and $p' \Rrightarrow  q'$ for some $p' $.
\end{enumerate}
\end{lemma}
\begin{proof}
\noindent \textbf{(1)} Assume $p \Rrightarrow  q$ and $p \stackrel{\alpha}{\longrightarrow}  p'$.
Clearly, $p \Rrightarrow_n  q$ for some $n$.
It proceeds by induction on $n$.
For the induction base $n = 0$, it holds trivially.
For the induction step $n = k+1$, we have $p \Rrightarrow_{k}  r \Rrightarrow_{1}  q$ for some $r$.
By IH, $r \stackrel{\alpha}{\longrightarrow}  r'$ and $p' \Rrightarrow  r'$ for some $r'$.
Moreover, for $r \Rrightarrow_{1}  q$, by Lemma~\ref{L:ONE_STEP_UNFOLDING_SYNTAX}, \ref{L:ONE_STEP_UNFOLDING_TAU_ACTION}(1) and \ref{L:ONE_STEP_UNFOLDING_VISIBLE_ACTION}(1), there exists $q'$ such that $q \stackrel{\alpha}{\longrightarrow}  q'$ and $r'  \Rrightarrow   q'$.
Obviously, we also have $p'  \Rrightarrow   q'$. \\

\noindent \textbf{(2)} Similar to item (1), but applying Lemmas~\ref{L:ONE_STEP_UNFOLDING_TAU_ACTION}(2) and \ref{L:ONE_STEP_UNFOLDING_VISIBLE_ACTION}(2) instead of Lemmas~\ref{L:ONE_STEP_UNFOLDING_TAU_ACTION}(1) and \ref{L:ONE_STEP_UNFOLDING_VISIBLE_ACTION}(1).
\end{proof}

A similar result also holds w.r.t $\stackrel{\epsilon}{\Longrightarrow}|$, that is

\begin{lemma}\label{L:MULTI_STEP_UNFOLDING_ACTION_II}
   For any $p,q$ such that $p \Rrightarrow  q$, we have
   \begin{enumerate}[(1)]
\renewcommand{\theenumi}{(\arabic{enumi})}
  \item if $p \stackrel{\epsilon}{\Longrightarrow}  |p'$ then $q \stackrel{\epsilon}{\Longrightarrow} | q'$ and $p' \Rrightarrow  q'$ for some $q' $;
  \item if $q \stackrel{\epsilon}{\Longrightarrow} | q'$ then $p \stackrel{\epsilon}{\Longrightarrow} | q'$ and $p' \Rrightarrow  q'$ for some $p' $.
\end{enumerate}
\end{lemma}
\begin{proof}
 By applying Lemma~\ref{L:MULTI_STEP_UNFOLDING_ACTION} finitely often.
\end{proof}

In fact, for any $p,q$ such that $p \Rrightarrow q$, it is to be expected that $p =_{RS} q$.
To verify it, we need to prove that $p \in F$ if and only if $q \in F$.
The next lemma will serve as a stepping stone in proving this.

\begin{convention}
  The arguments in the remainder of this paper often proceed by distinguishing some cases based on the last rule applied in an inference. For such argument, since rules about operations $\wedge$, $\vee$, $\parallel_A$ and $\Box$ are symmetric w.r.t their two operands (for instance, Rules $Rp_{11}$ and $Rp_{12}$, $Ra_4$ and $Ra_5$, and so on), we shall consider only one of two symmetric rules and omit another one.
\end{convention}

\begin{lemma}\label{L:ONE_STEP_UNFOLDING_FAILURE}
 For any $Y,E$  with $Y = t_Y \in E$ and context $C_{X}$ with at most one occurrence of the unfolded variable $X$,
    $C_{X}\{\langle Y|E \rangle/ X\} \in F  $ iff $C_{X}\{\langle t_Y|E \rangle/ X\} \in F  $.
\end{lemma}
\begin{proof}
We give proof only for the implication from left to right, the converse implication may be proved similarly and omitted.
  Assume  $C_{X}\{\langle Y|E \rangle/ X\} \in F $.
  It proceeds by induction on the depth of the inference  $Strip(\mathcal{P}_{\text{CLL}_R} ,M_{\text{CLL}_R} ) \vdash C_{X}\{\langle Y|E \rangle/ X\}F$, which is a routine case analysis on the form of $C_X$.
  We give the proof only for the case  $C_X \equiv B_X \wedge D_X$, the other cases are left to the reader.
 In this situation, the last rule applied in the inference is one of the following.\\

\noindent Case 1 $\frac{B_X \{\langle Y|E \rangle/X\}F}{B_X\{\langle Y|E \rangle/X\} \wedge D_X\{\langle Y|E \rangle/X\}F}$.

By IH, we get $B_X \{\langle t_Y|E \rangle/X\} \in F $.
Hence $C_X \{\langle t_Y|E \rangle/X\} \in  F $.\\

\noindent  Case 2 $\frac{B_X \{\langle Y|E \rangle/X\} \stackrel{a}{\longrightarrow} r,D_X \{\langle Y|E \rangle/X\} \not\stackrel{a}{\longrightarrow}, C_X \{\langle Y|E \rangle/X\} \not\stackrel{\tau}{\longrightarrow}}{B_X\{\langle Y|E \rangle/X\}\wedge D_X\{\langle Y|E \rangle/X\}F}$.


 By Lemma~\ref{L:ONE_STEP_UNFOLDING_TAU_ACTION}  and \ref{L:ONE_STEP_UNFOLDING_VISIBLE_ACTION}, we have $B_X \{\langle t_Y|E \rangle/X\} \stackrel{a}{\longrightarrow}  $, $D_X \{\langle t_Y|E \rangle/X\} \not\stackrel{a}{\longrightarrow} $ and $C_X \{\langle t_Y|E \rangle/X\} \not\stackrel{\tau}{\longrightarrow} $. So, $C_X\{\langle t_Y|E \rangle /X\} \in F $.\\

\noindent Case 3 $\frac{C_X \{\langle Y|E \rangle/X\} \stackrel{\alpha}{\longrightarrow}s, \{rF: C_X \{\langle Y|E \rangle/X\} \stackrel{\alpha}{\longrightarrow} r \}}{C_X\{\langle Y|E \rangle/X\}F}$.

Then $C_X \{\langle t_Y|E \rangle/X\} \stackrel{\alpha}{\longrightarrow} $ by Lemma~\ref{L:ONE_STEP_UNFOLDING_TAU_ACTION}(1) and \ref{L:ONE_STEP_UNFOLDING_VISIBLE_ACTION}(1).
Assume  $C_X \{\langle t_Y|E \rangle/X\} \stackrel{\alpha}{\longrightarrow} q$.
Thus it follows from Lemma~\ref{L:ONE_STEP_UNFOLDING_TAU_ACTION}(2) and \ref{L:ONE_STEP_UNFOLDING_VISIBLE_ACTION}(2) that there exists $C_X'$ with at most one occurrence of the unfolded variable $X$ such that
\[C_X \{\langle Y|E \rangle/X\} \stackrel{\alpha}{\longrightarrow} C_X'\{\langle Y|E \rangle/X\}\;\text{and}\; q \equiv C_X'\{\langle t_Y|E \rangle/X\}.\]
Then, by IH, $q \equiv C_X'\{\langle t_Y|E \rangle/X\} \in F $.
Hence $C_X\{\langle t_Y|E \rangle/X\} \in F $ by Theorem~\ref{L:LLTS}.\\

\noindent Case 4 $\frac{\{rF: C_X \{\langle Y|E \rangle/X\} \stackrel{\epsilon}{\Longrightarrow}| r \}}{C_X\{\langle Y|E \rangle/X\}F}$.

Assume $C_X \{\langle t_Y|E \rangle/X\} \stackrel{\epsilon}{\Longrightarrow}  |t$.
Repeated application of Lemma~\ref{L:ONE_STEP_UNFOLDING_TAU_ACTION}(2) enables us to get $C_X \{\langle Y|E \rangle/X\} \stackrel{\epsilon}{\Longrightarrow} | r$, $r \equiv C_X'\{\langle Y|E \rangle/X\}$ and $t \equiv C_X'\{\langle t_Y|E \rangle/X\}$ for some $r$ and context $C_X'$ with at most one occurrence of the unfolded free variable $X$.
Since $r  \equiv C_X'\{\langle Y|E \rangle/X\} \in F  $, we have $t \in F  $ by IH.
Then $C_X \{\langle t_Y|E \rangle/X\} \in F  $ by Theorem~\ref{L:LLTS}.
%
\end{proof}

Next we can show that the relation $\Rrightarrow$ preserves and respects the inconsistency.

\begin{lemma}\label{L:MULTI_STEP_UNFOLDING_FAILURE}
 For any $p,q$, if $p  \Rrightarrow q$, then $p \in F  $ iff $q \in F  $.
\end{lemma}
\begin{proof}
Suppose $p \Rrightarrow  q$.
Hence $p \Rrightarrow_n  q$ for some $n$.
Then, using Lemma~\ref{L:ONE_STEP_UNFOLDING_SYNTAX}  and \ref{L:ONE_STEP_UNFOLDING_FAILURE}, the proof is straightforward by induction on $n$.
\end{proof}

We now have the assertion below of the equivalence of $p$ and $q$ modulo $=_{RS}$ whenever $p \Rrightarrow q$.

\begin{lemma}\label{L:MULTI_STEP_UNFOLDING_IMPLIES_RS}
  If $p_1 \Rrightarrow p_2$ then $p_1 =_{RS} p_2$, in particular, $p_1 \approx_{RS} p_2$ whenever $p_1 \not\stackrel{\tau}{\longrightarrow}$.
\end{lemma}
\begin{proof}
  We only prove $p_1 \underset{\thicksim}{\sqsubset}_{RS} p_2$ whenever $p_1 \not\stackrel{\tau}{\longrightarrow}$, other proofs are straightforward and omitted.
  Set
  \[
  {\mathcal R} = \{(p,q): p \Rrightarrow q\;\text{and}\;p \not\stackrel{\tau}{\longrightarrow}\}.
  \]
  It suffices to prove that $\mathcal R$ is a stable ready simulation relation.
  Suppose $(p,q) \in {\mathcal R}$.
  Then, by Lemma~\ref{L:MULTI_STEP_UNFOLDING_ACTION} and \ref{L:MULTI_STEP_UNFOLDING_FAILURE}, it is evident that such pair satisfies (RS1), (RS2) and (RS4).
  For (RS3),
%
   suppose $p \stackrel{a}{\Longrightarrow}_F| p'$. Then $p \stackrel{a}{\longrightarrow}_F p'' \stackrel{\epsilon}{\Longrightarrow}_F| p'$ for some $p''$.
  By Lemma~\ref{L:MULTI_STEP_UNFOLDING_ACTION} and \ref{L:MULTI_STEP_UNFOLDING_FAILURE}, there exists $q''$ such that $q \stackrel{a}{\longrightarrow}_F q''$ and $p'' \Rrightarrow q''$.
  Further, by Lemma~\ref{L:FAILURE_TAU_I}, \ref{L:MULTI_STEP_UNFOLDING_ACTION_II} and \ref{L:MULTI_STEP_UNFOLDING_FAILURE}, $p' \Rrightarrow q'$ and $q'' \stackrel{\epsilon}{\Longrightarrow}_F|q'$ for some $q'$.
  Moreover, $(p',q') \in {\mathcal R}$, as desired.
%
%
\end{proof}

In the following, we shall generalize Lemma~\ref{L:ONE_ACTION_TAU} to the situation involving a sequence of $\tau$-labelled transitions.
Given a process $C_{\widetilde{X} }\{\widetilde{p}/\widetilde{X} \}$, by Lemma~\ref{L:ONE_ACTION_TAU}, any $\tau$-transition starting from $C_{\widetilde{X} }\{\widetilde{p}/\widetilde{X} \}$ may be caused by $C_{\widetilde{X}}$ itself or some $p_i$.
Thus, for a sequence of $\tau$-transitions, these two situations may occur alternately.
Based on Lemma~\ref{L:ONE_ACTION_TAU}, we can capture this as follows.

\begin{lemma}\label{L:MULTI_TAU_GF_STABLE}
  For any $C_{\widetilde{X} }$ and $\widetilde{p}$,
  if $C_{\widetilde{X} }\{\widetilde{p}/\widetilde{X} \} \stackrel{\epsilon}{\Longrightarrow} r$ then
  there exist $C_{\widetilde{X},\widetilde{Y}}'$ and $i_Y\leq |\widetilde{X}|,p_Y'(Y \in \widetilde{Y})$ such that

\noindent \textbf{(MS-$\tau$-1)} $\widetilde{X} \cap \widetilde{Y} = \emptyset$ and $Y$ is 1-active in $C_{\widetilde{X},\widetilde{Y} }'$ for each $Y \in \widetilde{Y}$;

\noindent \textbf{(MS-$\tau$-2)}   $p_{i_Y}\stackrel{\tau}{\Longrightarrow}  p_{Y}'$ for each $Y \in \widetilde{Y}$ and $r \equiv C_{\widetilde{X},\widetilde{Y} }'\{\widetilde{p}/\widetilde{X},\widetilde{p_{Y}'}/\widetilde{Y} \}$;

\noindent \textbf{(MS-$\tau$-3)}  for any $\widetilde{q} $ and $\widetilde{q_Y'}$ with $|\widetilde{q}| = |\widetilde{X}|$ and $Y \in \widetilde{Y}$,

  \textbf{(MS-$\tau$-3-i)} if $q_{i_Y} \stackrel{\epsilon}{\Longrightarrow} q_{Y}'$ for each $Y \in \widetilde{Y}$ then
    $C_{\widetilde{X} }\{\widetilde{q}/\widetilde{X} \} \stackrel{\epsilon}{\Longrightarrow}
      \Rrightarrow  C_{\widetilde{X},\widetilde{Y} }'\{\widetilde{q}/\widetilde{X},\widetilde{q_{Y}'}/\widetilde{Y} \}$;

    \textbf{(MS-$\tau$-3-ii)}  if $q_{i_Y} \stackrel{\tau}{\Longrightarrow}q_{Y}'$  for each $Y \in \widetilde{Y}$ then $C_{\widetilde{X} }\{\widetilde{q}/\widetilde{X} \} \stackrel{\epsilon}{\Longrightarrow}  C_{\widetilde{X},\widetilde{Y} }'\{\widetilde{q}/\widetilde{X},\widetilde{q_{Y}'}/\widetilde{Y} \}$;

\noindent \textbf{(MS-$\tau$-4)}  if $C_{\widetilde{X} }$ is stable then so is $C_{\widetilde{X},\widetilde{Y} }'$   and
      $C_{\widetilde{X}}\{\widetilde{q}/\widetilde{X} \}  \Rrightarrow C_{\widetilde{X},\widetilde{Y} }'\{\widetilde{q}/\widetilde{X},\widetilde{q_{i_Y}}/\widetilde{Y}\}$ for any $\widetilde{q} $;

\noindent \textbf{(MS-$\tau$-5)}  for each $ X \in \widetilde{X}$, if $X$ is strongly guarded in $C_{\widetilde{X}}$ then so it is in $C_{\widetilde{X},\widetilde{Y} }'$ and $X \not\equiv X_{i_Y}$ for each $Y \in \widetilde{Y}$;

\noindent \textbf{(MS-$\tau$-6)}  for each $X\in \widetilde{X}$ (or, $Y\in \widetilde{Y}$), if $X$ ($X_{i_Y}$ respectively) does not occur in the scope of any conjunction in $C_{\widetilde{X}}$ then neither does $X$ ($Y$ respectively) in $C_{\widetilde{X},\widetilde{Y} }'$;

\noindent \textbf{(MS-$\tau$-7)}  if $r$ is stable then so are $C_{\widetilde{X},\widetilde{Y}}'$ and $p_{Y}'$ for each $Y \in \widetilde{Y}$.
\end{lemma}
\begin{proof}
  Suppose $C_{\widetilde{X}}\{\widetilde{p}/\widetilde{X}\} (\stackrel{\tau}{\longrightarrow} )^n r(n \geq 0)$.
  We proceed by induction on $n$.
  For the inductive base $n=0$, the conclusion holds trivially by taking $C_{\widetilde{X},\widetilde{Y}}' \triangleq C_{\widetilde{X}}$ with $\widetilde{Y} = \emptyset$.

  For the inductive step, assume 
    $C_{\widetilde{X}}\{\widetilde{p}/\widetilde{X}\}  (\stackrel{\tau}{\longrightarrow} )^{k} s \stackrel{\tau}{\longrightarrow}  r$ for some $s$.
    For the transition $C_{\widetilde{X}}\{\widetilde{p}/\widetilde{X}\}  (\stackrel{\tau}{\longrightarrow} )^{k} s$, by IH, there exist $C_{\widetilde{X},\widetilde{Y}}'$ and $i_Y \leq |\widetilde{X}|,p_Y'(Y \in \widetilde{Y})$ that realize (MS-$\tau$-$l$) ($1 \leq l \leq 7$).
    In particular, we have $s \equiv C_{\widetilde{X},\widetilde{Y}}'\{\widetilde{p}/\widetilde{X}, \widetilde{p_{Y}'}/\widetilde{Y}\}$ due to (MS-$\tau$-2).
    Then, for the transition $s \stackrel{\tau}{\longrightarrow}r$, either the clause (1) or (2) in Lemma~\ref{L:ONE_ACTION_TAU} holds. The argument splits into two cases.\\

\noindent Case 1
    For the transition $s \stackrel{\tau}{\longrightarrow}r$, the clause (1) in Lemma~\ref{L:ONE_ACTION_TAU} holds.

    That is, for the transition $C_{\widetilde{X},\widetilde{Y}}'\{\widetilde{p}/\widetilde{X}, \widetilde{p_{Y}'}/\widetilde{Y}\} \equiv s \stackrel{\tau}{\longrightarrow}r$, there exists $C_{\widetilde{X},\widetilde{Y}}''$ satisfying (C-$\tau$-1,2,3) in Lemma~\ref{L:ONE_ACTION_TAU}.
    We shall check that $C_{\widetilde{X},\widetilde{Y}}''$, $\widetilde{i_Y}$ and $\widetilde{p_Y'}$ realize (MS-$\tau$-1) - (MS-$\tau$-7) w.r.t $C_{\widetilde{X}}\{\widetilde{p}/\widetilde{X}\} (\stackrel{\tau}{\longrightarrow} )^{k+1} r$.

    Since $C_{\widetilde{X},\widetilde{Y}}'$ satisfies (MS-$\tau$-1,5,6), it follows that $C_{\widetilde{X},\widetilde{Y}}''$ and $\widetilde{i_Y}$ realize (MS-$\tau$-1), (MS-$\tau$-5) and (MS-$\tau$-6) due to (C-$\tau$-3-i), (C-$\tau$-3-iii) and (C-$\tau$-3-iv) respectively.
    Moreover, as $C_{\widetilde{X},\widetilde{Y}}''$ satisfies (C-$\tau$-1) it follows immediately that (MS-$\tau$-2) holds.
    Since $C_{\widetilde{X},\widetilde{Y}}''$ satisfies (C-$\tau$-2), by Lemma~\ref{L:STABLE_CONTEXT_I}, $C_{\widetilde{X},\widetilde{Y}}'$ is not stable.
    Then neither is $C_{\widetilde{X}}$ because $C_{\widetilde{X},\widetilde{Y}}'$ satisfies (MS-$\tau$-4).
    Thus, $C_{\widetilde{X},\widetilde{Y}}''$ satisfies (MS-$\tau$-4) trivially.

    Next we verify (MS-$\tau$-3). Let $\widetilde{q}$ be any processes with $|\widetilde{q}|=|\widetilde{X}|$ and $q_{i_Y} \stackrel{\epsilon}{\Longrightarrow}q_{Y}'$  for each  $Y \in \widetilde{Y}$.

    \textbf{(MS-$\tau$-3-i)} Since $C_{\widetilde{X},\widetilde{Y}}'$ satisfies (MS-$\tau$-3-i), we have
    \[C_{\widetilde{X} }\{\widetilde{q}/\widetilde{X} \} \stackrel{\epsilon}{\Longrightarrow}
    t  \Rrightarrow  C_{\widetilde{X},\widetilde{Y} }'\{\widetilde{q}/\widetilde{X},\widetilde{q_{Y}'}/\widetilde{Y} \}\;\text{for some}\; t.\]
    Moreover, we have $C_{\widetilde{X},\widetilde{Y} }'\{\widetilde{q}/\widetilde{X},\widetilde{q_{Y}'}/\widetilde{Y} \} \stackrel{\tau}{\longrightarrow} C_{\widetilde{X},\widetilde{Y} }''\{\widetilde{q}/\widetilde{X},\widetilde{q_{Y}'}/\widetilde{Y} \}$ due to (C-$\tau$-2).
    Then it follows from Lemma~\ref{L:MULTI_STEP_UNFOLDING_ACTION} that
    \[t \stackrel{\tau}{\longrightarrow}t' \Rrightarrow C_{\widetilde{X},\widetilde{Y} }''\{\widetilde{q}/\widetilde{X},\widetilde{q_{Y}'}/\widetilde{Y} \}\;\text{for some}\;t'.\]
    Therefore, $C_{\widetilde{X} }\{\widetilde{q}/\widetilde{X} \} \stackrel{\epsilon}{\Longrightarrow}
    t  \stackrel{\tau}{\longrightarrow}t' \Rrightarrow C_{\widetilde{X},\widetilde{Y} }''\{\widetilde{q}/\widetilde{X},\widetilde{q_{Y}'}/\widetilde{Y} \}$, as desired.

     \textbf{(MS-$\tau$-3-ii)} It is straightforward as  $C_{\widetilde{X},\widetilde{Y}}'$  satisfies (MS-$\tau$-3-ii) and $C_{\widetilde{X},\widetilde{Y}}''$ satisfies (C-$\tau$-2).

     \textbf{(MS-$\tau$-7)} Suppose $r\equiv C_{\widetilde{X},\widetilde{Y}}''\{\widetilde{p}/\widetilde{X},\widetilde{p_Y'}/\widetilde{Y}\}\not\stackrel{\tau}{\longrightarrow}$.
     Then, since $C_{\widetilde{X},\widetilde{Y}}''$ satisfies (MS-$\tau$-1), by Lemmas~\ref{L:ONE_ACTION_TAU_GF} and \ref{L:STABLE_CONTEXT_I}, it is easy to see that both $C_{\widetilde{X},\widetilde{Y}}''$ and $\widetilde{p_Y'}$ 
     are stable.\\


\noindent Case 2
    For the transition $s \stackrel{\tau}{\longrightarrow}r$, the clause (2) in Lemma~\ref{L:ONE_ACTION_TAU} holds.

    Then there exist $i_0 \leq |\widetilde{X}|+|\widetilde{Y}|$, $C_{\widetilde{X},\widetilde{Y} }''$($\equiv C_{X_1,\dots,X_{|\widetilde{X}|},Y_{|\widetilde{X}|+1},\dots,Y_{|\widetilde{X}|+|\widetilde{Y}|}}''$) and $C_{\widetilde{X},\widetilde{Y},Z}'''$($\equiv C_{X_1,\dots,X_{|\widetilde{X}|},Y_{|\widetilde{X}|+1},\dots,Y_{|\widetilde{X}|+|\widetilde{Y}|},Z}'''$) with $Z \notin \widetilde{X}\cup \widetilde{Y}$ satisfying (P-$\tau$-1) - (P-$\tau$-4).
   In particular, by (P-$\tau$-3),
   \[C_{\widetilde{X},\widetilde{Y}}'' \equiv
   \left\{
     \begin{array}{ll}
       C_{\widetilde{X},\widetilde{Y},Z}'''\{X_{i_0}/Z\}, &  \text{if}\;1 \leq i_0 \leq |\widetilde{X}|; \\
       &\\
       C_{\widetilde{X},\widetilde{Y},Z}'''\{Y_{i_0}/Z\}, &  \text{if}\;|\widetilde{X}|+1 \leq i_0 \leq |\widetilde{X}|+|\widetilde{Y}|.
     \end{array}
   \right.\]

In case $|\widetilde{X}|+1 \leq i_0 \leq |\widetilde{X}|+|\widetilde{Y}|$, by (P-$\tau$-2), there exists $p'$ such that
\[p_{Y_{i_0}}' \stackrel{\tau}{\longrightarrow}p'\;\text{and}\;r \equiv C_{\widetilde{X},\widetilde{Y},Z}'''\{\widetilde{p}/\widetilde{X},\widetilde{p_Y'}/\widetilde{Y},p'/Z\}.\]
    Moreover, since $Y_{i_0}$ is 1-active in $C_{\widetilde{X},\widetilde{Y}}'$, by (P-$\tau$-1), we have $C_{\widetilde{X},\widetilde{Y}}' \equiv C_{\widetilde{X},\widetilde{Y}}''$.
    Further, since $Z$ is 1-active in $C_{\widetilde{X},\widetilde{Y},Z}'''$ and $C_{\widetilde{X},\widetilde{Y}}'' \equiv C_{\widetilde{X},\widetilde{Y},Z}'''\{Y_{i_0}/Z\}$, it is easy to see that $Y_{i_0}$ does not occur in $C_{\widetilde{X},\widetilde{Y},Z}'''$.
    Hence
\[r             \equiv  C_{\widetilde{X},\widetilde{Y}}'' \{\widetilde{p}/\widetilde{X},\widetilde{p_Y'}[p'/p_{Y_{i_0}}']/\widetilde{Y}\}
            \equiv  C_{\widetilde{X},\widetilde{Y}}' \{\widetilde{p}/\widetilde{X},\widetilde{p_Y'}[p'/p_{Y_{i_0}}']/\widetilde{Y}\}.\]
Then  it is not difficult to check that $C_{\widetilde{X},\widetilde{Y}}'$, $\widetilde{p_Y'}\,[p'/p_{Y_{i_0}}']$ and $\widetilde{i_Y}$ 
realize (MS-$\tau$-$l$)($1 \leq l \leq 7$) w.r.t the transition $C_{\widetilde{X}}\{\widetilde{p}/\widetilde{X}\} (\stackrel{\tau}{\longrightarrow} )^{k+1} r$, as desired.

    In case $1 \leq i_0 \leq |\widetilde{X}|$, by (P-$\tau$-2), there exists $p''$ such that  $p_{i_0}\stackrel{\tau}{\longrightarrow}p''$  and $r \equiv C_{\widetilde{X},\widetilde{Y},Z}'''\{\widetilde{p}/\widetilde{X},\widetilde{p_Y'}/\widetilde{Y},p''/Z\}$.
    Set
    \[i_Z \triangleq i_0\;\text{and}\; p_Z' \triangleq p''.\]
    In the following, we intend to verify that $C_{\widetilde{X},\widetilde{Y},Z }'''$, $i_U$ ($U \in \widetilde{Y} \cup \{Z\}$) and $|\widetilde{Y}|+1$-tuple $\widetilde{p_U'}$ with $U \in \widetilde{Y} \cup \{Z\}$ realize (MS-$\tau$-1) - (MS-$\tau$-7) w.r.t $C_{\widetilde{X}}\{\widetilde{p}/\widetilde{X}\} (\stackrel{\tau}{\longrightarrow} )^{k+1} r$.

    \textbf{(MS-$\tau$-1)} By (P-$\tau$-1), we have $C_{\widetilde{X},\widetilde{Y} }' \Rrightarrow C_{\widetilde{X},\widetilde{Y} }''$.
    Moreover, since $C_{\widetilde{X},\widetilde{Y}}'$ satisfy (MS-$\tau$-1), by Lemma~\ref{L:ONE_STEP_UNFOLDING_VARIABLE}(1), $Y$ is 1-active in $C_{\widetilde{X},\widetilde{Y} }''$ for each $Y \in \widetilde{Y}$.
    Further, by (P-$\tau$-3), each $Y\in \widetilde{Y}$ and $Z$ are 1-active in $C_{\widetilde{X},\widetilde{Y},Z }'''$.

    \textbf{(MS-$\tau$-2)} It is straightforward.

    \textbf{(MS-$\tau$-3)} Let $\widetilde{q}$ be any processes with $|\widetilde{q}| = |\widetilde{X}|$.

    \textbf{(MS-$\tau$-3-i)} Suppose $q_{i_U} \stackrel{\epsilon}{\Longrightarrow}q_{U}'$ for each $U \in \widetilde{Y} \cup \{Z\}$.
    Since $C_{\widetilde{X},\widetilde{Y}}'$ satisfies (MS-$\tau$-3-i), we have
    \[C_{\widetilde{X} }\{\widetilde{q}/\widetilde{X} \} \stackrel{\epsilon}{\Longrightarrow}
    t  \Rrightarrow  C_{\widetilde{X},\widetilde{Y} }'\{\widetilde{q}/\widetilde{X},\widetilde{q_{Y}'}/\widetilde{Y} \}\; \text{for some}\; t.\]
    It follows from  $q_{i_Z} \stackrel{\epsilon}{\Longrightarrow} q_{Z}'$ that $q_{i_Z} (\stackrel{\tau}{\longrightarrow})^m q_{Z}'$ for some $m \geq 0$.
    We shall distinguish two cases based on  $m$.

    In case $m=0$, we get $q_{i_Z} \equiv q_{Z}'$.
    Since $C_{\widetilde{X},\widetilde{Y},Z}'''$ satisfies (P-$\tau$-1) and (P-$\tau$-3), we have
    \[C_{\widetilde{X},\widetilde{Y}}'\{\widetilde{q}/\widetilde{X},\widetilde{q_{Y}'}/\widetilde{Y} \} \Rrightarrow  C_{\widetilde{X},\widetilde{Y} }''\{\widetilde{q}/\widetilde{X},\widetilde{q_{Y}'}/\widetilde{Y}\} \equiv C_{\widetilde{X},\widetilde{Y},Z }'''\{\widetilde{q}/\widetilde{X},\widetilde{q_{Y}'}/\widetilde{Y}, q_{Z}' /Z\}.
    \]
    Therefore $C_{\widetilde{X}}\{\widetilde{q}/\widetilde{X}\} \stackrel{\epsilon}{\Longrightarrow} t \Rrightarrow C_{\widetilde{X},\widetilde{Y} }''\{\widetilde{q}/\widetilde{X},\widetilde{q_{Y}'}/\widetilde{Y}\} \equiv C_{\widetilde{X},\widetilde{Y},Z }'''\{\widetilde{q}/\widetilde{X},\widetilde{q_{Y}'}/\widetilde{Y}, q_{Z}' /Z\}$, as desired.

    In case $m >0$, i.e., $q_{i_Z} \stackrel{\tau}{\longrightarrow}q'' \stackrel{
    \epsilon}{\Longrightarrow}q_{Z}'$ for some $q''$, by (P-$\tau$-4), we obtain
    \[C_{\widetilde{X},\widetilde{Y} }'\{\widetilde{q}/\widetilde{X},\widetilde{q_{Y}'}/\widetilde{Y} \} \stackrel{\tau}{\longrightarrow} C_{\widetilde{X},\widetilde{Y},Z }'''\{\widetilde{q}/\widetilde{X},\widetilde{q_{Y}'}/\widetilde{Y}, q'' /Z\}.
    \]
    Moreover, since $Z$ is 1-active, by Lemma~\ref{L:ONE_ACTION_TAU_GF}, we get
   \[C_{\widetilde{X},\widetilde{Y},Z}'''\{\widetilde{q}/\widetilde{X},\widetilde{q_{Y}'}/\widetilde{Y}, q'' /Z\} \stackrel{\epsilon}{\Longrightarrow} C_{\widetilde{X},\widetilde{Y},Z }'''\{\widetilde{q}/\widetilde{X},\widetilde{q_{Y}'}/\widetilde{Y}, q_{Z}' /Z\}.\]
   Then, by Lemma~\ref{L:MULTI_STEP_UNFOLDING_ACTION_II}, it follows from $t  \Rrightarrow  C_{\widetilde{X},\widetilde{Y} }'\{\widetilde{q}/\widetilde{X},\widetilde{q_{Y}'}/\widetilde{Y} \}$ that there exist $t'$ such that
   \[
   t \stackrel{\epsilon}{\Longrightarrow}t' \Rrightarrow C_{\widetilde{X},\widetilde{Y},Z }'''\{\widetilde{q}/\widetilde{X},\widetilde{q_{Y}'}/\widetilde{Y}, q_{Z}' /Z\}.
   \]
   Consequently, $C_{\widetilde{X}}\{\widetilde{q}/\widetilde{X}\} \stackrel{\epsilon}{\Longrightarrow} t \stackrel{\epsilon}{\Longrightarrow}t' \Rrightarrow C_{\widetilde{X},\widetilde{Y},Z }'''\{\widetilde{q}/\widetilde{X},\widetilde{q_{Y}'}/\widetilde{Y}, q_{Z}' /Z\}$.

   \textbf{(MS-$\tau$-3-ii)} Suppose $q_{i_U} \stackrel{\tau}{\Longrightarrow}q_{U}'$ for each $U \in \widetilde{Y} \cup \{Z\}$.
      Since $C_{\widetilde{X},\widetilde{Y}}'$ satisfies (MS-$\tau$-3-ii), we have
    \[C_{\widetilde{X} }\{\widetilde{q}/\widetilde{X} \} \stackrel{\epsilon}{\Longrightarrow}
      C_{\widetilde{X},\widetilde{Y} }'\{\widetilde{q}/\widetilde{X},\widetilde{q_{Y}'}/\widetilde{Y} \}.\]
      Moreover,  $q_{i_Z} \stackrel{\tau}{\longrightarrow}q'' \stackrel{
    \epsilon}{\Longrightarrow}q_{Z}'$ for some $q''$ because of $q_{i_Z} \stackrel{\tau}{\Longrightarrow} q_{Z}'$.
    Hence by (P-$\tau$-4)
    \[C_{\widetilde{X},\widetilde{Y} }'\{\widetilde{q}/\widetilde{X},\widetilde{q_{Y}'}/\widetilde{Y} \} \stackrel{\tau}{\longrightarrow} C_{\widetilde{X},\widetilde{Y},Z }'''\{\widetilde{q}/\widetilde{X},\widetilde{q_{Y}'}/\widetilde{Y}, q'' /Z\}.\]
    Further, since $Z$ is 1-active, it follows from Lemma~\ref{L:ONE_ACTION_TAU_GF} that
   \[C_{\widetilde{X},\widetilde{Y},Z}'''\{\widetilde{q}/\widetilde{X},\widetilde{q_{Y}'}/\widetilde{Y}, q'' /Z\} \stackrel{\epsilon}{\Longrightarrow} C_{\widetilde{X},\widetilde{Y},Z }'''\{\widetilde{q}/\widetilde{X},\widetilde{q_{Y}'}/\widetilde{Y}, q_{Z}' /Z\}.\]
   Consequently, $C_{\widetilde{X} }\{\widetilde{q}/\widetilde{X} \} \stackrel{\epsilon}{\Longrightarrow}C_{\widetilde{X},\widetilde{Y},Z }'''\{\widetilde{q}/\widetilde{X},\widetilde{q_{Y}'}/\widetilde{Y}, q_{Z}' /Z\}$, as desired.

   \textbf{(MS-$\tau$-4)} Assume $C_{\widetilde{X}}$ is stable.
   By (MS-$\tau$-4), $C_{\widetilde{X},\widetilde{Y}}'$ is stable and for any $\widetilde{q} $,
      $C_{\widetilde{X}}\{\widetilde{q}/\widetilde{X} \}  \Rrightarrow C_{\widetilde{X},\widetilde{Y} }'\{\widetilde{q}/\widetilde{X},\widetilde{q_{i_Y}}/\widetilde{Y}\}$.
      Moreover, by Lemma~\ref{L:MULTI_STEP_UNFOLDING_ACTION}, it follows from
       $C_{\widetilde{X},\widetilde{Y}}' \Rrightarrow C_{\widetilde{X},\widetilde{Y}}''$ (i.e., (P-$\tau$-1)) and $C_{\widetilde{X},\widetilde{Y},Z}'''\{X_{i_Z}/Z\} \equiv C_{\widetilde{X},\widetilde{Y}}'' $ (i.e., (P-$\tau$-3)) that $C_{\widetilde{X},\widetilde{Y},Z}'''$ is stable and
      \[C_{\widetilde{X},\widetilde{Y}}'\{\widetilde{q}/\widetilde{X},\widetilde{q_{i_Y}}/\widetilde{Y}\} \Rrightarrow C_{\widetilde{X},\widetilde{Y},Z}'''\{\widetilde{q}/\widetilde{X},\widetilde{q_{i_Y}}/\widetilde{Y},q_{i_Z}/Z\}.\]

   \textbf{(MS-$\tau$-5,6)} By Lemma~\ref{L:ONE_STEP_UNFOLDING_VARIABLE}(3)(5), they immediately follow from the fact that $C_{\widetilde{X},\widetilde{Y}}'$ satisfies (MS-$\tau$-5,6) and $C_{\widetilde{X},\widetilde{Y},Z}'''$ satisfies (P-$\tau$-1,3).

   \textbf{(MS-$\tau$-7)} Immediately follows from (MS-$\tau$-1), (MS-$\tau$-2) and Lemmas~\ref{L:ONE_ACTION_TAU_GF} and \ref{L:STABLE_CONTEXT_I}.
\end{proof}

\begin{lemma}\label{L:STABILIZATION}
  For any $\widetilde{p}$ and stable context $C_{\widetilde{X}}$, if,  for each $i \leq |\widetilde{X}|$, $p_i \stackrel{\epsilon}{\Longrightarrow} |p_i'$ then $C_{\widetilde{X}}\{\widetilde{p}/\widetilde{X}\} \stackrel{\epsilon}{\Longrightarrow} |q$ for some $q$.
\end{lemma}
\begin{proof}
    By Lemma~\ref{L:STABILIZATION_PRE} and  \ref{L:ONE_STEP_UNFOLDING_VARIABLE}(4), $C_{\widetilde{X}} \Rrightarrow  C_{\widetilde{X}}'$ for some $C_{\widetilde{X}}'$ such that each unguarded occurrence of any free variable in $C_{\widetilde{X}}'$ is unfolded.
    Moreover, since $C_{\widetilde{X}}$ is stable,  so is $C_{\widetilde{X}}'$ by $C_{\widetilde{X}}\{0/\widetilde{X}\} \Rrightarrow  C_{\widetilde{X}}'\{0/\widetilde{X}\}$ and Lemma~\ref{L:MULTI_STEP_UNFOLDING_ACTION}.

  Let $C_{\widetilde{X},\widetilde{Y}}'$ be the context obtained from $C_{\widetilde{X}}'$ by replacing simultaneously all unguarded and unfolded occurrences of free variables in $\widetilde{X}$ by distinct and fresh variables $\widetilde{Y}$.
  Here distinct occurrences are replaced by distinct variables. Clearly, we have
  \begin{enumerate}[(1)]
  \renewcommand{\theenumi}{(\arabic{enumi})}
    \item \  for each $Y \in \widetilde{Y}$, there exists exactly one $i_Y \leq |\widetilde{X}|$ such that $C_{\widetilde{X}}' \equiv C_{\widetilde{X},\widetilde{Y}}'\{\widetilde{X_{i_Y}}/\widetilde{Y}\}$,
    \item \ all variables in $\widetilde{Y}$ are 1-active in $C_{\widetilde{X},\widetilde{Y}}'$, and
    \item \ $C_{\widetilde{X},\widetilde{Y}}'$ is stable.
  \end{enumerate}
  Then $C_{\widetilde{X}}\{\widetilde{p}/\widetilde{X}\} \Rrightarrow  C_{\widetilde{X}}'\{\widetilde{p}/\widetilde{X}\}\equiv C_{\widetilde{X},\widetilde{Y}}'\{\widetilde{p}/\widetilde{X},\widetilde{p_{i_Y}}/\widetilde{Y}\} $, and by Lemma~\ref{L:ONE_ACTION_TAU_GF},  we obtain $C_{\widetilde{X},\widetilde{Y}}'\{\widetilde{p}/\widetilde{X},\widetilde{p_{i_Y}}/\widetilde{Y}\} \stackrel{\epsilon}{\Longrightarrow} C_{\widetilde{X},\widetilde{Y}}'\{\widetilde{p}/\widetilde{X},\widetilde{p_{i_Y}'}/\widetilde{Y}\}$.
  Further, since $C_{\widetilde{X},\widetilde{Y}}'$ and $\widetilde{p_{i_Y}'}$ are stable and $\widetilde{Y}$ contains all unguarded occurrences of variables in $C_{\widetilde{X},\widetilde{Y}}'$, we get $C_{\widetilde{X},\widetilde{Y}}'\{\widetilde{p}/\widetilde{X},\widetilde{p_{i_Y}'}/\widetilde{Y}\}\not\stackrel{\tau}{\longrightarrow}$ by Lemma~\ref{L:ONE_ACTION_TAU}.
  Hence, by Lemma~\ref{L:MULTI_STEP_UNFOLDING_ACTION_II}, $C_{\widetilde{X}}\{\widetilde{p}/\widetilde{X}\} \stackrel{\epsilon}{\Longrightarrow} |q \Rrightarrow C_{\widetilde{X},\widetilde{Y}}'\{\widetilde{p}/\widetilde{X},\widetilde{p_{i_Y}'}/\widetilde{Y}\}$ for some $q $.
\end{proof}

Given a process $C_{\widetilde{X}}\{\widetilde{p}/\widetilde{X}\}$ and its stable $\tau$-descendant $r$ (i.e., $C_{\widetilde{X}}\{\widetilde{p}/\widetilde{X}\} \stackrel{\epsilon}{\Longrightarrow} |r$), in general there exist more than one evolution paths from $C_{\widetilde{X}}\{\widetilde{p}/\widetilde{X}\}$ to $r$.
Since each $\tau$-labelled transition in $\text{CLL}_R$ activated by a single process, a natural conjecture arises at this point that there exist some \textquotedblleft canonical" evolution paths from $C_{\widetilde{X}}\{\widetilde{p}/\widetilde{X}\}$ to $r$ in which the context $C_{\widetilde{X}}$ first evolves itself into a stable context then $p_i$ evolves.
A weak version of this conjecture will be verified in Lemma~\ref{L:TAU_ACTION_NORMALIZATION}.
To this end, a preliminary result is given:

\begin{lemma}\label{L:PLACE_HOLDER}
  Let $t_1,t_2$ be two terms and $\widetilde{X}$ a tuple of variables such that any recursive variable occurring in $t_i$(with $i=1,2$) is not in $\widetilde{X}$,  and let $\widetilde{a_X.0}$ be a tuple of processes with fresh visible action $a_X$ for each $X \in \widetilde{X}$. Then
  \begin{enumerate}[(1)]
    \renewcommand{\theenumi}{(\arabic{enumi})}
    \item if $t_1\{\widetilde{a_X.0}/\widetilde{X}\} \equiv t_2\{\widetilde{a_X.0}/\widetilde{X}\} $ then $t_1 \equiv t_2$;
    \item if $t_1\{\widetilde{a_X.0}/\widetilde{X}\} \Rrightarrow_1 t_2\{\widetilde{a_X.0}/\widetilde{X}\} $ then $t_1\{\widetilde{r}/\widetilde{X}\} \Rrightarrow_1 t_2\{\widetilde{r}/\widetilde{X}\} $ for any $\widetilde{r}$.
  \end{enumerate}
\end{lemma}
\begin{proof}
\noindent \textbf{(1)} If $FV(t_1) \cap \widetilde{X} = \emptyset$ then $t_1\{\widetilde{a_X.0}/\widetilde{X}\} \equiv t_1 \equiv t_2\{\widetilde{a_X.0}/\widetilde{X}\} $.
  Further, since $a_X$ is fresh for each $X \in \widetilde{X}$, we have $FV(t_2) \cap \widetilde{X} = \emptyset$.
    Hence $t_1 \equiv t_2$.
    In the following, we consider the other case $FV(t_1) \cap \widetilde{X} \neq \emptyset$.
    We proceed by induction on $t_1$.\\

\noindent Case  1 $t_1 \equiv X_i$.

    Then $t_1\{\widetilde{a_X.0}/\widetilde{X}\} \equiv a_{X_i}.0 \equiv t_2\{\widetilde{a_X.0}/\widetilde{X}\} $.
    Hence $t_2 \equiv X_i$ due to the freshness of $a_{X_i}$.\\

 \noindent Case  2 $t_1 \equiv \alpha.s$.

   So  $t_1\{\widetilde{a_X.0}/\widetilde{X}\} \equiv \alpha.s\{\widetilde{a_X.0}/\widetilde{X}\} \equiv t_2\{\widetilde{a_X.0}/\widetilde{X}\} $.
   Since $\alpha \neq a_X$ for each $X \in \widetilde{X}$, there exists $s'$ such that $t_2 \equiv \alpha.s' $ and $s\{\widetilde{a_X.0}/\widetilde{X}\} \equiv  s'\{\widetilde{a_X.0}/\widetilde{X}\} $.
   By IH, we have $s \equiv  s'$.
   Hence $t_1 \equiv t_2$.\\

 \noindent Case  3 $t_1 \equiv s_1\odot s_2$ with $\odot \in \{\vee, \Box, \parallel_A, \wedge\}$.

 Then  $t_1\{\widetilde{a_X.0}/\widetilde{X}\} \equiv s_1\{\widetilde{a_X.0}/\widetilde{X}\} \odot s_2\{\widetilde{a_X.0}/\widetilde{X}\}\equiv t_2\{\widetilde{a_X.0}/\widetilde{X}\} $.
 Since $\widetilde{a_X.0}$ do not contain $\odot$, there exist $s_1',s_2'$ such that $t_2 \equiv s_1' \odot s_2'$, $s_1\{\widetilde{a_X.0}/\widetilde{X}\} \equiv  s_1'\{\widetilde{a_X.0}/\widetilde{X}\} $ and $s_2\{\widetilde{a_X.0}/\widetilde{X}\} \equiv  s_2'\{\widetilde{a_X.0}/\widetilde{X}\} $.
 Hence $s_1 \equiv  s_1'$ and $s_2 \equiv  s_2'$ by applying IH.\\

  \noindent Case  4 $t_1 \equiv \langle Y|E \rangle$ for some $E(V)$ with $Y \in V$.

  Then  $t_1\{\widetilde{a_X.0}/\widetilde{X}\} \equiv  \langle Y|E\{\widetilde{a_X.0}/\widetilde{X}\}\rangle \equiv t_2\{\widetilde{a_X.0}/\widetilde{X}\} $.
  So, $t_2 \equiv \langle Y|E' \rangle$ for some $E'(V)$ such that for each $Z \in V$, $t_Z\{\widetilde{a_X.0}/\widetilde{X}\} \equiv t_Z'\{\widetilde{a_X.0}/\widetilde{X}\}$ where $Z =t_Z \in E$ and $Z=t_Z' \in E'$.
  By IH, $t_Z\equiv  t_Z'$ for each $Z \in V$.
  Thus $t_1 \equiv \langle Y|E \rangle \equiv \langle Y|E' \rangle \equiv t_2$.\\

\noindent \textbf{(2)}
In case   $FV(t_1) \cap \widetilde{X} = \emptyset$, since one-step unfolding does not introduce fresh actions, we have $FV(t_2) \cap \widetilde{X} =\emptyset$.
Thus, $t_1 \equiv t_1\{\widetilde{a_X.0}/\widetilde{X}\} \Rrightarrow_1 t_2\{\widetilde{a_X.0}/\widetilde{X}\} \equiv t_2$, and hence $t_1 \equiv t_1\{\widetilde{r}/\widetilde{X}\} \Rrightarrow_1 t_2\{\widetilde{r}/\widetilde{X}\} \equiv t_2$ for any $\widetilde{r}$.
Next we consider the other case $FV(t_1) \cap \widetilde{X} \neq \emptyset$.
     It proceeds by induction on $t_1$.
     This is a routine case analysis on the format of $t_1$, we  handle only the case  $t_1 \equiv \langle Y|E \rangle$.

  In this case, $t_1\{\widetilde{a_X.0}/\widetilde{X}\}  \equiv \langle Y|E \rangle \{\widetilde{a_X.0}/\widetilde{X}\}$.
  By Def.~\ref{D:UNFOLDING}, the unique result of one-step unfolding of  $\langle Y|E \rangle \{\widetilde{a_X.0}/\widetilde{X}\}$ is $\langle t_Y|E \rangle \{\widetilde{a_X.0}/\widetilde{X}\}$ where $Y =t_Y \in E$.
  Thus, we get $ \langle t_Y|E \rangle \{\widetilde{a_X.0}/\widetilde{X}\}\equiv t_2\{\widetilde{a_X.0}/\widetilde{X}\} $.
  By item (1), we have $\langle t_Y|E \rangle   \equiv t_2 $, and hence $t_1\{\widetilde{r}/\widetilde{X}\} \Rrightarrow_1  \langle t_Y |E \rangle \{\widetilde{r}/\widetilde{X}\} \equiv t_2\{\widetilde{r}/\widetilde{X}\} $ for any $\widetilde{r}$.
\end{proof}

Having disposed of this preliminary step, we can now   verify a weak version of the conjecture mentioned above, which is sufficient for the aim of this paper.
At present, we do not know whether this result still holds if the requirement (1) in the next lemma is strengthened to $C_{\widetilde{X}}\{\widetilde{p}/\widetilde{X}\} \stackrel{\epsilon}{\Longrightarrow}  D_{\widetilde{X}}\{\widetilde{p}/\widetilde{X}\} \stackrel{\epsilon}{\Longrightarrow}  | r$.

\begin{lemma}\label{L:TAU_ACTION_NORMALIZATION}
 For any $C_{\widetilde{X}}$ and $\widetilde{p}$, if $C_{\widetilde{X}}\{\widetilde{p}/\widetilde{X}\} \stackrel{\epsilon}{\Longrightarrow} |r$ then there exists a stable context $D_{\widetilde{X}}$ such that
 \begin{enumerate}[(1)]
\renewcommand{\theenumi}{(\arabic{enumi})}
   \item $C_{\widetilde{X}}\{\widetilde{p}/\widetilde{X}\} \stackrel{\epsilon}{\Longrightarrow}  D_{\widetilde{X}}\{\widetilde{p}/\widetilde{X}\} \stackrel{\epsilon}{\Longrightarrow}  | r'  \Rrightarrow  r$ for some  $r'$, and
   \item  $C_{\widetilde{X}}\{\widetilde{q}/\widetilde{X}\} \stackrel{\epsilon}{\Longrightarrow}  D_{\widetilde{X}}\{\widetilde{q}/\widetilde{X}\}$ for any $\widetilde{q}$ with $|\widetilde{q}| = |\widetilde{X}|$.
 \end{enumerate}
\end{lemma}
\begin{proof}
  Suppose $C_{\widetilde{X}}\{\widetilde{p}/\widetilde{X}\} (\stackrel{\tau}{\longrightarrow} )^n|r$.
  It proceeds by induction on $n$.
  For the inductive base $n=0$, it follows from $C_{\widetilde{X}}\{\widetilde{p}/\widetilde{X}\} \equiv r \not\stackrel{\tau}{\longrightarrow}$ that $C_{\widetilde{X}} $ is stable by Lemma~\ref{L:STABLE_CONTEXT_I}.
   Then it is straightforward to verify that $C_{\widetilde{X}}$ itself is exactly what we seek.
  For the inductive step, assume $C_{\widetilde{X}}\{\widetilde{p}/\widetilde{X}\} \stackrel{\tau}{\longrightarrow}  t(\stackrel{\tau}{\longrightarrow} )^k|r$ for some $t$.
  Then, for the $\tau$-labelled transition $C_{\widetilde{X}}\{\widetilde{p}/\widetilde{X}\} \stackrel{\tau}{\longrightarrow}  t$, either the clause (1) or (2) in Lemma~\ref{L:ONE_ACTION_TAU} holds.
  The first alternative is easy to handle and is thus omitted.
    Next we consider the second alternative.

    In this situation, there exist $C_{\widetilde{X}}'$, $C_{\widetilde{X},Z}''$ with $Z \notin \widetilde{X}$ and $i_0\leq |\widetilde{X}|$ that satisfy (P-$\tau$-1) -- (P-$\tau$-4).
    By (P-$\tau$-2), we have
    \[t \equiv C_{\widetilde{X},Z}''\{\widetilde{p}/\widetilde{X},p'/Z\}\;\text{for some}\; p'\;\text{with}\; p_{i_0} \stackrel{\tau}{\longrightarrow}p'.\]
    Then, for $C_{\widetilde{X},Z}''\{\widetilde{p}/\widetilde{X},p'/Z\} (\stackrel{\tau}{\longrightarrow})^k|r$, by IH, there exists a stable context $D_{\widetilde{X},Z}'$ such that
    \[C_{\widetilde{X},Z}''\{\widetilde{p}/\widetilde{X},p'/Z\} \stackrel{\epsilon}{\Longrightarrow}  D_{\widetilde{X},Z}'\{\widetilde{p}/\widetilde{X},p'/Z\} \stackrel{\epsilon}{\Longrightarrow}  |r'  \Rrightarrow r\;\text{for some}\;  r'\tag{\ref{L:TAU_ACTION_NORMALIZATION}.1}\]
    and for any $q'$ and $\widetilde{q}$, we have \[C_{\widetilde{X},Z}''\{\widetilde{q}/\widetilde{X},q'/Z\} \stackrel{\epsilon}{\Longrightarrow}  D_{\widetilde{X},Z}'\{\widetilde{q}/\widetilde{X},q'/Z\}.\tag{\ref{L:TAU_ACTION_NORMALIZATION}.2}\]
    In particular, we have $C_{\widetilde{X},Z}''\{\widetilde{a_X.0} /\widetilde{X},a_Z.0/Z\} \stackrel{\epsilon}{\Longrightarrow}  D_{\widetilde{X},Z}'\{\widetilde{a_X.0} /\widetilde{X},a_Z.0/Z\}$ where distinct visible actions $\widetilde{a_X}$ and $a_Z$ are fresh.
    For this transition, applying Lemma~\ref{L:ONE_ACTION_TAU} finitely often (notice that, in this procedure, since $\widetilde{a_X.0}$ and $a_Z.0$ are stable, the clause (2) in Lemma~\ref{L:ONE_ACTION_TAU} is always false), then by clause (1) in Lemma~\ref{L:ONE_ACTION_TAU}, we get the sequence
    \begin{multline*}
    C_{\widetilde{X},Z}''\{\widetilde{a_X.0} /\widetilde{X},a_Z.0/Z\}\equiv C_{\widetilde{X},Z}^0\{\widetilde{a_X.0} /\widetilde{X},a_Z.0/Z\} \stackrel{\tau}{\longrightarrow}  C_{\widetilde{X},Z}^1\{\widetilde{a_X.0} /\widetilde{X},a_Z.0/Z\}\stackrel{\tau}{\longrightarrow}\\ \dots \stackrel{\tau}{\longrightarrow} C_{\widetilde{X},Z}^n\{\widetilde{a_X.0} /\widetilde{X},a_Z.0/Z\} \equiv D_{\widetilde{X},Z}'\{\widetilde{a_X.0} /\widetilde{X},a_Z.0/Z\}.
    \end{multline*}
    Here $n \geq 0$ and $C_{\widetilde{X},Z}^i$ satisfies (C-$\tau$-1,2,3) for each $1 \leq i \leq n$.
    Moreover, since $Z$ is 1-active in $C_{\widetilde{X},Z}''$, by (C-$\tau$-3-i), so is $Z$ in $C_{\widetilde{X},Z}^n$.
    We also have $C_{\widetilde{X},Z}^n \equiv D_{\widetilde{X},Z}'$ by Lemma~\ref{L:PLACE_HOLDER}.
    Hence we conclude that
    \[Z \;\text{is 1-active in}\;D_{\widetilde{X},Z}'. \tag{\ref{L:TAU_ACTION_NORMALIZATION}.3}\]
    Since $C_{\widetilde{X}}'$ and $C_{\widetilde{X},Z}''$ satisfy (P-$\tau$-1) and (P-$\tau$-3), for any $\widetilde{s}$, we get
    \[  C_{\widetilde{X}}\{\widetilde{s}/\widetilde{X}\}  \Rrightarrow C_{\widetilde{X}}' \{\widetilde{s}/\widetilde{X}\} \equiv C_{\widetilde{X},Z}''\{\widetilde{s}/\widetilde{X},s_{i_0}/Z\}. \tag{\ref{L:TAU_ACTION_NORMALIZATION}.4} \]
    In order to complete the proof, it suffices to find a stable context $D_{\widetilde{X}}$ satisfying conditions (1) and (2).
    In the following, we shall use $\widetilde{a_{X}.0}$  again to obtain such context.

    Since $\widetilde{a_{X}.0}$ and $D_{\widetilde{X},Z}'$ are stable, by  (\ref{L:TAU_ACTION_NORMALIZATION}.2),
    we get
    \[C_{\widetilde{X},Z}''\{\widetilde{a_{X}.0}/\widetilde{X},a_{X_{i_0}}.0/Z\} \stackrel{\epsilon}{\Longrightarrow}  |D_{\widetilde{X},Z}'\{\widetilde{a_{X}.0}/\widetilde{X},a_{X_{i_0}}.0/Z\}.\]
    Moreover, by (\ref{L:TAU_ACTION_NORMALIZATION}.4), we have $C_{\widetilde{X}}' \{\widetilde{a_X.0}/\widetilde{X}\} \equiv C_{\widetilde{X},Z}''\{\widetilde{a_X.0}/\widetilde{X},a_{X_{i_0}}.0/Z\}$.
    Thus, it follows that
    \[C_{\widetilde{X}}' \{\widetilde{a_X.0}/\widetilde{X}\}\stackrel{\epsilon}{\Longrightarrow}| D_{\widetilde{X},Z}'\{\widetilde{a_X.0}/\widetilde{X},a_{X_{i_0}}.0/Z\}.\]
    Then, since $\widetilde{a_X.0}$ are stable, by Lemma~\ref{L:MULTI_TAU_GF_STABLE}, there exists a stable context $B_{\widetilde{X}}$ such that
    \[B_{\widetilde{X}}\{\widetilde{a_{X}.0}/\widetilde{X}\} \equiv D_{\widetilde{X},Z}'\{\widetilde{a_{X}.0}/\widetilde{X},a_{X_{i_0}}.0/Z\} \tag{\ref{L:TAU_ACTION_NORMALIZATION}.5}\]
    and
    \[C_{\widetilde{X}}'\{\widetilde{s}/\widetilde{X}\} \stackrel{\epsilon}{\Longrightarrow}  B_{\widetilde{X}}\{\widetilde{s}/\widetilde{X}\} \;\text{for any}\;\widetilde{s}. \tag{\ref{L:TAU_ACTION_NORMALIZATION}.6}\]
    In addition, by (\ref{L:TAU_ACTION_NORMALIZATION}.4) and Lemma~\ref{L:MULTI_STEP_UNFOLDING_ACTION_II}, we have $C_{\widetilde{X}}\{\widetilde{a_X.0}/\widetilde{X}\}  \Rrightarrow C_{\widetilde{X}}' \{\widetilde{a_X.0}/\widetilde{X}\}$ and $C_{\widetilde{X}}\{\widetilde{a_X.0}/\widetilde{X}\}\stackrel{\epsilon}{\Longrightarrow}|t'  \Rrightarrow D_{\widetilde{X},Z}'\{\widetilde{a_{X}.0}/\widetilde{X},a_{X_{i_0}}.0/Z\} $ for some $t'$.
     Further, since $\widetilde{a_{X}.0}$ are stable, by Lemma~\ref{L:MULTI_TAU_GF_STABLE}, there exists a stable context $D_{\widetilde{X}}$ such that
      \[t' \equiv D_{\widetilde{X}}\{\widetilde{a_{X}.0}/\widetilde{X}\}   \Rrightarrow  D_{\widetilde{X},Z}'\{\widetilde{a_{X}.0}/\widetilde{X},a_{X_{i_0}}.0/Z\}  \tag{\ref{L:TAU_ACTION_NORMALIZATION}.7}\]
      and
     \[C_{\widetilde{X}}\{\widetilde{s}/\widetilde{X}\} \stackrel{\epsilon}{\Longrightarrow}  D_{\widetilde{X}}\{\widetilde{s}/\widetilde{X}\}\; \text{for any}\; \widetilde{s}. \tag{\ref{L:TAU_ACTION_NORMALIZATION}.8}\]
     Notice that, (\ref{L:TAU_ACTION_NORMALIZATION}.8) follows from (MS-$\tau$-3-ii) with $\widetilde{Y}= \emptyset$.
    In the following, we intend to prove that $D_{\widetilde{X}}$ is what we seek.
    It immediately follows from (\ref{L:TAU_ACTION_NORMALIZATION}.8) that $D_{\widetilde{X}}$ meets the requirement (2).
    We are left with the task of verifying that $D_{\widetilde{X}}$ satisfies the condition (1).
    So far, for any $\widetilde{s}$, we have the diagram below, where the first line follows from
    (\ref{L:TAU_ACTION_NORMALIZATION}.4),
    \[      \begin{array}{ccccc}
                  C_{\widetilde{X}}\{\widetilde{s}/\widetilde{X}\}\quad\quad\quad  &  \Rrightarrow  & C_{\widetilde{X}}' \{\widetilde{s}/\widetilde{X}\}\quad\quad\quad & \equiv  & C_{\widetilde{X},Z}''\{\widetilde{s}/\widetilde{X},s_{i_0}/Z\} \\
                     &  &  &   &  \\
                  \Downarrow \mspace{-3mu}\epsilon \;\text{by}\; (\ref{L:TAU_ACTION_NORMALIZATION}.8)&  & \Downarrow\mspace{-3mu}\epsilon \;\text{by}\; (\ref{L:TAU_ACTION_NORMALIZATION}.6)&   & \Downarrow\mspace{-3mu}\epsilon \;\text{by}\; (\ref{L:TAU_ACTION_NORMALIZATION}.2)\\
                   &  &  &   &  \\
                  D_{\widetilde{X}}\{\widetilde{s}/\widetilde{X}\} \quad\quad\quad &  \Rrightarrow  & B_{\widetilde{X}} \{\widetilde{s}/\widetilde{X}\} \quad\quad\quad& \equiv  & D_{\widetilde{X},Z}'\{\widetilde{s}/\widetilde{X},s_{i_0}/Z\} \\
                \end{array}
    \]
    Here the last line in the above follows from (\ref{L:TAU_ACTION_NORMALIZATION}.7) and (\ref{L:TAU_ACTION_NORMALIZATION}.5) using Lemma~\ref{L:PLACE_HOLDER}.
    Further, by Lemma~\ref{L:ONE_ACTION_TAU_GF} and $p_{i_0} \stackrel{\tau}{\longrightarrow}p'$, it follows from (\ref{L:TAU_ACTION_NORMALIZATION}.1) and (\ref{L:TAU_ACTION_NORMALIZATION}.3) that
    \[B_{\widetilde{X}} \{\widetilde{p}/\widetilde{X}\} \equiv D_{\widetilde{X},Z}'\{\widetilde{p}/\widetilde{X},p_{i_0}/Z\} \stackrel{\tau}{\longrightarrow} D_{\widetilde{X},Z}'\{\widetilde{p}/\widetilde{X},p'/Z\} \stackrel{\epsilon}{\Longrightarrow}|r'\Rrightarrow r.\]
    Finally, since $D_{\widetilde{X}}\{\widetilde{p}/\widetilde{X}\}   \Rrightarrow B_{\widetilde{X}} \{\widetilde{p}/\widetilde{X}\} $, by  Lemma~\ref{L:MULTI_STEP_UNFOLDING_ACTION_II}, we get $D_{\widetilde{X}}\{\widetilde{p}/\widetilde{X}\} \stackrel{\epsilon}{\Longrightarrow} |r''  \Rrightarrow r'  \Rrightarrow r$ for some $r''$, which, together with $C_{\widetilde{X}}\{\widetilde{p}/\widetilde{X}\} \stackrel{\epsilon}{\Longrightarrow}D_{\widetilde{X}}\{\widetilde{p}/\widetilde{X}\}$, implies that the stable context $D_{\widetilde{X}}$ also meets the requirement (1), as desired.
\end{proof}

The result below asserts that there exist another \textquotedblleft canonical\textquotedblright  evolution paths from $C_{\widetilde{X} }\{\widetilde{p}/\widetilde{X} \}$ to a given stable $\tau$-descendant $r$.
For these paths, an unstable $p_i$ evolves first provided that such $p_i$ is located in an active position.

\begin{lemma}\label{L:MULTI_TAU_ACTIVE_STABLE}
  For any $C_{\widetilde{X} }$ and $\widetilde{p}$,
  if $C_{\widetilde{X} }\{\widetilde{p}/\widetilde{X} \} \stackrel{\epsilon}{\Longrightarrow} |q$ and $X_i$ is 1-active in $C_{\widetilde{X}}$ for some $i \leq |\widetilde{X}|$, then there exists $p'$ such that $p_i \stackrel{\epsilon}{\Longrightarrow}|p'$ and $C_{\widetilde{X} }\{\widetilde{p}/\widetilde{X} \} \stackrel{\epsilon}{\Longrightarrow} C_{\widetilde{X}}\{\widetilde{p}\,[p'/p_i]/\widetilde{X}\} \stackrel{\epsilon}{\Longrightarrow} |q$.
\end{lemma}
\begin{proof}
  Suppose $C_{\widetilde{X} }\{\widetilde{p}/\widetilde{X} \} (\stackrel{\tau}{\longrightarrow})^n |q$ for some $n \geq 0$.
  We shall prove it by induction on $n$.
  For the inductive base $n=0$, we have $p_i \not\stackrel{\tau}{\longrightarrow}$  by Lemma~\ref{L:ONE_ACTION_TAU_GF}, and hence it holds trivially by taking $p' \equiv p_i$.
  For the inductive step $n =k+1$, suppose $C_{\widetilde{X} }\{\widetilde{p}/\widetilde{X} \} \stackrel{\tau}{\longrightarrow}r(\stackrel{\tau}{\longrightarrow})^k |q$ for some $r$.
    For the transition $C_{\widetilde{X} }\{\widetilde{p}/\widetilde{X} \} \stackrel{\tau}{\longrightarrow}r$, either the clause (1) or (2) in Lemma~\ref{L:ONE_ACTION_TAU} holds.

  For the first alternative, there exists a context $C_{\widetilde{X}}'$ such that
  \begin{enumerate}[({1.}1)]
  \renewcommand{\theenumi}{(1.\arabic{enumi})}
    \item \ $X_i$ is 1-active in $C_{\widetilde{X}}'$ (by (C-$\tau$-3-i)),
    \item \ $r\equiv C_{\widetilde{X}}'\{\widetilde{p}/\widetilde{X}\}$, and
    \item \ $C_{\widetilde{X} }\{\widetilde{s}/\widetilde{X} \} \stackrel{\tau}{\longrightarrow}C_{\widetilde{X} }'\{\widetilde{s}/\widetilde{X} \}$ for any $\widetilde{s}$.
  \end{enumerate}
  By (1.1), we can apply IH for the transition $r \equiv C_{\widetilde{X}}'\{\widetilde{p}/\widetilde{X}\} (\stackrel{\tau}{\longrightarrow})^k |q$, and hence
   there exists $p'$ such that $p_i \stackrel{\epsilon}{\Longrightarrow}|p'$ and $C_{\widetilde{X} }'\{\widetilde{p}/\widetilde{X} \} \stackrel{\epsilon}{\Longrightarrow} C_{\widetilde{X}}'\{\widetilde{p}\,[p'/p_i]/\widetilde{X}\} \stackrel{\epsilon}{\Longrightarrow} |q$.
  Moreover, since $X_i$ is 1-active in $C_{\widetilde{X}}$ and $p_i \stackrel{\epsilon}{\Longrightarrow}|p'$, we have $C_{\widetilde{X} }\{\widetilde{p}/\widetilde{X} \} \stackrel{\epsilon}{\Longrightarrow} C_{\widetilde{X}}\{\widetilde{p}\,[p'/p_i]/\widetilde{X}\}$ by Lemma~\ref{L:ONE_ACTION_TAU_GF}.
  We also have $C_{\widetilde{X}}\{\widetilde{p}\,[p'/p_i]/\widetilde{X}\} \stackrel{\tau}{\longrightarrow}C_{\widetilde{X}}'\{\widetilde{p}\,[p'/p_i]/\widetilde{X}\}$  by (1.3).
  Therefore, $C_{\widetilde{X} }\{\widetilde{p}/\widetilde{X} \} \stackrel{\epsilon}{\Longrightarrow} C_{\widetilde{X}}\{\widetilde{p}\,[p'/p_i]/\widetilde{X}\} \stackrel{\tau}{\longrightarrow}C_{\widetilde{X}}'\{\widetilde{p}\,[p'/p_i]/\widetilde{X}\} \stackrel{\epsilon}{\Longrightarrow} |q$, as desired.

  For the second alternative, there exist $C_{\widetilde{X}}'$, $C_{\widetilde{X},Z}''$ and $i_0 \leq |\widetilde{X}|$ such that
  \begin{enumerate}[({2.}1)]
  \renewcommand{\theenumi}{(2.\arabic{enumi})}
    \item \ $Z$ is 1-active in $C_{\widetilde{X},Z}''$,
    \item \ $r\equiv C_{\widetilde{X},Z}''\{\widetilde{p}/\widetilde{X},p_{i_0}'/Z\}$ for some $p_{i_0}'$ with $p_{i_0}\stackrel{\tau}{\longrightarrow}p_{i_0}'$, and
    \item  \ $C_{\widetilde{X} }\{\widetilde{s}/\widetilde{X}\} \stackrel{\tau}{\longrightarrow}C_{\widetilde{X},Z}''\{\widetilde{s}/\widetilde{X},s'/Z\}$ for any $\widetilde{s}$ and $s'$ with $s_{i_0}\stackrel{\tau}{\longrightarrow}s'$.
  \end{enumerate}
  In case $i_0 = i$, we have $C_{\widetilde{X}}\equiv C_{\widetilde{X}}'$ by (P-$\tau$-1), and hence $r\equiv C_{\widetilde{X}}\{\widetilde{p}\,[p_{i_0}'/p_i]/\widetilde{X}\}$ by (2.2) and (P-$\tau$-3).
  For the transition $r \equiv C_{\widetilde{X}}\{\widetilde{p}\,[p_{i_0}'/p_i]/\widetilde{X}\} (\stackrel{\tau}{\longrightarrow})^k|q$, by IH, there exists $p''$ such that $p_{i_0}'\stackrel{\epsilon}{\Longrightarrow}|p''$ and $C_{\widetilde{X}}\{\widetilde{p}\,[p_{i_0}'/p_i]/\widetilde{X}\} \stackrel{\epsilon}{\Longrightarrow} C_{\widetilde{X}}\{\widetilde{p}\,[p''/p_i]/\widetilde{X}\} \stackrel{\epsilon}{\Longrightarrow}|q$.
  Hence $p_{i_0}\stackrel{\tau}{\longrightarrow}p_{i_0}'\stackrel{\epsilon}{\Longrightarrow}|p''$ and $C_{\widetilde{X}}\{\widetilde{p} /\widetilde{X}\}\stackrel{\tau}{\longrightarrow}C_{\widetilde{X}}\{\widetilde{p}\,[p_{i_0}'/p_i]/\widetilde{X}\} \stackrel{\epsilon}{\Longrightarrow} C_{\widetilde{X}}\{\widetilde{p}\,[p''/p_i]/\widetilde{X}\} \stackrel{\epsilon}{\Longrightarrow}|q$.

  Next we consider the other case $i_0 \neq i$.
  Then for the transition $r\equiv C_{\widetilde{X},Z}''\{\widetilde{p}/\widetilde{X},p_{i_0}'/Z\}(\stackrel{\tau}{\longrightarrow})^k|q$, by IH, there exists $p'$ such that $p_i \stackrel{\epsilon}{\Longrightarrow}|p'$ and \[C_{\widetilde{X},Z}''\{\widetilde{p}/\widetilde{X},p_{i_0}'/Z\} \stackrel{\epsilon}{\Longrightarrow} C_{\widetilde{X},Z}''\{\widetilde{p}\,[p'/p_i]/\widetilde{X},p_{i_0}'/Z\} \stackrel{\epsilon}{\Longrightarrow} |q.\]
  In addition, since $X_i$ is 1-active in $C_{\widetilde{X}}$ and $p_i \stackrel{\epsilon}{\Longrightarrow}|p'$, by Lemma~\ref{L:ONE_ACTION_TAU_GF}, we obtain $C_{\widetilde{X} }\{\widetilde{p}/\widetilde{X} \} \stackrel{\epsilon}{\Longrightarrow} C_{\widetilde{X}}\{\widetilde{p}\,[p'/p_i]/\widetilde{X}\}$.
  Moreover, $C_{\widetilde{X}}\{\widetilde{p}\,[p'/p_i]/\widetilde{X}\} \stackrel{\tau}{\longrightarrow}C_{\widetilde{X},Z}''\{\widetilde{p}\,[p'/p_i]/\widetilde{X},p_{i_0}'/Z\}$  by (2.3).
  Thus
  \[C_{\widetilde{X} }\{\widetilde{p}/\widetilde{X} \} \stackrel{\epsilon}{\Longrightarrow} C_{\widetilde{X}}\{\widetilde{p}\,[p'/p_i]/\widetilde{X}\} \stackrel{\tau}{\longrightarrow} C_{\widetilde{X},Z}''\{\widetilde{p}\,[p'/p_i]/\widetilde{X},p_{i_0}'/Z\} \stackrel{\epsilon}{\Longrightarrow} |q,\]
  as desired.
\end{proof}


\section{Precongruence}

This section intends to establish a fundamental property that $\sqsubseteq_{RS}$ is a precongruence, that is, it is  substitutive w.r.t all operations in $\text{CLL}_R$.
This constitutes one of two main results of this paper.
Its proof is far from trivial and requires a solid effort.
As mentioned in Section 1, a distinguishing feature of LLTS is that it involves consideration of inconsistencies.
It is the inconsistency predicate $F$ that make everything become quite troublesome.
A crucial part in carrying out the proof is that we need to prove that $C_X\{q/X\} \in F$ implies $C_X\{p/X\} \in F$ whenever $p \sqsubseteq_{RS} q$.
Its argument will be divided into two steps.
First, we shall show that, for any stable process $p$, $C_X\{\tau.p/X\} \in F$ iff $C_X\{p/X\} \in F$.
Second, we intend to prove that $C_X\{q/X\} \in F$ implies $C_X\{p/X\} \in F$ whenever $p$ and $q$ are uniform w.r.t stability and  $p \sqsubseteq_{RS} q$.
%

\begin{mydefn}[Uniform w.r.t stability]
Two tuples $\widetilde{p}$ and $\widetilde{q}$ with $|\widetilde{q}| = |\widetilde{p}|$ are uniform w.r.t stability, in symbols 
$\widetilde{p} \bowtie \widetilde{q}$, if they are component-wise uniform w.r.t stability, that is, $p_i$ is stable iff $q_i$ is stable for each $i\leq |\widetilde{p}|$.
\end{mydefn}

An elementary property of this notion is given:

\begin{lemma}\label{L:STABLE_CONTEXT}
 The uniformity w.r.t stability are preserved under substitutions.
        That is, for any $\widetilde{p}$, $\widetilde{q}$ and $C_{\widetilde{X}}$, if $\widetilde{p} \bowtie \widetilde{q}$ then  $C_{\widetilde{X}}\{\widetilde{p}/\widetilde{X}\}  \bowtie C_{\widetilde{X}}\{\widetilde{q}/\widetilde{X}\}$.
\end{lemma}
\begin{proof}
  Immediately follows from Lemma~\ref{L:ONE_ACTION_TAU}.
\end{proof}

\noindent \textbf{Notation} For convenience, given tuples $\widetilde{p}$ and $\widetilde{q}$, for $R \in \{\sqsubseteq_{RS},\underset{\thicksim}{\sqsubset}_{RS}, \stackrel{\epsilon}{\Longrightarrow}|,\equiv \}$,  the notation $\widetilde{p}R \widetilde{q}$  means that $|\widetilde{p}|=|\widetilde{q}|$ and $p_i R q_i$ for each $i \leq |\widetilde{p}|$.


\begin{lemma}\label{L:SAME_ACTIONS}
  For any $C_{\widetilde{X}}$, $\widetilde{p}$ and $\widetilde{q}$ with $\widetilde{p} \sqsubseteq_{RS} \widetilde{q}$, if $C_{\widetilde{X}}\{\widetilde{p}/\widetilde{X}\}$ and $C_{\widetilde{X}}\{\widetilde{q}/\widetilde{X}\}$ are stable and $C_{\widetilde{X}}\{\widetilde{p}/\widetilde{X}\} \notin F $, then $C_{\widetilde{X}}\{\widetilde{p}/\widetilde{X}\} \stackrel{a}{\longrightarrow} $ iff $C_{\widetilde{X}}\{\widetilde{q}/\widetilde{X}\} \stackrel{a}{\longrightarrow} $ for any $a \in Act$.
\end{lemma}
\begin{proof}
We give the proof only for the implication from right to left, the same argument applies to the other implication.
Assume $C_{\widetilde{X}}\{\widetilde{q}/\widetilde{X}\} \stackrel{a}{\longrightarrow}  q'$.
  Then  there exist  $C_{\widetilde{X}}'$, $C_{\widetilde{X},\widetilde{Y}}'$ and $C_{\widetilde{X},\widetilde{Y}}''$ with $\widetilde{X} \cap \widetilde{Y} = \emptyset$  that satisfy (CP-$a$-1) -- (CP-$a$-4) in Lemma~\ref{L:ONE_ACTION_VISIBLE}.
        Hence, due to (CP-$a$-1) and (CP-$a$-3-i), there exist $i_Y \leq |\widetilde{X}|(Y \in \widetilde{Y})$ such that for any $\widetilde{r}$ with $|\widetilde{r}|=|\widetilde{X}|$
        \[C_{\widetilde{X}}\{\widetilde{r}/\widetilde{X}\}  \Rrightarrow C_{\widetilde{X}}' \{\widetilde{r}/\widetilde{X}\} \equiv C_{\widetilde{X},\widetilde{Y}}' \{\widetilde{r}/\widetilde{X},\widetilde{r_{i_Y}}/\widetilde{Y}\} .\tag{\ref{L:SAME_ACTIONS}.1}\]
        In particular, by Lemma~\ref{L:MULTI_STEP_UNFOLDING_ACTION}, it follows from $C_{\widetilde{X}}\{\widetilde{p}/\widetilde{X}\} \not\stackrel{\tau}{\longrightarrow}$ and  $C_{\widetilde{X}}\{\widetilde{q}/\widetilde{X}\} \not\stackrel{\tau}{\longrightarrow}$ that both  $C_{\widetilde{X},\widetilde{Y}}' \{\widetilde{p}/\widetilde{X},\widetilde{p_{i_Y}}/\widetilde{Y}\}$ and $C_{\widetilde{X},\widetilde{Y}}' \{\widetilde{q}/\widetilde{X},\widetilde{q_{i_Y}}/\widetilde{Y}\}$ are stable.
        Then, for each $Y \in \widetilde{Y}$, both $p_{i_Y}$  and $q_{i_Y}$ are stable by  Lemma~\ref{L:ONE_ACTION_TAU_GF} and (CP-$a$-2).
        Moreover, by (\ref{L:SAME_ACTIONS}.1) with $\widetilde{r} \equiv \widetilde{p}$ and Lemma~\ref{L:MULTI_STEP_UNFOLDING_FAILURE} and \ref{L:FAILURE_GF}, we have $p_{i_Y} \notin F $ for each $Y \in \widetilde{Y}$ due to $C_{\widetilde{X}}\{\widetilde{p}/\widetilde{X}\} \notin F $.
        Therefore, for each $Y \in \widetilde{Y}$, it follows from $\widetilde{p} \sqsubseteq_{RS} \widetilde{q}$ that $p_{i_Y} \underset{\thicksim}{\sqsubset}_{RS} q_{i_Y}$, and ${\mathcal I}(p_{i_Y})= {\mathcal I}(q_{i_Y})$ because of $p_{i_Y} \notin F$.
        Hence $C_{\widetilde{X}}\{\widetilde{p}/\widetilde{X}\} \stackrel{a}{\longrightarrow}$ by (CP-$a$-3-iii).
\end{proof}

In the following, we intend to show that, for any stable $p$, $C_X\{p/X\}$ and $C_X\{\tau.p/X\}$ are undifferentiated w.r.t consistency, which falls naturally into two parts: Lemmas~\ref{L:FAILURE_S_VS_NS} and \ref{L:FAILURE_NS_IMPLIES_S}.

\begin{lemma}\label{L:FAILURE_S_VS_NS}
  For any $C_X$ and stable $p$, $C_X\{p/X\}\notin F $ implies $C_X\{\tau.p/X\}\notin F $.
\end{lemma}
\begin{proof}
Let $p$ be any stable process. Set
  \[
        \Omega \triangleq \{B_{X}\{\tau.p/X\}:B_{X}\{p/X\}\notin F\;\text{and}\;B_X\;\text{is a context} \}.
  \]
  Similar to Lemma~\ref{L:RS_CON}, it suffices to prove that for any $t \in \Omega$, each proof tree of $tF$ has a proper subtree with root $sF$ for some $s \in \Omega$.
    Suppose that $C_X\{\tau.p/X\} \in \Omega$ and $\mathcal T$ is any proof tree of $Strip(\mathcal{P}_{\text{CLL}_R},M_{\text{CLL}_R}) \vdash  C_X\{\tau.p/X\}F$.
Hence $C_X\{p/X\} \notin F$.
    We distinguish six cases based on the form of $C_X$.\\

  \noindent Case 1 $C_{X}$ is   closed  or $C_{X} \equiv X$.

  In this situation, it is easy to see that  $C_{X}\{\tau.p/X\} \notin  F $.
  Hence there is no proof tree of $C_{X}\{\tau.p/X\}   F $.
  Thus the conclusion holds trivially.\\

%

  \noindent Case 2 $C_{X} \equiv \alpha.B_{X}$.

  Then the last rule applied in $\mathcal T$ is $\frac{B_{X}\{\tau.p/X\}F}{\alpha.B_{X}\{\tau.p/X\}F}$.
  Since $C_{X}\{p/X\} \notin F $, we get $B_{X}\{p/X\} \notin F $.
  Hence  $B_{X}\{\tau.p/X\} \in \Omega$; moreover, the node directly above the root of $\mathcal T$ is labelled with  $B_{X}\{\tau.p/X\}F$, as desired.\\

  \noindent Case 3 $C_{X} \equiv B_{X} \vee D_{X}$.

    Clearly, the last rule applied in $\mathcal T$ is $\frac{B_{X}\{\tau.p/X\}F,D_{X}\{\tau.p/X\}F}{B_{X}\{\tau.p/X\} \vee D_{X}\{\tau.p/X\}F}$.
     Since $C_{X}\{p/X\} \notin F $,   either $B_{X}\{p/X\} \notin F $ or $D_{X}\{p/X\} \notin F $.
     W.l.o.g, assume $B_{X}\{p/X\} \notin F $.
     Then $B_{X}\{\tau.p/X\}\in \Omega$.
     Moreover, it is obvious that $\mathcal T$ has a proper subtree with root $B_{X}\{\tau.p/X\}F$.\\

  \noindent Case 4 $C_{X} \equiv B_{X} \odot D_{X}$ with $\odot \in \{\Box, \parallel_A\}$.

  W.l.o.g, assume the last rule applied in $\mathcal T$ is $\frac{B_{X}\{\tau.p/X\}F}{B_{X}\{\tau.p/X\} \odot D_{X}\{\tau.p/X\}F}$.
  It is evident that $B_{X}\{p/X\} \notin F $ due to $C_{X}\{p/X\}  \notin F $.
  Hence $B_{X}\{\tau.p/X\} \in \Omega$, as desired.\\

  \noindent Case 5 $C_{X} \equiv \langle Y|E \rangle$.

   Then the last rule applied in $\mathcal T$ is
   \[\text{either}\; \frac{\langle t_Y|E \rangle \{\tau.p/X\}F}{\langle Y|E \rangle \{\tau.p/X\}F}\;\text{with}\; Y = t_Y \in E\;\text{or}\; \frac{\{r F:\langle Y|E \rangle \{\tau.p/X\} \stackrel{\epsilon}{\Longrightarrow}|r\}}{\langle Y|E \rangle \{\tau.p/X\} F}.\]

  For the first alternative, since $C_{X}\{p/X\} \equiv \langle Y|E \rangle \{p/X\} \notin F $, by Lemma~\ref{L:F_NORMAL}(8), we get $\langle t_Y|E \rangle \{p/X\} \notin F $.
  Hence $\langle t_Y|E \rangle \{\tau.p/X\} \in \Omega$.

  For the second alternative, since $C_{X}\{p/X\}   \notin F $, we get  $C_X \{p/X\} \stackrel{\epsilon}{\Longrightarrow}_F|q$ for some $q$.
  Moreover, by Lemma~\ref{L:MULTI_TAU_GF_STABLE}, it follows from $p \not\stackrel{\tau}{\longrightarrow}$ that there exists a stable context $C_X'$ such that
\[q \equiv C_{X}'\{p/X\}\;\text{and}\;C_X\{\tau.p/X\} \stackrel{\epsilon}{\Longrightarrow}  C_{X}'\{\tau.p/X\}. \tag{\ref{L:FAILURE_S_VS_NS}.1}\]
 Further, by Lemma~\ref{L:STABILIZATION} and $\tau.p \stackrel{\tau}{\longrightarrow} |p$, we get \[C_{X}'\{\tau.p/X\}\stackrel{\epsilon}{\Longrightarrow} |s\;\text{for some}\; s.\tag{\ref{L:FAILURE_S_VS_NS}.2}\]
  For the above transition, by Lemma~\ref{L:MULTI_TAU_GF_STABLE} again, there exists $C_{X,\widetilde{Z}}''$ with $X \notin \widetilde{Z}$  such that
\[s \equiv C_{X,\widetilde{Z}}''\{\tau.p/X,p/\widetilde{Z}\}\;\text{and}\;C_{X}'\{p/X\}  \Rrightarrow  C_{X,\widetilde{Z}}''\{p/X,p/\widetilde{Z}\}.\]
  Thus, by Lemma~\ref{L:MULTI_STEP_UNFOLDING_FAILURE}, we have $C_{X,\widetilde{Z}}''\{p/X,p/\widetilde{Z}\} \notin F $ because of $q \equiv C_{X}'\{p/X\} \notin F $.
  Set
  \[C_X''' \triangleq C_{X,\widetilde{Z}}''\{p/\widetilde{Z}\}.\]
  Then it follows from  $C_{X}'''\{p/X\} \equiv C_{X,\widetilde{Z}}''\{p/X,p/\widetilde{Z}\}\notin F  $ that $s \equiv C_{X}'''\{\tau.p/X\} \in \Omega$.
Moreover, $\mathcal T$ contains a proper subtree with root $sF$ due to (\ref{L:FAILURE_S_VS_NS}.1) and (\ref{L:FAILURE_S_VS_NS}.2).\\

  \noindent Case 6 $C_{X} \equiv B_{X} \wedge D_{X}$.

  Clearly, the last rule applied in $\mathcal T$ has one of the following formats.\\

 \noindent Case 6.1 $\frac{B_{X}\{\tau.p/X\}F}{B_{X}\{\tau.p/X\} \wedge  D_{X}\{\tau.p/X\}F}$.

  Similar to Case 4, omitted.\\

 \noindent Case 6.2 $\frac{B_{X}\{\tau.p/X\} \stackrel{a}{\longrightarrow} r,D_{X}\{\tau.p/X\} \not\stackrel{a}{\longrightarrow}, B_{X}\{\tau.p/X\}  \wedge D_{X}\{\tau.p/X\}  \not\stackrel{\tau}{\longrightarrow}}{B_{X}\{\tau.p/X\}  \wedge D_{X}\{\tau.p/X\}  F}$.

%

 In this situation,
 $B_X\{\tau.p/X\}$, $C_X  $ and $B_X$ are stable.
 Moreover, since $p$ is stable, so is $B_{X}\{p/X\}$.
 Due to $C_X\{p/X\} \notin F $, we obtain $B_X\{p/X\} \notin F $.
 Then, by Lemma~\ref{L:SAME_ACTIONS}, it follows from $p =_{RS} \tau.p$ and $B_{X}\{\tau.p/X\} \stackrel{a}{\longrightarrow} $ that
 \[B_{X}\{p/X\} \stackrel{a}{\longrightarrow}.\tag{\ref{L:FAILURE_S_VS_NS}.3}\]
 Similarly, it follows from $D_X\{\tau.p/X\} \not\stackrel{a}{\longrightarrow}$ that
 \[D_{X}\{p/X\} \not\stackrel{a}{\longrightarrow}.\tag{\ref{L:FAILURE_S_VS_NS}.4}\]
 In addition, since $B_{X} \wedge D_{X}$ and $p$ are stable, so is $B_{X}\{p/X\}  \wedge D_{X}\{p/X\}$.
 So, by (\ref{L:FAILURE_S_VS_NS}.3) and (\ref{L:FAILURE_S_VS_NS}.4), we get $C_{X}\{p/X\}\equiv B_{X}\{p/X\}  \wedge D_{X}\{p/X\} \in F $ by Rule $Rp_{10}$, which contradicts that $C_X\{\tau.p/X\} \in \Omega$.
 Hence this case is impossible. \\

 \noindent Case 6.3 $\frac{C_{X}\{\tau.p/X\}  \stackrel{\alpha}{\longrightarrow}s,\{rF:C_{X}\{\tau.p/X\}  \stackrel{\alpha}{\longrightarrow}r\}}{C_{X}\{\tau.p/X\}F}$.

 The argument splits into two cases based on $\alpha$.\\

\noindent Case 6.3.1 $\alpha = \tau$.

We distinguish two cases depending on whether $C_{X}$ is stable.\\

\noindent Case 6.3.1.1 $C_{X}$ is not stable.

 Since $C_{X}\{p/X\} \notin F $, we have $C_{X}\{p/X\} \stackrel{\epsilon}{\Longrightarrow}_F|p'$ for some $p'$.
 Moreover, by Lemma~\ref{L:TAU_ACTION_NORMALIZATION},  there exist $p''$ and stable $C_X^*$ such that
\[C_{X}\{p/X\} \stackrel{\epsilon}{\Longrightarrow}  C_{X}^*\{p/X\}  \stackrel{\epsilon}{\Longrightarrow}  |p''  \Rrightarrow  p'\]
and
\[C_{X}\{t/X\} \stackrel{\epsilon}{\Longrightarrow}  C_{X}^*\{t/X\}\;\text{for any}\; t.\]
 Further, since $C_{X}$ is not stable and $p \not \stackrel{\tau}{\longrightarrow}$, by Lemma~\ref{L:ONE_ACTION_TAU}, there exists $C_X'$ such that
 \[C_{X}\{p/X\} \stackrel{\tau}{\longrightarrow}  C_{X}'\{p/X\}\stackrel{\epsilon}{\Longrightarrow}  C_X^*\{p/X\} \;\text{and}\; C_{X}\{\tau.p/X\} \stackrel{\tau}{\longrightarrow}  C_{X}'\{\tau.p/X\}.\]
 Since $p' \notin F $ and $p'' \Rrightarrow p'$, by Lemma~\ref{L:MULTI_STEP_UNFOLDING_FAILURE}, we get $p'' \notin F $.
 Together with the transitions $C_{X}'\{p/X\} \stackrel{\epsilon}{\Longrightarrow} C_{X}^*\{p/X\}  \stackrel{\epsilon}{\Longrightarrow} |p''$, by Lemma~\ref{L:FAILURE_TAU_I}, this implies $C_{X}'\{p/X\} \notin F$.
 Hence $C_{X}'\{\tau.p/X\} \in \Omega$, and $\mathcal T$ has a proper subtree with root $C_{X}'\{\tau.p/X\}F$.\\

\noindent Case 6.3.1.2 $C_{X}$ is stable.

 Due to $C_{X}\{\tau.p/X\}\stackrel{\tau}{\longrightarrow}s$, either the clause (1) or (2) in  Lemma~\ref{L:ONE_ACTION_TAU} holds.
 Since $C_X$ is stable, by (C-$\tau$-2) in Lemma~\ref{L:ONE_ACTION_TAU}, it is easy to see that the clause (1) does not hold, and hence the clause (2) holds, that is, there exists   $C_{X,Z}'$ with $X \neq Z$  such that
 \[C_X\{\tau.p/X\} \stackrel{\tau}{\longrightarrow} C_{X,Z}'\{\tau.p/X,p/Z\}\;\text{and}\;C_X\{p/X\}  \Rrightarrow  C_{X,Z}'\{p/X,p/Z\}.\]
 Set \[C_X'' \triangleq C_{X,Z}'\{p/Z\}.\]
 Hence $\mathcal T$ has a proper subtree with root $C_{X}''\{\tau.p/X\}F$.
Moreover, by Lemma~\ref{L:MULTI_STEP_UNFOLDING_FAILURE}, it follows from $C_X\{p/X\} \notin F $ that $C_{X,Z}'\{p/X,p/Z\} \notin F $.
 Thus $C_X''\{\tau.p/X\} \equiv C_{X,Z}'\{\tau.p/X,p/Z\} \in \Omega$, as desired.\\

\noindent Case 6.3.2 $\alpha \in Act$.

    Then it is not difficult to know that both $C_X$ and $C_X\{p/X\}$ are stable.
 Moreover, since $C_X\{\tau.p/X\}\stackrel{\alpha}{\longrightarrow}$, $\tau.p =_{RS} p$ and $C_X\{p/X\} \notin F$, by Lemma~\ref{L:SAME_ACTIONS}, we get $C_X\{p/X\}\stackrel{\alpha}{\longrightarrow}$.
 Further, by Theorem~\ref{L:LLTS}, it follows from $C_X\{p/X\} \notin F $ that $C_X\{p/X\}\stackrel{\alpha}{\longrightarrow}_F q$ for some $q$.
 For such $\alpha$-labelled transition, by Lemma~\ref{L:ONE_ACTION_VISIBLE}, there exist $C_X'$, $C_{X,\widetilde{Z}}'$ and $C_{X,\widetilde{Z}}''$ with $X \notin \widetilde{Z}$  that realize (CP-$a$-1) -- (CP-$a$-4).

In order to complete the proof, we intend to prove that $\widetilde{Z} = \emptyset$.
        On the contrary, suppose $\widetilde{Z} \neq \emptyset$.
        Then, by (CP-$a$-2) and (CP-$a$-3-i), there exists an active occurrence of the variable $X$ in $C_X'$.
        So, by Lemma~\ref{L:ONE_ACTION_TAU_GF}, $C_X'\{\tau.p/X\} \stackrel{\tau}{\longrightarrow} $.
        Then, by Lemma~\ref{L:MULTI_STEP_UNFOLDING_ACTION}, it follows from $C_X\{\tau.p/X\}  \Rrightarrow C_X' \{\tau.p /X\}$ (i.e., (CP-$a$-1)) that $C_X\{\tau.p/X\} \stackrel{\tau}{\longrightarrow} $, which contradicts $C_X\{\tau.p/X\} \stackrel{\alpha}{\longrightarrow} $.

        Thus $\widetilde{Z} = \emptyset$, and hence $q \equiv C_{X,\widetilde{Z}}''\{p/X\}$ by (CP-$a$-3-ii).
        Since $C_X\{\tau.p/X\}$ is stable, by (CP-$a$-3-iii), we get $C_X\{\tau.p/X\}\stackrel{\alpha}{\longrightarrow} C_{X,\widetilde{Z}}''\{\tau.p/X\} $.
        Thus, $\mathcal T$ contains a proper subtree with root $C_{X,\widetilde{Z}}''\{\tau.p/X\}F$;
        moreover, $C_{X,\widetilde{Z}}''\{\tau.p/X\} \in \Omega$ due to $C_{X,\widetilde{Z}}''\{p/X\} \equiv q \notin F$.\\

 \noindent Case 6.4 $\frac{\{rF: B_{X}\{\tau.p/X\}  \wedge D_{X}\{\tau.p/X\} \stackrel{\epsilon}{\Longrightarrow}|r\}}{B_{X}\{\tau.p/X\}  \wedge D_{X}\{\tau.p/X\} F}$.

 Analogous to the second alternative in Case 5, omitted.
\end{proof}

In order to show the converse of the above result, the preliminary result below is given.
Here, for any finite set $S$ of processes, by virtue of the commutative and associative laws of external choice \cite{Zhang11}, we may introduce the notation of a generalized external choice (denoted by $\underset{p\in S}{\square}p$) by the standard method.

\begin{lemma}\label{L:MPLACE_HOLDER}
  Let $t_1,t_2$ be two terms and $\{X\} \cup \widetilde{Z}$ a tuple of variables such that none of recursive variable occurring in $t_i$(with $i=1,2$) is in $\{X\} \cup \widetilde{Z}$.
 Suppose that $Z$ is active in $t_1,t_2$ for each $Z \in \widetilde{Z}$ and
  \[T \triangleq
    \left\{
    \begin{array}{ll}
      \underset{Z \in \widetilde{Z}}{\square}\alpha.a_Z.0  &\text{if}\; \widetilde{Z} \neq \emptyset\\
      &\\
      a_X.0 & \text{otherwise}
    \end{array}
  \right.\]
where $a_X$ and $\widetilde{a_Z}$ are distinct fresh visible actions and $\alpha \in Act$. Then
  \begin{enumerate}[(1)]
    \renewcommand{\theenumi}{(\arabic{enumi})}
    \item if $t_1\{T/X,\widetilde{a_Z.0}/\widetilde{Z}\} \equiv t_2\{T/X,\widetilde{a_Z.0}/\widetilde{Z}\} $ then $t_1\{p/X,\widetilde{q}/\widetilde{Z}\} \equiv t_2\{p/X,\widetilde{q}/\widetilde{Z}\}$ for any $p$ and $\widetilde{q}$;
    \item if $t_1\{T/X,\widetilde{a_Z.0}/\widetilde{Z}\} \Rrightarrow_1 t_2\{T/X,\widetilde{a_Z.0}/\widetilde{Z}\} $ then $t_1\{p/X,\widetilde{q}/\widetilde{Z}\} \Rrightarrow_1 t_2\{p/X,\widetilde{q}/\widetilde{Z}\}$ for any $p$ and $\widetilde{q}$.
  \end{enumerate}
\end{lemma}
\begin{proof}
\noindent \textbf{(1)} 
    It proceeds by induction on $t_1$. We distinguish three cases as follows.\\

\noindent Case  1 $t_1$ is closed or $t_1$ is of the format $X$ or $\beta.s$ or $ s_1\vee s_2$ or $\langle Y|E \rangle$.

Since $Z$ is active in $t_1$ for each $Z \in \widetilde{Z}$, we get $\widetilde{Z} = \emptyset$.
Then it follows by Lemma~\ref{L:PLACE_HOLDER}.\\

 \noindent Case  2 $t_1 \equiv s_1\odot s_2$ with $\odot \in \{\parallel_A, \wedge\}$.

 Then  $t_1\{T/X,\widetilde{a_Z.0}/\widetilde{Z}\} \equiv s_1\{T/X,\widetilde{a_Z.0}/\widetilde{Z}\} \odot s_2\{T/X,\widetilde{a_Z.0}/\widetilde{Z}\}\equiv t_2\{T/X,\widetilde{a_Z.0}/\widetilde{Z}\}$.
 Since neither $\widetilde{a_Z.0}$ nor $T$ contain $\odot$, there exist $s_1',s_2'$ such that $s_1\{T/X,\widetilde{a_Z.0}/\widetilde{Z}\} \equiv  s_1'\{T/X,\widetilde{a_Z.0}/\widetilde{Z}\}$, $s_2\{T/X,\widetilde{a_Z.0}/\widetilde{Z}\}\equiv  s_2'\{T/X,\widetilde{a_Z.0}/\widetilde{Z}\}$ and $t_2 \equiv s_1' \odot s_2'$.
 Hence it immediately follows that  $t_1\{p/X,\widetilde{q}/\widetilde{Z}\} \equiv t_2\{p/X,\widetilde{q}/\widetilde{Z}\}$ for any $p$ and $\widetilde{q}$ by IH.\\

  \noindent Case  3 $t_1 \equiv s_1 \Box s_2$.

 Then  $t_1\{T/X,\widetilde{a_Z.0}/\widetilde{Z}\} \equiv s_1\{T/X,\widetilde{a_Z.0}/\widetilde{Z}\} \Box s_2\{T/X,\widetilde{a_Z.0}/\widetilde{Z}\}\equiv t_2\{T/X,\widetilde{a_Z.0}/\widetilde{Z}\}$.
Hence the topmost operator of $t_2\{T/X,\widetilde{a_Z.0}/\widetilde{Z}\}$ is an external choice $\Box$.
Clearly, such operator comes from either $T$ or $t_2$.
For the former, we get $t_2 \equiv X$.
If $|\widetilde{Z}| \leq 1$ then  $t_2\{T/X,\widetilde{a_Z.0}/\widetilde{Z}\}( \equiv T)$ does not contain the operator $\Box$ at all, a contradiction.
Next we treat the other case $|\widetilde{Z}| > 1$.
Clearly, $a_Z.0$ is guarded in $T$ for each $Z \in \widetilde{Z}$.
So $FV(t_1)\cap \widetilde{Z} = \emptyset$ and the numbers of visible actions of $s_1\{T/X,\widetilde{a_Z.0}/\widetilde{Z}\}$ and $s_2\{T/X,\widetilde{a_Z.0}/\widetilde{Z}\}$ are different from ones of two operands of the topmost external choice of $T$.
Hence this case is impossible and $t_2 \equiv s_1' \Box s_2'$ for some $s_1'$ and $s_2'$.
The rest of the proof is as in Case 2. \\

%

\noindent \textbf{(2)} If   $FV(t_1) \cap \widetilde{Z} = \emptyset$, it follows by Lemma~\ref{L:PLACE_HOLDER}.
Next we consider the other case $FV(t_1) \cap \widetilde{Z} \neq \emptyset$.
     It proceeds by induction on $t_1$.
    Since $Z$ is active in $t_1$ for each $Z \in \widetilde{Z}$, we get either $t_1 \equiv Z$  or $t_1 \equiv s_1 \odot s_2$ for some $s_1$ and $s_2$, where $Z \in \widetilde{Z}$ and $\odot \in \{\wedge, \parallel_A,\Box\}$.
     We give the proof only for the case $t_1\equiv s_1 \Box s_2$, the proofs for the remaining cases are straightforward and omitted.

It follows from $t_1\equiv s_1 \Box s_2$ that
\[t_1\{T/X,\widetilde{a_Z.0}/\widetilde{Z}\} \equiv s_1\{T/X,\widetilde{a_Z.0}/\widetilde{Z}\} \Box s_2\{T/X,\widetilde{a_Z.0}/\widetilde{Z}\} \Rrightarrow_1 t_2\{T/X,\widetilde{a_Z.0}/\widetilde{Z}\}.\]
So the topmost operator of $t_2\{T/X,\widetilde{a_Z.0}/\widetilde{Z}\}$ is an external choice $\Box$ which   comes from either $T$ or $t_2$.
Similar to Case~3 in the proof for item (1), we can make the conclusion that there exist $s_1',s_2'$ such that $t_2 \equiv s_1' \Box s_2'$.
Moreover, it is easily seen that either $s_1\{T/X,\widetilde{a_Z.0}/\widetilde{Z}\}$ or $s_2\{T/X,\widetilde{a_Z.0}/\widetilde{Z}\}$ triggers the unfolding from $t_1\{T/X,\widetilde{a_Z.0}/\widetilde{Z}\}$ to $t_2\{T/X,\widetilde{a_Z.0}/\widetilde{Z}\}$.
W.l.o.g, we consider the first alternative.
Then $s_1\{T/X,\widetilde{a_Z.0}/\widetilde{Z}\} \Rrightarrow_1  s_1'\{T/X,\widetilde{a_Z.0}/\widetilde{Z}\}$ and  $s_2\{T/X,\widetilde{a_Z.0}/\widetilde{Z}\}\equiv  s_2'\{T/X,\widetilde{a_Z.0}/\widetilde{Z}\}$.
Hence, by IH and item (1), for any $p$ and $\widetilde{q}$, we have $s_1\{p/X,\widetilde{q}/\widetilde{Z}\} \Rrightarrow_1  s_1'\{p/X,\widetilde{q}/\widetilde{Z}\}$ and  $s_2\{p/X,\widetilde{q}/\widetilde{Z}\}\equiv  s_2'\{p/X,\widetilde{q}/\widetilde{Z}\}$.
Therefore, $t_1\{p/X,\widetilde{q}/\widetilde{Z}\} \equiv s_1\{p/X,\widetilde{q}/\widetilde{Z}\} \Box s_2\{p/X,\widetilde{q}/\widetilde{Z}\} \Rrightarrow_1 t_2\{p/X,\widetilde{q}/\widetilde{Z}\}$.
\end{proof}

The next lemma establishes the converse of Lemma~\ref{L:FAILURE_S_VS_NS}.

\begin{lemma}\label{L:FAILURE_NS_IMPLIES_S}
  For any $C_X$ and stable process $p$, $C_X\{\tau.p/X\}\notin F $ implies $C_X\{p/X\}\notin F $.
\end{lemma}
\begin{proof}
Let $p$ be any stable process. Set
  \[
        \Omega \triangleq \{B_{X}\{p/X\}:B_{X}\{\tau.p/X\}\notin F\;\text{and}\;B_X\;\text{is a context} \}.
  \]
Assume $t \in \Omega$.
Then $t \equiv C_{X}\{p/X\}  $ for some $C_X$ such that $C_{X}\{\tau.p/X\}\notin F $.
Let $\mathcal T$ be any proof tree of $Strip(\mathcal{P}_{\text{CLL}_R} ,M_{\text{CLL}_R} ) \vdash C_{X}\{p/X\}F$.
Similar to Lemma~\ref{L:FAILURE_S_VS_NS}, it is sufficient to prove that $ \mathcal T$ has a proper subtree with root $sF$ for some $s \in \Omega$, which is a routine case analysis based on the last rule applied in $\mathcal T$.
Here we  treat only three primary cases.\\

  \noindent Case 1 $\frac{\{rF: C_X \{p/X\} \stackrel{\epsilon}{\Longrightarrow}|r\}}{C_X \{p/X\} F}$ with $C_X \equiv \langle Y|E \rangle$.

    Since $C_{X}\{\tau.p/X\} \notin F $, we get  $C_X\{\tau.p/X\} \stackrel{\epsilon}{\Longrightarrow}_F|q$ for some $q$.
    By Lemma~\ref{L:MULTI_TAU_GF_STABLE}, for this transition, there exists a stable context $C_{X,\widetilde{Z}}'$ satisfying (MS-$\tau$-1) -- (MS-$\tau$-7).
    In particular, since $p$ and $q$ are stable, by (MS-$\tau$-2,7), we have
    \[q \equiv C_{X,\widetilde{Z}}'\{\tau.p/X,p/\widetilde{Z}\}\notin F.\]
    Moreover,  since each $Z(\in \widetilde{Z})$ is 1-active in $C_{X,\widetilde{Z}}'$ (i.e., (MS-$\tau$-1)) and $\tau.p \stackrel{\tau}{\longrightarrow}p$, by Lemma~\ref{L:ONE_ACTION_TAU_GF}, we get $C_{X,\widetilde{Z}}'\{\tau.p/X,\tau.p/\widetilde{Z}\} \stackrel{\epsilon}{\Longrightarrow} C_{X,\widetilde{Z}}'\{\tau.p/X,p/\widetilde{Z}\} \equiv q \notin F$, which, by Lemma~\ref{L:FAILURE_TAU_I}, implies
  \[C_{X,\widetilde{Z}}'\{\tau.p/X,\tau.p/\widetilde{Z}\}\notin F .\tag{\ref{L:FAILURE_NS_IMPLIES_S}.1}\]
  Let $a_X$ be any fresh visible action.
  By (MS-$\tau$-3-i), it follows from $a_X.0 \stackrel{\epsilon}{\Longrightarrow} |a_X.0$ that there exists $s$ such that
  \[C_X\{a_X.0/X\} \stackrel{\epsilon}{\Longrightarrow} s  \Rrightarrow  C_{X,\widetilde{Z}}'\{ a_X.0/X, a_X.0/\widetilde{Z}\}. \tag{\ref{L:FAILURE_NS_IMPLIES_S}.2}\]
  Since $a_X.0$  and $C_{X,\widetilde{Z}}'$  are stable, so is $C_{X,\widetilde{Z}}'\{a_X.0/X,a_X.0/\widetilde{Z}\}$ by Lemma~\ref{L:ONE_ACTION_TAU}.
  Then, by Lemma~\ref{L:MULTI_STEP_UNFOLDING_ACTION}, $s$ is stable.
  Thus, for the transition in (\ref{L:FAILURE_NS_IMPLIES_S}.2), by Lemma~\ref{L:MULTI_TAU_GF_STABLE}, there exists a stable context $C_X^*$ such that
  \[s \equiv C_X^*\{a_X.0/X\}\;\text{and}\; C_X\{r/X\} \stackrel{\epsilon}{\Longrightarrow}C_X^*\{r/X\}\;\text{for any}\;r. \tag{\ref{L:FAILURE_NS_IMPLIES_S}.3}\]
  Then, by Lemma~\ref{L:PLACE_HOLDER}, it follows from  $s \equiv C_X^*\{a_X.0/X\}  \Rrightarrow  C_{X,\widetilde{Z}}'\{ a_X.0/X, a_X.0/\widetilde{Z}\}$ that
  \[C_X^*\{\tau.p/X\}  \Rrightarrow  C_{X,\widetilde{Z}}'\{ \tau.p/X, \tau.p/\widetilde{Z}\}.\]
  Hence $C_X^*\{\tau.p/X\} \notin F $ by  (\ref{L:FAILURE_NS_IMPLIES_S}.1) and Lemma~\ref{L:MULTI_STEP_UNFOLDING_FAILURE}, which implies $ C_X^*\{p/X\} \in \Omega$.
  Moreover, since $C_X^*$ and $p$ are stable, so is $C_X^*\{p/X\}$ by Lemma~\ref{L:ONE_ACTION_TAU}, which implies
  $C_X\{p/X\}\stackrel{\epsilon}{\Longrightarrow}|C_X^*\{p/X\}$ by (\ref{L:FAILURE_NS_IMPLIES_S}.3).
  Therefore, $\mathcal T$ has a proper subtree with root $C_X^*\{p/X\}  F $.\\

 \noindent Case 2 $\frac{B_{X}\{p/X\} \stackrel{a}{\longrightarrow}r, D_{X}\{p/X\} \not\stackrel{a}{\longrightarrow}, C_{X}\{p/X\}  \not\stackrel{\tau}{\longrightarrow}}{C_X\{p/X\}  F}$ with $C_X \equiv B_{X}  \wedge D_{X}$.

 Clearly, in this situation, both $B_X$ and $D_X$ are stable.
 Since $C_X\{\tau.p/X\} \notin F $, we have $C_X\{\tau.p/X\} \stackrel{\epsilon}{\Longrightarrow}_F |q$ for some $q$.
 So, there exist $s$ and $t$ such that $q\equiv s \wedge t$ and
 \[B_X\{\tau.p/X\} \stackrel{\epsilon}{\Longrightarrow}_F|s\;\text{and}\; D_X\{\tau.p/X\}  \stackrel{\epsilon}{\Longrightarrow}_F|t.\]
 Then, for these two transitions, by Lemma~\ref{L:MULTI_TAU_GF_STABLE}, there exist $B_{X,\widetilde{Y}}'$ and $D_{X,\widetilde{Z}}'$ satisfying (MS-$\tau$-1) -- (MS-$\tau$-7) respectively.
 In particular, since $p$, $B_X$ and $D_X$ are stable, by (MS-$\tau$-2,4,7), we have
 \begin{enumerate}[(1)]
 \renewcommand{\theenumi}{(\arabic{enumi})}
   \item \ $s \equiv B_{X,\widetilde{Y}}'\{\tau.p/X,p/\widetilde{Y}\}$ and $B_X\{p/X\} \Rrightarrow B_{X,\widetilde{Y}}'\{p/X,p/\widetilde{Y}\}$;
   \item \ $t \equiv D_{X,\widetilde{Z}}'\{\tau.p/X,p/\widetilde{Z}\}$ and $D_X\{p/X\} \Rrightarrow  D_{X,\widetilde{Z}}'\{p/X,p/\widetilde{Z}\}$.
 \end{enumerate}
 Hence, by Lemma~\ref{L:MULTI_STEP_UNFOLDING_ACTION}, it follows from $B_{X}\{p/X\} \stackrel{a}{\longrightarrow}$ and $D_{X}\{p/X\} \not\stackrel{a}{\longrightarrow}$ that \[B_{X,\widetilde{Y}}'\{p/X,p/\widetilde{Y}\} \stackrel{a}{\longrightarrow}\;\text{and}\; D_{X,\widetilde{Z}}'\{p/X,p/\widetilde{Z}\} \not\stackrel{a}{\longrightarrow}.\]
 Further, since $B_{X,\widetilde{Y}}'\{p/X,p/\widetilde{Y}\}$ and $B_{X,\widetilde{Y}}'\{\tau.p/X,p/\widetilde{Y}\}$ are stable, by Lemma~\ref{L:SAME_ACTIONS}, it follows from $\tau.p =_{RS}p$  and $s \equiv B_{X,\widetilde{Y}}'\{\tau.p/X,p/\widetilde{Y}\} \notin F$ that $B_{X,\widetilde{Y}}'\{\tau.p/X,p/\widetilde{Y}\} \stackrel{a}{\longrightarrow} $.
 Similarly, we also have $D_{X,\widetilde{Z}}'\{\tau.p/X,p/\widetilde{Z}\} \not\stackrel{a}{\longrightarrow} $.
 Hence $q\equiv s\wedge t \in F $ by Rule $Rp_{10}$, a contradiction.
 Thus this case is impossible.\\

 \noindent Case 3 $\frac{C_{X}\{p/X\}  \stackrel{\alpha}{\longrightarrow}r',\{r F:C_{X}\{p/X\}  \stackrel{\alpha}{\longrightarrow}r\}}{C_{X}\{p/X\}F}$ with $C_X \equiv B_{X}  \wedge D_{X}$.

Since $C_X\{\tau.p/X\} \notin F $, we have
\[C_X\{\tau.p/X\} \stackrel{\epsilon}{\Longrightarrow}_F|q \;\text{for some}\;q.\tag{\ref{L:FAILURE_NS_IMPLIES_S}.4}\]
Next we distinguish two cases based on $\alpha$.\\

\noindent Case 3.1  $\alpha = \tau$.

 By (\ref{L:FAILURE_NS_IMPLIES_S}.4) and Lemma~\ref{L:TAU_ACTION_NORMALIZATION}, there exist $t$ and stable context $C_X^*$ such that
 \[C_{X}\{\tau.p/X\} \stackrel{\epsilon}{\Longrightarrow}  C_{X}^*\{\tau.p/X\}  \stackrel{\epsilon}{\Longrightarrow}  |t  \Rrightarrow  q\notin F\]
 and
 \[C_{X}\{p/X\} \stackrel{\epsilon}{\Longrightarrow} C_{X}^*\{p/X\}.\]
 Moreover, since $p\not\stackrel{\tau}{\longrightarrow}$ and $\tau \in {\mathcal I}(C_X \{p/X\})$, by Lemma~\ref{L:ONE_ACTION_TAU}, there exists a context $C_X'$ such that
 \[C_{X}\{p/X\} \stackrel{\tau}{\longrightarrow}  C_{X}'\{p/X\} \stackrel{\epsilon}{\Longrightarrow}  C_{X}^*\{p/X\}\]
 and
 \[C_{X}\{\tau.p/X\} \stackrel{\tau}{\longrightarrow}  C_{X}'\{\tau.p/X\} \stackrel{\epsilon}{\Longrightarrow}  C_{X}^*\{\tau.p/X\}\stackrel{\epsilon}{\Longrightarrow}|t.\]
 Further, by Lemma~\ref{L:MULTI_STEP_UNFOLDING_FAILURE}, it follows from $q \notin F $ and $t  \Rrightarrow  q$ that $t \notin F $.
 Then, by Lemma~\ref{L:FAILURE_TAU_I} and the transition above, we have $C_{X}'\{\tau.p/X\} \notin F$.
 Hence $C_{X}'\{p/X\}\in \Omega$ and one of nodes directly above the root of $\mathcal T$ is labelled with $C_{X}'\{p/X\}F$, as desired.\\

\noindent Case 3.2  $\alpha \in Act$.

    In this case, $C_X$ is stable by Lemma~\ref{L:STABLE_CONTEXT_I}.
    By (\ref{L:FAILURE_NS_IMPLIES_S}.4) and Lemma~\ref{L:MULTI_TAU_GF_STABLE}, there exists a stable context $C_{X,\widetilde{Y}}'$ with $X \notin \widetilde{Y}$ that satisfies (MS-$\tau$-1) -- (MS-$\tau$-7).
    Then $q \equiv C_{X,\widetilde{Y}}'\{\tau.p/X,p/\widetilde{Y}\}$ due to $p \not\stackrel{\tau}{\longrightarrow}$ and (MS-$\tau$-2).
 Moreover, since $C_X$ is stable, by (MS-$\tau$-4), we have
 \[C_X\{r/X\} \Rrightarrow C_{X,\widetilde{Y}}'\{r/X,r/\widetilde{Y}\}\;\text{for any}\; r.\tag{\ref{L:FAILURE_NS_IMPLIES_S}.5}\]
 Then, by $C_X\{p/X\}\stackrel{\alpha}{\longrightarrow}$ and Lemma~\ref{L:MULTI_STEP_UNFOLDING_ACTION}, we get
 \[C_{X,\widetilde{Y}}'\{p/X,p/\widetilde{Y}\} \stackrel{\alpha}{\longrightarrow}.\tag{\ref{L:FAILURE_NS_IMPLIES_S}.6}\]
 Further, 
 by Lemma~\ref{L:SAME_ACTIONS}, we also have $C_{X,\widetilde{Y}}'\{\tau.p/X,p/\widetilde{Y}\} \stackrel{\alpha}{\longrightarrow} $ because of $\tau.p=_{RS}p$ and $q \equiv C_{X,\widetilde{Y}}'\{\tau.p/X,p/\widetilde{Y}\} \notin F$.
 Thus, by Theorem~\ref{L:LLTS}, we obtain
 \[ C_{X,\widetilde{Y}}'\{\tau.p/X,p/\widetilde{Y}\}  \stackrel{\alpha}{\longrightarrow}_F t\;\text{for some}\; t.\]
 For such $\alpha$-labelled transition, by Lemma~\ref{L:ONE_ACTION_VISIBLE}, there exist $C_{X,\widetilde{Y}}''$, $C_{X,\widetilde{Y},\widetilde{Z}}''$ and  $C_{X,\widetilde{Y},\widetilde{Z}}'''$ with $(\{X\} \cup \widetilde{Y}) \cap \widetilde{Z} =\emptyset$ that realize (CP-$a$-1) -- (CP-$a$-4).
   In particular, due to $\tau.p \not\stackrel{\alpha}{\longrightarrow}$ and (CP-$a$-3-ii),  there exist $p_Z'(Z \in \widetilde{Z})$ such that
\[p \stackrel{\alpha}{\longrightarrow}p_Z'\;\text{for each}\;Z \in \widetilde{Z}\;\text{and}\;t \equiv C_{X,\widetilde{Y},\widetilde{Z}}'''\{\tau.p/X,p/\widetilde{Y},\widetilde{p_Z'}/\widetilde{Z}\} \notin F. \tag{\ref{L:FAILURE_NS_IMPLIES_S}.7}\]
Moreover, by (CP-$a$-3-iii), for any $r,s$ and $s_Z'(Z \in \widetilde{Z})$ such that $s \stackrel{\alpha}{\longrightarrow} s_Z'$ for each $Z \in \widetilde{Z}$, we have
\[C_{X,\widetilde{Y}}'\{r/X,s/\widetilde{Y}\} \stackrel{\alpha}{\longrightarrow} C_{X,\widetilde{Y},\widetilde{Z}}'''\{r/X,s/\widetilde{Y},\widetilde{s_{Z}'}/\widetilde{Z}\} \;\text{whenever}\;C_{X,\widetilde{Y}}'\{r/X,s/\widetilde{Y}\}\;\text{is stable}. \tag{\ref{L:FAILURE_NS_IMPLIES_S}.8}\]
For each $Z \in \widetilde{Z} \cup \{X\}$, we fix a fresh and distinct visible action $a_Z$ and set
 \[ T \triangleq \left\{
                   \begin{array}{ll}
                     \underset{Z \in \widetilde{Z}}{\square}\alpha.a_Z.0,  &\text{if}\;\widetilde{Z} \neq \emptyset; \\
                     & \\
                     a_X.0,  & \text{otherwise}.
                   \end{array}
                 \right.
\]
 Since $T$ and $C_{X,\widetilde{Y}}'$ are stable, so is $C_{X,\widetilde{Y}}'\{T/X,T/\widetilde{Y}\}$ by Lemma~\ref{L:ONE_ACTION_TAU}.
 Then, by (\ref{L:FAILURE_NS_IMPLIES_S}.8), we have
 \[C_{X,\widetilde{Y}}'\{T/X,T/\widetilde{Y}\} \stackrel{\alpha}{\longrightarrow} C_{X,\widetilde{Y},\widetilde{Z}}'''\{T/X,T/\widetilde{Y}, \widetilde{a_Z.0}/\widetilde{Z}\}.\]
 So, by Lemma~\ref{L:MULTI_STEP_UNFOLDING_ACTION}, it follows from (\ref{L:FAILURE_NS_IMPLIES_S}.5) that there exists $t'$ such that
 \[C_X\{T/X\} \stackrel{\alpha}{\longrightarrow}t'\;\text{and}\;t'  \Rrightarrow  C_{X,\widetilde{Y},\widetilde{Z}}'''\{T/X,T/\widetilde{Y},\widetilde{a_Z.0}/\widetilde{Z}\}. \tag{\ref{L:FAILURE_NS_IMPLIES_S}.9}\]
Then, by Lemma~\ref{L:ONE_ACTION_VISIBLE}, it is not difficult to see that there exists a context $B_{X,\widetilde{Z}}$ that satisfies the conditions:
\begin{enumerate}[(a)]
\renewcommand{\theenumi}{(\alph{enumi}) }
   \item \ $t' \equiv B_{X,\widetilde{Z}}\{T/X,\widetilde{a_Z.0}/\widetilde{Z}\}$;
   \item \ none of $a_Z$ with $Z \in \widetilde{Z}$ occurs in $B_{X,\widetilde{Z}}$;
   \item \ for any $s$ and $s_Z'(Z \in \widetilde{Z})$ such that $s \stackrel{\alpha}{\longrightarrow} s_Z'$ for each $Z \in \widetilde{Z}$,
 \[C_X\{s/X\} \stackrel{\alpha}{\longrightarrow} B_{X,\widetilde{Z}}\{s/X,\widetilde{s_Z'}/\widetilde{Z}\} \;\text{whenever}\; C_X\{s/X\}\;\text{is stable}.\]
 \end{enumerate}

 Now we obtain the diagram
 \[      \begin{array}{ccc}
                    & \text{by}\; (\ref{L:FAILURE_NS_IMPLIES_S}.5) &  \\
                  C_{X}\{p/X\}  &  \Rrightarrow  & C_{X,\widetilde{Y}}'\{p/X,p/\widetilde{Y}\}   \\
                                     &  & \\
                 \quad\quad \quad\quad \downarrow \alpha \;\text{by (c)} &  & \quad\quad\quad\quad \downarrow\alpha  \;\text{by (\ref{L:FAILURE_NS_IMPLIES_S}.6)\;\text{and}\;(\ref{L:FAILURE_NS_IMPLIES_S}.8)}\\
                     & \text{by}\;(\ref{L:FAILURE_NS_IMPLIES_S}.9),\;(a)\;\text{and}\;\text{Lemma~\ref{L:MPLACE_HOLDER}}& \\
                  B_{X,\widetilde{Z}}\{p/X,\widetilde{p_Z'}/\widetilde{Z}\}  &  \Rrightarrow  &
                  C_{X,\widetilde{Y},\widetilde{Z}}''' \{p/X,p/\widetilde{Y},\widetilde{p_Z'}/\widetilde{Z}\}
                \end{array}
 \]
By Lemma~\ref{L:MPLACE_HOLDER}, we also have
\[B_{X,\widetilde{Z}}\{\tau.p/X,\widetilde{p_Z'}/\widetilde{Z}\}    \Rrightarrow
                  C_{X,\widetilde{Y},\widetilde{Z}}''' \{\tau.p/X,\tau.p/\widetilde{Y},\widetilde{p_Z'}/\widetilde{Z}\}.\tag{\ref{L:FAILURE_NS_IMPLIES_S}.10}\]
For each $Y \in \widetilde{Y}$, since $Y$ is 1-active in $C_{X,\widetilde{Y}}'$, by Lemma~\ref{L:ONE_STEP_UNFOLDING_VARIABLE}(1)(2) and $C_{X,\widetilde{Y}}'  \Rrightarrow C_{X,\widetilde{Y}}''$ (i.e., (CP-$a$-1)), so it is  in $C_{X,\widetilde{Y}}''$.
Moreover, by (CP-$a$-4-i,ii), for each $Y \in \widetilde{Y}\cap  FV(C_{X,\widetilde{Y},\widetilde{Z}}''')$, $Y$ is 1-active in $C_{X,\widetilde{Y},\widetilde{Z}}'''$.
Then, by Lemma~\ref{L:ONE_ACTION_TAU_GF}, we have
\[C_{X,\widetilde{Y},\widetilde{Z}}'''\{\tau.p/X,\tau.p/\widetilde{Y},\widetilde{p_Z'}/\widetilde{Z}\} \stackrel{\epsilon}{\Longrightarrow}C_{X,\widetilde{Y},\widetilde{Z}}'''\{\tau.p/X,p/\widetilde{Y},\widetilde{p_Z'}/\widetilde{Z}\}\]
which, together with (\ref{L:FAILURE_NS_IMPLIES_S}.7), implies $C_{X,\widetilde{Y},\widetilde{Z}}'''\{\tau.p/X,\tau.p/\widetilde{Y},\widetilde{p_Z'}/\widetilde{Z}\} \notin F $ by Lemma~\ref{L:FAILURE_TAU_I}.
Hence, by Lemma~\ref{L:MULTI_STEP_UNFOLDING_FAILURE}, it follows from (\ref{L:FAILURE_NS_IMPLIES_S}.10) that $B_{X,\widetilde{Z}}\{\tau.p/X,\widetilde{p_Z'}/\widetilde{Z}\}\notin F $.
Thus, $B_{X,\widetilde{Z}}\{p/X,\widetilde{p_Z'}/\widetilde{Z}\}\in \Omega$;
moreover, $\mathcal T$ has a proper subtree with root $B_{X,\widetilde{Z}}\{p/X,\widetilde{p_Z'}/\widetilde{Z}\}F$ due to (c) and (\ref{L:FAILURE_NS_IMPLIES_S}.7).
\end{proof}

Hitherto we have completed the first step mentioned at the beginning of this section.
Now we return to carry out the second step.
Before proving Lemma~\ref{L:FAILURE_CONGRUENCE}, a result concerning proof tree is given first.

\begin{lemma}\label{L:FAILURE_CONGRUENCE_PRE}
Let  $C_{\widetilde{X},\widetilde{Z}}$ be any context such that for each $Z \in \widetilde{Z}$, $Z$ is active and occurs at most once.
If $\widetilde{p}$, $\widetilde{q}$, $\widetilde{t}$, $\widetilde{s}$ and $\widetilde{r}$ are any processes such that
\begin{enumerate}[(a)]
\renewcommand{\theenumi}{(\alph{enumi})}
  \item \ $\widetilde{p}\sqsubseteq_{RS} \widetilde{q}$,
  \item \ $\widetilde{p} \bowtie \widetilde{q}$,
  \item \ $\widetilde{r} \stackrel{\epsilon}{\Longrightarrow} |\widetilde{t} $,
  \item \ $\widetilde{s}\underset{\thicksim}{\sqsubset}_{RS}\widetilde{t}$, and
  \item \ $C_{\widetilde{X},\widetilde{Z}}\{\widetilde{p}/\widetilde{X},\widetilde{s}/\widetilde{Z}\} \notin F $,
\end{enumerate}
 then, for any proof tree $\mathcal T$ for $Strip(\mathcal{P}_{\text{CLL}_R},M_{\text{CLL}_R}) \vdash C_{\widetilde{X},\widetilde{Z}}\{\widetilde{q}/\widetilde{X},\widetilde{r}/\widetilde{Z}\}F$, there exist  $C_{\widetilde{X},\widetilde{Z},\widetilde{Y}}^*$ and $p_Y'',q_Y''(Y\in \widetilde{Y})$  such that
  \begin{enumerate}[(1)]
  \renewcommand{\theenumi}{(\arabic{enumi})}
    \item \ $\mathcal T$ has a subtree with root $C_{\widetilde{X},\widetilde{Z},\widetilde{Y}}^*\{\widetilde{q}/\widetilde{X},\widetilde{t}/\widetilde{Z},\widetilde{q_Y''}/\widetilde{Y}\}F$,
    \item \ $C_{\widetilde{X},\widetilde{Z},\widetilde{Y}}^*\{\widetilde{p}/\widetilde{X},\widetilde{s}/\widetilde{Z},\widetilde{p_Y''}/\widetilde{Y}\} \notin F $, and
    \item \ $\widetilde{p_Y''}\underset{\thicksim}{\sqsubset}_{RS}\widetilde{q_Y''}$.
  \end{enumerate}
\end{lemma}
\begin{proof}
  It proceeds by induction on the depth of $\mathcal T$.
  We distinguish different cases depending on the form of $C_{\widetilde{X},\widetilde{Z}}$.\\

  \noindent Case 1 $C_{\widetilde{X},\widetilde{Z}}$ is closed or $C_{\widetilde{X},\widetilde{Z}} \equiv X_i$ or $C_{\widetilde{X},\widetilde{Z}} \equiv Z_j$ for some $i \leq |\widetilde{X}|$ and $j \leq |\widetilde{Z}|$.

  It is straightforward to show that this lemma holds trivially for such case.
  As a example, we consider the case $C_{\widetilde{X},\widetilde{Z}} \equiv Z_j$.
   Since $C_{\widetilde{X},\widetilde{Z}}\{\widetilde{p}/\widetilde{X},\widetilde{s}/\widetilde{Z}\}  \equiv s_j \notin F $ and $\widetilde{s} \underset{\thicksim}{\sqsubset}_{RS} \widetilde{t}$,
   we have $t_j \notin F $.
   Hence $r_j \stackrel{\epsilon}{\Longrightarrow}_F|t_j$ by Lemma~\ref{L:FAILURE_TAU_I}.
   So $C_{\widetilde{X},\widetilde{Z}}\{\widetilde{q}/\widetilde{X},\widetilde{r}/\widetilde{Z}\} \equiv r_j \notin F $.
   That is, there is no proof tree of $Strip(\mathcal{P}_{\text{CLL}_R},M_{\text{CLL}_R}) \vdash C_{\widetilde{X},\widetilde{Z}}\{\widetilde{q}/\widetilde{X},\widetilde{r}/\widetilde{Z}\}F$.
   Thus the conclusion holds trivially.\\

  \noindent Case 2 $C_{\widetilde{X},\widetilde{Z}}$ is of the format $\alpha.B_{\widetilde{X},\widetilde{Z}}$ or $ B_{\widetilde{X},\widetilde{Z}} \vee D_{\widetilde{X},\widetilde{Z}}$ or $\langle Y|E \rangle$.

  For these three formats, since each $Z(\in \widetilde{Z})$ is active in $C_{\widetilde{X},\widetilde{Z}}$, it is obvious that $\widetilde{Z} = \emptyset$.
  Thus $C_{\widetilde{X},\widetilde{Z}}\{\widetilde{q}/\widetilde{X},\widetilde{r}/\widetilde{Z}\} \equiv C_{\widetilde{X},\widetilde{Z}}\{\widetilde{q}/\widetilde{X},\widetilde{t}/\widetilde{Z}\}$.
  So, $\mathcal T$ has the root labelled with $C_{\widetilde{X},\widetilde{Z}}\{\widetilde{q}/\widetilde{X},\widetilde{t}/\widetilde{Z}\}F$.
  Therefore, the conclusion holds by setting $C_{\widetilde{X},\widetilde{Z},\widetilde{Y}}^*\triangleq C_{\widetilde{X},\widetilde{Z}}$ with $\widetilde{Y}= \emptyset$.\\

  \noindent Case 3 $C_{\widetilde{X},\widetilde{Z}} \equiv B_{\widetilde{X},\widetilde{Z}} \odot D_{\widetilde{X},\widetilde{Z}}$ with $\odot \in \{\Box, \parallel_A\}$.

  W.l.o.g, assume the last rule applied in  $\mathcal T$ is  \[\frac{B_{\widetilde{X},\widetilde{Z}}\{\widetilde{q}/\widetilde{X},\widetilde{r}/\widetilde{Z}\}F}{B_{\widetilde{X},\widetilde{Z}}\{\widetilde{q}/\widetilde{X},\widetilde{r}/\widetilde{Z}\} \odot D_{\widetilde{X},\widetilde{Z}}\{\widetilde{q}/\widetilde{X},\widetilde{r}/\widetilde{Z}\}F}.\]
  Then $\mathcal T$ has a proper subtree $\mathcal T'$ with root $B_{\widetilde{X},\widetilde{Z}}\{\widetilde{q}/\widetilde{X},\widetilde{r}/\widetilde{Z}\}F$.
  Since $C_{\widetilde{X},\widetilde{Z}}\{\widetilde{p}/\widetilde{X},\widetilde{s}/\widetilde{Z}\}  \notin F $, we get $B_{\widetilde{X},\widetilde{Z}}\{\widetilde{p}/\widetilde{X},\widetilde{s}/\widetilde{Z}\} \notin F $.
  Then the conclusion immediately follows by applying  IH on $\mathcal T'$.\\

  \noindent Case 4 $C_{\widetilde{X},\widetilde{Z}} \equiv B_{\widetilde{X},\widetilde{Z}} \wedge D_{\widetilde{X},\widetilde{Z}}$.

  The argument splits into four cases based on the last rule applied in $\mathcal T$.\\

 \noindent Case 4.1 $\frac{B_{\widetilde{X},\widetilde{Z}}\{\widetilde{q}/\widetilde{X},\widetilde{r}/\widetilde{Z}\}F}{B_{\widetilde{X},\widetilde{Z}}\{\widetilde{q}/\widetilde{X},\widetilde{r}/\widetilde{Z}\} \wedge D_{\widetilde{X},\widetilde{Z}}\{\widetilde{q}/\widetilde{X},\widetilde{r}/\widetilde{Z}\}F}$.

  Similar to Case 3, omitted.\\

 \noindent Case 4.2 $\frac{B_{\widetilde{X},\widetilde{Z}}\{\widetilde{q}/\widetilde{X},\widetilde{r}/\widetilde{Z}\} \stackrel{a}{\longrightarrow} r', \; D_{\widetilde{X},\widetilde{Z}}\{\widetilde{q}/\widetilde{X},\widetilde{r}/\widetilde{Z}\} \not\stackrel{a}{\longrightarrow}, \; C_{\widetilde{X},\widetilde{Z}}\{\widetilde{q}/\widetilde{X},\widetilde{r}/\widetilde{Z}\}   \not\stackrel{\tau}{\longrightarrow}}{B_{\widetilde{X},\widetilde{Z}}\{\widetilde{q}/\widetilde{X},\widetilde{r}/\widetilde{Z}\} \wedge D_{\widetilde{X},\widetilde{Z}}\{\widetilde{q}/\widetilde{X},\widetilde{r}/\widetilde{Z}\} F}$.



 For any $Z (\in \widetilde{Z})$ occurring in $C_{\widetilde{X},\widetilde{Z}}$, since $Z$ is active and $C_{\widetilde{X},\widetilde{Z}}\{\widetilde{q}/\widetilde{X},\widetilde{r}/\widetilde{Z}\}   \not\stackrel{\tau}{\longrightarrow}$, by Lemma~\ref{L:ONE_ACTION_TAU_GF}, we have $r_Z \not\stackrel{\tau}{\longrightarrow} $, and hence $r_Z \equiv t_Z$ because of (c).
 So, $C_{\widetilde{X},\widetilde{Z}}\{\widetilde{q}/\widetilde{X},\widetilde{r}/\widetilde{Z}\} \equiv C_{\widetilde{X},\widetilde{Z}}\{\widetilde{q}/\widetilde{X},\widetilde{t}/\widetilde{Z}\}$.
    Hence $\mathcal T$ has the root labelled with $C_{\widetilde{X},\widetilde{Z}}\{\widetilde{q}/\widetilde{X},\widetilde{t}/\widetilde{Z}\}F$.
    Clearly, the conclusion holds by setting $C_{\widetilde{X},\widetilde{Z},\widetilde{Y}}^*\triangleq C_{\widetilde{X},\widetilde{Z}}$ with $\widetilde{Y}=\emptyset$.\\

 \noindent Case 4.3 $\frac{C_{\widetilde{X},\widetilde{Z}}\{\widetilde{q}/\widetilde{X},\widetilde{r}/\widetilde{Z}\}   \stackrel{\alpha}{\longrightarrow}s',\;\{rF: C_{\widetilde{X},\widetilde{Z}}\{\widetilde{q}/\widetilde{X},\widetilde{r}/\widetilde{Z}\}  \stackrel{\alpha}{\longrightarrow}r\}}{C_{\widetilde{X},\widetilde{Z}}\{\widetilde{q}/\widetilde{X},\widetilde{r}/\widetilde{Z}\}F}$.

 If $\alpha \in Act$, the argument is similar to one of Case 4.2 and omitted.
 In the following, we handle the case $\alpha = \tau$.
 If $r_Z \not\stackrel{\tau}{\longrightarrow}$ for any $Z (\in \widetilde{Z})$ occurring in $C_{\widetilde{X},\widetilde{Z}}$,
 then the conclusion holds trivially by putting $C_{\widetilde{X},\widetilde{Z},\widetilde{Y}}^*\triangleq C_{\widetilde{X},\widetilde{Z}}$ with $\widetilde{Y} = \emptyset$.
 Next we consider the other case where $r_{Z_0} \stackrel{\tau}{\longrightarrow} $ for some $Z_0(\in \widetilde{Z})$ occurring in $C_{\widetilde{X},\widetilde{Z}}$.
 Then $r_{Z_0} \stackrel{\tau}{\longrightarrow} r'\stackrel{\epsilon}{\Longrightarrow}|t_{Z_0}$ for some $r'$ by (c);
 moreover, $Z_0$ is 1-active in $C_{\widetilde{X},\widetilde{Z}}$.
 Thus, by Lemma~\ref{L:ONE_ACTION_TAU_GF}, we get \[C_{\widetilde{X},\widetilde{Z}}\{\widetilde{q}/\widetilde{X},\widetilde{r}/\widetilde{Z}\}   \stackrel{\tau}{\longrightarrow}  C_{\widetilde{X},\widetilde{Z}}\{\widetilde{q}/\widetilde{X},\widetilde{r}\,[r'/r_{Z_0}]/\widetilde{Z}\}.\]
 So, $\mathcal T$ has a proper subtree $\mathcal T'$ with root $C_{\widetilde{X},\widetilde{Z}}\{\widetilde{q}/\widetilde{X},\widetilde{r}\,[r'/r_{Z_0}]/\widetilde{Z}\}F$.
 Since $\widetilde{r}\,[r'/r_{Z_0}] \stackrel{\epsilon}{\Longrightarrow} |\widetilde{t} $ and   $C_{\widetilde{X},\widetilde{Z}}\{\widetilde{p}/\widetilde{X},\widetilde{s}/\widetilde{Z}\} \notin F $, by IH, $\mathcal T'$ has a subtree with root $C_{\widetilde{X},\widetilde{Z},\widetilde{Y}}^*\{\widetilde{q}/\widetilde{X},\widetilde{t}/\widetilde{Z},\widetilde{q_Y''}/\widetilde{Y}\}F$ for some $C_{\widetilde{X},\widetilde{Z},\widetilde{Y}}^*$, $\widetilde{p_Y''}$ and $\widetilde{q_Y''}$ such that $C_{\widetilde{X},\widetilde{Z},\widetilde{Y}}^*\{\widetilde{p}/\widetilde{X},\widetilde{s}/\widetilde{Z},\widetilde{p_Y''}/\widetilde{Y}\} \notin F $ and $\widetilde{p_Y''}\underset{\thicksim}{\sqsubset}_{RS}\widetilde{q_Y''}$.\\

 \noindent Case 4.4 $\frac{\{rF: C_{\widetilde{X},\widetilde{Z}}\{\widetilde{q}/\widetilde{X},\widetilde{r}/\widetilde{Z}\}  \stackrel{\epsilon}{\Longrightarrow}|r\}}{C_{\widetilde{X},\widetilde{Z}}\{\widetilde{q}/\widetilde{X},\widetilde{r}/\widetilde{Z}\}F}$.

 It follows from $C_{\widetilde{X},\widetilde{Z}}\{\widetilde{p}/\widetilde{X},\widetilde{s}/\widetilde{Z}\} \notin F $ that \[C_{\widetilde{X},\widetilde{Z}}\{\widetilde{p}/\widetilde{X},\widetilde{s}/\widetilde{Z}\} \stackrel{\epsilon}{\Longrightarrow}_F| p'\;\text{for some}\; p'.\]
 Then, by Lemma~\ref{L:MULTI_TAU_GF_STABLE}, for such transition,
  there exist a stable context $C_{\widetilde{X},\widetilde{Z},\widetilde{Y}}'$ and $i_Y,p_Y'''(Y \in \widetilde{Y})$ that realize (MS-$\tau$-1) -- (MS-$\tau$-7).
   In particular, since each $s (\in \widetilde{s})$ is stable, by (MS-$\tau$-2,7), we have $i_Y \leq |\widetilde{X}|$ for each $Y \in \widetilde{Y}$ and
   \[p_{i_Y} \stackrel{\tau}{\Longrightarrow}|p_{Y}'''\;\text{for each}\;Y \in \widetilde{Y}\;\text{and}\;p'\equiv C_{\widetilde{X},\widetilde{Z},\widetilde{Y}}'\{\widetilde{p}/\widetilde{X},\widetilde{s}/\widetilde{Z},\widetilde{p_{Y}'''}/\widetilde{Y}\} \notin F.\]
   Then, by Lemma~\ref{L:FAILURE_GF}, it follows from (MS-$\tau$-1) 
   that, for each $Y \in \widetilde{Y}$,  $p_{Y}''' \notin F $ and hence $p_{i_Y}\stackrel{\tau}{\Longrightarrow}_F|p_{Y}'''$ by Lemma~\ref{L:FAILURE_TAU_I}.
   Further, since $\widetilde{p} \bowtie \widetilde{q}$ and $\widetilde{p}\sqsubseteq_{RS}\widetilde{q}$, there exist $q_{Y}'''(Y \in \widetilde{Y})$ such that
   \[q_{i_Y}\stackrel{\tau}{\Longrightarrow}_F|q_{Y}'''\;\text{and}\; p_{Y}'''\underset{\thicksim}{\sqsubset}_{RS} q_{Y}'''\;\text{for each}\;Y \in \widetilde{Y}.\]
   Then it follows from (MS-$\tau$-3-ii) 
   that
   \[C_{\widetilde{X},\widetilde{Z}}\{\widetilde{q}/\widetilde{X},\widetilde{t}/\widetilde{Z}\} \stackrel{\epsilon}{\Longrightarrow} C_{\widetilde{X},\widetilde{Z},\widetilde{Y}}'\{\widetilde{q}/\widetilde{X},\widetilde{t}/\widetilde{Z},\widetilde{q_{Y}'''}/\widetilde{Y}\}.\tag{\ref{L:FAILURE_CONGRUENCE_PRE}.1}\]   Moreover, since $Z$ is active and occurs at most once in $C_{\widetilde{X},\widetilde{Z}}$ for each $Z \in \widetilde{Z}$,  by   Lemma~\ref{L:ONE_ACTION_TAU_GF}, it follows from $\widetilde{r} \stackrel{\epsilon}{\Longrightarrow} \widetilde{t}$ that
    \[C_{\widetilde{X},\widetilde{Z}}\{\widetilde{q}/\widetilde{X},\widetilde{r}/\widetilde{Z}\} \stackrel{\epsilon}{\Longrightarrow} C_{\widetilde{X},\widetilde{Z}}\{\widetilde{q}/\widetilde{X},\widetilde{t}/\widetilde{Z}\}. \tag{\ref{L:FAILURE_CONGRUENCE_PRE}.2}\]
   Since $\widetilde{p} \bowtie \widetilde{q}$, $\widetilde{s}\underset{\thicksim}{\sqsubset}_{RS} \widetilde{t}$ and $\widetilde{p_{Y}'''}\underset{\thicksim}{\sqsubset}_{RS} \widetilde{q_{Y}'''}$, by $p'\equiv C_{\widetilde{X},\widetilde{Z},\widetilde{Y}}'\{\widetilde{p}/\widetilde{X},\widetilde{s}/\widetilde{Z},\widetilde{p_{Y}'''}/\widetilde{Y}\} \not\stackrel{\tau}{\longrightarrow}$ and Lemma~\ref{L:ONE_ACTION_TAU},  we can conclude that $C_{\widetilde{X},\widetilde{Z},\widetilde{Y}}'\{\widetilde{q}/\widetilde{X},\widetilde{t}/\widetilde{Z},\widetilde{q_{Y}'''}/\widetilde{Y}\}$ is stable.
   Hence $\mathcal T$ has a proper subtree with root $C_{\widetilde{X},\widetilde{Z},\widetilde{Y}}'\{\widetilde{q}/\widetilde{X},\widetilde{t}/\widetilde{Z},\widetilde{q_{Y}'''}/\widetilde{Y}\}F$ by  (\ref{L:FAILURE_CONGRUENCE_PRE}.1) and (\ref{L:FAILURE_CONGRUENCE_PRE}.2); moreover, we also have  $p' \equiv C_{\widetilde{X},\widetilde{Z},\widetilde{Y}}'\{\widetilde{p}/\widetilde{X},\widetilde{s}/\widetilde{Z},\widetilde{p_{Y}'''}/\widetilde{Y}\} \notin F $ and $\widetilde{p_{Y}'''}\underset{\thicksim}{\sqsubset}_{RS}\widetilde{q_{Y}'''}$.
   Consequently, $C_{\widetilde{X},\widetilde{Z},\widetilde{Y}}'$, $\widetilde{p_Y'''}$ and $\widetilde{q_Y'''}$ are what we seek.
\end{proof}

\begin{lemma}\label{L:FAILURE_CONGRUENCE}
   For any $C_{\widetilde{X}}$ and processes $\widetilde{r}$ and $\widetilde{s}$, if $\widetilde{r} \bowtie \widetilde{s}$  and $\widetilde{r}\sqsubseteq_{RS} \widetilde{s}$, then  $C_{\widetilde{X}}\{\widetilde{r}/\widetilde{X}\} \notin F $ implies $C_{\widetilde{X}}\{\widetilde{s}/\widetilde{X}\} \notin F $.
\end{lemma}
\begin{proof}
Set
\[
        \Omega = \{B_{\widetilde{X}}\{\widetilde{q}/\widetilde{X}\}:
        \widetilde{p} \bowtie \widetilde{q}, \;\widetilde{p} \sqsubseteq_{RS}\widetilde{q},\;
        B_{\widetilde{X}}\{\widetilde{p}/\widetilde{X}\}\notin F \;\text{and}\; B_X\;\text{is a context}\}.
\]
Let $C_{\widetilde{X}}\{\widetilde{q}/\widetilde{X}\} \in \Omega$ and $\mathcal T$ be any proof tree of $Strip(\mathcal{P}_{\text{CLL}_R} ,M_{\text{CLL}_R} ) \vdash C_{\widetilde{X}}\{\widetilde{q}/\widetilde{X}\}F$.
Similar to Lemma~\ref{L:FAILURE_S_VS_NS}, it suffices to show that $\mathcal T$ has a proper subtree with root $sF$ for some $s \in \Omega$.
We distinguish six cases based on the form of $C_{\widetilde{X}}$.\\

  \noindent Case 1 $C_{\widetilde{X}} $ is closed or $C_{\widetilde{X}} \equiv X_i$.

    In this situation, $C_{\widetilde{X}}\{\widetilde{q}/\widetilde{X}\} \notin F$ because of $C_{\widetilde{X}}\{\widetilde{p}/\widetilde{X}\} \notin F$ and $\widetilde{p} \sqsubseteq_{RS}\widetilde{q}$.
    Thus there is no proof tree of $C_{\widetilde{X}}\{\widetilde{q}/\widetilde{X}\}F$.
Hence the conclusion holds trivially.\\

  \noindent Case 2 $C_{\widetilde{X}} \equiv \alpha.B_{\widetilde{X}}$.

  Then the last rule applied in $\mathcal T$ is $\frac{B_{\widetilde{X}}\{\widetilde{q}/\widetilde{X}\}F}{\alpha.B_{\widetilde{X}}\{\widetilde{q}/\widetilde{X}\}F}$.
  Moreover $B_{\widetilde{X}}\{\widetilde{p}/\widetilde{X}\} \notin F $ due to $C_{\widetilde{X}}\{\widetilde{p}/\widetilde{X}\}  \notin F $.
  Hence $B_{\widetilde{X}}\{\widetilde{q}/\widetilde{X}\}\in \Omega$, as desired.\\

  \noindent Case 3 $C_{\widetilde{X}} \equiv B_{\widetilde{X}} \vee D_{\widetilde{X}}$.

     Obviously, the last rule applied in $\mathcal T$ is $\frac{B_{\widetilde{X}}\{\widetilde{q}/\widetilde{X}\}F,D_{\widetilde{X}}\{\widetilde{q}/\widetilde{X}\}F}{B_{\widetilde{X}}\{\widetilde{q}/\widetilde{X}\} \vee D_{\widetilde{X}}\{\widetilde{q}/\widetilde{X}\}F}$.
     Due to $C_{\widetilde{X}}\{\widetilde{p}/\widetilde{X}\} \notin F $,  we have either $B_{\widetilde{X}}\{\widetilde{p}/\widetilde{X}\} \notin F $ or $D_{\widetilde{X}}\{\widetilde{p}/\widetilde{X}\} \notin F $, which implies $B_{\widetilde{X}}\{\widetilde{q}/\widetilde{X}\} \in \Omega $ or $D_{\widetilde{X}}\{\widetilde{q}/\widetilde{X}\} \in \Omega $.
     Thus $\mathcal T$ contains a proper subtree with root $sF$ for some $s \in \Omega$.\\

  \noindent Case 4 $C_{\widetilde{X}} \equiv B_{\widetilde{X}} \odot D_{\widetilde{X}}$ with $\odot \in \{\Box, \parallel_A\}$.

  W.l.o.g, assume the last rule applied in $\mathcal T$ is  $\frac{B_{\widetilde{X}}\{\widetilde{q}/\widetilde{X}\}F}{B_{\widetilde{X}}\{\widetilde{q}/\widetilde{X}\} \odot D_{\widetilde{X}}\{\widetilde{q}/\widetilde{X}\}F}$.
  Since $C_{\widetilde{X}}\{\widetilde{p}/\widetilde{X}\}  \notin F $, we get $B_{\widetilde{X}}\{\widetilde{p}/\widetilde{X}\} \notin F $, which implies $B_{\widetilde{X}}\{\widetilde{q}/\widetilde{X}\} \in \Omega$, as desired.\\

  \noindent Case 5 $C_{\widetilde{X}} \equiv \langle Y|E \rangle$.

  Clearly, the last rule applied in $\mathcal T$ is
  \[\text{either}\; \frac{\langle t_Y|E \rangle \{\widetilde{q}/\widetilde{X}\}F}{\langle Y|E \rangle \{\widetilde{q}/\widetilde{X}\}F}\;\text{with}\;Y =t_Y \in E\;\text{or}\;\frac{\{r F:\langle Y|E \rangle \{\widetilde{q}/\widetilde{X}\} \stackrel{\epsilon}{\Longrightarrow}|r\}}{\langle Y|E \rangle \{\widetilde{q}/\widetilde{X}\} F}.\]

  For the first alternative, we have $\langle t_Y|E \rangle \{\widetilde{p}/\widetilde{X}\} \notin F $ because of $C_{\widetilde{X}}\{\widetilde{p}/\widetilde{X}\} \notin F $, and hence $\langle t_Y|E \rangle \{\widetilde{q}/\widetilde{X}\}  \in \Omega$.

  For the second alternative, due to $C_{\widetilde{X}}\{\widetilde{p}/\widetilde{X}\}  \notin F $, we get
  \[C_{\widetilde{X}}\{\widetilde{p}/\widetilde{X}\} \stackrel{\epsilon}{\Longrightarrow}_F|s\;\text{for some} \;s.\]
  For such transition, by Lemma~\ref{L:MULTI_TAU_GF_STABLE}, there exist $C_{\widetilde{X},\widetilde{Z}}'$ and $i_Z \leq |\widetilde{X}|,p_{Z}'(Z \in \widetilde{Z})$ that realize (MS-$\tau$-1) -- (MS-$\tau$-7).
  Amongst them, by (MS-$\tau$-2,7), we have
  \[p_{i_Z} \stackrel{\tau}{\Longrightarrow}|p_{Z}'\;\text{for each}\; Z \in \widetilde{Z} \;\text{and}\;s \equiv C_{\widetilde{X},\widetilde{Z}}'\{\widetilde{p}/\widetilde{X},\widetilde{p_{Z}'}/\widetilde{Z}\} \notin F. \tag{\ref{L:FAILURE_CONGRUENCE}.1}\]
  Thus, for each $Z \in  \widetilde{Z}$, by (MS-$\tau$-1) and Lemma~\ref{L:FAILURE_GF}, it follows that $p_{Z}'\notin F $,  and hence $p_{i_Z}\stackrel{\tau}{\Longrightarrow}_F |p_{Z}'$ by Lemma~\ref{L:FAILURE_TAU_I}.
  Further, since $\widetilde{p} \bowtie \widetilde{q}$, it follows from  $\widetilde{p} \sqsubseteq_{RS}\widetilde{q}$ that for each  $Z \in \widetilde{Z}$,
  \[q_{i_Z}\stackrel{\tau}{\Longrightarrow}_F|q_{Z}'\;\text{and}\;p_{Z}' \underset{\thicksim}{\sqsubset}_{RS} q_{Z}'\;\text{for some} \; q_{Z}'.\tag{\ref{L:FAILURE_CONGRUENCE}.2}\]
  Then $C_{\widetilde{X}}\{\widetilde{q}/\widetilde{X}\} \stackrel{\epsilon}{\Longrightarrow}   C_{\widetilde{X},\widetilde{Z}}'\{\widetilde{q}/\widetilde{X},\widetilde{q_{Z}'}/\widetilde{Z}\}$   by (MS-$\tau$-3-ii).
  In addition, since $\widetilde{p} \bowtie \widetilde{q}$ and $\widetilde{p_{Z}'} \underset{\thicksim}{\sqsubset}_{RS} \widetilde{q_{Z}'}$, by Lemma~\ref{L:ONE_ACTION_TAU},
  it follows from $s \equiv C_{\widetilde{X},\widetilde{Z}}'\{\widetilde{p}/\widetilde{X},\widetilde{p_{Z}'}/\widetilde{Z}\} \not\stackrel{\tau}{\longrightarrow}$ that $C_{\widetilde{X},\widetilde{Z}}'\{\widetilde{q}/\widetilde{X},\widetilde{q_{Z}'}/\widetilde{Z}\}$ is stable.
  Therefore
  \[C_{\widetilde{X}}\{\widetilde{q}/\widetilde{X}\} \stackrel{\epsilon}{\Longrightarrow}   | C_{\widetilde{X},\widetilde{Z}}'\{\widetilde{q}/\widetilde{X},\widetilde{q_{Z}'}/\widetilde{Z}\}.\]
  Hence $\mathcal T$ has a proper subtree with root $C_{\widetilde{X},\widetilde{Z}}'\{\widetilde{q}/\widetilde{X},\widetilde{q_{Z}'}/\widetilde{Z}\}F $;
  moreover $C_{\widetilde{X},\widetilde{Z}}'\{\widetilde{q}/\widetilde{X},\widetilde{q_{Z}'}/\widetilde{Z}\} \in \Omega$ due to (\ref{L:FAILURE_CONGRUENCE}.1) and (\ref{L:FAILURE_CONGRUENCE}.2).\\

  \noindent Case 6 $C_{\widetilde{X}} \equiv B_{\widetilde{X}} \wedge D_{\widetilde{X}}$.

  The argument  splits into five subcases depending on the last rule applied in $\mathcal T$.\\

 \noindent Case 6.1 $\frac{B_{\widetilde{X}}\{\widetilde{q}/\widetilde{X}\}F}{B_{\widetilde{X}}\{\widetilde{q}/\widetilde{X}\} \wedge  D_{\widetilde{X}}\{\widetilde{q}/\widetilde{X}\}F}$.

  Similar to Case 4, omitted.\\

 \noindent Case 6.2 $\frac{B_{\widetilde{X}}\{\widetilde{q}/\widetilde{X} \}\stackrel{a}{\longrightarrow} r, D_{\widetilde{X}}\{\widetilde{q}/\widetilde{X}\} \not\stackrel{a}{\longrightarrow}, C_{\widetilde{X}}\{\widetilde{q}/\widetilde{X}\}  \not\stackrel{\tau}{\longrightarrow}}{B_{\widetilde{X}}\{\widetilde{q}/\widetilde{X}\}  \wedge D_{\widetilde{X}}\{\widetilde{q}/\widetilde{X}\}  F}$.
%

Clearly, $B_{\widetilde{X}}\{\widetilde{p}/\widetilde{X}\} \notin F$ and $D_{\widetilde{X}}\{\widetilde{p}/\widetilde{X}\} \notin F$ due to $C_{\widetilde{X}}\{\widetilde{p}/\widetilde{X}\} \notin F$.
 Moreover, by $\widetilde{p} \bowtie \widetilde{q}$ and Lemma~\ref{L:ONE_ACTION_TAU}, we have $C_{\widetilde{X}}\{\widetilde{p}/\widetilde{X}\}   \not\stackrel{\tau}{\longrightarrow} $ 
because of $C_{\widetilde{X}}\{\widetilde{q}/\widetilde{X}\}   \not\stackrel{\tau}{\longrightarrow} $.
 Further, by Lemma~\ref{L:SAME_ACTIONS},  we also have  $B_{\widetilde{X}}\{\widetilde{p}/\widetilde{X}\} \stackrel{a}{\longrightarrow} $ and $D_{\widetilde{X}}\{\widetilde{p}/\widetilde{X}\} \not\stackrel{a}{\longrightarrow} $.
 So, $C_{\widetilde{X}}\{\widetilde{p}/\widetilde{X}\} \equiv B_{\widetilde{X}}\{\widetilde{p}/\widetilde{X}\}  \wedge D_{\widetilde{X}}\{\widetilde{p}/\widetilde{X}\} \in F $ by Rule $Rp_{10}$, which contradicts $C_{\widetilde{X}}\{\widetilde{q}/\widetilde{X}\} \in \Omega$.
 Hence such case is impossible.\\

 \noindent Case 6.3 $\frac{C_{\widetilde{X}}\{\widetilde{q}/\widetilde{X}\}   \stackrel{\tau}{\longrightarrow}r',\{rF:C_{\widetilde{X}}\{\widetilde{q}/\widetilde{X}\}  \stackrel{\tau}{\longrightarrow}r\}}{C_{\widetilde{X}}\{\widetilde{q}/\widetilde{X}\}F}$.

 It follows from $C_{\widetilde{X}}\{\widetilde{p}/\widetilde{X}\} \notin F $ that
 \[C_{\widetilde{X}}\{\widetilde{p}/\widetilde{X}\} \stackrel{\epsilon}{\Longrightarrow}_F|s\; \text{for some}\; s. \tag{\ref{L:FAILURE_CONGRUENCE}.3}\]
 Since $\widetilde{p} \bowtie \widetilde{q}$ and $C_{\widetilde{X}}\{\widetilde{q}/\widetilde{X}\} \stackrel{\tau}{\longrightarrow} $, by Lemma~\ref{L:ONE_ACTION_TAU}, we get $C_{\widetilde{X}}\{\widetilde{p}/\widetilde{X}\} \stackrel{\tau}{\longrightarrow} $.
 Then, by (\ref{L:FAILURE_CONGRUENCE}.3), we have
 \[C_{\widetilde{X}}\{\widetilde{p}/\widetilde{X}\} \stackrel{\tau}{\longrightarrow}_F t \stackrel{\epsilon}{\Longrightarrow}_F|s\;\text{for some}\;t.\]
 For the $\tau$-labelled transition leading to $t$, either the clause (1) or (2) in Lemma~\ref{L:ONE_ACTION_TAU} holds.

For the former, there exists $C_{\widetilde{X}}'$  such that $t \equiv  C_{\widetilde{X}}'\{\widetilde{p}/\widetilde{X}\}$ and $C_{\widetilde{X}}\{\widetilde{q}/\widetilde{X}\} \stackrel{\tau}{\longrightarrow}  C_{\widetilde{X}}'\{\widetilde{q}/\widetilde{X}\}$.
 Hence $C_{\widetilde{X}}'\{\widetilde{q}/\widetilde{X}\}F$  is one of premises of the last inferring step in $\mathcal T$.
 Moreover, it is evident that $C_{\widetilde{X}}'\{\widetilde{q}/\widetilde{X}\} \in \Omega$.

 For the latter, there exist  $C_{\widetilde{X}}'$, $C_{\widetilde{X},Z}''$ with $Z \notin \widetilde{X}$ and $i_0 \leq |\widetilde{X}|$ that realize (P-$\tau$-1) -- (P-$\tau$-4).
 In particular, by (P-$\tau$-2), we have
 \[t \equiv C_{\widetilde{X},Z}''\{\widetilde{p}/\widetilde{X},p'/Z\}\;\text{for some}\; p'\;\text{with}\;p_{i_0} \stackrel{\tau}{\longrightarrow} p'.\]
  Further, since $t \stackrel{\epsilon}{\Longrightarrow}_F|s$ and $Z$ is 1-active in $C_{\widetilde{X},Z}''$, by Lemma~\ref{L:MULTI_TAU_ACTIVE_STABLE} and \ref{L:FAILURE_TAU_I}, there exists $p''$ such that $p' \stackrel{\epsilon}{\Longrightarrow} |p''$ and
    \[t \equiv C_{\widetilde{X},Z}''\{\widetilde{p}/\widetilde{X},p'/Z\}\stackrel{\epsilon}{\Longrightarrow}_F C_{\widetilde{X},Z}''\{\widetilde{p}/\widetilde{X},p''/Z\} \stackrel{\epsilon}{\Longrightarrow}_F|s.\]
Moreover, $p'' \notin F $ by Lemma~\ref{L:FAILURE_GF}.
        Hence $p_{i_0} \stackrel{\tau}{\longrightarrow}_Fp'\stackrel{\epsilon}{\Longrightarrow}_F |p''$ by Lemma~\ref{L:FAILURE_TAU_I}.
        Since $\widetilde{p} \bowtie \widetilde{q}$, it follows from $\widetilde{p}\sqsubseteq_{RS} \widetilde{q}$ that
     \[q_{i_0} \stackrel{\tau}{\longrightarrow}_F q' \stackrel{\epsilon}{\Longrightarrow}_F|q''\;\text{and}\; p''\underset{\thicksim}{\sqsubset}_{RS} q''\;\text{for some}\;q'\;\text{and}\;q''.\tag{\ref{L:FAILURE_CONGRUENCE}.4}\]
        Then $C_{\widetilde{X}}\{\widetilde{q}/\widetilde{X}\} \stackrel{\tau}{\longrightarrow} C_{\widetilde{X},Z}''\{\widetilde{q}/\widetilde{X},q'/Z\}$  by (P-$\tau$-4).
        Therefore, $\mathcal T$ contains a proper subtree $\mathcal T'$ with root $C_{\widetilde{X},Z}''\{\widetilde{q}/\widetilde{X},q'/Z\}F$.
        In order to complete the proof, it is sufficient to show that $\mathcal T'$ contains a node labelled with $s'F$ for some $s' \in \Omega$.
        Since $Z$ is 1-active,  $\widetilde{p}\sqsubseteq_{RS} \widetilde{q}$, $\widetilde{p} \bowtie \widetilde{q}$, $ q' \stackrel{\epsilon}{\Longrightarrow} |q'' $,
   $p''\underset{\thicksim}{\sqsubset}_{RS}q''$ and $C_{\widetilde{X},Z}''\{\widetilde{p}/\widetilde{X},p''/Z\} \notin F $, by Lemma~\ref{L:FAILURE_CONGRUENCE_PRE}, there exist $C_{\widetilde{X},Z,\widetilde{Y}}^*$ and $q_Y''',p_Y'''(Y \in \widetilde{Y})$ such that
        \begin{enumerate}[({a.}1)]
\renewcommand{\theenumi}{(a.\arabic{enumi})}
          \item \  $\mathcal T'$ has a subtree with root  $C_{\widetilde{X},Z,\widetilde{Y}}^*\{\widetilde{q}/\widetilde{X},q''/Z,\widetilde{q_Y'''}/\widetilde{Y}\}F$,
          \item \  $C_{\widetilde{X},Z,\widetilde{Y}}^*\{\widetilde{p}/\widetilde{X},p''/Z,\widetilde{p_Y'''}/\widetilde{Y}\} \notin F$, and
          \item \  $\widetilde{p_Y'''} \underset{\thicksim}{\sqsubset}_{RS} \widetilde{q_Y'''} $.
        \end{enumerate}
        Clearly, $C_{\widetilde{X},Z,\widetilde{Y}}^*\{\widetilde{q}/\widetilde{X},q''/Z,\widetilde{q_Y'''}/\widetilde{Y}\} \in \Omega$ due to (a.2), (a.3) and (\ref{L:FAILURE_CONGRUENCE}.4), as desired. \\

\noindent Case 6.4 $ \frac{C_{\widetilde{X}}\{\widetilde{q}/\widetilde{X}\}   \stackrel{a}{\longrightarrow}r',\{rF:C_{\widetilde{X}}\{\widetilde{q}/\widetilde{X}\}  \stackrel{a}{\longrightarrow}r\}}{C_{\widetilde{X}}\{\widetilde{q}/\widetilde{X}\}F} (a \in Act)$.

Since $\widetilde{p} \bowtie \widetilde{q}$, by Lemma~\ref{L:ONE_ACTION_TAU}, it follows from $C_{\widetilde{X}}\{\widetilde{q}/\widetilde{X}\} \stackrel{a}{\longrightarrow}$ that $C_{\widetilde{X}}\{\widetilde{p}/\widetilde{X}\} $ is stable.
Further, since $\widetilde{p}\sqsubseteq_{RS} \widetilde{q}$ and $C_{\widetilde{X}}\{\widetilde{p}/\widetilde{X}\} \notin F$,  we get $C_{\widetilde{X}}\{\widetilde{p}/\widetilde{X}\} \stackrel{a}{\longrightarrow}$ by Lemma~\ref{L:SAME_ACTIONS}.
So, by Theorem~\ref{L:LLTS} and $C_{\widetilde{X}}\{\widetilde{p}/\widetilde{X}\} \notin F$, we have
\[C_{\widetilde{X}}\{\widetilde{p}/\widetilde{X}\} \stackrel{a}{\longrightarrow}_F t \stackrel{\epsilon}{\Longrightarrow}_F|s\;\text{for some}\;t\;\text{and}\;s. \tag{\ref{L:FAILURE_CONGRUENCE}.5}\]

On the one hand, for the $a$-labelled transition in (\ref{L:FAILURE_CONGRUENCE}.5), by Lemma~\ref{L:ONE_ACTION_VISIBLE}, there exist $C_{\widetilde{X}}'$, $C_{\widetilde{X},\widetilde{Y}}'$ and $C_{\widetilde{X},\widetilde{Y}}''$ that satisfy (CP-$a$-1) -- (CP-$a$-4).
In particular, by (CP-$a$-3-ii), there exist $i_Y \leq |\widetilde{X}|,p_{Y}'(Y \in \widetilde{Y})$ such that
\[p_{i_Y} \stackrel{a}{\longrightarrow}p_{Y}'\;\text{for each}\;Y \in \widetilde{Y}\;\text{and}\;t \equiv C_{\widetilde{X},\widetilde{Y}}''\{\widetilde{p}/\widetilde{X},\widetilde{p_{Y}'}/\widetilde{Y}\}.\]
Moreover, by (CP-$a$-1) and (CP-$a$-3-i), we have
\[C_{\widetilde{X}}\{\widetilde{p}/\widetilde{X}\}  \Rrightarrow C_{\widetilde{X}}'\{\widetilde{p}/\widetilde{X}\} \equiv C_{\widetilde{X},\widetilde{Y}}'\{\widetilde{p}/\widetilde{X},\widetilde{p_{i_Y}}/\widetilde{Y}\}.\]
Hence $C_{\widetilde{X},\widetilde{Y}}'\{\widetilde{p}/\widetilde{X},\widetilde{p_{i_Y}}/\widetilde{Y}\} \notin F$ by $C_{\widetilde{X}}\{\widetilde{p}/\widetilde{X}\} \notin F$ and Lemma~\ref{L:MULTI_STEP_UNFOLDING_FAILURE}.
Further, for each $Y \in \widetilde{Y}$, since $Y$ is 1-active in $C_{\widetilde{X},\widetilde{Y}}'$ (i.e., (CP-$a$-2)), by Lemma~\ref{L:FAILURE_GF}, $p_{i_Y} \notin F$.

On the other hand, for the transition $t \equiv C_{\widetilde{X},\widetilde{Y}}''\{\widetilde{p}/\widetilde{X},\widetilde{p_{Y}'}/\widetilde{Y}\} \stackrel{\epsilon}{\Longrightarrow}_F|s$ in (\ref{L:FAILURE_CONGRUENCE}.5), by Lemma~\ref{L:MULTI_TAU_ACTIVE_STABLE}, it follows from each $Y (\in \widetilde{Y})$ that is 1-active in $C_{\widetilde{X},\widetilde{Y}}''$ (i.e., (CP-$a$-2)) that there exist $p_{Y}''(Y \in \widetilde{Y})$ such that $p_{Y}'\stackrel{\epsilon}{\Longrightarrow}|p_{Y}''$ for each $Y \in \widetilde{Y}$ and
\[t \equiv C_{\widetilde{X},\widetilde{Y}}''\{\widetilde{p}/\widetilde{X},\widetilde{p_{Y}'}/\widetilde{Y}\} \stackrel{\epsilon}{\Longrightarrow}C_{\widetilde{X},\widetilde{Y}}''\{\widetilde{p}/\widetilde{X},\widetilde{p_{Y}'' }/\widetilde{Y}\} \stackrel{\epsilon}{\Longrightarrow}|s.\]
Since $s \notin F$, we obtain $C_{\widetilde{X},\widetilde{Y}}''\{\widetilde{p}/\widetilde{X},\widetilde{p_{Y}'' }/\widetilde{Y}\} \notin F$ by Lemma~\ref{L:FAILURE_TAU_I}, which implies that $p_{Y}'' \notin F$ for each $Y \in \widetilde{Y}$ due to Lemma~\ref{L:FAILURE_GF}.
Thus
\[p_{i_Y} \stackrel{a}{\longrightarrow}_F p_{Y}'  \stackrel{\epsilon}{\Longrightarrow}_F|p_{Y}'' \;\text{for each} \;Y \in \widetilde{Y}.\]
Since $\widetilde{p} \bowtie \widetilde{q}$, it follows from $\widetilde{p} \sqsubseteq_{RS} \widetilde{q}$ that, for each $Y \in \widetilde{Y}$, there exist $q_{Y}'$ and $q_{Y}'' $ such that
\[q_{i_Y} \stackrel{a}{\longrightarrow}_F q_{Y}'  \stackrel{\epsilon}{\Longrightarrow}_F|q_{Y}'' \;\text{and}\;p_{Y}''  \underset{\thicksim}{\sqsubset}_{RS} q_{Y}''. \tag{\ref{L:FAILURE_CONGRUENCE}.6}\]
Then $C_{\widetilde{X}}\{\widetilde{q}/\widetilde{X}\} \stackrel{a}{\longrightarrow} C_{\widetilde{X},\widetilde{Y}}''\{\widetilde{q}/\widetilde{X},\widetilde{q_{Y}'}/\widetilde{Y}\}$ by (CP-$a$-3-iii).
Hence $\mathcal T$ has a proper subtree $\mathcal T'$ with root $C_{\widetilde{X},\widetilde{Y}}''\{\widetilde{q}/\widetilde{X},\widetilde{q_{Y}'}/\widetilde{Y}\}F$.
    In order to complete the proof, it suffices to show that $\mathcal T'$ contains a node labelled with $s'F$ for some $s' \in \Omega$.
    Since each $Y (\in \widetilde{Y})$ is 1-active in $C_{\widetilde{X},\widetilde{Y}}''$,  $\widetilde{p}\sqsubseteq_{RS} \widetilde{q}$, $\widetilde{p} \bowtie \widetilde{q}$, $ \widetilde{q_Y'} \stackrel{\epsilon}{\Longrightarrow} |\widetilde{q_Y''} $,
   $\widetilde{p_Y''}\underset{\thicksim}{\sqsubset}_{RS}\widetilde{q_Y''}$ and
   $C_{\widetilde{X},\widetilde{Y}}''\{\widetilde{p}/\widetilde{X},\widetilde{p_Y''}/\widetilde{Y}\} \notin F $, by  Lemma~\ref{L:FAILURE_CONGRUENCE_PRE}, there exist $C_{\widetilde{X},\widetilde{Y},\widetilde{Z}}^*$ and $q_Z''',p_Z'''(Z \in \widetilde{Z})$ such that
        \begin{enumerate}[({b.}1)]
\renewcommand{\theenumi}{(b.\arabic{enumi})}
          \item \  $\mathcal T'$ has a subtree with root  $C_{\widetilde{X},\widetilde{Y},\widetilde{Z}}^*\{\widetilde{q}/\widetilde{X},\widetilde{q_{Y}''}/\widetilde{Y},\widetilde{q_Z'''}/\widetilde{Z}\}F$,
          \item \  $C_{\widetilde{X},\widetilde{Y},\widetilde{Z}}^*\{\widetilde{p}/\widetilde{X},\widetilde{p_{Y}''}/\widetilde{Y},\widetilde{p_Z'''}/\widetilde{Z}\} \notin F$, and
          \item \  $\widetilde{p_Z'''} \underset{\thicksim}{\sqsubset}_{RS} \widetilde{q_Z'''} $.
        \end{enumerate}
        Obviously, $C_{\widetilde{X},\widetilde{Y},\widetilde{Z}}^*\{\widetilde{q}/\widetilde{X},\widetilde{q_{Y}''}/\widetilde{Y},\widetilde{q_Z'''}/\widetilde{Z}\} \in \Omega$ due to (b.2), (b.3) and (\ref{L:FAILURE_CONGRUENCE}.6), as desired. \\

 \noindent Case 6.5 $\frac{\{rF: B_{\widetilde{X}}\{\widetilde{q}/\widetilde{X}\}  \wedge D_{\widetilde{X}}\{\widetilde{q}/\widetilde{X}\} \stackrel{\epsilon}{\Longrightarrow}|r\}}{B_{\widetilde{X}}\{\widetilde{q}/\widetilde{X}\}  \wedge D_{\widetilde{X}}\{\widetilde{q}/\widetilde{X}\} F}$.

 Similar to the second alternative in Case 5, omitted.
\end{proof}

In the remainder of this section, we shall prove that $\sqsubseteq_{RS}$ is indeed precongruent.
Let us first recall a distinct but equivalent formulation of $\sqsubseteq_{RS}$ due to Van Glabbeek \cite{Luttgen10}.

\begin{mydefn}\label{D:ALT_RS}
A relation ${\mathcal R} \subseteq T(\Sigma_{\text{CLL}_R})\times T(\Sigma_{\text{CLL}_R})$ is an alternative ready simulation relation, if for any $(p,q) \in {\mathcal R}$ and $a \in Act $\\
\textbf{(RSi)} $p \stackrel{\epsilon}{\Longrightarrow}_F|p'$ implies $\exists q'.q \stackrel{\epsilon}{\Longrightarrow}_F|q'\; \textrm{and}\;(p',q') \in {\mathcal R}$;\\
\textbf{(RSiii)} $p \stackrel{a}{\Longrightarrow}_F|p'$ and $p,q$ stable implies $\exists q'.q \stackrel{a}{\Longrightarrow}_F|q'\; \textrm{and}\;(p',q') \in {\mathcal R}$;\\
\textbf{(RSiv)} $p\notin F $ and $p,q$ stable implies ${\mathcal I}(p)={\mathcal I}(q)$.

We write $p \sqsubseteq_{ALT} q$ if there exists an alternative ready simulation relation $\mathcal R$ with $(p,q) \in \mathcal R$.
\end{mydefn}

The next proposition reveals that this definition agrees with the one given in Def.~\ref{D:RS}.

\begin{proposition}\label{P:COINCIDENCE}
  $\sqsubseteq_{RS} = \sqsubseteq_{ALT}$.
\end{proposition}
\begin{proof}
  See Prop.~13 in \cite{Luttgen10}.
\end{proof}

One advantage of Def.~\ref{D:ALT_RS} is that, given $p$ and $q$, we can prove $p\sqsubseteq_{RS}q$ by means of giving an alternative ready simulation relation relating them.
It is well known that up-to technique is a tractable way for such coinduction proof.
Here we introduce the notion of an alternative ready relation up to $\underset{\thicksim}{\sqsubset}_{RS}$ as follows.

\begin{mydefn}[ALT up to $\underset{\thicksim}{\sqsubset}_{RS}$]\label{D:ALT_RS_UPTO}
A relation ${\mathcal R} \subseteq T(\Sigma_{\text{CLL}_R})\times T(\Sigma_{\text{CLL}_R})$ is an alternative ready simulation relation up to $\underset{\thicksim}{\sqsubset}_{RS}$, if for any $(p,q) \in {\mathcal R}$ and $a \in Act $\\
\textbf{(ALT-upto-1)} $p \stackrel{\epsilon}{\Longrightarrow}_F|p'$ implies $\exists q'.q \stackrel{\epsilon}{\Longrightarrow}_F|q'\; \textrm{and}\;p' \underset{\thicksim}{\sqsubset}_{RS}{\mathcal R}\underset{\thicksim}{\sqsubset}_{RS}q'$;\\
\textbf{(ALT-upto-2)} $p \stackrel{a}{\Longrightarrow}_F|p'$ and $p,q$ stable implies $\exists q'.q \stackrel{a}{\Longrightarrow}_F|q'\; \textrm{and}\;p'\underset{\thicksim}{\sqsubset}_{RS}{\mathcal R}\underset{\thicksim}{\sqsubset}_{RS}q'$;\\
\textbf{(ALT-upto-3)} $p\notin F $ and $p,q$ stable implies ${\mathcal I}(p)={\mathcal I}(q)$.
\end{mydefn}

As usual, given a relation $\mathcal R$ satisfying the conditions (ALT-upto-1,2,3), in general, $\mathcal R$ in itself is not an alternative ready simulation relation.
But simple result below ensures that up-to technique based on the above notion is sound.

\begin{lemma}\label{L:ALT_UP_TO}
  If a relation $\mathcal R$ is an alternative ready simulation relation up to $\underset{\thicksim}{\sqsubset}_{RS}$ then ${\mathcal R}  \subseteq \sqsubseteq_{RS}$.
\end{lemma}
\begin{proof}
  By Prop.~\ref{P:COINCIDENCE}, it is sufficient to prove that the relation $\sqsubseteq_{RS} \circ{\mathcal R} \circ \sqsubseteq_{RS}$ is an alternative ready simulation. We leave it to the reader.
\end{proof}

Now we are ready to prove the main result of this section: $\sqsubseteq_{RS}$ is precongruent w.r.t all operations in $\text{CLL}_R$.
We shall divide the proof into the next two lemmas.

\begin{lemma}\label{L:CONGRUENCE_PRE}
 $C_X\{p/X\} =_{RS} C_X\{\tau.p/X\}$ for any context $C_X$ and stable process $p$.
\end{lemma}
\begin{proof}
Let $p$ be any stable process. First, we shall show that  $C_X\{p/X\} \sqsubseteq_{RS} C_X\{\tau.p/X\}$.
 Set
 \[
       {\mathcal R} \triangleq \{(B_X\{p/X\},B_X\{\tau.p/X\}):B_X\;\text{is a context}\}.
\]
  By Prop.~\ref{P:COINCIDENCE} and Lemma~\ref{L:ALT_UP_TO}, it is sufficient to prove that $\mathcal R$ is an alternative ready simulation relation up to $\underset{\thicksim}{\sqsubset}_{RS}$. Let $(C_X\{p/X\},C_X\{\tau.p/X\}) \in \mathcal R$.

  \textbf{(ALT-upto-1)} Assume that $C_X\{p/X\} \stackrel{\epsilon}{\Longrightarrow}_F|p'$.
  For this transition, since $p$ is stable, by Lemma~\ref{L:MULTI_TAU_GF_STABLE}, there exists a stable context $C_X'$ such that
  \[p' \equiv C_X'\{p/X\}\;\text{and}\;C_X\{\tau.p/X\} \stackrel{\epsilon}{\Longrightarrow}C_X'\{\tau.p/X\}.\tag{\ref{L:CONGRUENCE_PRE}.1}\]
  Moreover, by Lemma~\ref{L:STABILIZATION}, it follows from $\tau.p \stackrel{\tau}{\longrightarrow}|p$ that \[C_X'\{\tau.p/X\}\stackrel{\epsilon}{\Longrightarrow}|r\;\text{for some}\; r.\tag{\ref{L:CONGRUENCE_PRE}.2}\]
  For this transition, by Lemma~\ref{L:MULTI_TAU_GF_STABLE}, there exists a context $C_{X,\widetilde{Y}}''$ with $X \notin \widetilde{Y}$ such that $r\equiv C_{X,\widetilde{Y}}''\{\tau.p/X,p/\widetilde{Y}\}$ and
  \[p'\equiv C_X'\{p/X\}  \Rrightarrow C_{X,\widetilde{Y}}''\{p/X,p/\widetilde{Y}\}.\tag{\ref{L:CONGRUENCE_PRE}.3}\]
  Since $p' \notin F$, by Lemma~\ref{L:MULTI_STEP_UNFOLDING_FAILURE}, we get $C_{X,\widetilde{Y}}''\{p/X,p/\widetilde{Y}\} \notin F$.
  Further, by Lemma~\ref{L:FAILURE_S_VS_NS}, $r \equiv C_{X,\widetilde{Y}}''\{\tau.p/X,p/\widetilde{Y}\} \notin F$.
  So, by (\ref{L:CONGRUENCE_PRE}.1), (\ref{L:CONGRUENCE_PRE}.2) and Lemma~\ref{L:FAILURE_TAU_I}, we obtain
  \[C_X\{\tau.p/X\}\stackrel{\epsilon}{\Longrightarrow}_F| C_{X,\widetilde{Y}}''\{\tau.p/X,p/\widetilde{Y}\}.\]
  Moreover, by Lemma~\ref{L:MULTI_STEP_UNFOLDING_IMPLIES_RS}, it follows from (\ref{L:CONGRUENCE_PRE}.3) that \[p' \underset{\thicksim}{\sqsubset}_{RS} C_{X,\widetilde{Y}}''\{p/X,p/\widetilde{Y}\} {\mathcal R}C_{X,\widetilde{Y}}''\{\tau.p/X,p/\widetilde{Y}\}.\]

  \textbf{(ALT-upto-2)} Assume that $C_X\{p/X\}$ and $C_X\{\tau.p/X\}$ are stable and $C_X\{p/X\} \stackrel{a}{\Longrightarrow}_F|p'$.
  Hence $C_X\{p/X\} \stackrel{a}{\longrightarrow}_F r \stackrel{\epsilon}{\Longrightarrow}_F |p'$ for some $r$.
  Moreover, by Lemma~\ref{L:FAILURE_S_VS_NS} and $C_X\{p/X\} \notin F$, we have
\[C_X\{\tau.p/X\} \notin F.\tag{\ref{L:CONGRUENCE_PRE}.4}\]
  Next we intend to prove that $p$ does not involve in the transition $C_X\{p/X\} \stackrel{a}{\longrightarrow}_F r$. For this transition,
  by Lemma~\ref{L:ONE_ACTION_VISIBLE}, there exist $C_X'$, $C_{X,\widetilde{Y}}'$ and $C_{X,\widetilde{Y}}''$ that realize (CP-$a$-1) -- (CP-$a$-4).
  By (CP-$a$-1) and (CP-$a$-3-i), we have
  \[C_X\{\tau.p/X\} \Rrightarrow C_X'\{\tau.p/X\} \equiv C_{X,\widetilde{Y}}'\{\tau.p/X,\tau.p/\widetilde{Y}\}.\]
  If $\widetilde{Y} \neq \emptyset$ then, by (CP-$a$-2) and Lemma~\ref{L:ONE_ACTION_TAU_GF}, we have $C_{X,\widetilde{Y}}'\{\tau.p/X,\tau.p/\widetilde{Y}\} \stackrel{\tau}{\longrightarrow}$, and hence $C_X\{\tau.p/X\} \stackrel{\tau}{\longrightarrow}$ by Lemma~\ref{L:MULTI_STEP_UNFOLDING_ACTION}, which  contradicts that $C_X\{\tau.p/X\}$ is stable.
  Thus $\widetilde{Y} = \emptyset$, as desired.
  So, $r \equiv C_{X,\widetilde{Y}}''\{p/X\}$ by (CP-$a$-3-ii) and
\[C_X\{\tau.p/X\} \stackrel{a}{\longrightarrow}C_{X,\widetilde{Y}}''\{\tau.p/X\} \text{ by (CP-$a$-3-iii)}\;\text{and}\;C_X\{\tau.p/X\}\not\stackrel{\tau}{\longrightarrow}.\tag{\ref{L:CONGRUENCE_PRE}.5}\]
  Moreover, by (ALT-upto-1), it follows from $(C_{X,\widetilde{Y}}''\{p/X\},C_{X,\widetilde{Y}}''\{\tau.p/X\}) \in \mathcal R$ and $r \equiv C_{X,\widetilde{Y}}''\{p/X\}\stackrel{\epsilon}{\Longrightarrow}_F|p'$ that $C_{X,\widetilde{Y}}''\{\tau.p/X\} \stackrel{\epsilon}{\Longrightarrow}_F|q'$ and $p' \underset{\thicksim}{\sqsubset}_{RS} {\mathcal R} \underset{\thicksim}{\sqsubset}_{RS} q'$ for some $q'$.
Moreover, we also have $C_{X}\{\tau.p/X\} \stackrel{a}{\Longrightarrow}_F|q'$ due to (\ref{L:CONGRUENCE_PRE}.4) and (\ref{L:CONGRUENCE_PRE}.5), as desired.\\

    \textbf{(ALT-upto-3)} Immediately follows from Lemma~\ref{L:SAME_ACTIONS}.\\

Next we intend to prove $ C_X\{\tau.p/X\}  \sqsubseteq_{RS} C_X\{p/X\}$.  Set
 \[
       {\mathcal R} \triangleq \{(B_X\{\tau.p/X\},B_X\{p/X\}):B_X\;\text{is a context}\}.
\]
  Similarly, it is sufficient to prove that $\mathcal R$ is an alternative ready simulation relation up to $\underset{\thicksim}{\sqsubset}_{RS}$. Let $(C_X\{\tau.p/X\},C_X\{p/X\}) \in \mathcal R$.
  (ALT-upto-3) immediately follows from Lemma~\ref{L:SAME_ACTIONS}.
  In the following, we prove the other two conditions.

  \textbf{(ALT-upto-1)}
    Assume that $C_X\{\tau.p/X\} \stackrel{\epsilon}{\Longrightarrow}_F|p'$.
    For this transition, by Lemma~\ref{L:TAU_ACTION_NORMALIZATION}, there exist $r$ and stable context $C_X^*$ such that  $C_X\{p/X\} \stackrel{\epsilon}{\Longrightarrow}C_X^*\{p/X\}$ and
    \[C_X\{\tau.p/X\} \stackrel{\epsilon}{\Longrightarrow}C_X^*\{\tau.p/X\} \stackrel{\epsilon}{\Longrightarrow}|r  \Rrightarrow p'.\tag{\ref{L:CONGRUENCE_PRE}.6}\]
    Moreover, since $p$ is stable, so is $C_X^*\{p/X\}$ by Lemma~\ref{L:ONE_ACTION_TAU}.
    Due to $r \Rrightarrow p'$ and $p' \notin F$, by Lemma~\ref{L:MULTI_STEP_UNFOLDING_FAILURE}, we get $r \notin F$.
    Hence  $C_X^*\{\tau.p/X\} \notin F$ by (\ref{L:CONGRUENCE_PRE}.6) and Lemma~\ref{L:FAILURE_TAU_I}.
    Then $C_X^*\{p/X\} \notin F$  by Lemma~\ref{L:FAILURE_NS_IMPLIES_S}.
    Thus
    \[C_X\{p/X\} \stackrel{\epsilon}{\Longrightarrow}_F|C_X^*\{p/X\}.\]
    To complete the proof, it remains to prove that $p' \underset{\thicksim}{\sqsubset}_{RS}  {\mathcal R} \underset{\thicksim}{\sqsubset}_{RS} C_X^* \{p/X\}$.
    For the transition $C_X^*\{\tau.p/X\} \stackrel{\epsilon}{\Longrightarrow}|r$ in (\ref{L:CONGRUENCE_PRE}.6), by Lemma~\ref{L:MULTI_TAU_GF_STABLE}, there exists a stable context $C_{X,\widetilde{Y}}'^*$ such that $r \equiv C_{X,\widetilde{Y}}'^*\{\tau.p/X,p/\widetilde{Y}\} \Rrightarrow p'$ and $C_X^*\{p/X\} \Rrightarrow C_{X,\widetilde{Y}}'^*\{p/X,p/\widetilde{Y}\}$, which, by Lemma~\ref{L:MULTI_STEP_UNFOLDING_IMPLIES_RS}, implies
    \[p' \underset{\thicksim}{\sqsubset}_{RS} C_{X,\widetilde{Y}}'^*\{\tau.p/X,p/\widetilde{Y}\} {\mathcal R} C_{X,\widetilde{Y}}'^*\{p/X,p/\widetilde{Y}\} \underset{\thicksim}{\sqsubset}_{RS} C_X^* \{p/X\}.\]

  \textbf{(ALT-upto-2)} Assume that $C_X\{\tau.p/X\}$ and $C_X\{p/X\}$ are stable and $C_X\{\tau.p/X\} \stackrel{a}{\Longrightarrow}_F|p'$.
  Hence $C_X\{\tau.p/X\} \stackrel{a}{\longrightarrow}_F r \stackrel{\epsilon}{\Longrightarrow}_F |p'$ for some $r$.
  Moreover, by Lemma~\ref{L:FAILURE_NS_IMPLIES_S} and $C_X\{\tau.p/X\} \notin F$, we have
\[C_X\{p/X\} \notin F.\]
  For the $a$-labelled transition $C_X\{\tau.p/X\} \stackrel{a}{\longrightarrow}_F r$,
  by Lemma~\ref{L:ONE_ACTION_VISIBLE}, it is not difficult to see that there exists
  $C_X'$ such that
  \[C_X\{\tau.p/X\} \stackrel{a}{\longrightarrow} C_X'\{\tau.p/X\}\equiv r\;\text{and}\;C_X\{p/X\} \stackrel{a}{\longrightarrow} C_{X}'\{p/X\}.\]
  Moreover, by (ALT-upto-1), it follows from $(C_{X}'\{\tau.p/X\},C_{X}'\{p/X\}) \in \mathcal R$ and $r \equiv C_{X}'\{\tau.p/X\}\stackrel{\epsilon}{\Longrightarrow}_F|p'$ that $C_{X}'\{p/X\} \stackrel{\epsilon}{\Longrightarrow}_F|q'$ and $p' \underset{\thicksim}{\sqsubset}_{RS} {\mathcal R} \underset{\thicksim}{\sqsubset}_{RS} q'$ for some $q'$.
Clearly, we have $C_{X}\{p/X\} \stackrel{a}{\Longrightarrow}_F|q'$, as desired.
\end{proof}

\begin{lemma}\label{L:CONGRUENCE_PRE_UNIFORM}
   If $\widetilde{p} \bowtie \widetilde{q}$  and $\widetilde{p}\sqsubseteq_{RS} \widetilde{q}$ then  $C_{\widetilde{X}}\{\widetilde{p}/\widetilde{X}\}\sqsubseteq_{RS} C_{\widetilde{X}}\{\widetilde{q}/\widetilde{X}\}$ for any $C_{\widetilde{X}} $.
\end{lemma}
\begin{proof}
 Set
\[
       {\mathcal R} \triangleq \{(B_{\widetilde{X}}\{\widetilde{p}/\widetilde{X}\},B_{\widetilde{X}}\{\widetilde{q}/\widetilde{X}\}): \widetilde{p} \bowtie \widetilde{q},
       \widetilde{p}\sqsubseteq_{RS} \widetilde{q} \;\text{and}\; B_{\widetilde{X}}\;\text{is a context}\}.
\]
Similarly, it suffices to prove that $\mathcal R$ is an alternative ready simulation relation up to $\underset{\thicksim}{\sqsubset}_{RS}$. Suppose $(C_{\widetilde{X}}\{\widetilde{p}/\widetilde{X}\},C_{\widetilde{X}}\{\widetilde{q}/\widetilde{X}\})\in \mathcal R$. Then, by  Lemma~\ref{L:SAME_ACTIONS}, it is obvious that such pair satisfies the condition (ALT-upto-3). In the following, we consider two remaining conditions in turn.

      \textbf{(ALT-upto-1)} Assume that $C_{\widetilde{X}}\{\widetilde{p}/\widetilde{X}\} \stackrel{\epsilon}{\Longrightarrow}_F|s$.
      For this transition, by Lemma~\ref{L:MULTI_TAU_GF_STABLE}, there exist $C_{\widetilde{X},\widetilde{Y}}'$ and $i_Y \leq |\widetilde{X}|,p_Y'(Y\in \widetilde{Y})$  that satisfy (MS-$\tau$-1) -- (MS-$\tau$-7).
      In particular, by (MS-$\tau$-2,7), we have
      \[p_{i_Y}\stackrel{\tau}{\Longrightarrow}|p_{Y}'\;\text{for each}\;Y \in \widetilde{Y}\;\text{and}\; s \equiv C_{\widetilde{X},\widetilde{Y}}'\{\widetilde{p}/\widetilde{X},\widetilde{p_{Y}'}/\widetilde{Y}\} \notin F.\]
      Then, by (MS-$\tau$-1) and Lemma~\ref{L:FAILURE_GF}, $p_{Y}'\notin F$ and hence $p_{i_Y}\stackrel{\tau}{\Longrightarrow}_F|p_{Y}'$ by Lemma~\ref{L:FAILURE_TAU_I} for each $Y \in \widetilde{Y}$.
      Since $\widetilde{p} \bowtie \widetilde{q}$, it follows from $ \widetilde{p}\sqsubseteq_{RS} \widetilde{q}$ that there exist $q_{Y}'(Y \in \widetilde{Y})$ such that
      \[q_{i_Y}\stackrel{\tau}{\Longrightarrow}_F|q_{Y}'\;\text{and}\; p_{Y}'\underset{\thicksim}{\sqsubset}_{RS} q_{Y}'\;\text{for each}\;Y \in \widetilde{Y}.\tag{\ref{L:CONGRUENCE_PRE_UNIFORM}.1}\]
      So, by (MS-$\tau$-3-ii), we get
      \[C_{\widetilde{X}}\{\widetilde{q}/\widetilde{X}\} \stackrel{\epsilon}{\Longrightarrow} C_{\widetilde{X},\widetilde{Y}}'\{\widetilde{q}/\widetilde{X},\widetilde{q_{Y}'}/\widetilde{Y}\}.\]
      Moreover, by Lemma~\ref{L:ONE_ACTION_TAU}, it follows from $s \equiv C_{\widetilde{X},\widetilde{Y}}'\{\widetilde{p}/\widetilde{X},\widetilde{p_{Y}'}/\widetilde{Y}\} \not\stackrel{\tau}{\longrightarrow}$, $\widetilde{p} \bowtie \widetilde{q}$ and  $\widetilde{p_{Y}'}\underset{\thicksim}{\sqsubset}_{RS} \widetilde{q_{Y}'}$ that \[C_{\widetilde{X},\widetilde{Y}}'\{\widetilde{q}/\widetilde{X},\widetilde{q_{Y}'}/\widetilde{Y}\} \not\stackrel{\tau}{\longrightarrow}.\]
      In addition, by Lemma~\ref{L:FAILURE_CONGRUENCE} and $C_{\widetilde{X},\widetilde{Y}}'\{\widetilde{p}/\widetilde{X},\widetilde{p_{Y}'}/\widetilde{Y}\} \notin F$, we get $C_{\widetilde{X},\widetilde{Y}}'\{\widetilde{q}/\widetilde{X},\widetilde{q_{Y}'}/\widetilde{Y}\} \notin F$.
      Hence, by Lemma~\ref{L:FAILURE_TAU_I}, we obtain
      \[C_{\widetilde{X}}\{\widetilde{q}/\widetilde{X}\} \stackrel{\epsilon}{\Longrightarrow}_F | C_{\widetilde{X},\widetilde{Y}}'\{\widetilde{q}/\widetilde{X},\widetilde{q_{Y}'}/\widetilde{Y}\}.\]
      Clearly, $(C_{\widetilde{X},\widetilde{Y}}'\{\widetilde{p}/\widetilde{X},\widetilde{p_{Y}'}/\widetilde{Y}\},C_{\widetilde{X},\widetilde{Y}}'\{\widetilde{q}/\widetilde{X},\widetilde{q_{Y}'}/\widetilde{Y}\})\in \underset{\thicksim}{\sqsubset}_{RS} {\mathcal R} \underset{\thicksim}{\sqsubset}_{RS}$ due to (\ref{L:CONGRUENCE_PRE_UNIFORM}.1) and the reflexivity of $\underset{\thicksim}{\sqsubset}_{RS}$.\\

  \textbf{(ALT-upto-2)} Let $C_{\widetilde{X}}\{\widetilde{p}/\widetilde{X}\}$ and $C_{\widetilde{X}}\{\widetilde{q}/\widetilde{X}\}$ be stable and $C_{\widetilde{X}}\{\widetilde{p}/\widetilde{X}\} \stackrel{a}{\Longrightarrow}_F|s$.
  Then
\[C_{\widetilde{X}}\{\widetilde{p}/\widetilde{X}\} \stackrel{a}{\longrightarrow}_F r \stackrel{\epsilon}{\Longrightarrow}_F|s\;\text{for some}\; r. \tag{\ref{L:CONGRUENCE_PRE_UNIFORM}.2}\]
Moreover, by Lemma~\ref{L:FAILURE_CONGRUENCE}, it follows from $\widetilde{p} \bowtie \widetilde{q}$,  $\widetilde{p}\sqsubseteq_{RS} \widetilde{q}$ and $C_{\widetilde{X}}\{\widetilde{p}/\widetilde{X}\} \notin F$ that
\[C_{\widetilde{X}}\{\widetilde{q}/\widetilde{X}\} \notin F.\tag{\ref{L:CONGRUENCE_PRE_UNIFORM}.3}\]
For the  transition $C_{\widetilde{X}}\{\widetilde{p}/\widetilde{X}\} \stackrel{a}{\longrightarrow} r$, by Lemma~\ref{L:ONE_ACTION_VISIBLE}, there exist $C_{\widetilde{X}}'$, $C_{\widetilde{X},\widetilde{Y}}'$ and $C_{\widetilde{X},\widetilde{Y}}''$ that satisfy (CP-$a$-1) -- (CP-$a$-4).
In particular, by (CP-$a$-3-ii), there exist $i_Y \leq |\widetilde{X}|,p_{Y}'(Y \in \widetilde{Y})$  such that
$p_{i_Y} \stackrel{a}{\longrightarrow}p_{Y}'$ for each $Y \in \widetilde{Y}$ and $r \equiv C_{\widetilde{X},\widetilde{Y}}''\{\widetilde{p}/\widetilde{X},\widetilde{p_{Y}'}/\widetilde{Y}\}$.
Moreover, by (CP-$a$-1) and (CP-$a$-3-i), we have
\[C_{\widetilde{X}}\{\widetilde{p}/\widetilde{X}\}  \Rrightarrow C_{\widetilde{X}}'\{\widetilde{p}/\widetilde{X}\} \equiv C_{\widetilde{X},\widetilde{Y}}'\{\widetilde{p}/\widetilde{X},\widetilde{p_{i_Y}}/\widetilde{Y}\}.\]
Hence $C_{\widetilde{X},\widetilde{Y}}'\{\widetilde{p}/\widetilde{X},\widetilde{p_{i_Y}}/\widetilde{Y}\} \notin F$ by $C_{\widetilde{X}}\{\widetilde{p}/\widetilde{X}\} \notin F$ and Lemma~\ref{L:MULTI_STEP_UNFOLDING_FAILURE}.
Further, since  each $Y (\in \widetilde{Y})$ is 1-active in $C_{\widetilde{X},\widetilde{Y}}'$, by Lemma~\ref{L:FAILURE_GF}, we get
\[p_{i_Y} \notin F\;\text{for each}\; Y \in \widetilde{Y}.\tag{\ref{L:CONGRUENCE_PRE_UNIFORM}.4}\]
For the transition $r \equiv C_{\widetilde{X},\widetilde{Y}}''\{\widetilde{p}/\widetilde{X},\widetilde{p_{Y}'}/\widetilde{Y}\}  \stackrel{\epsilon}{\Longrightarrow}|s$ in (\ref{L:CONGRUENCE_PRE_UNIFORM}.2), by Lemma~\ref{L:MULTI_TAU_ACTIVE_STABLE}, it follows that for each $Y \in \widetilde{Y}$, there exists $p_{Y}''$ such that $p_{Y}'\stackrel{\epsilon}{\Longrightarrow}|p_{Y}''$ and \[C_{\widetilde{X},\widetilde{Y}}''\{\widetilde{p}/\widetilde{X},\widetilde{p_{Y}'}/\widetilde{Y}\} \stackrel{\epsilon}{\Longrightarrow} C_{\widetilde{X},\widetilde{Y}}''\{\widetilde{p}/\widetilde{X},\widetilde{p_{Y}''}/\widetilde{Y}\} \stackrel{\epsilon}{\Longrightarrow}|s.\]
Then $C_{\widetilde{X},\widetilde{Y}}''\{\widetilde{p}/\widetilde{X},\widetilde{p_{Y}''}/\widetilde{Y}\} \notin F$ due to $s \notin F$ and Lemma~\ref{L:FAILURE_TAU_I}, and hence $p_{Y}''\notin F$ for each $Y \in \widetilde{Y}$ by Lemma~\ref{L:FAILURE_GF}. Therefore, by (\ref{L:CONGRUENCE_PRE_UNIFORM}.4) and Lemma~\ref{L:FAILURE_TAU_I}, we have
\[p_{i_Y} \stackrel{a}{\longrightarrow}_F p_{Y}'  \stackrel{\epsilon}{\Longrightarrow}_F|p_{Y}''\;\text{for each} \;Y \in \widetilde{Y}.\]
So it follows from $\widetilde{p} \bowtie \widetilde{q}$ and $\widetilde{p} \sqsubseteq_{RS} \widetilde{q}$ that for each $Y \in \widetilde{Y}$, there exist $q_{Y}'$ and $q_{Y}''$ such that $q_{i_Y} \stackrel{a}{\longrightarrow}_F q_{Y}'  \stackrel{\epsilon}{\Longrightarrow}_F|q_{Y}''$ and $p_{Y}'' \underset{\thicksim}{\sqsubset}_{RS} q_{Y}''$.
By (CP-$a$-3-iii), we get
\[C_{\widetilde{X}}\{\widetilde{q}/\widetilde{X}\} \stackrel{a}{\longrightarrow} C_{\widetilde{X},\widetilde{Y}}''\{\widetilde{q}/\widetilde{X},\widetilde{q_{Y}'}/\widetilde{Y}\}. \tag{\ref{L:CONGRUENCE_PRE_UNIFORM}.5}\]
Further, by Lemma~\ref{L:ONE_ACTION_TAU_GF} and (CP-$a$-2), we obtain \[C_{\widetilde{X},\widetilde{Y}}''\{\widetilde{q}/\widetilde{X},\widetilde{q_{Y}'}/\widetilde{Y}\} \stackrel{\epsilon}{\Longrightarrow} C_{\widetilde{X},\widetilde{Y}}''\{\widetilde{q}/\widetilde{X},\widetilde{q_{Y}''}/\widetilde{Y}\}. \tag{\ref{L:CONGRUENCE_PRE_UNIFORM}.6}\]
Clearly, $(C_{\widetilde{X},\widetilde{Y}}''\{\widetilde{p}/\widetilde{X},\widetilde{p_{Y}''}/\widetilde{Y}\},C_{\widetilde{X},\widetilde{Y}}''\{\widetilde{q}/\widetilde{X},\widetilde{q_{Y}''}/\widetilde{Y}\}) \in \mathcal R$.
So, by  $C_{\widetilde{X},\widetilde{Y}}''\{\widetilde{p}/\widetilde{X},\widetilde{p_{Y}''}/\widetilde{Y}\} \stackrel{\epsilon}{\Longrightarrow}_F |s$ and (ALT-upto-1), there exists $t$ such that $C_{\widetilde{X},\widetilde{Y}}''\{\widetilde{q}/\widetilde{X},\widetilde{q_{Y}''}/\widetilde{Y}\} \stackrel{\epsilon}{\Longrightarrow}_F|t$ and $s \underset{\thicksim}{\sqsubset}_{RS}{\mathcal R} \underset{\thicksim}{\sqsubset}_{RS} t$;
moreover, we also have $C_{\widetilde{X}}\{\widetilde{q}/\widetilde{X}\} \stackrel{a}{\Longrightarrow}_F|t$ due to (\ref{L:CONGRUENCE_PRE_UNIFORM}.3), (\ref{L:CONGRUENCE_PRE_UNIFORM}.5), (\ref{L:CONGRUENCE_PRE_UNIFORM}.6) and Lemma~\ref{L:FAILURE_TAU_I}, as desired.
\end{proof}

We are now in a position to state the main result of this section.

\begin{theorem}[Precongruence]\label{T:CONGRUENCE}
  If $p\sqsubseteq_{RS} q$ then  $C_X\{p/X\}\sqsubseteq_{RS} C_X\{q/X\}$ for any context $C_X$.
\end{theorem}
\begin{proof}
By  Lemma~\ref{L:CONGRUENCE_PRE} and \ref{L:CONGRUENCE_PRE_UNIFORM}, it immediately follows from $\tau.p =_{RS} p \sqsubseteq_{RS}q =_{RS} \tau.q$.
\end{proof}

As an immediate consequence of this theorem, we also have
\begin{corollary}
  If $\widetilde{p}\sqsubseteq_{RS} \widetilde{q}$ then  $C_{\widetilde{X}}\{\widetilde{p}/\widetilde{X}\}\sqsubseteq_{RS} C_{\widetilde{X}}\{\widetilde{q}/\widetilde{X}\}$ for any context $C_{\widetilde{X}}$.
\end{corollary}
\begin{proof}
  Applying Theorem~\ref{T:CONGRUENCE} finitely many times.
\end{proof}

\section{Unique solution of equations}

This section focuses on the solutions of equations.
Especially, we shall prove that the equation $X=_{RS}t_X$ has at most one consistent solution modulo $=_{RS}$ provided that $X$ is strongly guarded and does not occur in the scope of any conjunction in $t_X$;
moreover, the process $\langle X|X =t_X \rangle$ is indeed the unique consistent solution whenever such equation has a consistent solution.
We begin with giving two results on the inconsistency predicate $F$.

\begin{lemma}\label{L:FAILURE_EQUATION_UNIQUE_PRE}
 For any stable processes $p,q \notin F $ and context $C_X$ such that $X$ does not occur in the scope of any conjunction, if $C_X\{p/X\} \in F $ then $C_X\{q/X\} \in F $.
\end{lemma}
\begin{proof}
Assume that $C_X\{p/X\} \in F $ and $\mathcal T$ is any proof tree of $Strip(\mathcal{P}_{\text{CLL}_R} ,M_{\text{CLL}_R} ) \vdash C_{X}\{p/ X\}F$.
We proceed by induction on the depth of $\mathcal T$.
  The argument is a routine case analysis on $C_X$.
  Moreover, since $X$ does not occur in the scope of any conjunction, in addition to that $C_X$ is closed, the form of $C_X$ is one of the following: $X$, $\alpha.B_X$, $B_X \odot D_X$ with $\odot \in \{\vee, \Box, \parallel_A\}$ and $\langle Y|E \rangle$.
  Here, we give the proof only for the case   $C_X \equiv \langle Y|E \rangle$, the other cases are straightforward and omitted.

  In case $C_{X} \equiv \langle Y|E \rangle$, the last rule applied in $\mathcal T$ is
  \[\text{either}\;\frac{\langle t_Y|E \rangle \{p/X\}F}{\langle Y|E \rangle \{p/X\}F}\;\text{with}\;Y =t_Y \in E\;\text{or}\;\frac{\{rF:\langle Y|E \rangle \{p/X\} \stackrel{\epsilon}{\Longrightarrow}|r\}}{\langle Y|E \rangle \{p/X\} F}.\]

  For the first alternative, we have $\langle t_Y|E\rangle\{q/X\} \in F $ by IH, and hence $C_{X}\{q/X\} \equiv \langle Y|E\rangle\{q/X\} \in F $.

  For the second alternative, assume $\langle Y|E \rangle \{q/X\} \stackrel{\epsilon}{\Longrightarrow} |s$.
  Since $q$ is stable, by Lemma~\ref{L:MULTI_TAU_GF_STABLE}, $s\equiv C_X'\{q/X\}$ for some stable $C_X'$ such that $X$ does not occur in the scope of any conjunction in $C_X'$ and $\langle Y|E \rangle \{p/X\} \stackrel{\epsilon}{\Longrightarrow} C_X'\{p/X\}$.
  Moreover, since $p$ is stable, so is $C_X'\{p/X\}$.
  Thus there exists a proper subtree of $\mathcal T$ with root $C_X'\{p/X\} F $.
  So, by IH, $s\equiv C_X'\{q/X\} \in F $.
  Hence $C_X\{q/X\} \in F $ by Theorem~\ref{L:LLTS}, as desired.
\end{proof}

This result is of independent interest, but its principal use is that it will serve as an important step in demonstrating the next lemma, which reveals that the above result still holds if it is deleted from the hypotheses that $q$ and $p$ are stable.

\begin{lemma}\label{L:FAILURE_EQUATION_UNIQUE}
 For any processes $p,q \notin F $ and context $C_X$ such that $X$ does not occur in the scope of any conjunction, if $C_X\{p/X\} \in F $ then $C_X\{q/X\} \in F $.
\end{lemma}
\begin{proof}
Suppose that $C_X\{p/X\} \in F $.
We proceed by induction on the depth of the proof tree $\mathcal T$ of $Strip(\mathcal{P}_{\text{CLL}_R} ,M_{\text{CLL}_R} ) \vdash C_{X}\{p/ X\}F$.
Similar to the preceding lemma, we handle only the case   $C_{X} \equiv \langle Y|E \rangle$.
In this situation, the last rule applied in $\mathcal T$ is
\[\text{either}\; \frac{\langle t_Y|E \rangle \{p/X\}F}{\langle Y|E \rangle \{p/X\}F}\;\text{with}\;Y =t_Y \in E\;\text{or}\; \frac{\{rF:\langle Y|E \rangle \{p/X\} \stackrel{\epsilon}{\Longrightarrow}|r\}}{\langle Y|E \rangle \{p/X\} F}.\]
The argument for the former is the same as the one in Lemma~\ref{L:FAILURE_EQUATION_UNIQUE_PRE} and omitted.
In the following, we consider the latter and suppose $\langle Y|E \rangle \{q/X\} \stackrel{\epsilon}{\Longrightarrow} |s$.
By Theorem~\ref{L:LLTS}, it is not difficult to see that, to complete the proof, it suffices to prove that  $s \in F$.
  By Lemma~\ref{L:TAU_ACTION_NORMALIZATION},  there exist $t$ and stable context $C_X^*$ such that
  \[\langle Y|E \rangle \{q/X\} \stackrel{\epsilon}{\Longrightarrow} C_X^*\{q/X\} \stackrel{\epsilon}{\Longrightarrow} |t \Rrightarrow s\]
  and
  \[\langle Y|E \rangle \{r/X\} \stackrel{\epsilon}{\Longrightarrow} C_X^*\{r/X\}\;\text{for any}\; r. \tag{\ref{L:FAILURE_EQUATION_UNIQUE}.1}\]
  In particular, we have $\langle Y|E\rangle \{a_X.0/X\} \stackrel{\epsilon}{\Longrightarrow}  C_X^*\{a_X.0/X\}$ where $a_X$ is a fresh visible action.
    For this transition, applying Lemma~\ref{L:ONE_ACTION_TAU} finitely many times (notice that, in this procedure, since  $a_X.0$ is stable, the clause (2) in Lemma~\ref{L:ONE_ACTION_TAU} is always false), and by the clause (1) in Lemma~\ref{L:ONE_ACTION_TAU}, we get the sequence
\begin{multline*}
    \langle Y|E\rangle \{a_X.0/X\} \equiv  C_{X}^0\{a_X.0/X\}\stackrel{\tau}{\longrightarrow}  C_{X}^1\{a_X.0/X\}\stackrel{\tau}{\longrightarrow}\\
\dots   \stackrel{\tau}{\longrightarrow} C_{X}^n\{a_X.0/X\} \equiv C_{X}^*\{a_X.0/X\}.
\end{multline*}
    Here $n \geq 0$ and for each $1 \leq i \leq n$, $C_{X}^i$ satisfies (C-$\tau$-1,2,3) in Lemma~\ref{L:ONE_ACTION_TAU}.
    Since $X$ does not occur in the scope of any conjunction in $\langle Y|E \rangle$, by (C-$\tau$-3-iv), neither does $X$ in $C_X^n$.
    In addition, by Lemma~\ref{L:PLACE_HOLDER}, we have $C_{X}^n \equiv C_X^*$.
    Hence  $X$ does not occur in the scope of any conjunction in $C_X^*$.


  If $p$ is stable then so is  $C_X^*\{p/X\} $ by Lemma~\ref{L:ONE_ACTION_TAU}.
  Thus, by (\ref{L:FAILURE_EQUATION_UNIQUE}.1), $C_X^*\{p/X\}F $ is one of premises in the last inferring step in $\mathcal T$.
  Hence $C_X^*\{q/X\} \in F $ by applying IH.
  Then $t \in F $ by Lemma~\ref{L:FAILURE_TAU_I}.
  Further, by Lemma~\ref{L:MULTI_STEP_UNFOLDING_FAILURE}, it follows from $t \Rrightarrow s$ that $s \in F $, as desired.

  Next we consider the other case where $p$ is not stable.
  In this situation, due to $p \notin F $, we have
  \[p\stackrel{\tau}{\Longrightarrow}_F|p^*\;\text{for some}\;p^*.\tag{\ref{L:FAILURE_EQUATION_UNIQUE}.2}\]
  In the following, we distinguish two cases based on whether $q$ is stable.\\

\noindent Case  1  $q$ is stable.

 Then, for the transition $\langle Y|E \rangle \{q/X\} \stackrel{\epsilon}{\Longrightarrow} |s$, by Lemma~\ref{L:MULTI_TAU_GF_STABLE}, we have  $s \equiv C_X'\{q/X\}$ for some stable $C_X'$ such that $X$ does not occur in the scope of any conjunction and $C_X\{p/X\}\stackrel{\epsilon}{\Longrightarrow} C_X'\{p/X\}$.
  Moreover, by Lemma~\ref{L:STABILIZATION}, it follows from (\ref{L:FAILURE_EQUATION_UNIQUE}.2) that
  \[C_X'\{p/X\}\stackrel{\epsilon}{\Longrightarrow} |p'\;\text{for some}\;p'.\]
  For this transition, by Lemma~\ref{L:MULTI_TAU_GF_STABLE}, there exist  a stable context $C_{X,\widetilde{Y}}'^*$ and stable processes $p_Y'(Y \in \widetilde{Y})$ that realize (MS-$\tau$-1) -- (MS-$\tau$-7).
  In particular, by (MS-$\tau$-3-ii), it follows from (\ref{L:FAILURE_EQUATION_UNIQUE}.2) that
  \[C_X'\{p/X\}\stackrel{\epsilon}{\Longrightarrow} C_{X,\widetilde{Y}}'^*\{p/X,p^*/\widetilde{Y}\}.\]
Then, since 
$C_{X,\widetilde{Y}}'^*$, $p$ and $p^*$ are stable,   by Lemma~\ref{L:ONE_ACTION_TAU}, so is $C_{X,\widetilde{Y}}'^*\{p/X,p^*/\widetilde{Y}\}$.
  Thus, $C_{X,\widetilde{Y}}'^*\{p/X,p^*/\widetilde{Y}\}  F$ is one of premises of the last inferring step in $\mathcal T$.
  Moreover, by (MS-$\tau$-2), $p' \equiv C_{X,\widetilde{Y}}'^*\{p/X,\widetilde{p_Y'}/\widetilde{Y}\}$.
  Then, by (MS-$\tau$-6) and IH, we obtain
  \[C_{X,\widetilde{Y}}'^*\{q/X,p^*/\widetilde{Y}\} \in F .\]
  Further, by (MS-$\tau$-6) and Lemma~\ref{L:FAILURE_EQUATION_UNIQUE_PRE}, we get
  \[C_{X,\widetilde{Y}}'^*\{q/X,q/\widetilde{Y}\} \in F .\]
  In addition, due to the stableness of $C_X'$, by (MS-$\tau$-4), we have
  \[C_X'\{q/X\}  \Rrightarrow C_{X,\widetilde{Y}}'^*\{q/X,q/\widetilde{Y}\}.\]
  Hence $s \equiv C_X'\{q/X\} \in F $ by Lemma~\ref{L:MULTI_STEP_UNFOLDING_FAILURE}, as desired.\\

\noindent  Case 2 $q$ is not stable.

 By Lemma~\ref{L:MULTI_TAU_GF_STABLE}, for the transition $\langle Y|E \rangle \{q/X\} \stackrel{\epsilon}{\Longrightarrow}|s$, there exist a stable context $C_{X,\widetilde{Z}}'$ and $q_Z'(Z \in \widetilde{Z})$ that satisfy (MS-$\tau$-1) -- (MS-$\tau$-7).
 Amongst them, by (MS-$\tau$-2,7), 
 \[q \stackrel{\tau}{\Longrightarrow}|q_Z'\;\text{for each}\;Z \in \widetilde{Z}\;\text{and}\; s \equiv C_{X,\widetilde{Z}}'\{q/X, \widetilde{q_Z'}/\widetilde{Z}\}. \tag{\ref{L:FAILURE_EQUATION_UNIQUE}.3}\]
  If $q_Z' \in F $ for some $Z \in \widetilde{Z}$ then by Lemma~\ref{L:FAILURE_GF}, we get $s \in F $ (notice that each $Z$ in $\widetilde{Z}$ is 1-active), as desired.
  In the following, we handle the other case where
  \[q_Z'\notin F\;\text{for each}\; Z \in \widetilde{Z}.\tag{\ref{L:FAILURE_EQUATION_UNIQUE}.4}\]
  By (MS-$\tau$-3-ii), it follows from (\ref{L:FAILURE_EQUATION_UNIQUE}.2) that
  \[C_X\{p/X\}\stackrel{\epsilon}{\Longrightarrow} C_{X,\widetilde{Z}}'\{p/X,p^*/\widetilde{Z}\}.\]
 Since $p  \stackrel{\tau}{\longrightarrow}$, $q  \stackrel{\tau}{\longrightarrow}$, $p^* \not \stackrel{\tau}{\longrightarrow}$, $q_Z'  \not\stackrel{\tau}{\longrightarrow}$ for each $Z \in \widetilde{Z}$ and $s \equiv C_{X,\widetilde{Z}}'\{q/X, \widetilde{q_Z'}/\widetilde{Z}\}\not\stackrel{\tau}{\longrightarrow}$, by Lemma~\ref{L:ONE_ACTION_TAU}, $C_{X,\widetilde{Z}}'\{p/X,p^*/\widetilde{Z}\}$ is stable.
  Hence $\mathcal T$ has a proper subtree with root $C_{X,\widetilde{Z}}'\{p/X,p^*/\widetilde{Z}\}F$.
  Then $C_{X,\widetilde{Z}}'\{q/X,p^*/\widetilde{Z}\} \in F$ by (MS-$\tau$-6) and IH.
  Further, by Lemma~\ref{L:FAILURE_EQUATION_UNIQUE_PRE}, it follows from (\ref{L:FAILURE_EQUATION_UNIQUE}.3) and (\ref{L:FAILURE_EQUATION_UNIQUE}.4) that $s\equiv C_{X,\widetilde{Z}}'\{q/X,\widetilde{q_Z'}/\widetilde{Z}\} \in F$, as desired.
\end{proof}

We shall use the notation $Dep(\mathcal T)$ to denote the depth of a given proof tree $\mathcal T$.
Given $p$, $q$ and $\alpha \in Act_{\tau}$, for any proof tree $\mathcal T$ of $Strip(\mathcal{P}_{\text{CLL}_R},M_{\text{CLL}_R}) \vdash p \stackrel{\alpha}{\longrightarrow}q$, it is evident that $\mathcal T$ involves only rules in Table~\ref{Ta:OPERATIONAL_RULES}.
Moreover, since each rule in Table~\ref{Ta:OPERATIONAL_RULES} has only finitely many premises, it is not difficult to show that $Dep(\mathcal T) < \omega$ by induction on the depth of $\mathcal T$.
This makes it legitimate to use arithmetical expressions with the form like $\underset{{\mathcal T} \in \Omega}{\sum}Dep({\mathcal T})$ where $\Omega$ is a finite set and each ${\mathcal T} \in \Omega$ is a proof tree for some labelled transition $p \stackrel{\alpha}{\longrightarrow} r$.

\begin{mydefn}
  Given $p\stackrel{\epsilon}{\Longrightarrow}_Fq$ and a finite set $\Omega$ of proof trees, we say that $\Omega$ is a \emph{proof forest} for $p\stackrel{\epsilon}{\Longrightarrow}_Fq$ if there exist $p_i(0 \leq i \leq n)$ such that

\begin{enumerate}[(1)]
 \renewcommand{\theenumi}{(\arabic{enumi})}
  \item $p \equiv p_0 \stackrel{\tau}{\longrightarrow}_F p_1\stackrel{\tau}{\longrightarrow}_F \dots \stackrel{\tau}{\longrightarrow}_F p_n \equiv q$,
  \item  for each $i < n$, $\Omega$ contains exactly one proof tree for $Strip(\mathcal{P}_{\text{CLL}_R},M_{\text{CLL}_R}) \vdash p_i \stackrel{\tau}{\longrightarrow}p_{i+1}$, and
  \item  for each $\mathcal T \in \Omega$, $\mathcal T$ is a proof tree for $Strip(\mathcal{P}_{\text{CLL}_R},M_{\text{CLL}_R}) \vdash p_i \stackrel{\tau}{\longrightarrow}p_{i+1}$ for some $i<n$.
\end{enumerate}
The depth of $\Omega$ is defined as $Dep(\Omega) \triangleq \underset{{\mathcal T} \in \Omega}{\sum}Dep({\mathcal T})$.
Similarly, we may define the notion of a proof forest for $p \stackrel{a}{\Longrightarrow}_Fq$.
\end{mydefn}

 It is obvious that $p\stackrel{\epsilon}{\Longrightarrow}_Fq$ (or, $p\stackrel{a}{\Longrightarrow}_Fq$) holds if and only if there exists a proof forest for it. The following lemma will prove extremely useful in establishing the main result in this section and its proof involves induction on the depths of proof forests.

\begin{lemma}\label{L:UNIQUE_SOLUTION}
  Let $C_X$ be any context where $X$ is strongly guarded and does not occur in the scope of any conjunction.
  For any processes $p,q \notin F$ with $p \bowtie q$, if $p =_{RS} C_X\{p/X\}$ and $q =_{RS} C_X\{q/X\}$ then $p =_{RS} q$.
\end{lemma}
\begin{proof}
  Suppose $p,q \notin F$ with $p \bowtie q$, $p=_{RS}C_X\{p/X\}$ and $q =_{RS}C_X\{q/X\}$.
  It is sufficient to prove that $p \sqsubseteq_{RS} q$.
  Put
  \begin{multline*}
       {\mathcal R} \triangleq \{(B_X\{p/X\},B_X\{q/X\}):  X\;\text{does not occur in the scope of any conjunction in }B_X\}.
  \end{multline*}
  By Prop.~\ref{P:COINCIDENCE} and Lemma~\ref{L:ALT_UP_TO}, it suffices to prove that $\mathcal R$ is an alternative ready simulation relation up to $\underset{\thicksim}{\sqsubset}_{RS}$.
Let $(B_X\{p/X\},B_X\{q/X\})\in \mathcal R$.

   (\textbf{ALT-upto-1}) Assume that $B_X\{p/X\}   \stackrel{\epsilon}{\Longrightarrow}_F| p'$ and $\Omega$ is any proof forest for it.
  Hence $B_X\{p/X\} \equiv p_0 \stackrel{\tau}{\longrightarrow}_F p_1 \stackrel{\tau}{\longrightarrow}\dots p_{n-1}\stackrel{\tau}{\longrightarrow}_F|p_n \equiv p' $ for some $p_i(0 \leq i \leq n)$, and $\Omega$ exactly consists of proof trees ${\mathcal T}_i(0 \leq i < n)$  for $Strip(\mathcal{P}_{\text{CLL}_R},M_{\text{CLL}_R}) \vdash p_i \stackrel{\tau}{\longrightarrow}p_{i+1}$.
  We intend to prove that there exists $q'$ such that $B_X\{q/X\} \stackrel{\epsilon}{\Longrightarrow}_F|q'$ and $p' \underset{\thicksim}{\sqsubset}_{RS}{\mathcal R}\underset{\thicksim}{\sqsubset}_{RS} q'$ by induction on $Dep(\Omega)$.
  It is a routine case analysis on $B_X$.
  We treat only three cases as examples.\\

  \noindent Case 1 $B_X \equiv X$.

  Then $B_X\{p/X\} \equiv p \stackrel{\epsilon}{\Longrightarrow}_F|p'$.
  Thus it follows from $p =_{RS} C_X\{p/X\}$ that
  \[C_X\{p/X\} \stackrel{\epsilon}{\Longrightarrow}_F|s\;\text{and}\; p'\underset{\thicksim}{\sqsubset}_{RS}s\;\text{for some}\;s.\]
  Since $X$ is strongly guarded and does not occur in the scope of any conjunction in $C_X$, by Lemma~\ref{L:MULTI_TAU_GF_STABLE}, there exists a stable context $C_X'$ such that
  \begin{enumerate}[({a.}1)]
\renewcommand{\theenumi}{(a.\arabic{enumi})}
    \item \ $s \equiv C_X'\{p/X\}$,
    \item \ $X$ is strongly guarded and does not occur in the scope of any conjunction in $C_X'$, and
    \item \ $C_X\{q/X\} \stackrel{\epsilon}{\Longrightarrow}  C_X'\{q/X\}$.
  \end{enumerate}
  Since $s \equiv C_X'\{p/X\} \not\stackrel{\tau}{\longrightarrow}$, by (a.2) and Lemma~\ref{L:ONE_ACTION_VISIBLE_GUARDED}, we have $C_X'\{q/X\} \not\stackrel{\tau}{\longrightarrow}$.
  Moreover, by Lemma~\ref{L:FAILURE_EQUATION_UNIQUE}, $C_X'\{q/X\} \notin F $ follows from $C_X'\{p/X\} \notin F $ and $p,q \notin F $.
  Hence $C_X\{q/X\} \stackrel{\epsilon}{\Longrightarrow}_F|C_X'\{q/X\}$ by (a.3) and Lemma~\ref{L:FAILURE_TAU_I}.
  Further, it follows from $q =_{RS} C_X\{q/X\}$ that
  \[q \stackrel{\epsilon}{\Longrightarrow}_F| q'\;\text{and}\;C_X'\{q/X\}\underset{\thicksim}{\sqsubset}_{RS}q'\;\text{for some}\;q'.\]
  Therefore, $B_X\{q/X\} \equiv q \stackrel{\epsilon}{\Longrightarrow}_F|q'$ and $p' \underset{\thicksim}{\sqsubset}_{RS} s \equiv C_X'\{p/X\}{\mathcal R}C_X'\{q/X\}\underset{\thicksim}{\sqsubset}_{RS}q'$.\\

  \noindent Case 2 $B_X \equiv \langle Y|E \rangle$.

  If $ \langle Y|E \rangle \{p/X\}$ is stable then so is $\langle Y|E \rangle \{q/X\}$ by $p \bowtie q$ and  Lemma~\ref{L:ONE_ACTION_TAU}.
  Moreover, by Lemma~\ref{L:FAILURE_EQUATION_UNIQUE}, we have $\langle Y|E \rangle \{q/X\} \notin F $ because of $\langle Y|E \rangle \{p/X\} \notin F $.
Hence $\langle Y|E \rangle \{q/X\} \stackrel{\epsilon}{\Longrightarrow}_F| \langle Y|E \rangle \{q/X\}$ and $(\langle Y|E \rangle \{p/X\},\langle Y|E \rangle \{q/X\}) \in \underset{\thicksim}{\sqsubset}_{RS}{\mathcal R}\underset{\thicksim}{\sqsubset}_{RS}$ due to the reflexivity of $\underset{\thicksim}{\sqsubset}_{RS}$.

  Next we handle the other case where $ \langle Y|E \rangle \{p/X\}$ is not stable.
  Clearly, the last rule applied in  $\mathcal T_0$ is
  \[\frac{\langle t_Y|E \rangle \{p/X\}\stackrel{\tau}{\longrightarrow}p_1}{\langle Y|E \rangle \{p/X\}\stackrel{\tau}{\longrightarrow}p_1}\;\text{with}\; Y = t_Y \in E.\]
  Thus, $\mathcal T_0$ contains a proper subtree, say $\mathcal T_0'$, which is a proof tree of $Strip(\mathcal{P}_{\text{CLL}_R},M_{\text{CLL}_R}) \vdash \langle t_Y|E \rangle \{p/X\} \stackrel{\tau}{\longrightarrow}p_1$  and $Dep({\mathcal T}_0')<Dep({\mathcal T}_0)$.
  Thus $\Omega' \triangleq \{{\mathcal T}_0',{\mathcal T}_i: 1\leq i \leq n-1\}$ is a proof forest for $\langle t_Y|E \rangle \{p/X\} \stackrel{\epsilon}{\Longrightarrow}_F|p'$;
  moreover
\[Dep(\Omega') < Dep(\Omega).\]
  Then, by Lemma~\ref{L:ONE_STEP_UNFOLDING_VARIABLE}(5) and IH, we have  $\langle t_Y|E \rangle \{q/X\} \stackrel{\epsilon}{\Longrightarrow}_F|q'$ and $p' \underset{\thicksim}{\sqsubset}_{RS}{\mathcal R}\underset{\thicksim}{\sqsubset}_{RS} q'$ for some $q'$.
  Moreover, we also have $B_X\{q/X\} \equiv \langle Y|E \rangle \{q/X\}\stackrel{\epsilon}{\Longrightarrow}_F|q'$, as desired.\\

  \noindent Case 3 $B_X \equiv D_X \Box D_X'$.

  If $B_X\{p/X\}$ is stable then we can proceed analogously to  Case 2 with $\langle Y|E \rangle \{p/X\} \not\stackrel{\tau}{\longrightarrow}$. In the following, we consider the case $B_X\{p/X\} \stackrel{\tau}{\longrightarrow}$.

  For the transitions $ D_X\{p/X\} \Box D_X'\{p/X\} \equiv p_0 \stackrel{\tau}{\longrightarrow}_F \dots \stackrel{\tau}{\longrightarrow}_F| p_n \equiv p'(n \geq 1)$, there exist two  sequences of processes $t_0(\equiv D_X\{p/X\}),\dots,t_n$ and $s_0(\equiv D_X'\{p/X\}),\dots,s_n$ such that $t_n,s_n$ are consistent and stable, $p_n\equiv t_n \Box  s_n$, and for each $0\leq i<n$, $p_i \equiv t_i \Box s_i$ and the last rule applied in ${\mathcal T}_i$ is
\[
\text{either}\;\frac{t_i \stackrel{\tau}{\longrightarrow} t_{i+1}}{t_i \Box  s_i \stackrel{\tau}{\longrightarrow} t_{i+1} \Box s_{i+1}}\;\text{or}\;   \frac{s_i \stackrel{\tau}{\longrightarrow} s_{i+1}}{t_i \Box s_i \stackrel{\tau}{\longrightarrow} t_{i+1} \Box s_{i+1}}.
\]
For the former, $s_{i+1} \equiv s_i$ and ${\mathcal T}_i$ contains a proper subtree ${\mathcal T}_i'$ which is a proof tree for $Strip(\mathcal{P}_{\text{CLL}_R},M_{\text{CLL}_R}) \vdash t_i \stackrel{\tau}{\longrightarrow} t_{i+1}$.
We use $\Omega_1$ to denote the (finite) set of all these proof trees ${\mathcal T}_i'$.
Similarly, for the latter, $t_{i+1} \equiv t_i$ and ${\mathcal T}_i$ contains a proper subtree ${\mathcal T}_i''$ which is a proof tree for $Strip(\mathcal{P}_{\text{CLL}_R},M_{\text{CLL}_R}) \vdash s_i \stackrel{\tau}{\longrightarrow} s_{i+1}$.
We use $\Omega_2$ to denote the (finite) set of all these proof trees ${\mathcal T}_i''$.
   It is obvious that $\Omega_1$ is a proof forest for $D_X\{p/X\} \stackrel{\epsilon}{\Longrightarrow}_F |t_n$; moreover,
\[  Dep(\Omega_1) <   Dep(\Omega).\]
   Thus, by IH, we have  $D_X\{q/X\} \stackrel{\epsilon}{\Longrightarrow}_F|q_1'$ and $t_n \underset{\thicksim}{\sqsubset}_{RS}{\mathcal R}\underset{\thicksim}{\sqsubset}_{RS} q_1'$ for some $q_1'$.
   Similarly, for the transition $D_X'\{p/X\} \stackrel{\epsilon}{\Longrightarrow}_F |s_n$, we also have $D_X'\{q/X\} \stackrel{\epsilon}{\Longrightarrow}_F|q_2'$ and $s_n \underset{\thicksim}{\sqsubset}_{RS}{\mathcal R}\underset{\thicksim}{\sqsubset}_{RS} q_2'$ for some $q_2'$.
   Then, by Theorem~\ref{L:pre_precongruence}, it is easy to check that
   $p' \equiv t_n \Box s_n \underset{\thicksim}{\sqsubset}_{RS}{\mathcal R}\underset{\thicksim}{\sqsubset}_{RS} q_1' \Box q_2'$.
   Moreover, we also have  $B_X\{q/X\} \equiv D_X\{q/X\} \Box D_X'\{q/X\} \stackrel{\epsilon}{\Longrightarrow}_F|q_1'\Box q_2'$.\\

   (\textbf{ALT-upto-2}) Suppose that $B_X\{p/X\} $ and $B_X\{q/X\} $ are stable.
 Let $B_X\{p/X\}\stackrel{a}{\Longrightarrow}_F|p'$ and $\Omega$ be its proof forest.
So, there exist $ p_0,\dots,p_n (n \geq 1)$ such that
\[B_{X}\{p/X\} \equiv p_0 \stackrel{a}{\longrightarrow}_F p_1 \stackrel{\tau}{\longrightarrow}_F \dots \stackrel{\tau}{\longrightarrow}_F| p_n \equiv p',\tag{\ref{L:UNIQUE_SOLUTION}.1}\]
 and $\Omega$ exactly consists of proof trees ${\mathcal T}_i$ for $Strip(\mathcal{P}_{\text{CLL}_R},M_{\text{CLL}_R}) \vdash p_i \stackrel{\alpha_i}{\longrightarrow}p_{i+1}$ for $i < n$, where $\alpha_0 = a$ and $\alpha_j = \tau (1 \leq j < n)$.
We want to prove that there exists $q'$ such that $B_X\{q/X\} \stackrel{a}{\Longrightarrow}_F|q'$ and $p' \underset{\thicksim}{\sqsubset}_{RS}{\mathcal R}\underset{\thicksim}{\sqsubset}_{RS} q'$ by induction on $ Dep(\Omega)$.
Since $B_X\{p/X\}$ is stable and $X$ does not occur in the scope of any conjunction in $B_X$,  the topmost operator of $B_X$ is neither disjunction nor conjunction.
Thus, we distinguish five cases based on the form of $B_X$.\\

\noindent Case 1 $B_X \equiv X$.

  Due to $B_X\{p/X\} \equiv p \stackrel{a}{\Longrightarrow}_F|p'$, we have $p \notin F$.
  Moreover, since $p(\equiv B_X\{p/X\})$ is stable, we get $p \stackrel{\epsilon}{\Longrightarrow}_F|p$.
  Hence it follows from $p =_{RS} C_X\{p/X\}$ that
  \[C_X\{p/X\} \stackrel{\epsilon}{\Longrightarrow}_F|s\;\text{and}\;p\underset{\thicksim}{\sqsubset}_{RS}s\;\text{for some}\; s.\]
  Further, since $X$ is strongly guarded and does not occur in the scope of any conjunction in $C_X$, by Lemma~\ref{L:MULTI_TAU_GF_STABLE}, there exists a stable context $C_X'$ such that
  \begin{enumerate}[{(b.}1)]
\renewcommand{\theenumi}{(b.\arabic{enumi})}
    \item \ $X$ is strongly guarded and does not occur in the scope of any conjunction in $C_X'$,
    \item \ $s \equiv C_X'\{p/X\}$, and
    \item \ $C_X\{q/X\} \stackrel{\epsilon}{\Longrightarrow} C_X'\{q/X\}$.
  \end{enumerate}
  Then it follows from $p \underset{\thicksim}{\sqsubset}_{RS} s\equiv C_X'\{p/X\}$ and $p \stackrel{a}{\Longrightarrow}_F|p'$ that
    \[C_X'\{p/X\} \stackrel{a}{\Longrightarrow}_F|s' \;\text{and}\; p'\underset{\thicksim}{\sqsubset}_{RS}s'\;\text{for some}\;s'.\]
Since $p\not\stackrel{\tau}{\longrightarrow}$, by (b.1), Lemma~\ref{L:ONE_ACTION_VISIBLE_GUARDED} and \ref{L:MULTI_TAU_GF_STABLE}, there exists a stable context $C_X''$ such that
\begin{enumerate}[({c.}1)]
\renewcommand{\theenumi}{(c.\arabic{enumi})}
  \item \ $s'\equiv C_X''\{p/X\}$,
  \item \ $X$ does not occur in the scope of any conjunction in $C_X''$, and
  \item \ $C_X'\{q/X\} \stackrel{a}{\longrightarrow}\stackrel{\epsilon}{\Longrightarrow} C_X''\{q/X\}$.
\end{enumerate}
Moreover, since  $q(\equiv B_X\{q/X\})$ is stable, so is $C_X''\{q/X\}$.
Then, by (b.3) and (c.3), we have
\[C_X\{q/X\} \stackrel{\epsilon}{\Longrightarrow} |C_X'\{q/X\}  \stackrel{a}{\Longrightarrow}|C_X''\{q/X\}.\]
Further, by Lemma~\ref{L:FAILURE_EQUATION_UNIQUE} and \ref{L:FAILURE_TAU_I}, it follows from $p,q,C_X\{p/X\},C_X'\{p/X\},C_X''\{p/X\}\notin F$ that
\[C_X\{q/X\} \stackrel{\epsilon}{\Longrightarrow}_F |C_X'\{q/X\}  \stackrel{a}{\Longrightarrow}_F|C_X''\{q/X\}.\tag{\ref{L:UNIQUE_SOLUTION}.2}\]
%
  Then, since $C_X\{q/X\}=_{RS}q$ and $q\not\stackrel{\tau}{\longrightarrow}$, we get
  \[C_X'\{q/X\}\underset{\thicksim}{\sqsubset}_{RS}q.\]
  Further, due to (\ref{L:UNIQUE_SOLUTION}.2), it follows that
  \[B_X\{q/X\}(\equiv q) \stackrel{a}{\Longrightarrow}_F|q'\;\text{and}\;C_X''\{q/X\}\underset{\thicksim}{\sqsubset}_{RS}q'\;\text{for some}\;q'.\]
  Moreover, $p' \underset{\thicksim}{\sqsubset}_{RS} s'\equiv C_X''\{p/X\}{\mathcal R}C_X''\{q/X\} \underset{\thicksim}{\sqsubset}_{RS}q'$, as desired.\\

 \noindent Case 2 $B_X \equiv \alpha.D_X$.

  So $\alpha = a$ and $D_X\{p/X\}  \stackrel{\epsilon}{\Longrightarrow}_F|p'$.
  Clearly, $(D_X\{p/X\},D_X\{q/X\})\in R$.
  By (ALT-upto-1), there exists $q'$ such that  $D_X\{q/X\}  \stackrel{\epsilon}{\Longrightarrow}_F|q'$ and $p' \underset{\thicksim}{\sqsubset}_{RS}{\mathcal R}\underset{\thicksim}{\sqsubset}_{RS} q'$.
  Moreover, it is evident that $\alpha.D_X\{q/X\}  \stackrel{a}{\Longrightarrow}_F|q'$.\\

\noindent Case 3 $B_X \equiv D_X \Box D_X'$.

   W.l.o.g, assume that the last rule applied in $\mathcal T_0$ is $\frac{D_X\{p/X\} \stackrel{a}{\longrightarrow}p_1,\;D_X'\{p/X\} \not\stackrel{\tau}{\longrightarrow}}{D_X\{p/X\} \Box D_X'\{p/X\} \stackrel{a}{\longrightarrow}p_1}$.
   Then $\mathcal T_0$ has a proper subtree, say $\mathcal T_0'$, which is a proof tree for $Strip(\mathcal{P}_{\text{CLL}_R},M_{\text{CLL}_R}) \vdash D_X\{p/X\} \stackrel{a}{\longrightarrow}p_1$.
Clearly, $\Omega' \triangleq \{{\mathcal T}_0',{\mathcal T}_i:1 \leq i \leq n-1\}$ is a proof forest for $D_X\{p/X\} \stackrel{a}{\Longrightarrow}_F|p'$ and $Dep(\Omega') <  Dep(\Omega)$.
   Moreover, since $B_X\{q/X\}$ is stable, so are $D_X\{q/X\}$ and $D_X'\{q/X\}$.
   Then, by IH, we have  $D_X\{q/X\} \stackrel{a}{\Longrightarrow}_F|q'$ and $p' \underset{\thicksim}{\sqsubset}_{RS}{\mathcal R}\underset{\thicksim}{\sqsubset}_{RS} q'$ for some $q'$.
   Moreover, $D_X'\{p/X\} \notin F$ because of $B_X\{p/X\} \notin F$, which, by Lemma~\ref{L:FAILURE_EQUATION_UNIQUE}, implies $D_X'\{q/X\} \notin F$.
   Hence $B_X\{q/X\} \equiv D_X\{q/X\} \Box D_X'\{q/X\} \notin F$, and
   $B_X\{q/X\} \equiv D_X\{q/X\} \Box D_X'\{q/X\}\stackrel{a}{\Longrightarrow}_F|q'$, as desired.\\

\noindent Case 4 $B_X \equiv D_X \parallel_A D_X'$.

    Then the last rule applied in $\mathcal T_0$ is one of the following three formats:
       \begin{enumerate}[(1)]
       \renewcommand{\theenumi}{(\arabic{enumi})}
     \item  $\frac{D_X\{p/X\} \stackrel{a}{\longrightarrow} t_1, D_X'\{p/X\} \stackrel{a}{\longrightarrow} s_1}{B_{X}\{p/X\} \parallel_A  D_{X}\{p/X\} \stackrel{a}{\longrightarrow} t_1 \parallel_A s_1}$ with $a \in A$ and $p_1 \equiv t_1 \parallel_A s_1$;
     \item  $\frac{D_{X}\{p/X\} \stackrel{a}{\longrightarrow} t_1,\; D_{X}'\{p/X\} \not\stackrel{\tau}{\longrightarrow} }{D_{X}\{p/X\} \parallel_A  D_{X}'\{p/X\} \stackrel{a}{\longrightarrow} t_1 \parallel_A D_{X}'\{p/X\}}$ with $a \notin A$ and $p_1 \equiv t_1 \parallel_A D_{X}'\{p/X\}$;
     \item   $\frac{D_{X}'\{p/X\} \stackrel{a}{\longrightarrow} s_1,\;D_{X}\{p/X\} \not\stackrel{\tau}{\longrightarrow} }{D_{X}\{p/X\} \parallel_A  D_{X}'\{p/X\} \stackrel{a}{\longrightarrow} D_{X}\{p/X\} \parallel_A s_1}$ with $a \notin A$ and $p_1 \equiv D_{X}\{p/X\} \parallel_A s_1$.
   \end{enumerate}
   We treat the first one, and the proof of the later two runs, as in Case 3.
  Clearly, $\mathcal T_0$ has two proper subtrees $\mathcal T_0'$ and $\mathcal T_0''$, which are proof trees for $D_X\{p/X\} \stackrel{a}{\longrightarrow} t_1$ and $ D_X'\{p/X\} \stackrel{a}{\longrightarrow} s_1$ respectively.
   Moreover, for the transitions $ p_1 \stackrel{\tau}{\longrightarrow} \dots \stackrel{\tau}{\longrightarrow}| p_n$, there exist two processes sequences $t_1,\dots,t_n$ and $s_1,\dots,s_n$ such that $t_n,s_n$ are stable, $p_n\equiv t_n\parallel_A s_n$, and for each $1\leq i<n$, $p_i \equiv t_i \parallel_A s_i$ and the last rule applied in ${\mathcal T}_i$ is
\[
\text{either}\;\frac{t_i \stackrel{\tau}{\longrightarrow} t_{i+1}}{t_i\parallel_A  s_i \stackrel{\tau}{\longrightarrow} t_{i+1} \parallel_A s_{i+1}}\;\text{or}\;   \frac{s_i \stackrel{\tau}{\longrightarrow} s_{i+1}}{t_i\parallel_A  s_i \stackrel{\tau}{\longrightarrow} t_{i+1} \parallel_A s_{i+1}}.
\]
For the former,  $s_{i+1} \equiv s_i$ and ${\mathcal T}_i$ contains a proper subtree ${\mathcal T}_i'$ which is a proof tree for $Strip(\mathcal{P}_{\text{CLL}_R},M_{\text{CLL}_R}) \vdash t_i \stackrel{\tau}{\longrightarrow} t_{i+1}$.
We use $\Omega_1$ to denote the (finite) set of all these proof tree ${\mathcal T}_i'$.
Similarly, for the latter,  $t_{i+1} \equiv t_i$ and ${\mathcal T}_i$ contains a proper subtree ${\mathcal T}_i''$ which is a proof tree for $Strip(\mathcal{P}_{\text{CLL}_R},M_{\text{CLL}_R}) \vdash s_i \stackrel{\tau}{\longrightarrow} s_{i+1}$.
We use $\Omega_2$ to denote the (finite) set of all these proof tree ${\mathcal T}_i''$.
   Clearly, $\Omega' \triangleq \{{\mathcal T}_0'\} \cup \Omega_1$ is a proof forest for $D_X\{p/X\} \stackrel{a}{\Longrightarrow}_F |t_n$ and $Dep(\Omega') < Dep(\Omega)$.
   Thus, by IH, we have  $D_X\{q/X\} \stackrel{a}{\Longrightarrow}_F|q_1'$ and $t_n \underset{\thicksim}{\sqsubset}_{RS}{\mathcal R}\underset{\thicksim}{\sqsubset}_{RS} q_1'$ for some $q_1'$.
   Similarly, for the transition $D_X'\{p/X\} \stackrel{a}{\Longrightarrow}_F |s_n$, we also have $D_X'\{q/X\} \stackrel{a}{\Longrightarrow}_F|q_2'$ and $s_n \underset{\thicksim}{\sqsubset}_{RS}{\mathcal R}\underset{\thicksim}{\sqsubset}_{RS} q_2'$ for some $q_2'$.
   Therefore, by Theorem~\ref{L:pre_precongruence}, we obtain
   $p' \equiv t_n \parallel_A s_n \underset{\thicksim}{\sqsubset}_{RS}{\mathcal R}\underset{\thicksim}{\sqsubset}_{RS} q_1' \parallel_A q_2'$.
   Moreover, it is not difficult to see that  $B_X\{q/X\} \equiv D_X\{q/X\} \parallel_A D_X'\{q/X\} \stackrel{a}{\Longrightarrow}_F|q_1'\parallel_A q_2'$ because of    $B_X\{q/X\}\not\stackrel{\tau}{\longrightarrow}$, $D_X\{q/X\} \stackrel{a}{\Longrightarrow}_F|q_1'$ and $D_X'\{q/X\} \stackrel{a}{\Longrightarrow}_F|q_2'$.\\

   \noindent Case 5 $B_X \equiv \langle Y|E \rangle$.

   Clearly, the last rule applied in  $\mathcal T_0$ is
   $\frac{\langle t_Y|E \rangle \{p/X\}\stackrel{a}{\longrightarrow}p_1}{\langle Y|E \rangle \{p/X\}\stackrel{a}{\longrightarrow}p_1}.$
  Hence $\mathcal T_0$ contains a proper subtree, say $\mathcal T_0'$, which is a proof tree for $Strip(\mathcal{P}_{\text{CLL}_R},M_{\text{CLL}_R}) \vdash \langle t_Y|E \rangle \{p/X\} \stackrel{a}{\longrightarrow}p_1$, and $Dep({\mathcal T}_0')<Dep({\mathcal T}_0)$.
  So, $\Omega' \triangleq \{{\mathcal T}_0',{\mathcal T}_i:1 \leq i <n\}$ is a proof forest for $\langle t_Y|E \rangle \{p/X\} \stackrel{a}{\Longrightarrow}_F|p'$ and $Dep(\Omega')< Dep(\Omega)$.
   Then, by IH, we have  $\langle t_Y|E \rangle \{q/X\} \stackrel{a}{\Longrightarrow}_F|q'$ and $p' \underset{\thicksim}{\sqsubset}_{RS}{\mathcal R}\underset{\thicksim}{\sqsubset}_{RS} q'$ for some $q'$; moreover, $B_X\{q/X\} \equiv \langle Y|E \rangle \{q/X\}\stackrel{a}{\Longrightarrow}_F|q'$, as desired.\\

   (\textbf{ALT-upto-3}) Let $B_X\{p/X\}$ and $B_X\{q/X\}$ be stable and $B_X\{p/X\} \notin F $.
  We shall prove ${\mathcal I}(B_X\{p/X\})\supseteq {\mathcal I}(B_X\{q/X\})$, the converse inclusion may be proved in a similar manner and is omitted.
  Assume that $B_X\{q/X\}\stackrel{a}{\longrightarrow} q'$.
  Then, for such $a$-labelled transition, by Lemma~\ref{L:ONE_ACTION_VISIBLE}, there exist $B_{X}'$, $B_{X,\widetilde{Y}}'$ and $B_{X,\widetilde{Y}}''$ with $X \notin \widetilde{Y}$  that satisfy (CP-$a$-1) -- (CP-$a$-4).
  In case  $\widetilde{Y} = \emptyset$, it immediately follows from (CP-$a$-3-iii) that $B_X\{p/X\}\stackrel{a}{\longrightarrow} B_{X,\widetilde{Y}}''\{p/X\}$.

  Next we handle the case $\widetilde{Y} \neq \emptyset$.
  In this situation, by (CP-$a$-3-iii), to complete the proof, it suffices to prove that ${\mathcal I}(p) = {\mathcal I}(q)$.
        By (CP-$a$-1) and (CP-$a$-3-i), we have
        \[B_{X}\{r/X\}  \Rrightarrow B_{X,\widetilde{Y}}' \{r/X,r/\widetilde{Y}\}\;\text{for any}\;r.\]
        Then, since $B_{X}\{p/X\}$ and $B_{X}\{q/X\}$ are stable, by $\widetilde{Y} \neq \emptyset$, (CP-$a$-2) and Lemmas~\ref{L:MULTI_STEP_UNFOLDING_ACTION} and \ref{L:ONE_ACTION_TAU_GF}, it follows that both $p$ and $q$ are stable.
        Hence $p \stackrel{\epsilon}{\Longrightarrow}_F|p$ by $p \notin F$.
        Then, due to  $p =_{RS} C_X\{p/X\}$, we have
        \[C_X\{p/X\} \stackrel{\epsilon}{\Longrightarrow}_F|s\;\text{and}\;p\underset{\thicksim}{\sqsubset}_{RS} s\;\text{for some}\;s.\]
        For the transition above, since $X$ is strongly guarded in $C_X$, by Lemma~\ref{L:MULTI_TAU_GF_STABLE}, there exists a stable context $D_X$ such that
        \begin{enumerate}[({d.}1)]
\renewcommand{\theenumi}{(d.\arabic{enumi})}
          \item \ $s \equiv D_X\{p/X\} \not\stackrel{\tau}{\longrightarrow}$,
          \item \ $X$ is strongly guarded and does not occur in the scope of any conjunction in $D_X$, and
          \item \ $C_X\{q/X\} \stackrel{\epsilon}{\Longrightarrow} D_X\{q/X\}$.
        \end{enumerate}
     Hence ${\mathcal I}(p)= {\mathcal I}(D_X\{p/X\})$ by (d.1), $p\underset{\thicksim}{\sqsubset}_{RS} s$ and $p \notin F$.
  Moreover, by (d.1), (d.2) and Lemma~\ref{L:ONE_ACTION_VISIBLE_GUARDED}, we have $D_X\{q/X\}\not\stackrel{\tau}{\longrightarrow}$ and
        \[{\mathcal I}(p)={\mathcal I}(D_X\{p/X\}) = {\mathcal I}(D_X\{q/X\}).\]
        We also obtain $D_X\{q/X\} \notin F$ by $p\notin F$, $q \notin F$, $s \equiv D_X\{p/X\} \notin F $ and Lemma~\ref{L:FAILURE_EQUATION_UNIQUE}.
        So, $C_X\{q/X\} \stackrel{\epsilon}{\Longrightarrow}_F|D_X\{q/X\}$ by Lemma~\ref{L:FAILURE_TAU_I}.
        Further, it follows from $q =_{RS} C_X\{q/X\}$ and $q \not\stackrel{\tau}{\longrightarrow}$ that $D_X\{q/X\} \underset{\thicksim}{\sqsubset}_{RS} q$.
        Hence ${\mathcal I}(D_X\{q/X\})={\mathcal I}(q)$ because of $D_X\{q/X\} \notin F$. Therefore, ${\mathcal I}(p) = {\mathcal I}(D_X\{p/X\}) = {\mathcal I}(D_X\{q/X\})={\mathcal I}(q) $, as desired.
\end{proof}

The next lemma is the crucial step in the demonstrating the assertion that $\langle X|X=t_X\rangle$ is a consistent solution of a given equation $X=_{RS}t_X$ whenever consistent solutions exist.

\begin{lemma}\label{L:UNIQUE_SOLUTION_EXISTENCE}
  For any term $t_X$ where $X$ is strongly guarded and does not occur in the scope of any conjunction, if $q =_{RS} t_X\{q/X\}$ for some $q \notin F$ then $\langle X | X= t_X \rangle \notin F$.
\end{lemma}
\begin{proof}
  Assume $p =_{RS} t_X\{p/X\}$ for some $p \notin F$. Then  $t_X\{p/X\} \notin F$.
Set
\[\Omega=
\left\{B_Y\{\langle X|X=t_X \rangle/Y\}:
  \begin{array}{l}
        B_Y\{p/Y\} \notin F\;\text{and}\;Y\;\text{does not occur in the scope of} \\
         \text{any conjunction in}\;B_Y
  \end{array}
\right\}.
\]
  It is obvious that $\langle X|X=t_X \rangle \in \Omega$ by taking $B_Y \triangleq Y$.
Thus we intend to show that $\Omega \cap F = \emptyset$.
Assume $C_Y\{\langle X|X=t_X \rangle/Y\} \in \Omega$.
Let $\mathcal T$ be any proof tree for $Strip(\mathcal{P}_{\text{CLL}_R} ,M_{\text{CLL}_R} ) \vdash C_Y\{\langle X|X=t_X \rangle/Y\}F$.
Similar to Lemma~\ref{L:FAILURE_S_VS_NS}, it is sufficient to prove that $ \mathcal T$ has a proper subtree with root $sF$ for some $s \in \Omega$, which is a routine case analysis based on the last rule applied in $\mathcal T$.
Here we  treat only two cases as examples.\\

\noindent Case 1 $C_Y \equiv Y$.

       Then $C_Y\{\langle X|X=t_X \rangle/Y\} \equiv  \langle X|X=t_X \rangle $.
       Clearly, the last rule applied in $\mathcal T$ is
\[\text{either}\; \frac{\langle t_X|X=t_X \rangle F}{\langle X|X=t_X \rangle F}\;\text{or}\;\frac{\{rF:\langle X|X=t_X \rangle \stackrel{\epsilon}{\Longrightarrow}|r\}}{\langle X|X=t_X \rangle F}.\]

For the former, $\mathcal T$ has a proper subtree with root $\langle t_X|X=t_X \rangle F$;
moreover, $\langle t_X|X=t_X \rangle \equiv t_X\{\langle X|X=t_X\rangle/X\} \in \Omega$ due to $t_X\{p/X\} \notin F$, as desired.

For the latter, if $\langle X|X=t_X \rangle\not\stackrel{\tau}{\longrightarrow}$, then, in $\mathcal T$, the unique node directly above the root is labelled with $\langle X|X=t_X \rangle F$;
moreover $\langle X|X=t_X \rangle \in \Omega$, as desired.
 In the following, we consider the nontrivial case $\langle X|X=t_X \rangle \stackrel{\tau}{\longrightarrow}$.
Since $t_X\{p/X\} \notin F$, by Theorem~\ref{L:LLTS}, we get $t_X\{p/X\} \stackrel{\epsilon}{\Longrightarrow}_F|p'$ for some $p'$. For this transition, since $X$ is strongly guarded and does not occur in the scope of any conjunction, by Lemma~\ref{L:MULTI_TAU_GF_STABLE}, there exists a stable context $B_X$ such that
\begin{enumerate}[({a.}1)]
   \renewcommand{\theenumi}{(a.\arabic{enumi})}
  \item \ $X$ is strongly guarded and does not occur in the scope of any conjunction,
  \item \ $p' \equiv B_X\{p/X\}$, and
  \item \ $t_X\{\langle X|X=t_X \rangle/X\} \stackrel{\epsilon}{\Longrightarrow}B_X\{\langle X|X=t_X \rangle/X\}$.
\end{enumerate}
Since $p'\equiv B_X\{p/X\}\not\stackrel{\tau}{\longrightarrow}$, by (a.1) and Lemma~\ref{L:ONE_ACTION_VISIBLE_GUARDED},  $B_X\{\langle X|X=t_X \rangle/X\}\not\stackrel{\tau}{\longrightarrow}$.
Then it follows from (a.3) and $\langle X|X=t_X \rangle \stackrel{\tau}{\longrightarrow}$ that $ \langle X|X=t_X \rangle \stackrel{\epsilon}{\Longrightarrow}|B_X\{\langle X|X=t_X \rangle/X\}$.
Hence  $\mathcal T$ has a proper subtree with root $B_X\{\langle X|X=t_X \rangle/X\}F$;
moreover, $B_X\{\langle X|X=t_X \rangle/X\} \in \Omega$ because of $p' \notin F$, (a.1) and (a.2).\\

\noindent Case 2 $C_Y \equiv \langle Z |E\rangle$.

Here, $C_Y\{\langle X|X=t_X \rangle/Y\} \equiv  \langle Z|E \{\langle X|X=t_X\rangle /Y\} \rangle $.
      Then the last rule applied in $\mathcal T$ is
\[\text{either}\;\frac{\langle t_Z|E\rangle \{\langle X|X=t_X \rangle/Y\} F}{\langle Z|E \rangle \{\langle X|X=t_X \rangle/Y\} F}(Z=t_Z \in E)\;\text{or}\;\frac{\{rF:\langle Z|E \rangle \{\langle X|X=t_X \rangle/Y\}  \stackrel{\epsilon}{\Longrightarrow}|r\}}{\langle Z|E \rangle \{\langle X|X=t_X \rangle/Y\}  F}.\]

For the first alternative, by Lemma~\ref{L:F_NORMAL}(8), it follows from 
$\langle Z|E \rangle \{p/Y\} \notin F$ that $\langle t_Z|E \rangle \{p/Y\} \notin F$.
Since $Y$ does not occur in the scope of any conjunction in $\langle Z|E \rangle$, by Lemma~\ref{L:ONE_STEP_UNFOLDING_VARIABLE}(5), neither does it in $\langle t_Z|E \rangle$.
Therefore $\langle t_Z|E\rangle \{\langle X|X=t_X \rangle/Y\} \in \Omega$, as desired.

For the second alternative, since $\langle Z|E \rangle \{p/Y\} \notin F$ and $p =_{RS} t_X\{p/X\}$, we get $\langle Z|E \rangle \{t_X\{p/X\}/Y\} \notin F$  by Theorem~\ref{T:CONGRUENCE}.
So $\langle Z|E \rangle \{t_X\{p/X\}/Y\} \stackrel{\epsilon}{\Longrightarrow}_F|p'$ for some $p'$.
Then, for this transition, by Lemma~\ref{L:MULTI_TAU_GF_STABLE}, there exist processes $q_W(W \in \widetilde{W})$ and a context $D_{Y,\widetilde{W}}$ with $Y \notin \widetilde{W}$ such that
\begin{enumerate}[({b.}1)]
  \renewcommand{\theenumi}{(b.\arabic{enumi})}
  \item \ $t_X\{p/X\} \stackrel{\tau}{\Longrightarrow}|q_W\;\text{for each}\;W\in \widetilde{W}\;\text{and}\;p'\equiv D_{Y,\widetilde{W}}\{t_X\{p/X\}/Y,\widetilde{q_W}/\widetilde{W}\}$,

  \item \ $Y$ and each $W(\in \widetilde{W})$ are strongly guarded and do not occur in the scope of any conjunction in $D_{Y,\widetilde{W}}$, and 
  \item \ $\langle Z|E \rangle \{r/Y\} \stackrel{\epsilon}{\Longrightarrow} D_{Y,\widetilde{W}}\{r/Y,\widetilde{r_W}/\widetilde{W}\}$ for any $r$ and $r_W(W \in \widetilde{W})$ such that $r \stackrel{\tau}{\Longrightarrow}r_W$ for each $W \in \widetilde{W}$. 
\end{enumerate}
Then, since $X$ is strongly guarded and does not occur in the scope of any conjunction in $t_X$, by Lemma~\ref{L:MULTI_TAU_GF_STABLE} and \ref{L:ONE_ACTION_VISIBLE_GUARDED}, for each transition $t_X\{p/X\} \stackrel{\tau}{\Longrightarrow}|q_W$,
there exists a stable context $t_X^W$ such that
\begin{enumerate}[({c.}1)]
   \renewcommand{\theenumi}{(c.\arabic{enumi})}
  \item \ $X$ is strongly guarded and does not occur in the scope of any conjunction in $t_X^W$,
  \item \ $q_W \equiv t_X^W\{p/X\}$, and
  \item \ $ t_X\{\langle X|X=t_X \rangle/X\} \stackrel{\tau}{\Longrightarrow}|t_X^W\{\langle X|X=t_X \rangle/X\}$.
\end{enumerate}
%
For the simplicity of notation, we let $Q_W$ stand for $t_X^W\{\langle X|X=t_X \rangle/X\}$ for each $W \in \widetilde{W}$.
So, by (c.3), $\langle X|X=t_X \rangle\stackrel{\tau}{\Longrightarrow}|Q_W$ for each $W \in \widetilde{W}$.
Hence it follows from (b.3) that
\[
\langle Z|E \rangle \{ \langle X|X=t_X\rangle /Y\} \stackrel{\epsilon}{\Longrightarrow}  D_{Y,\widetilde{W}}\{ \langle X|X=t_X\rangle /Y,\widetilde{Q_W}/\widetilde{W}\}. \tag{\ref{L:UNIQUE_SOLUTION_EXISTENCE}.1}
\]
By (b.2) and (c.1), it is not difficult to see that $X$ is strongly guarded and does not occur in the scope of any conjunction in $D_{Y,\widetilde{W}}\{t_X/Y,\widetilde{t^W_X}/\widetilde{W}\}$.
So, by Lemma~\ref{L:ONE_ACTION_VISIBLE_GUARDED} and $p' \equiv D_{Y,\widetilde{W}}\{t_X/Y,\widetilde{t^W_X}/\widetilde{W}\}\{p/X\}\not\stackrel{\tau}{\longrightarrow}$, we get \[
D_{Y,\widetilde{W}}\{t_X/Y,\widetilde{t^W_X}/\widetilde{W}\}\{\langle X|X=t_X\rangle/X\}  \not\stackrel{\tau}{\longrightarrow}.
\]
Hence $D_{Y,\widetilde{W}}\{ \langle X|X=t_X\rangle /Y,\widetilde{Q_W}/\widetilde{W}\}\not\stackrel{\tau}{\longrightarrow}$  by Lemma~\ref{L:ONE_ACTION_TAU} and  ${\mathcal I}(\langle X|X=t_X\rangle )={\mathcal I}(t_X\{\langle X|X=t_X\rangle /X\})$.
Then $\mathcal T$ has a proper subtree with root $D_{Y,\widetilde{W}}\{ \langle X|X=t_X\rangle /Y,\widetilde{Q_W}/\widetilde{W}\}F$ due to (\ref{L:UNIQUE_SOLUTION_EXISTENCE}.1).
Moreover, by Theorem~\ref{T:CONGRUENCE} and $p=_{RS}t_X\{p/X\}$, it follows from $p'\equiv D_{Y,\widetilde{W}}\{t_X\{p/X\}/Y,\widetilde{t_X^W\{p/X\}}/\widetilde{W}\} \notin F$ that $D_{Y,\widetilde{W}}\{p/Y,\widetilde{t_X^W\{p/X\}}/\widetilde{W}\} \notin F$.
Set
\[D'_Y \triangleq D_{Y,\widetilde{W}}\{\widetilde{t_X^W\{Y/X\}}/\widetilde{W}\}.\]
Therefore, $D_{Y,\widetilde{W}}\{ \langle X|X=t_X\rangle /Y,\widetilde{Q_W}/\widetilde{W}\} \equiv D'_Y\{\langle X|X=t_X\rangle /Y\}   \in \Omega$, as desired.
\end{proof}

We now have the assertion below which states that given an equation $X =_{RS}t_X$ satisfying some conditions, $\langle X|X=t_X\rangle$ is the unique consistent solution whenever consistent solutions exist.

\begin{theorem}[Unique solution]\label{T:UNIQUE_SOLUTION}
  For any $p,q \notin F $ and $t_X$ where $X$ is strongly guarded and does not occur in the scope of any conjunction, if $p =_{RS} t_X\{p/X\}$ and $q =_{RS} t_X\{q/X\}$ then $p =_{RS} q$.
Moreover, $\langle X | X=t_X \rangle$ is the unique consistent solution modulo $=_{RS}$ for the equation $X =_{RS} t_X$ whenever  consistent solutions exist.
\end{theorem}
\begin{proof}
    If $p \bowtie q$ then $p =_{RS} q$ follows from Lemma~\ref{L:UNIQUE_SOLUTION}, otherwise, w.l.o.g, we assume that $p$ is stable and $q$ is not.
    By Theorem~\ref{T:CONGRUENCE}, $\tau.p =_{RS} p =_{RS} t_X\{p/X\} =_{RS} t_X\{\tau.p/X\}$.
    Then, by Lemma~\ref{L:UNIQUE_SOLUTION}, it follows from $\tau.p,q \notin F$, $\tau.p \bowtie q$, $\tau.p =_{RS} t_X\{\tau.p/X\}$ and $q =_{RS} t_X\{q/X\}$ that $\tau.p =_{RS} q$ . Hence $p =_{RS} q$.

Suppose that $X=_{RS}t_X$ has consistent solutions.
It is obvious that $\langle X|X=t_X \rangle =_{RS} t_X\{\langle X|X=t_X \rangle /X\}$ due to $\langle X|X=t_X \rangle \Rrightarrow_1 \langle t_X|X=t_X \rangle  \equiv t_X\{\langle X|X=t_X \rangle /X\}$ and Lemma~\ref{L:MULTI_STEP_UNFOLDING_IMPLIES_RS}.
Further, by Lemma~\ref{L:UNIQUE_SOLUTION_EXISTENCE}, $\langle X | X=t_X \rangle$ is the unique consistent solution of the equation $X =_{RS} t_X$.
\end{proof}

As an immediate consequence, we have

\begin{corollary}
  For any term $t_X$ where $X$ is strongly guarded and does not occur in the scope of any conjunction, then the equation $X =_{RS}t_X$ has consistent solutions iff $\langle X|X=t_X \rangle \notin F$.
\end{corollary}
\begin{proof}
  Immediately by Theorem~\ref{T:UNIQUE_SOLUTION}.
\end{proof}

We  conclude this section with providing a brief discussion.
For Theorem~\ref{T:UNIQUE_SOLUTION}, the condition that $X$ is strongly guarded can not be relaxed to that $X$ is weakly guarded.
For instance, consider the equation $X =_{RS} \tau.X$, it has infinitely many consistent solutions.
In fact, for any $p$, it always holds that $p =_{RS} \tau.p$.
Moreover, the condition that $p,q\notin F$ is also necessary.
For example, both $\langle X| X = a.X \rangle$ and $\bot$ are  solutions of the equation $X =_{RS} a.X$, but they are not equivalent modulo $=_{RS}$.


\section{Conclusions and future work}

This paper considers recursive operations over LLTSs in pure process-algebraic style and
a process calculus $\text{CLL}_R$, which is obtained from CLL by adding recursive operations, is proposed.
We show that the behavioral relation $\sqsubseteq_{RS}$ is precongruent w.r.t all operations in $\text{CLL}_R$, which reveals that this calculus supports compositional reasoning.
Moreover, we also provide a theorem on the uniqueness of consistent solution of a given equation $X =_{RS} t_X$ where $X$ is required to be strongly guarded and does not occur in the scope of any conjunction in $t_X$.

Although CLL contains logic operators $\wedge$ and $\vee$ over processes, due to lack of modal operators, it does not afford describing abstract properties of concurrent systems.
As we know, some modal operators could be characterized by equations and fixpoints \cite{Bradfield01}.
Fortunately, under the mild condition that the set of actions $Act$ is finite, we can integrate standard temporal operators \emph{always} and \emph{unless} into $\text{CLL}_R$ \cite{Zhu13} but this requires us to strengthen Theorem~\ref{T:UNIQUE_SOLUTION} by removing the restriction that \textquotedblleft recursive variables do not occur in the scope of any conjunction in recursive specifications\textquotedblright.
We leave the strengthened Theorem~\ref{T:UNIQUE_SOLUTION} as a open problem.

In this paper, we adopt a proof method \emph{well-ordered proof tree contradiction} to obtain properties of $F$-predicate. It reflects a kind of principle \emph{negation as failure} and it is different from proof method \emph{witnesses} adopted by L\"{u}ttgen and Vogler.
Their method requires one to find proofs (i.e., witness set) to illustrate the existence of properties.
However, the way of constructing witnesses is similar to our method.

Future work could proceed along two directions.
Firstly, we will add hiding operator to $\text{CLL}_R$. As an important feature of LLTS \cite{Luttgen10}, hiding in the presence of recursion may lead to divergence and introduce inconsistency by (LTS2).
For example, $\langle X | X= a.X \rangle \notin F$ but $\langle X |X= a.X \rangle \backslash a \in F$.
The other direction of future work is to find a (ground) complete proof system for regular processes in $\text{CLL}_R$ along lines adopted in \cite{Milner89,Baeten08}.
Here a process is regular if its LTS has only finitely many states and transitions.


\begin{thebibliography}

\bibitem[Aceto 2001]{Aceto01}
  Aceto, L., Fokkink, W.J. and Verhoef, C. (2001) Structural operational semantics.
  In Bergstra~J.A., Ponse~A. and Smolka~S.A. (editors), \emph{Handbook of Process Algebra}, Chapter \textbf{3}, 197-292. Elsevier Science.

\bibitem[Andersen \emph{et al.} 1994]{Andersen94}
   Andersen, H.R.,  Stirling, C. and  Winskel, G. (1994)
   A compositional proof system for the modal $\mu $-calculus.
    \emph{Proceeding of the 9th Annual IEEE Symposium on Logic in Computer Science}, 144-153. IEEE Computer Society Press.


 \bibitem[Baeten and Bravetti 2008]{Baeten08}
    Baeten, J.C.M. and Bravetti, M. (2008) A ground-complete axiomatisation of finite-state processes in a generic process algebra. \emph{Math. Sturuct. in Comp. Science} \textbf{18}, 1057-1089.


  \bibitem[Bergstra \emph{et al.} 2001]{Bergstra01}
      Bergstra, J.A.,  Fokkink, W. and Ponse, A. (2001) Process algebra with recursive operations. In Bergstra, J.A., Ponse, A. and Smolka, S.A. (editors),
    \emph{Handbook of Process Algebra}, Chapter \textbf{5},  333-389. Elsevier Science.


   \bibitem[Bloom 1994]{Bloom94}
    Bloom, B. (1994) Ready simulation, bisimulation, and the semantics of the CCS-like languages.
    Ph.D dissertation, Massachusetts Institute of Technology, Cambridge, Mass., Aug.


 \bibitem[Bol and Groote 1996]{Bol96}
    Bol, R. and Groote, J.F. (1996) The meaning of negative premises in transition system specifications. \emph{Journal of the ACM} \textbf{43}, 863-914.

\bibitem[Bradfield and Stirling 2001]{Bradfield01}
  Bradfield, J.C. and Stirling, C. (2001) Modal logics and mu-calculi: an introduction.
  In Bergstra~J.A., Ponse~A. and Smolka~S.A. (editors), \emph{Handbook of Process Algebra}, Chapter \textbf{4}, 293-330. Elsevier Science.




\bibitem[Cleaveland and L\"{u}ttgen 2000]{Cleaveland00}
    Cleaveland, R. and L\"{u}ttgen, G. (2000) A semantic theory for heterogeneous system design. In \emph{FSTTCS 2000, LNCS} \textbf{1974}, 312-324. Springer-Verlag.

\bibitem[Cleaveland and L\"{u}ttgen 2002]{Cleaveland02}
    Cleaveland, R. and L\"{u}ttgen, G. (2002) A logical process calculus.  In  \emph{{EXPRESS 2002},   {ENTCS}} {\textbf{68} 2}. {Elsevier Science}.


%


\bibitem[Gelfond and Lifchitz 1988]{Gelfond88}
    Gelfond, M. and  Lifchitz, V. (1988) The stable model semantics for logic programming.
   In {Kowalski, R.} and {Bowen, K.} (editors),
   {\emph{Proceedings of the 5th International Conference on Logic Programming}}, 1070-1080.
   {MIT Press}.

%



%

 \bibitem[Groote 1992]{Groote92}
   {Groote, J.F.} and
   {Vaandrager, F.} (1992)
   {Structured operational semantics and bisimulation as a congruence}.
   {\emph{Information and Compuation}}
   {\textbf{100}},
   {202-260}.


 \bibitem[Groote 1993]{Groote93}
   {Groote, J.F.} (1993)
   {Transition system specifications with negative premises}.
   {\emph{Theoretical Computer Science}}
   {\textbf{118}}, {263-299}.

%

 \bibitem[L\"{u}ttgen and Vogler 2007]{Luttgen07}
   {L\"{u}ttgen, G.} and
   {Vogler, W.} (2007)
   {Conjunction on processes: full-abstraction via ready-tree semantics}.
   {\emph{Theoretical Computer Science}}
   {\textbf{373}}
  ({1-2}),
   {19-40}.

 \bibitem[L\"{u}ttgen and Vogler 2010]{Luttgen10}
   {L\"{u}ttgen, G.} and
   {Vogler, W.} (2010)
   {Ready simulation for concurrency: it's logical}.
   {\emph{Information and computation}}
   {\textbf{208}},
   {845-867}.


 \bibitem[L\"{u}ttgen and Vogler 2011]{Luttgen11}
   {L\"{u}ttgen, G.} and
   {Vogler, W.} (2011)
   {Safe reasoning with Logic LTS}.
   {\emph{Theoretical Computer Science}}
   {\textbf{412}},
   {3337-3357}.

   \bibitem[Milner 1983]{Milner83}
   Milner, R. (1983) Calculi for synchrony and asynchrony. \emph{Theoretical Computer Science} \textbf{25} (3), 267-310.

   \bibitem[Milner 1989]{Milner89}
   {Milner, R.} (1989)
   {A complete axiomatization for observational congruence of finite-state behaviours}.
   {\emph{Information and Computation}}
   {\textbf{81}},
   {227-247}.
%

 \bibitem[Nicola and Hennessy 1983]{Nicola83}
   {De Nicola, R.} and  {Hennessy, M.} (1983)
   {Testing equivalences for processes}.
   {\emph{Theoretical Computer Science}}
   {\textbf{34}},
   {83-133}.

%

\bibitem[Plotkin 1981]{Plotkin81}
   {Plotkin, G.} (1981)
   {A structural approach to operational semantics}.
   {Report DAIMI FN-19},
   {Computer Science Department, Aarhus University}.
  Also in
   {\emph{Journal of Logic and Algebraic Programming}}
   {\textbf{60}}
  ({2004}),
   {17-139}.

%
%

 \bibitem[Verhoef 95]{Verhoef95}
   {Verhoef, C.} (1995)
   {A congruence theorem for structured operational semantics with predicates and negative premisess}.
   {\emph{Nordic Journal of Computing}}
   {\textbf{2} (2)},
   {274-302}.

 \bibitem[Zhang \emph{et al.} 2011]{Zhang11}
   {Zhang, Y., Zhu, Z.H., Zhang, J.J. and Zhou, Y.} (2011)
   {A process algebra with logical operators}.
   {arXiv:1212.2257}.

 \bibitem[Zhu \emph{et al.} 2013]{Zhu13}
   {Zhu, Z.H.},
   {Zhang, Y.} and
   {Zhang, J.J.} (2013)
   {Merging process algebra and action-based computation tree logic}.
   {Submitted to \emph{Theoretical Computer Science}, arXiv:1212.6813}.
\end{thebibliography}
\end{document}